\documentclass[a4paper,twocolumn,superscriptaddress,nofootinbib,10pt,floatfix]{revtex4-2}

\usepackage{amsmath}
\usepackage{amssymb}
\usepackage{amsthm}
\usepackage{graphicx}
\usepackage[dvipsnames]{xcolor}
\usepackage{csquotes}
\usepackage{mathtools}
\usepackage{booktabs}
\usepackage{enumitem}
\usepackage{tikz-cd} 
\usepackage{adjustbox}
\usepackage{verbatim}
\usepackage{listings}
\usepackage{cancel}

\usepackage[bookmarks=true,colorlinks=true,linkcolor=orange,citecolor=orange,urlcolor=orange,bookmarks]{hyperref}

\makeatletter
\def\maketitle{
\@author@finish
\title@column\titleblock@produce
\suppressfloats[t]}
\makeatother

\newcommand{\diff}{\ensuremath\mathrm{d}}

\newcommand{\Diff}[1]{\ensuremath\mathrm{d}#1~}
\newcommand{\norm}[1]{\lVert#1\rVert}
\newcommand{\Norm}[1]{\left\lVert#1\right\rVert}
\newcommand{\supp}{\operatorname{supp}}

\newcommand{\Id}{\ensuremath{\mathbb{I}}}

\newcommand*{\pipe}{\ensuremath{\, | \, }}
\newcommand*{\fatpipe}{\ensuremath{\, \| \, }}
\newcommand*{\Tr}{\operatorname{Tr}}
\newcommand*{\argmax}{\operatornamewithlimits{argmax}}
\newcommand*{\argmin}{\operatornamewithlimits{argmin}}

\newcommand{\bra}[1]{\langle #1|}
\newcommand{\ket}[1]{|#1\rangle}
\newcommand{\ketbra}[2]{\ket{#1}\!\bra{#2}}

\newcommand{\braket}[2]{\langle #1|#2\rangle}

\providecommand{\myvec}[1]{\ensuremath{\boldsymbol{#1}}}

\providecommand{\ll}{\ensuremath{\myvec{l}}}

\providecommand{\ppsi}{\ensuremath{\myvec{\psi}}}

\providecommand{\hatf}{\ensuremath{\hat{f}}}

\providecommand{\hatw}{\ensuremath{\hat{w}}}

\providecommand{\calA}{\ensuremath{\mathcal{A}}}
\providecommand{\calB}{\ensuremath{\mathcal{B}}}

\providecommand{\calE}{\ensuremath{\mathcal{E}}}
\providecommand{\calF}{\ensuremath{\mathcal{F}}}

\providecommand{\calI}{\ensuremath{\mathcal{I}}}

\providecommand{\calM}{\ensuremath{\mathcal{M}}}
\providecommand{\calN}{\ensuremath{\mathcal{N}}}
\providecommand{\calO}{\ensuremath{\mathcal{O}}}

\providecommand{\calQ}{\ensuremath{\mathcal{Q}}}

\providecommand{\calS}{\ensuremath{\mathcal{S}}}
\providecommand{\calT}{\ensuremath{\mathcal{T}}}
\providecommand{\calU}{\ensuremath{\mathcal{U}}}

\providecommand{\calX}{\ensuremath{\mathcal{X}}}

\providecommand{\bbE}{\ensuremath{\mathbb{E}}}

\providecommand{\bbI}{\ensuremath{\mathbb{I}}}

\providecommand{\bbN}{\ensuremath{\mathbb{N}}}

\providecommand{\bbP}{\ensuremath{\mathbb{P}}}

\providecommand{\bbR}{\ensuremath{\mathbb{R}}}

\providecommand{\sfS}{\ensuremath{\mathsf{S}}}

\newcommand{\ie}{\textit{i.e.}\ }

\newcommand\notesmaller[1][\notesmallerfrac]{%
  \fontsize{#1\dimexpr\f@size pt\relax}{#1\dimexpr\f@baselineskip pt\relax}%
  \selectfont\ignorespaces%
}
\def\notesmallerfrac{0.9}

\usepackage[nameinlink,capitalize]{cleveref}
\DeclarePairedDelimiterX\abs[1]{\lvert}{\rvert}{{#1}}

\AtBeginDocument{
\heavyrulewidth=.08em
\lightrulewidth=.05em
\cmidrulewidth=.03em
\belowrulesep=.65ex
\belowbottomsep=0pt
\aboverulesep=.4ex
\abovetopsep=0pt
\cmidrulesep=\doublerulesep
\cmidrulekern=.5em
\defaultaddspace=.5em
}

\newtheorem{theorem}{Theorem}
\newtheorem{lemma}[theorem]{Lemma}
\newtheorem{proposition}[theorem]{Proposition}
\newtheorem{definition}[theorem]{Definition}
\newtheorem{corollary}[theorem]{Corollary}
\newtheorem{remark}[theorem]{Remark}
\newtheorem{observation}[theorem]{Observation}
\newtheorem{conjecture}{Conjecture}
\newtheorem{openproblem}[conjecture]{Open Problem}

\newtheorem{stheorem}{Theorem}
\newtheorem{slemma}[stheorem]{Lemma}
\newtheorem{sproposition}[stheorem]{Proposition}

\newtheorem{scorollary}[stheorem]{Corollary}

\newtheoremstyle{red}{}{}{\color{red}}{}{\color{red}\bfseries}{.}{ }{}
\theoremstyle{red}

\let\oldcite\cite
\renewcommand{\cite}[1]{\mbox{\oldcite{#1}}}

\newcommand{\EV}{\operatornamewithlimits{\bbE}}

\newcommand{\dccqs}{Dahlem Center for Complex Quantum Systems, Freie Universit{\"a}t Berlin, 14195 Berlin, Germany}
\newcommand{\hzb}{Helmholtz-Zentrum Berlin f{\"u}r Materialien und Energie, 14109 Berlin, Germany}
\newcommand{\hhi}{Fraunhofer Heinrich Hertz Institute, 10587 Berlin, Germany}
\newcommand{\ucph}{Department of Mathematical Sciences, University of Copenhagen, 2100 K{\o}benhavn, Denmark}
\newcommand{\ens}{Ecole Normale Superieure de Lyon, 69342 Lyon Cedex 07, France}

\newcommand{\papertitle}{Quantum metrology in the finite-sample regime}

\begin{document}	

\title{\papertitle}
\date{\today}
	
\author{Johannes~Jakob~Meyer}
\affiliation{\dccqs}
    
\author{Sumeet~Khatri}
\affiliation{\dccqs}
    
\author{Daniel~Stilck~Fran\c{c}a}
\affiliation{\dccqs}
\affiliation{\ucph}
\affiliation{\ens}

\author{Jens~Eisert}
\affiliation{\dccqs}
\affiliation{\hzb}
\affiliation{\hhi}

\author{Philippe~Faist}
\affiliation{\dccqs}

\let\oldaddcontentsline\addcontentsline 
\renewcommand{\addcontentsline}[3]{\oldaddcontentsline{#1}{lot}{#3}} 

\begin{abstract}
    In quantum metrology, one of the major applications of quantum technologies, the ultimate precision of estimating an unknown parameter is often stated in terms of the Cramér-Rao bound. 
    Yet, the latter is no longer guaranteed to carry an operational meaning in the regime where few measurement samples are obtained, which we illustrate through a simple example.
    We instead propose to quantify the quality of a metrology protocol by the probability of obtaining an estimate with a given accuracy.
    This approach, which we refer to as probably approximately correct (PAC) metrology, ensures operational significance in the finite-sample regime. The accuracy guarantees hold for any value of the unknown parameter, unlike the Cramér-Rao bound which assumes it is approximately known.
    We establish a strong connection to multi-hypothesis testing with quantum states, which allows us to 
    derive an analogue of the Cramér-Rao bound which contains explicit corrections relevant to the finite-sample regime.    
    We further study the asymptotic behavior of the success probability of the estimation procedure for many copies of the state and apply our framework to the example task of phase estimation with an ensemble of spin-1/2 particles.    
    Overall, our operational approach allows the study of quantum metrology in the finite-sample regime and opens up a plethora of new avenues for research at the interface of quantum information theory and quantum metrology.
\end{abstract}

\maketitle

Metrology, the scientific study of measurements, has naturally evolved to encompass the realm of quantum theory. Quantum metrology seeks to realize practical advantages by harnessing quantum effects. The growing quantum technologies sector, especially, holds high expectations for achieving unparalleled sensitivity with quantum sensors.
Anticipated applications range from the calibration of atomic clocks over gravitational-wave detection to potential medical uses~\cite{giovannetti_advances_2011,Paris,RevModPhys.89.035002,RevModPhys.90.035005,BraunsteinOld}. 
It is crucial that the theory of quantum metrology accommodates the emerging technological capabilities of near-term quantum sensors, which necessitates an in-depth understanding of their performance in realistic settings, where the size of experiments might be limited.

A standard question in quantum metrology is to determine the value of an unknown parameter that has been encoded in a quantum state. For instance, suppose we wish to estimate the difference of time $t$ between two events.  
One might prepare an initial clock state $\rho_0$, \emph{e.g.}, an ensemble of spin-$\tfrac{1}{2}$ particles, in some standard state when the first event occurs, let the system evolve under its natural dynamics -- say, a magnetic field of fixed strength -- resulting in a state $\rho(t)$,
and perform a measurement on the system when the second event occurs. 
The accuracy to which $t$ is determined can be improved through suitable
choices of the initial state, the dynamics, and the final measurement.
A similar scheme can be employed to sense the value of an unknown parameter
in a Hamiltonian, such as the strength of an external field. In this case,
one lets the system evolve under the unknown Hamiltonian for a fixed amount
of time. In either case, the problem reduces to estimating the value of a
parameter $t$ among a parametrized set of states $t \mapsto \rho(t)$.

\begin{figure}
    \centering
    \includegraphics{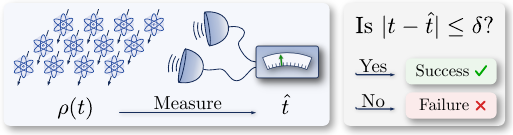}
    \caption{The setup we consider for quantum metrology in the finite-sample regime
    consists in applying a measurement onto the quantum state $\rho(t)$,
    and inferring an estimate $\hat{t}$ of the value of the unknown parameter $t$
    from the outcome of the measurement. The estimation process is successful if the
    estimate $\hat{t}$ and the true parameter value $t$ differ by at most some fixed
    error tolerance $\delta$. Our approach quantifies the probability that the
    estimation procedure is successful.
In contrast, the standard quantum Cram\'er-Rao bound quantifies the variance
of the outcome $\hat{t}$ for a given estimation procedure that reveals
the correct parameter in expectation;
its operational meaning is guaranteed only after collecting many outcomes.
}
    \label{fig:intro-super-simple-setup}
\end{figure}

A standard treatment of this problem proceeds as
follows~\cite{BraunsteinOld,giovannetti_advances_2011,Paris,meyer2021fisher}.
One assumes that $t$ is already known to be close to some value $t_0$.
The task is to refine one's knowledge of $t$ by accessing the expectation value of some observable.
A central result in quantum metrology quantifies the variance $\sigma^2$ of a quantum measurement whose expectation value is equal to $t$. 
The \emph{quantum Cram\'er-Rao bound} states that~\cite{cramer_book,rao1945,helstrom_minimum_1967,BraunsteinOld}
\begin{align}\label{eqn:qcrb}
    \sigma^2 \geq \frac{1}{\calF},
\end{align}
where $\calF$ is the \emph{quantum Fisher information}, a quantity that measures how distinguishable the states of the family $\rho(t)$ are around $t_0$. 
Furthermore, there exists a quantum measurement which achieves equality in Eq.~\eqref{eqn:qcrb}.
One thus frequently resorts to the quantum Fisher information
as a measure of sensitivity, including when quantifying
the advantages of using entangled states for quantum sensing~\cite{giovannetti_advances_2011}, the effect of noise on the sensitivity of probe states~\cite{DemkowiczDobrzanski2012NatComm_elusive,faist2022time-energy}, as well as the advantages of using quantum error correction in metrology~\cite{zhou_achieving_2018}.
The quantum Fisher information naturally generalizes the classical
Fisher information and enjoys the geometrical interpretation of being the metric tensor
associated with the fidelity of quantum states~\cite{BraunsteinOld,liu_quantum_2020,meyer2021fisher}. 

In this work, %
we consider the regime where few measurement samples are available.
This regime is increasingly expected to be relevant when considering
the limited capabilities of quantum sensors in the near term.
Specifically, we revisit some of the founding assumptions that lead to the
quantum Cram\'er-Rao bound which are difficult to justify in the few-sample regime.
First, the expectation value of an observable can only be
reliably estimated if sufficiently many samples are available. Thus, access to
few samples of the measurement that achieves equality in the quantum Cram\'er-Rao
bound might not provide meaningful information about the unknown parameter.
Second, few samples from a quantum measurement are unlikely to yield the degree of precision that is compatible with the assumption that the parameter is already approximately known.
Relaxing this second assumption furthermore enables us to consider general
families of states $\rho(t)$ without the smoothness properties required to apply the
quantum Cram\'er-Rao bound.

We establish a general finite-sample analysis of quantum metrology rooted
in fundamental principles of quantum information theory.
We consider a general one-parameter family of states $t \mapsto \rho(t)$, where $t$
is known to belong to some real interval $I \subseteq \mathbb{R}$ (Fig.~\ref{fig:intro-super-simple-setup}).
We then consider a quantum measurement whose outcome leads to an estimate
\smash{$\hat{t}$} of the value $t$.
In our model, the measurement is applied only once. Access to a finite number
$n$ of samples is modeled by explicitly considering the copies of the state in
parallel, $\rho^{\otimes n}(t)$.
The estimation procedure is successful if
\smash{$\hat{t}$} is within some fixed \emph{estimation error tolerance} $\delta$ of the true parameter value $t$.
We then ask, 
\emph{\enquote{What is the probability that our estimate of the underlying parameter is within a given estimation error tolerance around the true value?}} and
\emph{\enquote{What is the smallest estimation error tolerance such that this success probability is at least some given threshold?}}
For these questions to have a well-defined answer, we can either assume prior knowledge about the underlying parameter $t$ or we take the worst-case among all possible values of $t$. In our work, we explore both settings in depth.

We show that computing the optimal success probability over all possible measurements for a fixed estimation tolerance belongs to
a class of convex optimization problems known as \emph{semi-infinite programs}, which are 
essentially semi-definite programs with an infinite number of
semi-definite constraints. We explicitly show how the semi-infinite program reduces to a semi-definite program upon discretization.

We also establish close connections between metrology in the finite-sample regime and
\emph{multi-hypothesis testing of quantum 
states}~\cite{audenaert2007discriminating,li2016discriminating,audenaert2014upper,khatri2020principles}.
In quantum multi-hypothesis testing, one receives an unknown state from a fixed set
of quantum states, and seeks to identify which state was provided. %
The metrological task considered here can be intuitively understood as 
a continuous version of multi-hypothesis testing of quantum states, where we seek to identify the value of an unknown parameter $t$ in the family of states $\rho(t)$.
In contrast to the discrete multi-hypothesis task, it is impossible to determine the
value of $t$ exactly given the parameter's continuous nature. Instead, the
parameter $t$ should be determined up to some fixed precision,
quantified by $\delta$.
We make this intuitive connection rigorous by proving upper bounds on the success probability of the metrological task in terms of the success probability of a related multi-hypothesis testing task. 
More specifically, we show that determining the parameter $t$ to precision $\delta$ is at least as hard as distinguishing quantum states corresponding to parameters that are at least $2\delta$ apart.
Along the same vein, we express quantities of interest, such as the success probability of our estimation procedure, in terms of known single-shot entropy measures such as the conditional
min-entropy~\cite{PhDRenner2005,khatri2020principles}. We also connect our framework to known estimation lower bounds in terms of the the hypothesis testing relative entropy~\cite{walter_heisenberg_2014}.

We then connect the finite-sample regime to the many-sample regime as follows.
We study the behavior of the success probability of the estimation where a finite number $n$ of copies of the state are available, and consider the limit $n\to\infty$.
In this regime, we prove an upper bound on the 
rate with which the success probability approaches one
in terms of the Chernoff divergence of quantum states.
This result extends known distinguishability rates in multi-hypothesis
testing~\cite{li2016discriminating}.

Exploiting the connection to multi-hypothesis testing allows us to formulate an analogue of the quantum Cramér-Rao bound that is valid in the single-shot regime.  
The estimation error tolerance $\delta$ replaces the standard deviation $\sigma$ on the left-hand side of the inequality of Eq.~\eqref{eqn:qcrb}, and we obtain correction terms on the right-hand side that depend on the desired success probability and properties of the set of states $\rho(t)$.

We then turn to an alternative setting, in which the parameter to be estimated is
accessed through the use of a parameter-dependent quantum channel. This
setting offers richer estimation strategies than the parameter-dependent state
setting. 
For instance, an estimation strategy may interleave the application of the parametrized channel on a probe system with interactions with a memory system.
We extend a selection of our earlier results to such general strategies, like the formulation of the optimal success probability as a convex problem as well as the rigorous connection to multi-hypothesis testing of quantum channels. Such generalizations become possible by viewing such strategies in their entirety as quantum combs~\cite{CDP09,GW07} or general strategies of indefinite causal order~\cite{chiribella2013noncausal}.

We further consider the task of estimating the parameter $t$ in a family
of pure states $\ket{\psi(t)}$ for which the parameter $t$ corresponds to
time, and whose evolution is governed by a fixed
Hamiltonian. Furthermore, $t$ is to be estimated globally
over the entire period of the Hamiltonian~\cite{holevo2011probabilistic}.
We establish a closed-form expression of the worst-case estimation success probability, exploiting the group-covariant structure of this set of states with respect to time evolution.
We finally consider examples of this setting on an ensemble of spin-$\tfrac{1}{2}$ particles.
We numerically compute the optimal success probability, as well as the optimal estimation error tolerance, for a collection of states.
The GHZ state fails in the global estimation setting considered here, despite
the state being optimal for local estimation.
We compare the sensitivity of a selection of states, including
a standard ensemble of spins prepared in a superposition
of a ground and an excited state (\emph{i.e.}, a spin-coherent state) as well as
a uniform superposition over all distinct energy levels (the Holland-Burnett 
state~\cite{Holland1993PRL_interferometric}).
We also determine the state that achieves optimal success probability, for any $n$ and
for any fixed estimation error tolerance. 

To further motivate our approach,
consider the following example~\cite{Safranek2017PRA_discontinuities,Zhou2019arXiv_exact,faist2022time-energy} (see Section~\ref{ssec:cramer_rao_finite_sample} of the supplementary material for details).
Alice prepares a particle in the state 
$\rho_0=\ketbra{+}{+}$, where $\ket{\pm} = (\ket{0} \pm \ket{1})/\sqrt{2}$.
The particle evolves according to the Hamiltonian $H = (\omega/2)\, Z$ for some fixed $\omega$, causing it to rotate in the $X$-$Y$-plane of the Bloch sphere.
At time $t$, Alice sends the particle instantaneously over to Bob through a completely dephasing channel %
acting in the Pauli-$X$ basis, defined as
\smash{$\rho_0\mapsto \bra{+}\rho_0\ket{+}\,\ketbra{+}{+}
+ \bra{-}\rho_0\ket{-}\,\ketbra{-}{-}$}. As a consequence,
Bob thus receives the state \smash{$\rho(t)=\cos^2(\omega t/2)\ketbra{+}{+}+\sin^2(\omega t/2)\ketbra{-}{-}$}.
The quantum Fisher information that Bob has with respect to $t$ is~\cite{Safranek2017PRA_discontinuities,Zhou2019arXiv_exact}
\begin{align}
    \calF = \begin{cases}
        \omega^2 & \textup{if $t \notin (\pi/\omega)\mathbb{Z}$},\\
        0 & \textup{if $t \in (\pi/\omega)\mathbb{Z}$}.
    \end{cases}
\end{align}
That is, $\calF$ is constant equal to $\omega^2$
except in a discrete set of points where $\calF = 0$.
While the discontinuity at $t \in (\pi/\omega)\mathbb{Z}$ is concerning given
the operational nature of the quantum Fisher information, it can be
attributed to the vanishing first-order expansion of $\rho(t)$ at those points and
therefore to a failure of the first-order approximation of the
curve $\rho(t)$~\cite{Safranek2017PRA_discontinuities,Zhou2019arXiv_exact}.
Consider now a point $t\approx 0$ with $t>0$ arbitrarily small. The quantum
Cram\'er-Rao bound guarantees the existence of a measurement \smash{$\hat{T}$}
with expectation value \smash{$\langle \hat{T} \rangle = t$} and with variance
\smash{$\langle \Delta \hat{T} \rangle ^2 = 1/\omega^2$}. %
This operator is
\smash{$\hat{T} = t\Id  + \omega^{-1}\bigl( -\tan(\omega t/2)\ketbra++ + \cot(\omega t/2) \ketbra-- \bigr)$}
(compare Section~\ref{ssec:cramer_rao_finite_sample} of the supplementary material).
The eigenvalue of \smash{$\hat{T}$} associated with $\ket-$ diverges as $\sim 1/t$.
In fact, both eigenvalues
contribute significantly to the expectation value and variance of \smash{$\hat{T}$};
the effect of the divergent  eigenvalue associated with $\ket-$  is kept finite only thanks to the
corresponding outcome happening with vanishingly small probability $\sim t^2$.
That is, a measurement of \smash{$\hat{T}$} almost certainly yields the  outcome $\ket+$; 
the  outcome $\ket-$, necessary for an accurate estimation of the expectation value,
only occurs after an expected $\sim 1/t^2$ number of samples.
Therefore, a naive estimation of the expectation value of this observable
yields little useful information on $t$ if fewer
than $\sim 1/t^2$ samples are collected.
One of the main goals of this work is to develop a rigorous and precise analysis of the
accuracy limits of sensing a parameter in the regime where the number of samples is
insufficient to accurately estimate the expectation value of the sensing observable
given through the Cram\'er-Rao bound.

Our inherently operational, information-theoretic approach to the estimation
task guarantees an operational meaning to the estimation error achieved by a given
measurement, in contrast to the variance of an observable whose operational meaning is
ensured only in the asymptotic limit of many samples.  
Our approach is strongly inspired
by recent advancements in \emph{single-shot quantum information
theory}~\cite{PhDRenner2005,BeyondIID,Tomamichel_book,khatri2020principles},
whose aim is to quantify the resource requirements of information-theoretic tasks beyond
the traditional regime %
where many \emph{independent and identically distributed (i.i.d.)} copies of a quantum state are available.
The approach of quantifying the performance of a quantum metrology protocol through the probability of obtaining a sufficiently accurate estimate is also similar in spirit to the \emph{de-facto} standard approach to computational learning theory, namely \emph{probably approximately correct (PAC) learning}~\cite{valiant_theory_1984}. As such, we will also refer to our framework as \emph{probably approximately correct (PAC) metrology}.
Our approach can also be understood as constructing sets known as
\emph{confidence intervals} in the field of statistics, and characterizing the effect
of different choices of quantum measurements on their size. Our approach is thus
closely related to confidence region estimation of quantum states~\cite{BlumeKohout2012arXiv_Tomo,Christandl2012_Tomo,walter_heisenberg_2014,Faist2016PRL_practical,Wang2019PRL_polytopes}.

Our framework enables the study of estimation procedures that can interpolate between
\emph{local estimation}, as in the context of the quantum Cram\'er-Rao bound, and
\emph{global estimation}, where the possible values of the underlying parameters are not constrained to a very small neighborhood of a known value.
Intuitively, the local setting corresponds to the case where the possible range of values for the parameter in question is small compared to the right hand side of the quantum Cramér-Rao bound Eq.~\eqref{eqn:qcrb}.
The global setting requires states to remain distinguishable over the full range of values that the unknown parameter might take.
Probe states that are accurate in the local estimation setting are not necessarily accurate
for global estimation. For instance, the $n$-qubit GHZ state
$\ket{\textup{GHZ}} = (\ket{00\ldots0} + \ket{11\ldots1})/\sqrt{2}$ is optimal for local
estimation of a parameter $t$ of a non-interacting ensemble of spin-$\tfrac{1}{2}$ particles.
However, its very short period $2\pi/n$ prohibits us from distinguishing values of $t$ spaced by more than $2\pi/n$. In contrast, the state $\ket{+}^{\otimes n} =
[(\ket{0}+\ket{1})/\sqrt{2}]^{\otimes n}$ has a period of $2\pi$ and is capable of identifying
greater time intervals at the cost of a worse accuracy in the setting of local estimation.

While the general approach of using the quantum Fisher information can be
extended to the global estimation regime by 
considering Bayesian prior information about the underlying 
parameter~\cite{durkin2007localglobalinterferometry,Paris,liu2016valid,rubio_bayesian_2020,rubio_global_2021,boeyens_uninformed_2021}, this does similarly suffer from possible issues in the few-shot regime we outlined above.
Our alternative approach, on the other hand side, can interpolate between the local and the global setting, both in the presence and absence of prior information about the underlying parameter, and is applicable in the few-shot setting.

Our approach furthermore does not suffer from apparent inconsistencies that 
can arise when the family of quantum states $\rho(t)$ is not sufficiently
well-behaved (\emph{e.g.}, if the derivative vanishes),
in contrast to the singularities and divergences that the 
quantum Fisher information is prone to in such
cases~\cite{Safranek2017PRA_discontinuities,Zhou2019arXiv_exact}.

\begin{figure}
    \centering
    \includegraphics{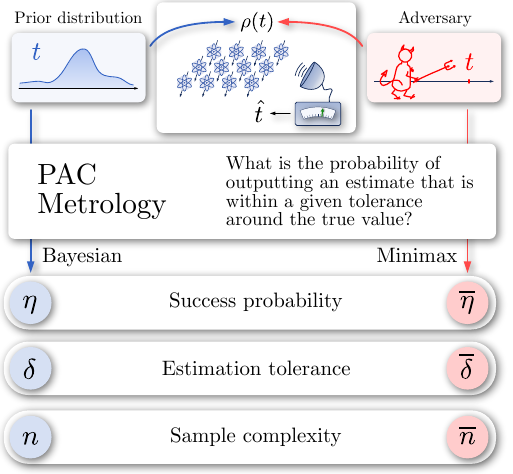}
    \caption{The fundamental quantities forming the  \emph{probably approximately correct (PAC) metrology} framework quantifying non-asymptotic quantum metrology. 
    The success probability, defined intuitively through the statement given in the above figure, forms the cornerstone of the framework. The tolerance is obtained from it by answering the question: \emph{\enquote{What is the smallest tolerance that still guarantees a given success probability?}} Similarly, the sample complexity is the answer to the question: \emph{\enquote{How many repetitions of an experiment do I need to perform to achieve a certain tolerance at a fixed success probability?}}
    We rigorously define them in Definition~\ref{def:success_probability_no_optimization} (success probability), Definition~\ref{def:metrological_tolerance} (tolerance) and Definition~\ref{def:sample_complexity} (sample complexity).
    The relation to the theory of PAC learning~\cite{valiant_theory_1984} lies in the spirit of how the quality of a protocol is quantified through the success probability -- beyond that, there is no overlap between the frameworks.}
    \label{fig:pac_metrology_quantities}
\end{figure}

Our framework %
is summarized in Fig.~\ref{fig:pac_metrology_quantities}.  We identify three key quantities of interest: the success probability at fixed estimation error tolerance, the best achievable tolerance at fixed success probability, and the sample complexity, which quantifies the minimum number of experimental repetitions needed to achieve a desired success probability and tolerance.
We discuss two scenarios of these measures, one assuming information in the form of a prior -- referred to as \emph{Bayesian} -- and one that captures guarantees that can be made agnostic to the underlying parameter -- referred to as \mbox{\emph{minimax}}.
The minimax setting enables a rigorous treatment of the lack of any prior knowledge
on the unknown parameter. %
In particular, attempts to capture this lack of knowledge in the Bayesian
setting, \emph{e.g.}\ by picking a uniform prior, fail to achieve the worst-case statements
that are enabled by the minimax setting. The minimax setting indeed leads to
guarantees that hold even when %
an adversary can %
choose a parameter value that a particular metrology strategy is least likely to work for. 

\begin{figure*}
    \centering
    \includegraphics{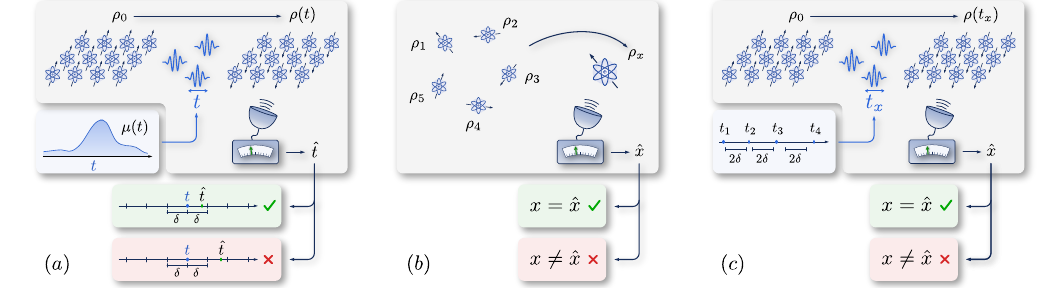}
    \caption{Metrological tasks consist of estimating some physical properties of a quantum system. A basic protocol involves preparing a probe state $\rho_0$, letting it interact with the quantum system of interest, and then measuring the output state. We describe the most general protocols in Section~\ref{sec:parametrized_quantum_channels}. 
    $(a)$ In quantum parameter estimation, the parameter $t\in\mathbb{R}$ indicates the physical property of interest, which may be governed by an underlying random process with probability density function $\mu$. In our PAC metrology framework, we identify the success of the estimation procedure by correctly identifying the parameter $t$ up to a given tolerance $\delta$. 
    $(b)$ In quantum hypothesis testing, a set of quantum states with labels $x \in \calX$ is given and the task is to design a measurement that maximizes the probability of correctly identifying the label $x$.
    $(c)$ The core contribution of our work is a rigorous connection between quantum hypothesis testing and quantum parameter estimation. Specifically, we show that the parameter estimation problem of $(a)$ is at least as hard as distinguishing states that are mutually at least $2\delta$ apart in the hypothesis testing setting of $(b)$ -- see Theorem~\ref{thm:succ_prob_upper_bound_mht_delta_window}.}
    \label{fig:front_figure}
\end{figure*}

Prior works focusing on metrology with finite repetitions~\cite{bahadur1967rates,bahadur1986distribution,bahadur1971limittheorems,spokoiny2012parameterfinitesamples,hayashi2002two,sugiyama2011tomography,sugiyama_thesis} in both the classical and quantum case usually take a \enquote{large deviation} perspective, such that the number of samples is understood to be finite but still large. A notion of success probability already appeared in Ref.~\cite{walter_lower_2014}. The authors of Refs.~\cite{hayashi2002two,Hayashi2} have quantified asymptotic properties related to the achievable precision under statistical assumptions on the estimators, something that can be tightened with our results as we explain in the supplementary material. References~\cite{sugiyama2011tomography,sugiyama2013tomography,sugiyama_thesis} have extended this approach to quantum state tomography. 
The problem of optimizing the metrological tolerance with a given guarantee on the success probability has been treated in Refs.~\cite{walter_heisenberg_2014,walter_lower_2014,sugiyama_precision-guaranteed_2015,yang2018stopwatch}, where some bounds have been given.
In a spirit similar to our work, connections between quantum metrology and hypothesis testing have been used to obtain precision bounds in the standard approach to quantum metrology~\cite{tsang2012ziv-zakai}. The reverse direction of using metrology bounds to quantify channel discrimination has also been explored~\cite{albarelli_probe_2022}.
Our analysis of the phase estimation problem has overlaps with work by Imai and Hayashi~\cite{imai_fourier_2009}, where the asymptotic distribution of phase estimates is analyzed.
Alternative methods for multi-parameter estimation are reviewed
in~\cite{GutaBeyondFisher}.
The authors of Ref.~\cite{Salmon} have studied the case of a fixed measurement with a focus on the admissibility of said measurement. The connection to previous work is explored in greater detail in Section~\ref{ssec:relation_to_prior_art} of the supplementary material.

After summarizing the main contributions of our work in Section~\ref{sec:main_results}, we outline the non-asymptotic framework for quantum metrology in Section~\ref{sec:quantifying_metrological_performance} and discuss optimizing over metrology protocols in Section~\ref{sec:optimal_metrology_protocols}. We describe the optimal post-processing in the practically relevant case of a fixed quantum measurement in Section~\ref{sec:opt_fixed_measurement}. We go on to describe the intimate connection to
hypothesis testing in Section~\ref{sec:metrology_as_hypothesis_testing} and show in Section~\ref{sec:asymptotics_succ_prob} how it can be used to understand the achievable asymptotic performance in general quantum metrology problems. Section~\ref{sec:optimal_tolerance} is dedicated to the optimal metrological tolerance for a fixed success probability and relates it to the hypothesis testing relative entropy. We showcase the various applications of our quantifiers in Section~\ref{sec:minimax_analysis_phase_estimation} where we perform a minimax analysis of the phase estimation problem. We discuss the generalization of our definitions and results to multivariate quantum metrology and their relation to learning from quantum systems in Section~\ref{sec:learning_from_quantum_systems}. Our work concludes with a detailed outline of future directions in Section~\ref{sec:future_directions} and a discussion of our results in Section~\ref{sec:discussion}.

\section{Overview of our main results}\label{sec:main_results}
We establish the following main results.

\vspace{.3\baselineskip}
\noindent\emph{$\triangleright$ A framework for quantum metrology in the finite-sample regime.}

\vspace{.3\baselineskip}
The framework of \emph{PAC metrology} established in Section~\ref{sec:quantifying_metrological_performance} and summarized in Fig.~\ref{fig:pac_metrology_quantities} constitutes an approach to quantum metrology that is both operational and valid in the single- and few-shot settings. We give rigorous definitions in Definition~\ref{def:success_probability_no_optimization} (success probability), Definition~\ref{def:metrological_tolerance} (tolerance) and Definition~\ref{def:sample_complexity} (sample complexity). 
We show that the optimization of the success probability over metrological protocols can be carried out as a convex optimization problem, concretely a semi-infinite program~\cite{SemiInfinite,Vandenberghe1998}, an infinite reading of a semi-definite program~\cite{BV96,BookBoyd2004ConvexOptimization}. We further establish properties of the proposed quantifiers, relate them to entropic quantities and detail their practical computation.

\vspace{.3\baselineskip}
\noindent\emph{$\triangleright$ A rigorous connection to hypothesis testing.}

\vspace{.3\baselineskip}
In Theorem~\ref{thm:succ_prob_upper_bound_mht_delta_window}, we establish rigorously that the task of estimating a parameter encoded in a state $\rho(t)$ is at least as hard as performing a quantum multi-hypothesis test between parametrized states $\{ \rho(t_i) \}_{i=1}^M$ associated to parameter values that are separated by at least twice the tolerance $\delta$, \textit{i.e.}\ $|t_i - t_j| > 2\delta$ for $i\neq j$ (see Fig.~\ref{fig:front_figure}).  
We also extend this upper bound to the case of parametrized quantum channels in Corollary~\ref{corr:succ_prob_upper_bound_mht_access_strategies}. We make use of this theorem to derive a simple relation of success probability to the fidelity of states in Corollary~\ref{corr:two_point_fidelity_bound} that forms the base of our further results. 

\vspace{.3\baselineskip}
\noindent\emph{$\triangleright$ Asymptotic rates of quantum metrology.}

\vspace{.3\baselineskip}
We exploit the hypothesis testing bound of Theorem~\ref{thm:succ_prob_upper_bound_mht_delta_window} to study the rate with which the success probability approaches one when using more and more copies of the same quantum state $\rho^{\otimes n}(t)$ while maintaining a fixed tolerance.
In particular, we provide upper and lower bounds on error rates in Theorems~\ref{thm:rate_upper_bound_delta} and~\ref{thm:rate_lower_bound_delta_window} and show the exact asymptotic rate for commuting problems in Corollary~\ref{corollary:asymptotic_rate_commuting}.

\vspace{.3\baselineskip}
\noindent\emph{$\triangleright$ A single-shot Cramér-Rao-like bound.}

\vspace{.3\baselineskip}
The metrological tolerance, which quantifies the smallest deviation of metrological estimates that still guarantees a given success probability, fulfills a role similar in spirit to the standard deviation in the asymptotic framework of quantum metrology. We exemplify this similarity by giving a bound in Theorem~\ref{thm:cramer_rao_like_bound} that resembles the Cramér-Rao bound but is valid in the single-shot setting. It establishes, among other insights, that the best achievable tolerance under many copies of the same state is $\delta = O(1 / \sqrt{n \min_t \calF(t)})$, similar in scaling to the quantum Cramér-Rao bound on the standard deviation.

\vspace{.3\baselineskip}
\noindent\emph{$\triangleright$ A finite-sample analysis of phase estimation.}

\vspace{.3\baselineskip}
We devote Section~\ref{sec:minimax_analysis_phase_estimation} to a minimax analysis of the phase estimation problem on an ensemble of spin-$\tfrac{1}{2}$ particles, \emph{i.e.}\ the estimation of a phase imprinted by a Hamiltonian evolution. For the general case of a covariant Hamiltonian evolution, we give the measurement achieving the optimal success probability and provide a formula in closed form in Theorem~\ref{thm:pgm_is_minimax_optimal}. This allows us to find the optimal probe state in the phase estimation scenario and to establish the optimal asymptotic rate of the error probability in Theorem~\ref{thm:opt_minimax_rate_entangled_phase_est}. We perform exhaustive numerics that showcase the differences of the single-shot analysis from the asymptotic framework, highlighting that the optimal probe states in this instance of global estimation are notably different from the optimal states for local estimation. We furthermore give evidence that in the setting of i.i.d.\ copies of the same state, the quantum Cramér-Rao bound gives a faithful estimate of the achievable minimax tolerance.

\section{A framework for quantum metrology in the finite-sample regime}
\label{sec:quantifying_metrological_performance}

We consider the task of estimating an unknown parameter that is encoded in a quantum state.
Let $t\mapsto \rho(t)$ be a one-parameter family of quantum states, where
the parameter $t$ belongs to some fixed real interval $t \in I \subseteq \mathbb{R}$.
For technical convenience, we assume that the interval $I$ is given as the domain of
the function $t \mapsto \rho(t)$ and henceforth omit explicit mention of $I$, all while assuming
that $t$ belongs to the domain of $t \mapsto \rho(t)$.

We first review the general abstract basics of Bayesian parameter estimation and the
alternative minimax parameter estimation setting.
These definitions work independently of the actual metrological problem and, as we show later, can be easily generalized to multivariate quantum metrology and metrology of quantum channels
We then consider the setting depicted in Fig.~\ref{fig:intro-super-simple-setup}, where a parameter is to be extracted from a parametrized quantum state through a quantum measurement. This allows us to rigorously establish our framework as outlined in Fig.~\ref{fig:pac_metrology_quantities}.

Through this development, we establish measures of performance that have a direct operational meaning in the non-asymptotic setting where only few experimental repetitions can be performed.
Our approach revolves around the question: \emph{\enquote{What is the probability \smash{$\eta$} of outputting an estimate that is within a given tolerance \smash{$\delta$} around the true value?}} 

We first see how we can answer this question in the general setting of Bayesian parameter estimation.
Suppose the %
value of the underlying parameter is distributed according to a
prior distribution $\mu(t)$, $t \sim \mu(t)$.
Given a value of $t$, we assume that the probability of our estimation procedure
producing the estimate $\tau$ is distributed according to $\nu(\tau \pipe t)$.
Then we can compute the \emph{Bayesian success probability} as 
\begin{align}\label{equ:intro_bayesian_success_probability}
    \eta = \int \diff \mu(t) \, \diff \nu(\tau \pipe t) \, w_{\delta}(t-\tau),
\end{align}
where $w_{\delta}(t-\tau) \coloneqq \chi[ |t-\tau| \leq \delta]$
represents a \emph{window} of size $\delta$ around the true value, with $\chi(\cdot)$ the indicator function that is equal to one when its argument is true and zero otherwise.

An alternative setting applies to the case where we have no prior information about
the parameter $t \in I$.  
Suppose that for a fixed value of $t$, the probability of
our estimation procedure producing the estimate $\tau$ is again distributed according
to $\nu(\tau \pipe t)$. 
In this case, the probability of success, in the worst case over $t$,
is determined as
\begin{align}
    \overline\eta
    &= \inf_{t\in I} \int \diff \nu(\tau \pipe t) \, w_{\delta}(t-\tau).
    \label{eq:intro_worstcase_success_probability}
\end{align}
While this setting is radically different from Bayesian estimation
on the conceptual level, we exploit a close relation between these settings at
the technical level in order to simplify our derivations. Specifically, the
quantity in Eq.~\eqref{eq:intro_worstcase_success_probability} can be expressed as
the Bayesian success probability of Eq.~\eqref{equ:intro_bayesian_success_probability} minimized
over all possible priors with support in $I$:
\begin{align}
    \overline\eta
    &= \inf_{\mu\colon \mu(I)=1} \int \diff \mu(t) \, \diff \nu(\tau \pipe t) \, w_{\delta}(t-\tau).
    \label{eq:intro_worstcase_success_probability_minimax}
\end{align}
Indeed, the minimum in Eq.~\eqref{eq:intro_worstcase_success_probability_minimax} is achieved by
a prior $\mu(t)$ that is concentrated at the time $t$ where the minimum in Eq.~\eqref{eq:intro_worstcase_success_probability} is achieved.
In other words, we might consider an adversary who gets to choose the prior $\mu(t)$ according
to which the parameter value is distributed. The worst thing that can happen is that an adversary chooses a very unfortunate prior. In this case, we can still guarantee the
success probability lower bounded by $\overline{\eta}$.
We refer to this setting as the \emph{minimax} setting, following standard terminology
in statistics. The name stems from the two optimizations that are involved when we
optimize Eq.~\eqref{eq:intro_worstcase_success_probability_minimax} over possible
estimation procedures: one optimization ranges over the estimation procedure
and the other one over $\mu(t)$.

The window function $w_\delta(t-\tau)$ in the definitions above
identifies the successful events as those where $t$ and $\tau$ differ by at most $\delta$.
More general window functions can be employed to quantify alternative definitions of
success. For instance, if the parameter represents an angle \mbox{$t \in [0, 2\pi]$}, a meaningful window function would identify $t$ and $\hat{t}$ as $\delta$-close under the topology of the unit circle.

We now turn to our specific setup in quantum metrology as depicted in Fig.~\ref{fig:intro-super-simple-setup}. 
We %
consider the %
setting where the parameter in question is encoded in a set of states $t \mapsto \rho(t)$.
In this case %
any prediction must be obtained from performing some sort of quantum measurement
on the given state and subsequently classically post-processing the outcome
of the measurement into a prediction.
We can combine both of these elements into a POVM $\tau \mapsto Q(\tau)$ such that
\begin{align}
    \nu(\tau \pipe t) = \Tr[ \rho(t) Q(\tau) ].
\end{align}

This leads us to the following formal definition.

\begin{definition}[Success probability]\label{def:success_probability_no_optimization}
For a given tolerance $\delta$, a set of states $\rho(t)$, possibly with prior $\mu(t)$, and a measurement $Q(\tau)$, the Bayesian success probability is given by
\begin{align}
\eta(\delta, \mu, \rho, Q)  
&\coloneqq \int \diff \mu(t) \, \diff \tau \, w_{\delta}( t - \tau) \Tr [\rho(t) Q(\tau)].  \label{eqn:Bayesian_succ_prob_no_opt}
\end{align}
The minimax success probability is given by
\begin{align}
\overline{\eta}(\delta, \rho, Q)  
&\coloneqq \inf_t \int \diff \tau \, w_{\delta}( t - \tau) \Tr [\rho(t) Q(\tau)].       
\end{align}
\end{definition}
We observe that the success probability can be more compactly written using the convolution notation
\begin{align}
    (w * \rho)(t) \coloneqq \int \diff \tau \, w(t-\tau)\rho(\tau)
\end{align}
as
\begin{align}
    \eta(\delta, \mu, \rho, Q)  &= \int \diff \mu(t) \, \Tr [ \rho(t)(w_{\delta} * Q)(t)  ]\\
    &= \int \diff t \, \Tr [ (w_{\delta} * [\mu \cdot \rho])(t) Q(t) ],\nonumber
\end{align}
and we will use this notation in the rest of this work.

The success probability of Definition~\ref{def:success_probability_no_optimization} quantifies the probability that our metrology protocol outputs a correct estimate. This is very reminiscent of the strategy used to quantify the performance of learning algorithms pioneered by Valiant~\cite{valiant_theory_1984}, which was coined as \emph{probably approximately correct (PAC) learning}. As this naming also conveys the essence of our approach to quantum metrology, we refer to it as \emph{probably approximately correct (PAC) metrology}.

As we show in Section~\ref{ssec:properties_of_quantifiers} of the supplementary material, both the Bayesian and minimax success probabilities have basic continuity properties in all their arguments and allow for intuitive majorization relations.

In addition to the success probability as a measure of metrological performance, it is equally fair and operationally relevant to reverse the question and ask: \emph{\enquote{What is the smallest tolerance \smash{$\delta$} that still guarantees a success probability of \smash{$\eta$}?}} 
We condense this reasoning into the following rigorous definition:
\begin{definition}[Estimation tolerance]
\label{def:metrological_tolerance}
For a given success probability $\eta$, a set of states $\rho(t)$, possibly with prior $\mu(t)$, and a measurement $Q(\tau)$, the Bayesian estimation error tolerance is given by
\begin{align}%
    \delta(\eta, \mu, \rho, Q) \coloneqq \inf \mathstrut\{ \delta' \geq 0 \pipe \eta(\delta', \mu, \rho, Q) \geq \eta \}.
\end{align}
The minimax estimation error tolerance is given by
\begin{align}
\overline\delta(\eta, \rho, Q) \coloneqq \inf \mathstrut\{ \delta' \geq 0 \pipe \overline\eta(\delta', \rho, Q) \geq \eta \}.
\end{align}
\end{definition}
Looking at the estimation error tolerance instead of the success probability allows
for a simpler comparison with the standard bounds encountered in quantum metrology,
as the tolerance has comparable meaning to the standard deviation of an estimator, 
which is the target of the quantum Cramér-Rao bound. 

Both the success probability and the tolerance introduced in Definitions~\ref{def:success_probability_no_optimization}
and~\ref{def:metrological_tolerance} are truly \emph{single-shot} quantities,
in that they consider a single outcome of the quantum measurement.
Usually, however, the desired performance can only be achieved by
performing multiple repetitions of the same experiment. 
We %
model multiple repetitions of an experiment
by having access to $n$ copies of the parametrized state, 
\textit{i.e.}, \smash{$t \mapsto \rho^{\otimes n}(t)$}.  
The measurement $Q(\tau)$ is then collectively measure the $n$ copies
of the state.

The multi-copy scenario leads us to a third operationally relevant question:
\emph{\enquote{How many repetitions \smash{$n$} of my experiment do I need to perform to obtain a desired tolerance \smash{$\delta$} with a fixed success probability \smash{$\eta$}?}}
We can cast this \emph{sample complexity} into the following definition:
\begin{definition}[Sample complexity]\label{def:sample_complexity}
For a given success probability $\eta$ and tolerance $\delta$, a set of states $\rho(t)$, possibly with prior $\mu(t)$, and a sequence of measurements $\{ Q^{(n)}(\tau)\}$ the Bayesian sample complexity is given by
\begin{align}%
\begin{split}
    &n(\eta, \delta, \mu, \rho, \{ Q^{(n)}\}) \\
    &\qquad\coloneqq \min \mathstrut\{ n' \in \bbN \pipe \eta(\delta, \mu, \rho^{\otimes n'}, Q^{(n')}) \geq \eta \}.
\end{split}
\end{align}
The minimax sample complexity is given by
\begin{align}
\overline{n}(\eta, \delta, \rho, \{  Q^{(n)}\}) \coloneqq \min \mathstrut\{ n' \pipe \overline\eta(\delta, \rho^{\otimes n'}, Q^{(n')}) \geq \eta \}.
\end{align}
\end{definition}
Results in the context of more general metrology tasks, like state tomography or Hamiltonian learning are therefore usually phrased in terms of the sample complexity. We explore this connection and the multivariate generalization of this framework in Section~\ref{sec:learning_from_quantum_systems}, where we show how our definitions and results generalize to these settings and how the languages can be compared.

Practical settings in quantum metrology often involve estimating an unknown parameter
present in the dynamics of a system, rather than directly encoded into the state itself. 
Such dynamics might involve interactions with an environment system, or another quantum system
whose properties we seek to estimate.  Formally, the task becomes that of estimating an unknown
parameter present in a quantum channel by applying the unknown channel on suitable inputs and
performing suitable subsequent measurements.
In simple cases, the channel parameter estimation problem can reduce to a state estimation problem:
One prepares a fixed initial state $\rho_0$ and sends it through the channel, resulting in a state
$\rho(t)$; the task is now to estimate the parameter $t$ encoded in the quantum state.
The channel estimation problem, however, provides a richer landscape of estimation strategies when more than one copy of the channel is available.
We discuss %
this setting in Section~\ref{sec:parametrized_quantum_channels}.

With the preceding definitions that form our few-shot framework of PAC metrology, we have established a set of quantities that capture the performance of finite-sample quantum metrology protocols. Here, we discuss the optimal values these quantities can take, when we optimize over all possible metrological prescriptions. This brings us to the following definitions of \emph{optimal} counterparts of the Definitions~\ref{def:success_probability_no_optimization},~\ref{def:metrological_tolerance} and~\ref{def:sample_complexity}.
\begin{definition}[Optimal Bayesian quantities]
The \emph{optimal Bayesian success probability} is obtained by optimizing the Bayesian success probability over all possible POVMs: %
\begin{align}\label{eqn:opt_success_probability_bayesian}
    \eta^{*}(\delta, \mu, \rho) &\coloneqq \sup_{Q(\tau)} \eta(\delta, \mu, \rho, Q).
\end{align}
We use it to define the \emph{optimal Bayesian tolerance}
\begin{align}
    \delta^{*}(\eta, \mu, \rho) \coloneqq \inf \mathstrut\{ \delta' \geq 0 \pipe \eta^{*}(\delta', \mu, \rho) \geq \eta \}
\end{align}
and \emph{optimal Bayesian sample complexity}
\begin{align}%
    n^{*}(\eta, \delta, \mu, \rho) \coloneqq \min \mathstrut\{ n' \in \bbN \pipe \eta^{*}(\delta, \mu, \rho^{\otimes n'}) \geq \eta \}.
\end{align}
\end{definition}
The optimal minimax quantities are defined analogously.
\begin{definition}[Optimal minimax quantities]
The \emph{optimal minimax success probability} is obtained by optimizing the minimax success probability over all possible POVMs: %
\begin{align}
    \overline{\eta}^{*}(\delta, \rho) &\coloneqq \sup_{Q(\tau)} \overline{\eta}(\delta, \rho, Q).
\end{align}
We use it to define the \emph{optimal minimax estimation tolerance}
\begin{align}
    \overline\delta^{*}(\eta, \rho) \coloneqq \inf \mathstrut\{ \delta' \geq 0 \pipe \overline\eta^{*}(\delta', \rho) \geq \eta \}
  \label{eq:optimal-minimax-estimation-tolerance-defn}
\end{align}
and \emph{optimal minimax sample complexity}
\begin{align}%
    \overline{n}^{*}(\eta, \delta, \rho) \coloneqq \min \mathstrut\{ n' \in \bbN \pipe \overline\eta^{*}(\delta, \rho^{\otimes n'}) \geq \eta \}.
\end{align}
\end{definition}
The optimal Bayesian and minimax success probabilities have desirable properties such as convexity in their arguments and data-processing under noise channels, a point we elaborate on in Section~\ref{ssec:properties_of_quantifiers} of the supplementary material. 

We do not have an \emph{a-priori} restriction on the domain of the parameter $t$ when calculating the success probability. However, it is often easier to compute it, both analytically and numerically, if we restrict it to a finite interval. While this is trivially giving a bound in the minimax case, the following lemma ensures that we can also use the restriction to a subinterval to compute bounds for the optimized Bayesian success probability:
\begin{lemma}[Subdivision trick]\label{lem:subdivision_trick}
For a given tolerance $\delta$ and a set of states $\rho(t)$ with prior distribution $\mu(t)$, define the restriction of a prior to an interval $I \subseteq \bbR$ as
\begin{align}
    \mu|_{I}(t) \coloneqq \chi[t \in I]\mu(t) / \mu(I).
\end{align}
Let us further denote with $I_t$ the interval of size $T$ centered at $t$. Then, we have that
\begin{align}
    \eta^{*}(\delta, \rho, \mu) &\leq \frac{1}{T} \int \diff t \, \mu(I_t) \,  \eta^{*}(\delta, \rho, \mu|_{I_t}) \\
    &\leq \max_t \eta^{*}(\delta, \rho, \mu|_{I_t}).
    \nonumber
\end{align}
\end{lemma}
The proof is provided in Section~\ref{ssec:proof_subdivision_trick} of the supplementary material.
The above immediately implies corresponding subdivision lower bounds on the optimal tolerance and optimal sample complexity.

\section{Computing optimal measurements}\label{sec:optimal_metrology_protocols}

\subsection{Generally optimal measurements through convex optimization}
Of the three quantities, the success probability is the most amenable to optimization. It is linear in the chosen measurement, and we can show that the optimization over the measurement $Q(\tau)$ can be cast into the form of a convex problem~\cite{BookBoyd2004ConvexOptimization}. Specifically, it assumes the form of a \emph{semi-infinite problem}~\cite{SemiInfinite,Vandenberghe1998}, so a \emph{semi-definite problem} with infinitely many objective variables and finitely many constraints, or the other way around. For the theory of semi-definite programming to largely take over, the involved functions must be analytically defined, smooth, and convex~\cite{Vandenberghe1998}, which can be safely assumed here. In what follows, we refer to such infinite readings of semi-definite problems as convex problems. This does not only provide a path to more easily compute them, but is also a reliable tool to prove upper and lower bounds. 

\begin{proposition}[POVM optimization]\label{prop:sdp_formulation}
For a given tolerance $\delta$ and a set of states $\rho(t)$ with prior distribution $\mu(t)$, the optimal success probability $\eta^*(\delta,\mu,\rho)$ can be computed using the convex program
\begin{align}
    \max_{Q(t) \geq 0}\mathstrut \left\{\left. \int \diff t \, \Tr[ (w_{\delta} * [\mu \cdot \rho])(t) Q(t) ] \, \right| \, \int \diff t \, Q(t) = \bbI \right\}.
\end{align}
The optimal success probability is equally characterized by the convex program 
\begin{align}\label{eq:sdp_dual_formulation_Bayesian}
    \min_{X \geq 0}\mathstrut \big\{\Tr[X] \, \big| \, X \geq (w_{\delta} * [\mu \cdot \rho])(t)\text{\upshape{ for all }}t \big\}.
\end{align}
The minimax success probability can be computed in a similar way by additionally optimizing over all priors. 
\end{proposition}
The second convex program in Eq.~\eqref{eq:sdp_dual_formulation_Bayesian} is derived through the notion of duality in semi-infinite programming.
The proof of the above proposition, along with a more detailed statement is provided in Proposition~\ref{prop:minimax_sdp} of the supplementary material, where we also explicitly give the convex formulation of the minimax success probability. 

From a practical standpoint, it is also important that the above semi-infinite program can actually be implemented numerically while having guarantees on the quality of the approximation.
\begin{proposition}
[Discretization]
There exists a discretization of the convex-program of Proposition~\ref{prop:sdp_formulation} that yields a semi-definite program that can be solved with standard tools. If $\rho(t)$ is Lipschitz with respect to the trace distance, the error of the discretization can be made arbitrarily small by choosing a suitable scale of the discretization.
\end{proposition}
We give the detailed statements of the above in Section~\ref{ssec:discretization_bouds_sip} of the supplementary material.
The optimal tolerance and sample complexity can in principle be computed by combining the above semi-definite program with binary search~\cite{oliveira_enhancement_2020}.

\subsection{%
Maximum-likelihood-inspired post-processing of the outcome of a fixed measurement
}
\label{sec:opt_fixed_measurement}

The general structure of the solutions to the convex problem of Proposition~\ref{prop:sdp_formulation} is unclear \emph{a priori}. 
Consequently, the optimal measurement for a given metrological problem might be either impractical or impossible to implement in an experiment. %
Here, we study the performance of strategies that naturally model strategies that deal with experimental restrictions on the possible measurements that can be performed.

As shown in Fig.~\ref{fig:post_processing_figure}, we consider an estimation strategy that begins by applying a fixed quantum measurement described by a
POVM $\{ M(\lambda) \}_\lambda$.  For instance, the POVM might represent a projective measurement
in a fixed basis. Subsequently, the procedure infers from the outcome $\lambda$
an estimate $\tau^*(\lambda)$ for the value
$t$ that is encoded in the measured state.
We focus on a particular post-processing strategy $\tau^*(\lambda)$ inspired by the maximum-likelihood estimation technique. 

\begin{figure}
    \centering
    \includegraphics{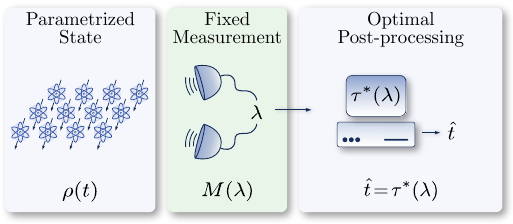}
    \caption{Metrology setting with fixed measurement as discussed in Section~\ref{sec:opt_fixed_measurement}. A parametrized state $\rho(t)$ is measured with a POVM $\{ M(\lambda) \}_{\lambda}$ fixed by, \emph{e.g.}, experimental feasibility restrictions. We show that the the optimal way of post-processing the measurement outcome $\lambda$ into predictions of the underlying parameter $t$ is given by a smoothed analogue of maximum a posterior estimation in the Bayesian case. A similar strategy also manages to provide a guarantee on the achievable success probability in the minimax case.}
    \label{fig:post_processing_figure}
\end{figure}

If we only fix the quantum-to-classical measurement, the question of the optimal post-processing of the measurement outcomes for such a fixed measurement of the quantum system arises. The authors of Ref.~\cite{Salmon} studied such a setting, but focusing on more fundamental point of view on when a quantum-to-classical measurement can be considered admissible in a statistical sense.%
A fixed measurement can be discrete, like a measurement in a specific basis, or continuous, like a pretty good measurement. We can model both cases in a unified way by assuming a POVM with continuous outcomes $\{ M(\lambda) \}_{\lambda}$. 
For a given set of states $\rho(t)$, we can define the \emph{likelihood function}
\begin{align}\label{eqn:def_likelihood}
    \Lambda(\lambda \pipe t) &\coloneqq \Tr[ \rho(t) M(\lambda)]
\end{align}
which captures the distribution over measurement outcomes for a fixed value of the underlying parameter $t$.
The \emph{joint distribution} of measurement outcomes and underlying parameters is then given by
\begin{align}
    (\lambda, t) \sim \mu(t) \Lambda(\lambda \pipe t).
\end{align}
If we denote the marginal distribution of the measurement outcomes $\lambda$ as
\begin{align}
    \nu(\lambda) \coloneqq \int \diff \mu(t) \, \Lambda(\lambda \pipe t),
\end{align}
then the joint distribution is related to the \emph{posterior distribution} of the underlying parameter given an observed measurement outcome $\lambda$ as
\begin{align}\label{eqn:def_posterior}
    P(t \pipe \lambda) \coloneqq \frac{ \mu(t) \Lambda(\lambda \pipe t)}{\int \diff \mu(t) \, \Lambda(\lambda \pipe t)}.
\end{align}
With this notation, we equivalently have that the joint distribution can be written as
\begin{align}
    (\lambda, t) \sim \nu(\lambda) P(t \pipe \lambda).
\end{align}

As the measurement $\{ M(\lambda)\}$ is fixed, the only thing left to optimize is the prediction we make when observing a certain measurement outcome $\lambda$. We denote this prediction as $\tau^{*}(\lambda)$ and refer to this function as a \emph{prediction strategy}. The POVM $Q(\tau)$ associated to the combination of fixed POVM $\{ M(\lambda) \}_{\lambda}$ together with a prediction strategy $\tau^{*}(\lambda)$ is given by collecting all POVM effects $M(\lambda)$ associated to a particular prediction $\tau$:
\begin{align}
    Q_{M, \tau^{*}}(\tau) = \int \diff \lambda \, \chi[\tau^{*}(\lambda) = \tau] M(\lambda).
\end{align}
With this POVM, it is straightforward to see that the Bayesian success probability takes the form
\begin{align}
    \eta(\delta, \mu, \rho, Q_{M, \tau^{*}}) &= \int \diff \mu(t) \, \diff \lambda \, w_{\delta}(t - \tau^{*}(\lambda)) \Lambda(\lambda \pipe t) \nonumber\\
    &= \int \diff \nu(\lambda) \, ( w_{\delta} *  P(\cdot \pipe \lambda))(\tau^{*}(\lambda))\label{eqn:succ_prob_pred_strategy}.
\end{align}
In the second line, we see that the posterior distribution of $t$ given the observed measurement outcomes is critical to the success probability. 

A look at Eq.~\eqref{eqn:succ_prob_pred_strategy} reveals that the optimal prediction strategy is to always predict the $\tau^{*}$ that maximizes the \emph{smoothed posterior probability} $w_{\delta} * P(\cdot \pipe \lambda)$. In accordance with the naming conventions of classical statistics, we refer to this prediction rule 
\begin{align}
    \tau^{*}_{\mathrm{SMAP}}(\lambda) &\coloneqq \argmax_{\tau} \mathstrut(w_{\delta} * P(\cdot \pipe \lambda))(\tau)
\end{align}
as the \emph{smoothed maximum a posteriori} (SMAP) estimate. This represents a smoothed version of maximum a posteriori estimation, the Bayesian generalization of maximum likelihood estimation.
The so achieved Bayesian success probability is consequently
\begin{align}
    \eta(\delta, \mu, \rho, Q_{M, \tau_{\mathrm{SMAP}}^{*}}) &= \int \diff \nu(\lambda) \, \max_{\tau} \mathstrut(w_{\delta} * P(\cdot \pipe \lambda))(\tau)\nonumber \\
&= \int \diff \nu(\lambda) \, \mathstrut \lVert w_{\delta} * P(\cdot \pipe \lambda) \rVert_{\infty}.\label{eqn:succ_prob_for_smap}
\end{align}
The relation to the function infinity norm allows us to derive some simple upper bounds on the success probability from Young's convolution inequality, as stated in Lemma~\ref{slem:young_conv_inequ} of the supplementary material. We will later make use of smoothed maximum a posteriori estimation to obtain lower bounds on the asymptotic error rate of quantum metrology protocols. 
To summarize:
\begin{theorem}[Optimal post-processing of a fixed measurement in the Bayesian setting]
    For a given tolerance $\delta$, state set $\rho(t)$ with prior $\mu(t)$ and a fixed POVM $\{ M(\lambda) \}_{\lambda}$, the optimal Bayesian success probability of an estimation
    strategy that measures $\{M(\lambda)\}$ and post-processes the result is achieved by the
    smoothed maximum a posteriori estimator:
    \begin{align}
        \sup_{\tau^*} \eta(\delta, \mu, \rho, Q_{M, \tau^{*}})
        = \eta(\delta, \mu, \rho, Q_{M, \tau_{\mathrm{SMAP}}^{*}}).
    \end{align}
\end{theorem}

Having established the optimal strategy in the Bayesian setting, we now turn to the minimax case. 
With the definitions introduced earlier, the minimax success probability associated to a prediction strategy $\tau^{*}$ is given by
\begin{align}
    \overline\eta(\delta, \rho, Q_{M, \tau^{*}}) &= \min_t \int \diff \lambda \, w_{\delta}(t - \tau^{*}(\lambda)) {\Lambda}(\lambda \pipe t).
\end{align}
The minimum over $t$ is \emph{outside} the integration over the different measurement outcomes $\lambda$, which means that, contrary to the Bayesian case, we cannot optimize the prediction for each $\lambda$ independently.
We can, however, still establish a lower bound on the minimax success probability for the optimal prediction strategy. To do so, we need to go via the minimax error probability:
\begin{align}
    1 &- \overline\eta(\delta, \rho, Q_{M, \tau^{*}}) \\
    &= 1 - \min_t \int \diff \lambda \, w_{\delta}(t - \tau^{*}(\lambda)) {\Lambda}(\lambda \pipe t) \nonumber \\
    &= \max_t \int \diff \lambda \, \Lambda(\lambda \pipe t) - \min_t \int \diff \lambda \, w_{\delta}(t - \tau^{*}(\lambda)) {\Lambda}(\lambda \pipe t) \nonumber\\
    &= \max_t \int \diff \lambda \, [ 1 - w_{\delta}(t-\tau^{*}(\lambda))] {\Lambda}(\lambda \pipe t)\nonumber.
\end{align}
Here, we have exploited that $\int \diff \lambda \, \Lambda(\lambda \pipe t) = 1$ for all $t$. We have thus reformulated the minimax error probability as a function of the complement of the window function. We can then exchange integration and maximization to obtain the upper bound on the error
\begin{align}
    1 - \overline\eta(\delta, \rho, Q_{M, \tau^{*}}) \leq\int \diff \lambda \,  \max_t \mathstrut [ 1 - w_{\delta}(t-\tau^{*}(\lambda))] {\Lambda}(\lambda \pipe t)
\end{align}
The above bound can now be optimized independently for all $\lambda$ which gives us a strategy analogous to the 
maximum a posterior estimation we have used in the Bayesian case. We will refer to as \emph{smoothed minimax complementary likelihood (SMCL)} estimate, 
\begin{align}
    \tau_{\mathrm{SMCL}}^{*}(\lambda) &\coloneqq \argmin_{\tau} \max_t \mathstrut [1-w_{\delta}(t - \tau)] {\Lambda}(\lambda \pipe t).
\end{align}
This strategy amounts to choosing $\tau$ such that the largest values of $\Lambda(\lambda \pipe t)$ are contained in the window centered around $\tau$.
Using this strategy then establishes the bound
\begin{align}
    1 &- \sup_{\tau^{*}} \overline\eta(\delta, \rho, Q_{M, \tau^{*}}) \\ &\leq 
    \int \diff \lambda \, \min_{\tau} \max_t \mathstrut [1-w(t-\tau)] \Lambda(\lambda \pipe t)\nonumber
\end{align}
on the minimax error probability that can be achieved using the optimal post-processing.

\section{Finite-sample quantum metrology as continuous hypothesis testing}\label{sec:metrology_as_hypothesis_testing}

One might imagine that, intuitively, determining a parameter encoded in a quantum system is intimately related to the task of distinguishing states for different values of the parameter. This reasoning has already been used in the standard approach to quantum metrology to obtain asymptotic bounds, for example through the 
\emph{quantum Ziv-Zakai bound}~\cite{tsang2012ziv-zakai}. These approaches, however, are still hampered by the limitations of the standard approach to quantum metrology. As we outline below, by choosing the success probability as a measure of metrological performance, the connection to the task of distinguishing quantum states is much more natural and fundamental.

The optimal probability of success for distinguishing a set of $N$ quantum states $\{ \rho_i \}_{i=1}^N$ with prior probabilities $\{ p_i \}_{i=1}^N$ -- a test between $N$ quantum hypotheses -- is given by\footnote{This is the definition of \emph{symmetric} hypothesis testing, where all states are treated equally and the average success probability is used as quantifier. In \emph{asymmetric} hypothesis testing, the states are treated independently.}
\begin{align}
    P_s^{*}(\{ p_i \rho_i \}_{i=1}^N) \coloneqq \sup_{\substack{\text{POVMs}\\\{ Q_i \}_{i=1}^N}  }  \sum_{i=1}^N p_i \Tr[ \rho_i Q_i ] .
\end{align}
Similarly, the optimal minimax success probability is defined as
\begin{align}
    \overline{P}_s^{*}(\{ \rho_i \}_{i=1}^N) \coloneqq \sup_{\substack{\text{POVMs}\\\{ Q_i \}_{i=1}^N}  } \min_{1 \leq i \leq N} \Tr[ \rho_i Q_i ].
\end{align}
Already, Yuen \emph{et al.}~\cite{yuen1975optimum} have studied an extension of this definition with an additional \emph{cost matrix} $C_{ij}$ and the associated \emph{optimal expected cost}
\begin{align}
    \inf_{\substack{\text{POVMs}\\\{ Q_i \}_{i=1}^N}  } \sum_{i=1}^N\sum_{j=1}^N  p_i  C_{ij} \Tr[ \rho_i Q_j]
\end{align}
and have shown that this constitutes a semi-infinite problem.
A look at Proposition~\ref{prop:sdp_formulation} shows that our notion of Bayesian success probability can be understood as a continuous version of this definition in which the complement of the window function, $1 - w_{\delta}$, takes the role of the cost function. 

Because of this close resemblance of metrological problems to generalized multi-hypothesis testing problems, it is unsurprising that every quantum multi-hypothesis testing problem can be written as a particular metrology problem with suitably chosen prior distribution and tolerance (see Section~\ref{ssec:mht_as_metrology} of the supplementary material). 
We, however, also establish a result in the other direction that bounds the optimal Bayesian and minimax success probability for a metrological problem through multi-hypothesis testing.
It is intuitively clear that being able to determine a parameter $t$ to a precision $\delta$ means that we must be able to sufficiently well distinguish between states that are at least $2\delta$ apart. We make this reasoning rigorous in the below theorem.
\begin{theorem}[Hypothesis testing bound]\label{thm:succ_prob_upper_bound_mht_delta_window}
For a given tolerance $\delta$, fix any set $\calS = \{(\lambda, s)\}$ of prior probabilities $\lambda \in [0,1]$ and shifts $s \in \bbR$ such that for all distinct $s, s' \in \calS$ we have that $|s-s'| > 2\delta$ and $\sum_{\lambda \in \calS} \lambda = 1$. Then, for a state set $\rho(t)$ with prior $\mu(t)$ we have the upper bound
\begin{align}
    \eta^{*}(\delta, \mu, \rho) &\leq \int \diff t \, P_s^{*}(\{ \lambda \, \mu(t+s) \rho(t+s)\}_{(\lambda, s) \in \calS}).
\end{align}
Optimizing over the prior probabilities $\lambda$ then yields the analogous upper bound in the minimax setting
\begin{align}
    \overline{\eta}^{*}(\delta, \rho) &\leq \inf_t \overline{P}_s^{*}(\{  \rho(t+s)\}_{s \in \calS}).
\end{align}
\end{theorem}
We established that a metrological problem is at least as hard as determining whether an adversary has manipulated the clock used for the experiment by shifting its time by one of the values $s$ in the set $\calS$ with probability $\lambda$. The tolerance $\delta$ gives us a lower bound on the distance between the different shifts and therefore works in our favor. 
The proof of the theorem is presented in Section~\ref{ssec:upper_bound_state_discr} of the supplementary material. It also generalizes to parametrized quantum channels as discussed in Section~\ref{sec:parametrized_quantum_channels}.
\begin{figure}
    \centering
    \includegraphics{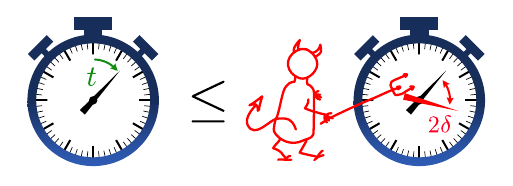}
    \caption{A visualization of Theorem~\ref{thm:succ_prob_upper_bound_mht_delta_window}. Telling the time with a certain tolerance $\delta$ is at least as hard as deducing if an adversary manipulated the clock, under the guarantee that the manipulation is not smaller than twice the tolerance $\delta$.}
    \label{fig:devil}
\end{figure} 

We note that application of Theorem~\ref{thm:succ_prob_upper_bound_mht_delta_window} is especially useful analytically when restricting to two shifts, because then the analytic expression for the success probability of binary hypothesis testing due to Helstrom~\cite{helstrom1969quantum} can be used. Looking at the minimax case with two shifts, we further obtain what can be considered as an analogue of Le Cam's two point method~\cite{yu_assouad_1997}:
\begin{corollary}[Two-point method]
\label{thm:success-probability-upper-bound-via-two-point-method}
    For a given tolerance $\delta$ and state set $\rho(t)$, we have the upper bound on the minimax success probability
    \begin{align}
        \overline{\eta}^{*}(\delta, \rho) &\leq \inf_{|t-t'| > 2\delta} \overline{P}_s^{*}(\rho(t), \rho(t')).
    \end{align}
\end{corollary}

The above result, together with the Fuchs-van-de-Graaf inequalities allows us to further deduce a relation to the quantum fidelity 
\begin{align}\label{eqn:def_fidelity}
    F(\rho, \sigma) \coloneqq \Tr[ (\sigma^{1/2} \rho \sigma^{1/2})^{1/2}]
\end{align}
which we will use later to give a concise sample complexity bound, as shown in Section~\ref{ssec:upper_bound_from_mht} of the supplementary material.
\begin{corollary}\label{corr:two_point_fidelity_bound}
For a given tolerance $\delta$ and state set $\rho(t)$, we have the minimax error probability lower bound
\begin{align}
    1 - \overline{\eta}^{*}(\delta, \rho) &\geq \frac{1}{4} \sup_{|t-t'| > 2\delta} F(\rho(t), \rho(t'))^2.
\end{align}
\end{corollary}

We can use similar reasoning as in the derivation of Theorem~\ref{thm:succ_prob_upper_bound_mht_delta_window} to obtain a lower bound that makes use of \emph{asymmetric} hypothesis testing. The resulting bounds are presented in Section~\ref{ssec:upper_bound_succ_prob_asym_ht} of the supplementary material.

We explore other directions that yield metrology bounds of different forms in Section~\ref{ssec:further_succ_prob_bounds} of the supplementary material. Along the way, we present some bounds for quantum multi-hypothesis testing in Section~\ref{ssec:bounds_for_quantum_mht} of the supplementary material.

\section{Asymptotic behavior of the success probability in the \texorpdfstring{i.i.d.\@}{i.i.d.} regime}
\label{sec:asymptotics_succ_prob}
In this section, we consider our framework in the limit where many
independent and identically distributed (i.i.d.\@)
copies of the state are available.  We aim to understand this
asymptotic limit of our framework in order to connect our finite-sample approach
to standard approaches in metrology and information theory where many copies
of the state are assumed to be available.
We exploit
Theorem~\ref{thm:succ_prob_upper_bound_mht_delta_window} %
to compute an upper bound on the asymptotic rate with which both the Bayesian and minimax success probability approach one.
Formally, we define the asymptotic error rate as
\begin{align}\label{eqn:def_asymptotic_rate}
    R^{*}(\delta, \mu, \rho) &\coloneqq \lim_{n \to \infty}  -\frac{1}{n} \log \left( 1 -\eta^{*}(\delta, \mu, \rho^{\otimes n}) \right),
\end{align}
when the limit exists. The minimax rate $\overline{R}^{*}(\delta, \rho)$ is defined analogously.  

We relate the asymptotic rate of a metrological problem to the asymptotic rate of hypothesis testing~\cite{audenaert2008asymptotic,li2016discriminating}, which is given in terms of the \emph{Chernoff divergence}
\begin{align}
    C(\rho, \sigma) \coloneqq -\inf_{0 \leq s \leq 1} \log  \Tr[ \rho^s \sigma^{1-s} ].
\end{align}
By analogy to multi-hypothesis testing, where the asymptotic rate is given by the smallest Chernoff divergence between two states that are to be discriminated, we obtain the following upper bound:
\begin{theorem}[Upper bound on asymptotic rate]\label{thm:rate_upper_bound_delta}
For a given tolerance $\delta$ and a set of states $\rho(t)$, possibly with prior $\mu(t)$, the Bayesian and the minimax rate obey the upper bounds
\begin{align}
    \overline{R}^{*}(\delta, \rho) \leq R^{*}(\delta, \mu, \rho) \leq \inf_{|t - t'| > 2\delta} C(\rho(t), \rho(t')),
\end{align}
where the optimization is over time values that have non-vanishing support in the possible priors.
\end{theorem}
We emphasize that for the above statement and the ones are to follow in this section, we implicitly assumed that the support of $\mu$ encompasses the whole admissible set of values for the parameter $t$ that is used to compute the minimax success probability. Restrictions to smaller admissible sets of parameters are possible and the results carry over straightforwardly.
The above theorem establishes that a metrological problem is asymptotically at most as hard as the hardest binary hypothesis testing problem of two states whose associated time values are at least $2\delta$ apart. The proof of the theorem uses Theorem~\ref{thm:succ_prob_upper_bound_mht_delta_window} together with Laplace's principle and is presented in Section~\ref{ssec:upper_bound_asymptotic_rate_mht} of the supplementary material.

In quantum multi-hypothesis testing, the smallest pairwise rate is actually achievable~\cite{li2016discriminating} and, because of the results below on the commuting case, we believe this to also be true in the case of quantum metrology, at least under suitable regularity assumptions. While we were not yet able to prove this general statement, we succeeded in establishing a lower bound that guarantees an asymptotic rate equal to the best pairwise hypothesis testing rate for a fixed measurement sequence. Formally, we consider a fixed sequence of measurements $\{ M^{(n)} \}$ for $n \in \bbN$ with outcomes $\lambda$ (compare to Section~\ref{sec:opt_fixed_measurement}), and denote the channel that maps states to their output distributions over $\lambda$ as
\begin{align}
    \calM^{(n)}[\rho] = \int \diff \lambda \, |\lambda \rangle\!\langle \lambda | \, \Tr[ \rho M^{(n)}(\lambda)].
\end{align}
This sequence achieves the following rate for binary state discrimination:
\begin{align}
\begin{split}
    &R(\rho, \sigma, \{ M^{(n)} \}_{n \in \bbN})  \coloneqq \\
    &\quad \lim_{n\to \infty} - \frac{1}{n} \log \left(1 - P_s^{*}(\calM^{(n)}[\rho^{\otimes n}], \calM^{(n)}[\sigma^{\otimes n}])\right).
    \end{split}
\end{align}
With this notation in place, we have the following theorem, which adds the additional assumption that the set of states is continuous:
\begin{theorem}[Asymptotic rate with fixed measurement scheme]\label{thm:rate_lower_bound_delta_window}
For a given tolerance $\delta > 0$, a continuous set of states $\rho(t)$, possibly with prior $\mu(t)$ of full support, and any fixed measurement sequence $\{ M^{(n)}\}$ for $n\in\bbN$, we have 
\begin{align}
    R^{*}(\delta, \mu, \rho) \geq \overline{R}^{*}(\delta, \rho) \geq \inf_{|t-t'|>2\delta} R(\rho(t), \rho(t'), \{ M^{(n)} \}).
\end{align}
\end{theorem}
In the above theorem, the possible values of $t$ are constrained to the domain of $\rho(t)$ and priors that have full support on said domain. Priors with restricted support fulfill the theorem with $t$ and $t'$ contained in the restricted domain.
The proof uses a discretization argument together with the smoothed maximum a posteriori estimation technique introduced in Section~\ref{sec:opt_fixed_measurement} and is given in Section~\ref{ssec:bin_hyp_testing} of the supplementary material. In the case where an optimal measurement sequence is known, however, this lower bound already achieves the rate. This is particularly true in the commuting (\textit{i.e.}, classical) case:
\begin{corollary}[Asymptotic rate for commuting states]\label{corollary:asymptotic_rate_commuting}
    For a given tolerance $\delta > 0$, a continuous set of states $\rho(t)$ such that $[\rho(t), \rho(t')] = 0$ for all $t, t'$, possibly with prior $\mu(t)$, we have that
\begin{align}
    R^{*}(\delta, \mu, \rho) = \overline{R}^{*}(\delta, \rho) = \inf_{|t-t'|>2\delta} C(\rho(t), \rho(t')),
\end{align}
where the optimization is over time values that have non-vanishing support in the possible priors.
\end{corollary}
We believe that the strong connection we established here to hypothesis testing between multiple quantum states serves as a motivation to further explore the connection between quantum metrology and quantum information tasks. This should be a fruitful endeavor that allows researchers both with a background in quantum metrology as well as in quantum information processing to make an impact on quantum metrology.

\section{Relation to entropic quantities}\label{sec:entropic_quantities}
In the single-shot approach to quantum information theory~\cite{Tomamichel_book,khatri2020principles}, many quantum information processing tasks can be quantified through generalized notions of \emph{entropy}. In this section, we explore how our definitions relate to some of these concepts.

We start by giving an alternative definition of the optimal Bayesian success probability.
To do so, we define the \emph{conditional min-entropy} of a bipartite and positive semi-definite classical-quantum operator~\cite{koenig09operationalminmax}
\begin{align}
    H_{\min}(T \pipe S)_X \coloneqq -\inf_{\sigma_S} D_{\max}(X_{TS}\pipe \mathbb{I}_T \otimes \sigma_S),
\end{align}
where the \emph{max-relative entropy}~\cite{datta2009minmaxrelent} of two positive semi-definite operators $X$ and $Y$ is given by
\begin{align}
    D_{\max}(X \fatpipe Y) \coloneqq \log\mathstrut\lVert{Y^{-1/2}X  Y^{-1/2}}\rVert_{\infty}.
\end{align}
We then have the following corollary of the dual formulation of the optimal success probability in Proposition~\ref{prop:sdp_formulation}.
\begin{corollary}[Relation to conditional min-entropy]\label{corollary:relation_to_cond_min_ent}
For a given tolerance $\delta$ and a set of states $\rho(t)$ with prior distribution $\mu(t)$, define the classical-quantum operator
\begin{align}
    X_{TS} \coloneqq\int\diff t \, \ketbra{t}{t}_{T}\otimes (w_{\delta}\ast[\mu\cdot\rho])(t)_S,
\end{align}
where $T$ is the time register and $S$ the system register. Then
\begin{align}
    -\log\eta^{*}(\delta,\mu,\rho) &= H_{\min}(T\pipe S)_X.
\end{align}
\end{corollary}
Another angle on the same fact is given by defining the \emph{max-relative entropy radius}~\cite{audenaert2014upper} of a set of positive semi-definite operators $\calX$ as 
\begin{align}
    r_{\max}(\calX) \coloneqq \inf_{\sigma\vphantom{\calX}} \max_{x \in \calX} D_{\max}( x \fatpipe \sigma).
\end{align}
In this case, the optimal success probability can be understood as the max relative entropy radius of the set of states \enquote{smoothed} using the window function
\begin{align}
    \log \eta^{*}(\delta,\mu,\rho) &= r_{\max}( \{ (w_{\delta} * [\mu \cdot \rho])(t) \} ).
\end{align}
Both relations can be thought of continuous analogues of the known connection between the success probability of multi-hypothesis testing 
and the conditional min-entropy of classical-quantum states~\cite{koenig09operationalminmax,audenaert2014upper}. In addition to the various known operational meanings of the conditional min-entropy~\cite{PhDRenner2005,koenig09operationalminmax,dupuis2014decoupling,duan2016zeroerror,fang2020channelsimulation}, the relation of Corollary~\ref{corollary:relation_to_cond_min_ent} endows it with yet another operational meaning, this time in the context of quantum metrology. Notably, the operator $X_{TS}$ of Corollary~\ref{corollary:relation_to_cond_min_ent} is not normalized, unlike in the case of multi-hypothesis testing, due to the window function $w_{\delta}$. Another instance in which the conditional min-entropy of a non-normalized operator has an operational meaning is in the context of channel simulation~\cite{duan2016zeroerror,fang2020channelsimulation}. We provide the proofs of both relations in Section~\ref{ssec:entropy_measures} of the supplementary material.

Finally, we exploit the connection between quantum metrology and hypothesis testing
derived in Ref.~\cite{walter_heisenberg_2014} to give an alternative lower bound on the optimal
minimax success probability $\overline\eta^*$.  
The lower bound is expressed in terms of the
hypothesis testing relative entropy~\cite{Hiai1991CMP_proper,Dupuis2013_DH,datta2013smoothmax,Wang2012PRL_oneshot,Tomamichel2013_hierarchy,khatri2020principles} 
\begin{align}
    D_{\mathrm{h}}^{\eta}(\rho \, \| \,  \sigma) \coloneqq - \log \beta_{\mathrm{h}}^{\eta}(\rho \, \| \,  \sigma) ,
\end{align}
where the asymmetric hypothesis testing error $\beta_{\mathrm{h}}^{\eta}$ is defined as
\begin{align}\label{eqn:def_asym_hyp_test_error}
    \beta_{\mathrm{h}}^{\eta}(\rho \fatpipe \sigma) \coloneqq \inf_{0 \leq M \leq \bbI} \{ \Tr[ M \sigma ] \pipe \Tr[M \rho] \geq \eta \}.  
\end{align}
We have the following proposition, reminiscent of a comparable result in Ref.~\cite{walter_heisenberg_2014}.
\begin{proposition}[Lower bound on estimation tolerance in terms of the hypothesis testing relative entropy]
\label{prop:lower_bound_on_delta_via_hypothesis_testing}
For a set of states $\rho(t)$ and any $0 \leq \eta \leq \overline{\eta}^{*}(\delta, \rho)$, the minimax tolerance satisfies 
\begin{align}
    \overline{\delta}(\eta, \rho) \geq \frac{1}{2}\int \diff t \, \exp\left( -D_{\mathrm{h}}^{\eta}(\rho(t) \, \| \,  \sigma)\right) 
\end{align}
for any state $\sigma$.
\end{proposition}
The proof exploits the SDP dual formulation of the hypothesis testing relative entropy and is given next to the more general statement of Theorem~\ref{sthm:window_width_lower_bound_integral_t_fixed_eta} of the supplementary material.
We provide a proof that is independent of Ref.~\cite{walter_heisenberg_2014} for completeness.

\section{
A Cramér-Rao-like bound in the finite-sample regime
}
\label{sec:optimal_tolerance}

In this section, we want to further our understanding of the metrological tolerance defined in Definition~\ref{def:metrological_tolerance} and its ultimate limits when optimizing over metrological protocols. Basic exploration of the definition of (optimal) tolerance was performed in Ref.~\cite{yang2018stopwatch}, where the authors proposed a dimension-dependent lower bound that allowed them to conclude that the tolerance can only decrease asymptotically as $O(1/n)$, which corresponds to Heisenberg scaling. Ref.~\cite{sugiyama_precision-guaranteed_2015} derived an upper bound in the case of i.i.d.\ repetitions with a fixed measurement. 

In the asymptotic approach to quantum metrology, the most important tool is the \emph{quantum Cramér-Rao bound}, which gives a lower bound on the standard deviation of any locally unbiased estimate, relating it to the inverse square root of the quantum Fisher information. This inherently geometric quantity measures how quickly the quantum states $\rho(t)$ change when the parameter $t$ is altered. As the first contribution of this section, we derive a lower bound that fulfills a similar role for the minimax tolerance.

We derive the bound from Corollary~\ref{corr:two_point_fidelity_bound}, which we rephrase in terms of the sandwiched Rényi relative entropies, which are defined in terms of a constant $\alpha \in (0,1) \cup (1, \infty)$ as~\cite{khatri2020principles}
\begin{align}\label{eqn:def_sandwiched_renyi_relative_entropy}
    \tilde{D}_{\alpha}(\rho \fatpipe \sigma) \coloneqq \frac{1}{\alpha - 1} \log \Tr\left[\left( \sigma^{\frac{1-\alpha}{2\alpha}} \rho \sigma^{\frac{1-\alpha}{2\alpha}} \right)^{\alpha}
    \right].
\end{align}
In our case, the value $\alpha = \frac{1}{2}$ is crucial, as it connects to the fidelity of quantum states defined in Eq.~\eqref{eqn:def_fidelity}:
\begin{align}\label{eqn:sandwiched_renyi_12_fidelity}
    \tilde{D}_{\frac{1}{2}}(\rho \fatpipe \sigma) &= -\frac{1}{2} \log \Tr[( \sigma^{{1}/{2}} \rho \sigma^{{1}/{2}})^{{1}/{2}}] \\
    &= -\frac{1}{2} \log F(\rho, \sigma)\nonumber.
\end{align}
With this notation at hand, Corollary~\ref{corr:two_point_fidelity_bound} can be rephrased as 
\begin{align}\label{eqn:corollary_10_as_log_fidelity}
\log \left( \frac{1}{4(1 - \overline{\eta}^{*}(\delta, \rho))}\right) \leq 4 \inf_{|t-t'|>2\delta} \tilde{D}_{\frac{1}{2}}(\rho(t), \rho(t')).
\end{align}

Before we come to the formal statement, let us outline how this bound comes about in the case of i.i.d.\ copies. We can make the choice $t' = t + 2\delta$ in the above bound. For sufficiently small $\delta$ we can perform a Taylor expansion. We denote the Taylor expansion as
\begin{align}\label{eqn:expansion_log_fidelity}
\begin{split}
    &\tilde{D}_{\frac{1}{2}}(\rho(t)\fatpipe \rho(t+\tau)) \\
    &\qquad= \frac{1}{2} f_2(t) \tau^2 + \frac{1}{3!} f_3(t) \tau^3+ \frac{1}{4!} f_4(t) \tau^4 + \dots,
\end{split}
\end{align}
where we envision $\tau = 2\delta$. The values $f_k(t)$ are the Taylor coefficients of the function $\tau \mapsto \tilde{D}_{\frac{1}{2}}(\rho(t)\fatpipe \rho(t+\tau))$ at $\tau = 0$, an explicit formula is given in Eq.~\eqref{eqn:def_f_k_taylor_coeffs}.
Because of the relationship between the sandwiched Rényi relative entropy and the fidelity, we have that the coefficient $f_2(t)$ is a constant multiple of the quantum Fisher information $f_2(t) = \frac{1}{8} \calF(t)$, see Section~\ref{ssec:tolerance_lower_bound_symmetric_ht} of the supplementary material.

Let us now drop the explicit time dependence. If we could ignore the higher order terms in the Taylor expansion, we could simply choose $2 \delta = \tau \propto 1/\sqrt{f_2}$ to render the error probability lower bound constant, which would give us the desired scaling of the bound. We will, however, need to work a bit harder to get something analytically meaningful.
Let us choose $\tau = \gamma/\sqrt{f_2}$ in the above expansion. Then, as long as $2 \delta = \tau$ is within the radius of convergence of the Taylor series, we have
\begin{align}\begin{split}
    &\tilde{D}_{\frac{1}{2}}(\rho(t)\fatpipe \rho(t+\tau)) \\
    &\qquad= \frac{1}{2} \gamma^2 + \frac{1}{3!} \frac{f_3}{f_2^{3/2}} \gamma^3+ \frac{1}{4!} \frac{f_4}{f_2^{2}} \gamma^4 + \dots\,
    .
\end{split}
\end{align}
We immediately observe that the validity of the second order approximation depends on the ratios ${f_p}/{f_2^{p/2}}$.
To get some intuition for these ratios, it is rather instructive to look at the case of i.i.d.\ 
copies, \textit{i.e.}, $\rho(t) \to \rho^{\otimes n}(t)$. In this case, the additivity of the sandwiched Rényi relative entropy implies that
\begin{align}
    \tilde{D}_{\frac{1}{2}}(\rho^{\otimes n}(t)\fatpipe \rho^{\otimes n}(t+\tau)) = n \tilde{D}_{\frac{1}{2}}(\rho(t), \rho(t+\tau)).
\end{align}
This in turn means that all coefficients $f_p$ grow linearly in $n$, \ie $f_p \to n f_p$ which implies that
\begin{align}
    \frac{f_p}{f_2^{p/2}} \to \frac{1}{n^{\frac{p-2}{2}}}\frac{f_p}{f_2^{p/2}},
\end{align}
meaning that higher order ratios decay more quickly, justifying the second order approximation in the limit of large $n$. Note that this is the same regime in which the asymptotic Cramér-Rao bound is attainable as well.

With the intuition we gained, we can now make sense of the following theorem.
\begin{theorem}[Non-asymptotic Cramér-Rao-like bound]\label{thm:cramer_rao_like_bound}
For a given smooth set of states $\rho(t)$, we define
\begin{align}\label{eqn:def_f_k_taylor_coeffs}
    f_k(t) \coloneqq -\frac{1}{2} \left.\frac{\partial^k}{\partial \tau^k} \log F(\rho(t), \rho(t + \tau))\right|_{\mathrlap{\tau=0}}
\end{align}
and the coefficient
\begin{align}
    q \coloneqq  \sup_{t} \sup_{3 \leq p \in \bbN} \left| \frac{f_p(t)}{f_2^{p/2}(t)} \right|^{\mathrlap{\frac{1}{p-2}}}.
\end{align}
We then have for any desired minimax success probability $\overline{\eta}>3/4$ that
\begin{align}
    \overline{\delta}^{*}(\overline{\eta}, \rho) \geq \frac{\Gamma}{\sqrt{\inf_t \calF(t)}},
\end{align}
where
\begin{align}
      \frac{q}{6\sqrt{2}} \log \frac{1}{4(1-\overline{\eta})} > \sqrt{\log \frac{1}{4(1-\overline{\eta})}} - \Gamma > 0.
\end{align}
The bound holds as long as $\Gamma/{\sqrt{\inf_t \calF(t)}}$ does not exceed the smallest convergence radius of the Taylor expansion of the sandwiched Rényi relative entropy.
\end{theorem}
We note that the coefficient $q$ is a measure of closeness to a Gaussian shape of the fidelity curve $\tau \mapsto F(\rho(t), \rho(t+\tau))$ and again emphasize that $f_2(t)$ relates to the quantum Fisher information as $f_2(t) = \frac{1}{8}\calF(t)$.
The proof is presented in Section~\ref{ssec:tolerance_lower_bound_symmetric_ht} of the supplementary material.
As we discussed above, in the case of i.i.d.\ copies $\rho(t) \to \rho^{\otimes n}(t)$, we have $q = O(1/\sqrt{n})$ giving a bound that is asymptotically constant. 

The above gives a lower bound for the optimal minimax tolerance that, asymptotically, has the expected scaling both in the quantum Fisher information and the desired logarithmic dependence on the success probability. It furthermore is valid in the single-shot setting, where the finite-size corrections depend on the higher-order derivatives of the sandwiched Rényi relative entropy. We find it conceivable that the ratio of third to second derivative could give the factual second-order asymptotics, but to conclude that it would be necessary to find a matching upper bound (\ie a protocol) with similar performance guarantees. 

It is also interesting to gather some intuition about the workings of the above bound. First of all, the right-hand-side involves the quantum Fisher information, which is a local quantity that captures how much states change \emph{infinitesimally}. 
This quantity yet puts a bound on $\overline\delta$, which quantifies the estimation accuracy globally over the range of possible parameter values. This is because we effectively reduce to a setting where the quantum Fisher information captures the dominant contributions to the sandwiched Rényi relative entropy even at non-infinitesimal distances -- the case when it is dominated by the second order expansion. The coefficient $q$ measures how close we are to this setting. 

The proof of the above theorem hinges on the convergence of the Taylor expansion of the sandwiched Rényi relative entropy. As the convergence radius is the distance to the closest pole of the function, we can conclude that the smallest radius of convergence is 
\begin{align}
    r \coloneqq \min \{ |\tau| \pipe F(\rho(t), \rho(t+\tau)) = 0 \text{ for some }t \}.
\end{align}
This shows that in the case of i.i.d.\ copies, for example, the radius of convergence is independent of the number of copies. We thus do not expect the radius of convergence to be an issue in practically relevant scenarios.

We further observe that the expansion we use is very reminiscent of the Edgeworth/Gram-Charlier series expansion method in statistics~\cite{wallace_asymptotic_1958}, where the ratio \smash{$f_p/f_2^{p/2}$} can be understood as the normalized $p$-th cumulant. It is an important direction of research to further our understanding of the higher derivatives of the sandwiched Rényi relative entropy. 

\section{Bounds on the optimal sample complexity}\label{sec:sample_complexity}

The sample complexity defined in Definition~\ref{def:sample_complexity} captures the number of copies of a quantum system needed to achieve a target tolerance with a guaranteed success probability. It is a quantity for which bounds follow in a relatively straightforward manner from our previous results. This is because any bound that relates the tolerance and the success probability to each other can be used to establish a sample complexity bound. In this way, the sample complexity is -- in a way -- mathematically secondary to success probability and tolerance.

We can use Corollary~\ref{corr:two_point_fidelity_bound} in the form of Eq.~\eqref{eqn:corollary_10_as_log_fidelity} to obtain the following concise minimax sample complexity lower bound that involves the sandwiched Rényi relative entropy of order $1/2$ introduced in Eq.~\eqref{eqn:sandwiched_renyi_12_fidelity}.
\begin{corollary}[Two-point sample complexity bound]\label{corr:two_point_fidelity_sample_complexity_bound}
For a given tolerance $\delta$ and state set $\rho(t)$, we have the following lower bound on the minimax sample complexity
\begin{align}
\overline{n}^{*}(\eta, \delta, \rho) \geq \frac{1}{4 \displaystyle\inf_{|t-t'| > 2\delta}\Tilde{D}_{\frac{1}{2}}(\rho(t)\fatpipe \rho(t'))}\log \frac{1}{4(1-\eta)}.
\end{align}
\end{corollary}
The above result concerns the sample complexity in the setting of parametrized states. When talking about sample complexities, we usually present them in Big-$O$ notation, where the relevant limits are $\delta \to 0$ and $\eta \to 1$. In the limit $\delta \to 0$, the above bound is dominated by close values $t' = t + 2\delta$, motivating a Taylor expansion. As we have discussed in detail in Section~\ref{sec:optimal_tolerance}, we have 
\begin{align}
    \tilde{D}_{\frac{1}{2}}(\rho(t) \fatpipe \rho(t+2\delta)) = \frac{1}{4}\calF(t) \delta^2 + O(\delta^3),
\end{align}
where $\calF(t)$ is the quantum Fisher information at time $t$.
This means Corollary~\ref{corr:two_point_fidelity_sample_complexity_bound} immediately implies an i.i.d.\ sample complexity lower bound of
\begin{align}
    \overline{n}^{*}(\eta, \delta, \rho) \geq O\left( \frac{1}{\min_t\calF(t)} \frac{1}{\delta^2} \log \frac{1}{1-\eta} \right).
\end{align}

The above result applies to i.i.d.\ states. It is, however, equally important to have a sample complexity bound that does not rely on the i.i.d.\ structure of the underlying state. This can, for example, happen when quantum metrology with quantum channels is performed. In that case, we deal with a parametrized family of states $\rho^{(n)}(t)$ for $n \in \bbN$.
To get results about this case from Corollary~\ref{corr:two_point_fidelity_bound}, we have to work harder. Luckily, we already performed the heavy lifting in Section~\ref{sec:optimal_tolerance} and can build on our non-asymptotic Cramér-Rao like bound of Theorem~\ref{thm:cramer_rao_like_bound}. Said result hinges on two quantities, first the smallest quantum Fisher information $\inf_t \calF(t)$ and second the coefficient $q$ that quantifies the validity of a quadratic approximation to the sandwiched Rényi relative entropy.

Many results in the existing literature on quantum metrology concern themselves with the asymptotic scaling of the quantum Fisher information, which is usually $\calF = O(n)$ in the standard quantum limit and $\calF = O(n^2)$ in the Heisenberg limit~\cite{kurdzialek2023using}. We can use results of this type in the form of the following Corollary:
\begin{corollary}[Sample complexity scaling bound]\label{corr:sample_complexity_scaling_bound}
In the setting of Theorem~\ref{thm:cramer_rao_like_bound}, where $\rho(t)$ is replaced with a parametrized family of states $\{ \rho^{(n)}(t)\}_{n\in \bbN}$, assume that we have the asymptotic scalings in $n$:
\begin{align}
    \inf_t \calF^{(n)}(t) &= O(n^{\alpha})\\
    q^{(n)} &= o(1),
\end{align}
for $\alpha > 0$. Then we have that
\begin{align}
    \overline{n}(\eta, \delta) \geq O\left( \left[ \frac{\log\frac{1}{1-\eta}}{\delta^{2}} \right]^{\frac{1}{\alpha}}\right).
\end{align}
\end{corollary}
The proof of the above Corollary is exhibited in Section~\ref{ssec:optimal_sample_complexity} of the supplementary material.
As was shown in Section~\ref{sec:optimal_tolerance}, in the case of i.i.d.\ copies, we have $\alpha = 1$ and $q = O(n^{-1/2}) = o(1)$, reproducing the scaling of Corollary~\ref{corr:two_point_fidelity_sample_complexity_bound} up to a worse dependence on the inverse error probability. 

\begin{figure*}
    \centering
    \includegraphics{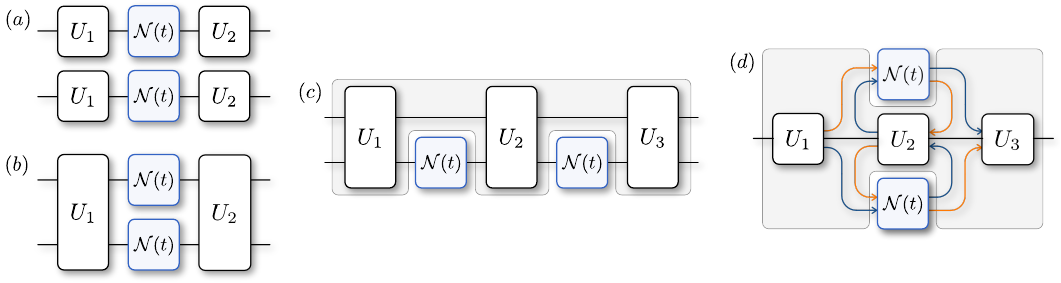}
    \caption{The different types of protocols that can be realized when accessing multiple (in this case, two) copies of the same parametrized channel $\calN(t)$. We distinguish: 
    $(a)$ the \emph{i.i.d.\ case} when the same single-shot protocol is repeated multiple times. 
    $(b)$ the \emph{parallel} case where the probe state can be entangled by a suitable unitary transformation. 
    $(c)$ the \emph{adaptive} case where an auxiliary system can be used as a memory to coherently adjust processing based on the outcomes of the first use of the channel. The gray shaded area represents an instance of a \emph{quantum comb}. 
    $(d)$ the case of \emph{indefinite causal order}, where a superposition of orders of invocations of the quantum channel can be used to boost the sensitivity. The blue and orange arrows represent the two orders of invocation. The solid black line in the middle indicates an auxiliary system. The gray shaded area represents a general strategy with indefinite causal order.}
    \label{fig:types_of_protocols_shadows_wide}
\end{figure*}

\section{
Metrology of quantum channels in the finite-sample regime
}
\label{sec:parametrized_quantum_channels}

So far in this work, we have considered the task of estimating a parameter $t$ encoded in a set of states $\rho(t)$ and we have presented results on the optimal success probability, tolerance and sample complexity for this task. In practice, as alluded to in Fig.~\ref{fig:front_figure}, parametrized quantum states arise from the interaction of some probe system with some physical system whose properties in the form of the parameter $t$ we wish to determine. 
In this case, the primary object of interest is not a set of parametrized states $t\mapsto\rho(t)$, but a set of parametrized \emph{quantum channels} $t\mapsto \mathcal{N}(t)$, which represent the evolution of the probe system and its properties. 
The goal is then to estimate the parameter $t$, given access to the quantum channel $\mathcal{N}(t)$, with the same goals as before: optimize the estimation success probability, the tolerance, and the sample complexity, \textit{i.e.}\ the number of times the channel is accessed. %

\emph{Access models.}
As quantum channels represent evolutions of quantum systems, the possible ways of interacting with multiple copies of them to extract the underlying parameter are much richer than in the case of a parametrized quantum state where the only way of interacting is to choose a suitable POVM. As we show in Fig.~\ref{fig:types_of_protocols_shadows_wide}, a variety of \emph{access models} for multiple copies of quantum channels can be distinguished. 
In the simplest case, the same single-shot protocol is repeated multiple times and the individual outcomes are processed classically, effectively reducing to the case of the parametrized quantum state -- we therefore refer to it as the \emph{i.i.d.\ case}. In the \emph{parallel} case, we use $n$ copies of the channel in parallel, but possibly with entangled inputs and measurements. If \emph{adaptive} processing is allowed, then every use of the channel can be followed by a round of adaptive quantum processing before the next channel use. Such causally-ordered strategies are modeled by \emph{quantum combs}~\cite{GW07,CDP09}. Finally, the most general conceivable access to $n$ copies of a quantum channel allows for \emph{indefinite causal order} of the channel uses, \emph{e.g.}\ through the use of a \emph{quantum switch}~\cite{chiribella2013quantum}. Such strategies can lead to an asymptotic quadratic advantage over the Heisenberg limit for infinite-dimensional systems~\cite{zhao_quantum_2020,kurdzialek2023using}, and there exists a strict performance hierarchy in the finite-dimensional case~\cite{liu2023metrologyhierarchy}.

If we fix a particular way of interacting with $n$ copies of a parametrized quantum channel in a particular access model and subsequently measuring a POVM that predicts the underlying parameter, we will refer to this as a \emph{strategy}. Luckily, we can give a formal description of access models and strategies via the Choi representation formalism. In this framework, every \emph{strategy} within an access model for $n$ copies of the channel, $\sfS_n$, is a function $\tau \mapsto P_n(\tau)$, $P_n \in \mathsf{S}_n$, that maps possible predictions to positive semi-definite operators.
With this, we can extend Definition~\ref{def:success_probability_no_optimization} of the success probability as follows. 

\begin{definition}[Success probability (channels)]\label{def:success_probability_strategies}
For a given tolerance $\delta$, a set of channels $\calN(t)$ of which we can access $n$ copies, possibly with prior $\mu(t)$, and a strategy $P_n(\tau)$, the Bayesian success probability is given by
\begin{align}
\begin{split}
&\eta(\delta, \mu, \calN, P_n)   \\
&\qquad \coloneqq \int \diff \mu(t) \, \diff \tau \, w_{\delta}( t - \tau) \Tr [C[\calN(t)]^{\otimes n} P_n(\tau)],
\end{split}
\end{align}
where $C[\calN(t)]$ is the Choi representation of the quantum channel $\calN(t)$.
The minimax success probability is given by
\begin{align}
\begin{split}
&\overline{\eta}(\delta, \calN, P_n)  \\
&\qquad\coloneqq \inf_t \int \diff \tau \, w_{\delta}( t - \tau) \Tr [C[\calN(t)]^{\otimes n} P_n(\tau)].  
\end{split}     
\end{align}
\end{definition}
The Bayesian and minimax tolerance sand sample complexities in the channel case are then defined similar to Definitions~\ref{def:metrological_tolerance} and~\ref{def:sample_complexity} from the above defined Bayesian and minimax success probabilities.

The natural next step is now to define the \emph{optimal} (Bayesian) success probability relative to an access model as
\begin{align}
    \eta^{*}(\delta, \mu, \calN, \sfS_n) \coloneqq \sup_{P_n \in \sfS_n} \eta(\delta, \mu, \calN, P_n).
\end{align}
The optimal Bayesian tolerance and sample complexity, as well as the corresponding minimax quantities are defined analogously. 

In Section~\ref{sec:optimal_metrology_protocols}, we have shown that the optimal Bayesian and minimax success probabilities can be computed by solving a convex optimization problem without duality gap. As evidenced in Fig.~\ref{fig:types_of_protocols_shadows_wide}, protocols that involve a parametrized quantum channel have more moving parts that can and need to be optimized. Let us, for example, take the simplest case of only one use of the parametrized quantum channel. 
In this case we have to optimize over both the probe state $\rho_0$ that is fed into the quantum channel and the measurement $Q(\tau)$. Naively, the objective is then a nonlinear function of the arguments of the optimization $\rho_0$ and $Q(\tau)$ and we would not expect that this can be cast as a semi-infinite program. However, if we change our perspective and combine the preparation of the probe state and the measurement into a single object represented by a parametrized quantum comb, we can exploit the convexity of the set of quantum combs to again cast the computation of the optimal success probability as a semi-definite program. This reasoning then immediately means that also the case of adaptively interacting with $n$ copies of the channel can be efficiently optimized for. The same is true for strategies involving indefinite causal order as we summarize in the following proposition.

\begin{proposition}[Joint optimization]
For a given tolerance $\delta$, a set of channels $\calN(t)$ of which we can access $n$ copies, possibly with prior $\mu(t)$, the optimal success probability $\eta^*(\delta,\mu,\calN, \sfS_n)$ and optimal minimax success probability $\eta^{*}(\delta, \calN, \sfS_n)$ can be computed using a semi-definite program without duality gap for both adaptive $\sfS_n = \sfS_n^{\mathrm{ada}}$ and indefinitely causally ordered $\sfS_n = \sfS_n^{\mathrm{ico}}$ access.
\end{proposition}
The detailed statements of the convex programs and their duals are given in Section~\ref{ssec:optimization} of the supplementary material. 

\emph{Connection to hypothesis testing.}
The core contribution of this work is the rigorous connection of PAC metrology with quantum hypothesis testing given in Section~\ref{sec:metrology_as_hypothesis_testing}. There, we gave an upper bound on the success probability through the success probability of corresponding multi-hypothesis testing problems between quantum states. As we show below, these results also carry over to the case of parametrized quantum channels. In this case, the reduction is to multi-hypothesis testing between channels under different access models. Formally, we define the optimal Bayesian success probability of testing $n$ copies of the quantum channels $\{ \calN_i \}_{i=1}^N$ with prior probabilities $\{ p_i \}_{i=1}^N$ under the access model $\sfS_n$ as
\begin{align}
    P_s^{*}(\{ p_i \calN_i \}_{i=1}^N, \sfS_n) \coloneqq
    \sup_{ \{P_i \}_{i=1}^N \subset \sfS_n } \sum_{i=1}^N p_i \Tr [ C[\calN_i]^{\otimes n} P_i ]
\end{align}
under the condition that the set $\{ P_i \}_{i=1}^N$ corresponds to a valid combination of processing and POVM. The minimax statement is given analogously by choosing the prior probabilities adversarially for the chosen strategy. The fact that our results from the state case carry over to this more general case comes with little surprise when we realize that upon fixing the strategy $P(\tau)$ to the optimal strategy and executing it right until before the final measurement is performed, we obtain a parametrized set of states to which we can then apply Theorem~\ref{thm:succ_prob_upper_bound_mht_delta_window}, leading to the following result.
\begin{corollary}[Hypothesis testing bound (channels)]\label{corr:succ_prob_upper_bound_mht_access_strategies}
For a given tolerance $\delta$, fix any set $\calS = \{(\lambda, s)\}$ of prior probabilities $\lambda \in [0,1]$ and shifts $s \in \bbR$ such that for all distinct $s, s' \in \calS$ we have that $|s-s'| > 2\delta$ and $\sum_{\lambda \in \calS} \lambda = 1$. Then, for a set of channels $\calN(t)$ of which we can access $n$ copies with prior $\mu(t)$ and a fixed access model $S_n$ we have the upper bound
\begin{align}
\begin{split}
    &\eta^{*}(\delta, \mu, \calN, \sfS_n) 
    \\ &\qquad\leq \int \diff t \, P_s^{*}(\{ \lambda \, \mu(t+s) \calN(t+s)\}_{(\lambda, s) \in \calS}, \sfS_n).
\end{split}
\end{align}
Optimizing over the prior probabilities $\lambda$ then yields the analogous upper bound in the minimax setting
\begin{align}
    \overline{\eta}^{*}(\delta, \calN, \sfS_n) &\leq \inf_t \overline{P}_s^{*}(\{  \calN(t+s)\}_{s \in \calS}, \sfS_n).
\end{align}
\end{corollary}

Because of the richer structure embodied by different access models for multiple copies of the same parametrized channel $\calN(t)$, we can define multiple types of asymptotic rates to generalize the analysis carried out in Section~\ref{sec:asymptotics_succ_prob}. Of particular interest to us are the rates corresponding to i.i.d.\ strategies ${R}^{*}_{\mathrm{iid}}$ (item $(a)$ in Fig.~\ref{fig:types_of_protocols_shadows_wide}), because there we can make use of our results on the asymptotics of the state case, and the rates corresponding to parallel strategies ${R}^{*}_{\mathrm{par}}$ (item $(b)$ in Fig.~\ref{fig:types_of_protocols_shadows_wide}).

\section{Phase estimation of a pure state Hamiltonian evolution}
\label{sec:covariant_setting}

In this section, we analyze one of the most prototypical scenarios of quantum metrology, namely phase estimation with pure states. We analyze the minimax success probability, as it represents the most stringent achievable guarantees. It is further important to emphasize that our analysis takes the perspective of \emph{global} estimation, contrary to the \emph{local} estimation routinely seen in the literature.

As a first step, we establish a general result on the minimax success probability in the $U(1)$-\emph{group-covariant} setting, which applies beyond phase estimation. We consider a set of states $|\psi(t)\rangle$ generated by unitary evolution of a pure initial probe state $|\psi \rangle$ under a Hamiltonian $H$, reflecting  the evolution of a closed quantum system
\begin{align}
    |\psi(t)\rangle &= e^{-i t H} |\psi\rangle = U(t)\ket{\psi}.
\end{align}
To ensure that $t$ can be understood as a \enquote{phase}, $H$ is assumed to be such that all differences between eigenvalues are integer-valued, in which case the recurrence time of the Hamiltonian is guaranteed to be $2\pi$. 
Let now $H$ decompose as $H = \sum_{\lambda} \lambda \Pi_{\lambda}$, where $\lambda$ are the different eigenvalues and $\Pi_{\lambda}$ are the projectors onto the possibly degenerate eigenspaces. Then, we can expand
\begin{align}
    \ket{\psi} = \sum_{\lambda} \psi_{\lambda} \ket{\psi_{\lambda}},
\end{align}
where we defined the normalized projections of $\ket{\psi}$ onto the eigenspaces of $H$ such that $\Pi_{\lambda}\ket{\psi} = \psi_{\lambda}\ket{\psi_{\lambda}}$. 

For such a covariant set of states $\psi(t) = |\psi(t) \rangle\!\langle \psi(t)|$, the following theorem establishes that the \emph{pretty good measurement}~\cite{Bel75,belavkin1975PGM,Hol78,belavkin1988design,HW94,hughston1993PGM} is minimax optimal, and we obtain a closed form solution for the optimal minimax success probability. 
This result is well in line with known results on the optimality of the pretty good measurement in other covariant state discrimination and parameter estimation tasks~\cite{holevo_covariant_1979,ban1997optimum,hayashi1998asymptotic,chiribella_extremal_2004,holevo2011probabilistic,chiribella2004covariant,chiribella_extremal_2006}. 
\begin{theorem}[Minimax optimal measurement]\label{thm:pgm_is_minimax_optimal}
For a state set $\psi(t)$ given by a pure initial probe state $\rho_0 = |\psi\rangle\!\langle \psi |$ evolving under a Hamiltonian with integer eigenvalue differences for time $t \in [0, 2\pi]$, the pretty good measurement
\begin{align}
    Q_{\mathrm{PGM}}(t) &= R^{-1/2} \psi(t) R^{-1/2}, \text{ where } R = \int \diff t \, \psi(t),
\end{align}
achieves the optimal minimax success probability, equal to
\begin{align}
    \overline{\eta}^*(\delta, \psi) &= \sum_{\smash{\lambda, \lambda'}} |\psi_{\lambda}| |\psi_{\lambda'}| \hatw_{\delta}({\lambda - \lambda'}),
\end{align}
where \smash{$\hat{w}_{\delta}({\omega})  = \sin(\delta \omega) / (\pi \omega)$} is the Fourier transform of the rectangular window $w_{\delta}$ at frequency $\omega$.
\end{theorem}
The above theorem establishes a direct relation between the amplitudes of the probe state $\ket{\psi}$ and 
the minimax success probability.
It especially shows that only the spectrum of the Hamiltonian and the \emph{absolute values} of the amplitudes matter. The result holds for any window function. The proof exploits strong duality and complementary slackness to establish a formula for the optimal dual variable and is presented in Section~\ref{ssec:pure_covariant_evolution} of the supplementary material.

One important consequence of Theorem~\ref{thm:pgm_is_minimax_optimal} is that it greatly simplifies the search for a minimax optimal probe state. If we arrange the absolute values of the amplitudes in a vector $\ppsi \coloneqq ( |\psi_{\lambda}| )_{\lambda}$ and construct the matrix associated to the Fourier transform, $W_{\lambda ,  \lambda'} \coloneqq \hat{w}_{\delta}(\lambda - \lambda')$, then the optimal minimax success probability is given by the quadratic form
\begin{align}
    \overline{\eta}^*(\delta, \psi) &= \langle \ppsi, W \ppsi \rangle.
\end{align}
The optimal probe state is hence obtained by solving the following optimization problem:
\begin{align}\label{eqn:opt_probe_state}
    \ppsi^{*} \coloneqq \argmax_{\ppsi} \{ \langle \ppsi, W \ppsi \rangle \pipe \psi_{\lambda} \geq 0 \text{ for all } \lambda, \lVert \ppsi \rVert_2 = 1 \}.
\end{align}
Because of the positivity constraint on the entries of the vector, this is in general an NP-hard optimization problem~\cite{murty1987np-complete}.

\section{Phase estimation with an ensemble of \texorpdfstring{spin-$\tfrac{1}{2}$}{spin-1/2} particles}
\label{sec:minimax_analysis_phase_estimation}   
We now turn our attention to the special phase estimation on a spin chain~\cite{demkowicz-dobrzanski_optimal_2011}. 
This is a covariant problem in the above sense with the single-spin Hamiltonian given by $H = \operatorname{diag}(0,1)$. If we have $n$ spins separately evolving under this Hamiltonian, the effective Hamiltonian is given by summing up the local terms on the individual copies
\begin{align}
    H_n = \sum_{i=1}^n H_i, \ \ H_i = \bbI^{\otimes i-1} \otimes H \otimes \bbI^{\otimes n- i}.
\end{align}
As we have seen in Theorem~\ref{thm:pgm_is_minimax_optimal}, only the spectrum of the Hamiltonian matters. For the Hamiltonian $H_n$, it is given by $\operatorname{spec}(H_n) = \{ 0, 1, \dots, n\}$, which grows linearly in $n$. We can therefore treat the equivalent problem of a Hamiltonian with spectral decomposition $H = \sum_{k = 0}^n k |k \rangle\!\langle k|$, where the eigenstates $\ket{k}$ are understood to be any eigenstate of the Hamiltonian $H_n$ with energy $k$, \textit{e.g.},   $\ket{3}$ could be $\ket{001101}$ for $n=6$. 

Applying Theorem~\ref{thm:pgm_is_minimax_optimal} allows us to compute the optimal minimax success probability (Fig.~\ref{fig:error_probability_phase_estimation_plot}) and the optimal minimax tolerance (Fig.~\ref{fig:tolerance_phase_estimation_publication_loglog}) for different kinds of probe states. We note that our analysis of the asymptotics of the minimax success probability has significant overlap with prior work by Imai and Hayashi in Ref.~\cite{imai_fourier_2009}. They discuss the asymptotic distribution of phase estimates and discuss the asymptotic rate.

\begin{figure}
    \centering
    \includegraphics{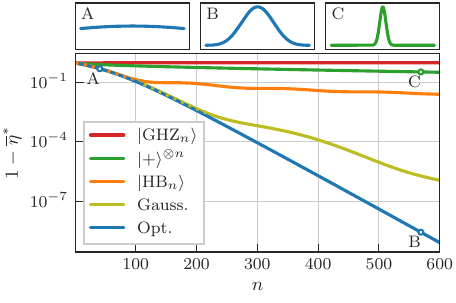}
    \caption{Optimal error probability of quantum metrology that can be guaranteed for any prior distribution for different probes in the phase estimation scenario for $\overline\delta = 0.04$. We compare a generalized GHZ state (Eq.~\eqref{eqn:def_ghz_n}, red), a tensor power of plus states (green), the Holland-Burnett state (Eq.~\eqref{eqn:def_HB_n}, orange), the Gaussian state (Eq.~\eqref{eqn:def_gauss_probe}, yellow) and the optimal state (Eq.~\eqref{eqn:opt_probe_state}, blue). The generalized GHZ state never performs well because it cannot resolve the time globally. For $n \ll 1/\delta$, the Holland-Burnett state performs almost optimally but has comparable asymptotic performance to the tensor power probe. The Gaussian probe performs almost optimally in a larger regime than the Holland-Burnett and has an intermediary asymptotic. The overset plots show the energy profiles of the optimal probe for $n = 41$ (A), the optimal probe for $n = 561$ (B) and the tensor power probe for $n=561$ (C). The asymptotic rate in the i.i.d. case is consistent with $R \approx \delta^2$, whereas the optimal entangled rate is consistent with $R \approx \delta$. We present additional numerical results for different values of $\overline\delta$ in Section~\ref{ssec:additional_numerics} of the supplementary material.}
    \label{fig:error_probability_phase_estimation_plot}
\end{figure}

\begin{figure}
    \centering
    \includegraphics{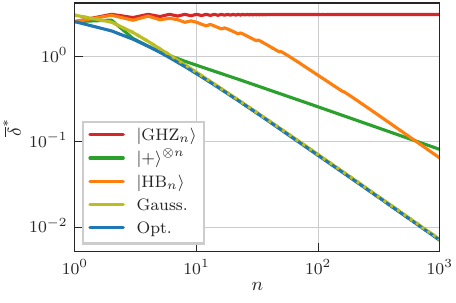}
    \caption{Optimal tolerance of quantum metrology that can be guaranteed for any prior distribution for different probes in the phase estimation scenario for fixed success probability $\overline\eta = 0.99$. We compare a generalized GHZ state (Eq.~\eqref{eqn:def_ghz_n}, red), an tensor power of plus states (green), the Holland-Burnett state (Eq.~\eqref{eqn:def_HB_n}, orange), the Gaussian state (Eq.~\eqref{eqn:def_gauss_probe}, yellow) and the optimal state (Eq.~\eqref{eqn:opt_probe_state}, blue). 
    The generalized GHZ state never performs well because it cannot resolve the time globally. The tensor power of plus states performs adequately for small $n$, but only achieves standard quantum limited scaling $O(1/\sqrt{n})$ asymptotically.
    For small $n$, the Holland-Burnett state does not perform satisfactorily, but achieves the optimal Heisenberg scaling $O(1/n)$ in the asymptotic limit. The Gaussian probe performs almost optimally except for very small $n$. We present additional numerical results for different values of $\eta$ in Section~\ref{ssec:additional_numerics} of the supplementary material.}
    \label{fig:tolerance_phase_estimation_publication_loglog}
\end{figure}

Our first and most obvious candidate for a probe state is a generalized \emph{Greenberger-Horne-Zeilinger (GHZ)} state
\begin{align}\label{eqn:def_ghz_n}
    \ket{\mathrm{GHZ}_n} := \frac{1}{\sqrt{2}}( \ket{0} + \ket{n} ),
\end{align}
which is optimal in the standard approach to quantum metrology~\cite{PhysRevA.46.R6797,PhysRevA.50.67,huelga_improvement_1997}. 
However, in the minimax setting, it fails spectacularly -- with a minimax success probability amounting to random guessing and a similarly high tolerance (see Figs.~\ref{fig:error_probability_phase_estimation_plot} and~\ref{fig:tolerance_phase_estimation_publication_loglog}). 
This is an immediate consequence of the fact that the standard approach to quantum metrology is concerned with \emph{local} estimation. In our case, however, the probe needs to be able to perform well in a task of \emph{global} estimation, \textit{i.e.}, the probe state should allow us to discern values in the whole interval $[0, 2\pi]$. The recurrence time of $\sim 2\pi/n$ of the generalized GHZ state, means it can very well resolve small differences in values, but not larger ones. A clock can serve as a good metaphor for this phenomenon: If we want to tell the time, we need to make use of the hour, minute and second hand. In this picture, the generalized GHZ state corresponds to a clock with only a second hand -- which is very suitable if you want to time a short sprint but useless when telling the time of the day.

The analogy of a clock inspires the use of another state, namely the \emph{Holland-Burnett (HB)} state
\begin{align}\label{eqn:def_HB_n}
    \ket{\mathrm{HB}_n} := \frac{1}{\sqrt{n+1}} \sum_{k=0}^n \ket{k},
\end{align}
which consists of an equal superposition of all energy eigenstates. Metaphorically, this state uses all the available hands of the clock equally. 
The Holland-Burnett state indeed has a much more desirable performance. As shown in Fig.~\ref{fig:error_probability_phase_estimation_plot}, this probe state achieves almost optimal success probability in the regime where $n \ll 1/\delta$. This behavior can be explained by expanding the Fourier transform of the window function, $\hat{w}_{\delta}(\omega) = \frac{\delta}{\pi}\operatorname{sinc}(\delta \omega)$, around $\omega = 0$, because the largest frequency scales as $O(n)$:
\begin{align}
    \hat{w}_{\delta}(\omega) = \frac{\delta}{\pi} \left[ 1 - O(\delta^2 \omega^2)\right].
\end{align}
The zeroth order contribution to this term is given by
\begin{align}
    \overline{\eta}^{*}(\delta, \psi) &= \frac{\delta}{\pi} \sum_{\lambda, \lambda'} |\psi_{\lambda}| |\psi_{\lambda'}| + O(\delta^2 \omega^2)\\
    &= \frac{\delta}{\pi} \lVert \ppsi \rVert_1^2 + O(\delta^2 \omega^2).\nonumber
\end{align}
In this limit, a probe that maximizes the one-norm of the amplitude vector is clearly optimal, which corresponds to the Holland-Burnett state. This state was also identified in recent work as a suitable probe state for (multi-)phase estimation~\cite{chesi_protocol_2023}.
In the case of the optimal tolerance, see Fig.~\ref{fig:tolerance_phase_estimation_publication_loglog}, we observe the inverse of this behavior: For small $n$, the performance is not satisfactory. Asymptotically, however, the Holland-Burnett state achieves the same Heisenberg scaling $\overline{\delta}^{*} = O(1/n)$ as the optimal probe. As we show in additional numerics presented in Section~\ref{ssec:additional_numerics} of the supplementary material, the critical value of $n$ at which the Holland-Burnett state starts to enter the Heisenberg-scaling regime increases with increasing success probability.

Next, we have analyzed the performance of a separable probe state. We chose i.i.d.\ copies of the optimal single-spin probe state, the $\ket{+} = (\ket{0} + \ket{1})/\sqrt{2}$ state. We observe in Fig.~\ref{fig:error_probability_phase_estimation_plot} that the success probability achieved with this state reaches towards unity much more slowly than the optimal probe state, with an asymptotic rate quadratically smaller. When looking at the optimal tolerance in Fig.~\ref{fig:tolerance_phase_estimation_publication_loglog}, we observe the expected asymptotic scaling of the standard quantum limit $\overline{\delta}^{*} = O(1/\sqrt{n})$.  

We additionally compare the aforementioned probe states with a \emph{Gaussian probe} whose amplitudes have a Gaussian shape: 
\begin{align}\label{eqn:def_gauss_probe}
\psi_{\lambda} \propto \exp\left( -\frac{1}{2}\frac{2\delta}{n+1}\left( \lambda - \frac{n}{2}  \right)^2\right).
\end{align}
The optimality of the choice of the standard deviation, $\sqrt{(n+1)/2\delta}$, is discussed in Section~\ref{ssec:minimax_analysis_of_phase_estimation} of the supplementary material. Regarding the success probability, we observe in Fig.~\ref{fig:error_probability_phase_estimation_plot} that the Gaussian probe performs close to optimally in a larger regime than the Holland-Burnett state and achieves better asymptotics, but also does not match the optimal probe. In the case of the tolerance, we do, however, observe in Fig.~\ref{fig:tolerance_phase_estimation_publication_loglog} that it nearly reproduces the optimal probe state. As we show in additional numerics presented in Section~\ref{ssec:additional_numerics} of the supplementary material, a gap in tolerance opens between the optimal and the Gaussian probe when the target success probability is increased, but the Gaussian probe preserves Heisenberg scaling and a good performance. 

Last but not least, we study the optimal probe state. Normally, we would need to solve the optimization problem of Eq.~\eqref{eqn:opt_probe_state}. However, in the special case we encounter here, in which $W$ is defined through the Fourier transform of a rectangular window function and the eigenvalue spectrum has no gaps, we can build on prior work studying a similar problem in the context of classical signal processing~\cite{slepian_prolate_1978}, where $W$ is referred to as the \emph{prolate matrix}. In Ref.~\cite{slepian_prolate_1978}, Slepian establishes that the largest eigenvector of the matrix $W$ is given by the so-called \emph{discrete prolate spheroidal sequence (DPSS)} of zeroth order. While he studies the problem without the positivity constraint on the eigenvector, we can build on a different result of Slepian to show that the largest eigenvector is always non-negative. In the case of phase estimation, we can therefore compute the optimal probe by finding the eigenvector associated to the largest eigenvalue of $W$. 

In our numerical investigations, we observe in Fig.~\ref{fig:error_probability_phase_estimation_plot} that the success probability tends towards unity with an asymptotic rate quadratically greater than what is possible with the separable probe and also outperforms the Gaussian probe significantly. It further achieves a clear Heisenberg scaling $\overline\delta^{*} = O(1/n)$ for the tolerance as evident in Fig.~\ref{fig:tolerance_phase_estimation_publication_loglog}. To get a feeling for the amplitude distributions of the different probes, we plot in Fig.~\ref{fig:error_probability_phase_estimation_plot} the amplitude distribution over the eigenvalues of the optimal probe states for $n = 41$ and $n = 561$. These plots clearly show that the optimal probe state for small $n$ has a flat spectrum, whereas asymptotically a moderately concentrated shape is optimal. We compare this with the corresponding shape of the tensor power probe at $n = 561$, which is much more concentrated, explaining its inferior performance.

Our numerical results make it quite clear that the setting of optimizing the success probability for a fixed tolerance and of optimizing the tolerance for a fixed success probability are qualitatively different. We especially see that probes that perform well in one setting do not necessarily perform well in the other. 

\begin{figure}
    \centering
    \includegraphics{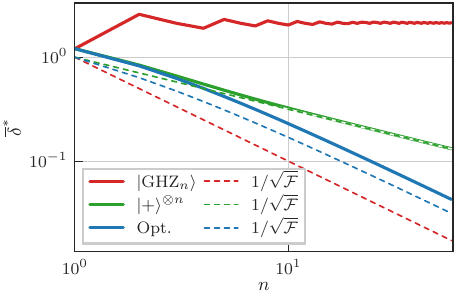}
    \caption{Optimal tolerance of quantum metrology that can be guaranteed for any prior distribution for different probes in the phase estimation scenario compared to the quantum Cramér-Rao bound embodied by the inverse square root of the quantum Fisher information. The success probability is fixed to $\overline\eta = \operatorname{erf}(1/\sqrt{2}) \approx 0.6827$, representing the probability that the value of a normally distributed random variable is within one standard deviation of the mean.
    We compare a generalized GHZ state (Eq.~\eqref{eqn:def_ghz_n}, red), an tensor power of plus states (green), and the optimal state (Eq.~\eqref{eqn:opt_probe_state}, blue). The quantum Cramér-Rao bound is shown in dashed lines.
    We observe that the quantum Cramér-Rao bound faithfully predicts the achievable precision in the case of i.i.d.\ copies for the chosen success probability, underpinning the interpretation that the quantum Cramér-Rao quantifies the i.i.d.\ case with a Gaussian shape for the fidelity curve. Both for the GHZ state and the optimal state, the quantum Cramér-Rao bound is way too optimistic, showcasing a decoupling of the quantum Cramér-Rao bound from the achievable performance in the few-shot and entangled regime.}
    \label{fig:tolerance_phase_estimation_publication_loglog_with_qfi}
\end{figure}

Finally, we also want to shine a light on the relation of the minimax tolerance with the quantum Cramér-Rao bound. To this end, in Fig.~\ref{fig:tolerance_phase_estimation_publication_loglog_with_qfi}, we plot the achievable tolerance for a subset of the states presented above together with the quantum Cramér-Rao bound for a fixed success probability of $\overline\eta = \operatorname{erf}(1/\sqrt{2}) \approx 0.6827$, representing the probability that the value of a normally distributed random variable is within one standard deviation of its mean. We observe that the quantum Cramér-Rao bound only faithfully predicts the achievable precision in the case of i.i.d.\ copies, but is overly optimistic otherwise. This underscores the interpretation that the quantum Cramér-Rao bound necessitates a degree of regularity of the underlying problem to be tight, and its connection to maximum-likelihood estimation which is optimal for Gaussian distributions. The sometimes overly optimistic estimate obtained from the quantum Cramér-Rao bound is especially evident for the GHZ state, which achieves the maximum quantum Fisher information, but the actual achievable tolerance is very bad. Nevertheless, these numerical experiments also suggest a positive result about the quantum Cramér-Rao bound. We observe that there exist settings where it gives a good measure of metrological precision, even in the global and non-asymptotic regime -- this is indicated by the fact that the agreement with the achievable precision is already very good at $n\approx 10$ repetitions. These results highlight that there are regimes where the quantum Cramér-Rao bound faithfully predicts the achievable precision, but that it is too optimistic in the case of few shots and entangled strategies.

\emph{Analytical results.}
The numerical observations presented in Fig.~\ref{fig:error_probability_phase_estimation_plot} motivate an analytical study of the asymptotics of the minimax error probability, extending the results of Section~\ref{sec:metrology_as_hypothesis_testing}. There, we established results on the asymptotic rate for i.i.d.\ copies of the same state, which corresponds to the case of the tensor power probe discussed above. We observe that when entangled probe states are allowed -- corresponding to the parallel setting -- the asymptotic rate is much improved, as the rates we observe numerically are consistent with
\begin{align}
    \overline{R}^{*}_{\mathrm{iid}}(\overline\delta) &\approx \overline\delta^2, \\
    \overline{R}^{*}_{\mathrm{par}}(\overline\delta) &\approx \overline\delta.
\end{align}
In the following, we make this observation rigorous.

As we already argued above, the optimal probe for phase estimation on a spin chain is given by the \emph{discrete prolate spheroidal sequence (DPSS)} of zeroth order~\cite{slepian_prolate_1978}. We can combine two results by Slepian to obtain the following result on the optimal rate in the parallel case: 
\begin{theorem}[Optimal minimax rate]\label{thm:opt_minimax_rate_entangled_phase_est}
    For a given minimax tolerance $0 < \overline\delta < \pi/2$, the parallel minimax error rate is given by
    \begin{align*}
        \overline{R}^{*}_{\mathrm{par}}(\overline\delta) &= \log \left( \frac{1 + \sin \frac{\overline\delta}{2}}{1 - \sin \frac{\overline\delta}{2}} \right) \\
        &= \overline\delta + O(\overline\delta^3).
    \end{align*}
\end{theorem}
The proof combines two results from Ref.~\cite{slepian_prolate_1978} and exploits the Perron-Frobenius theorem to establish positivity of the DPSS of zeroth order and is given in Section~\ref{ssec:minimax_analysis_of_phase_estimation} of the supplementary material.
While the DPSS has no closed-form, there exist efficient approximations involving the modified Bessel function of the first kind of zeroth order $I_0$~\cite{walden_accurate_1989}. In our case, the optimal probe can thus be approximated by choosing 
\begin{align}
    \psi_{\lambda} \propto I_0\left(\frac{\overline\delta n}{2} \sqrt{1 - \left(\frac{2 \lambda + 1}{n + 1} - 1\right)^2}\right)
\end{align}
and subsequently normalizing. 

We can use the result on the rate for i.i.d.\ probes of Theorem~\ref{thm:rate_upper_bound_delta} to calculate
\begin{align}
    \overline{R}^{*}_{\mathrm{iid}}(\overline\delta) \leq - \log \cos^2(\overline\delta) \approx \overline\delta^2
\end{align}
for small $\delta$, see Theorem~\ref{sthm:upper_bound_iid_rate_phase_estimation} of the supplementary material. Combined with the above Theorem~\ref{thm:opt_minimax_rate_entangled_phase_est} this implies that entangled strategies have a quadratic advantage in the asymptotic minimax rate, which can be understood as the rate analogue of the dichotomy between the standard quantum and Heisenberg limits.%

We further perform a theoretical analysis of the optimal standard deviation for the Gaussian probe. We can use tail bound estimates for the Gaussian distribution to prove that the optimal choice of standard deviation (see Eq.~\eqref{eqn:def_gauss_probe}) achieves half the optimal rate.
\begin{theorem}[Minimax rate for Gaussian probes]\label{thm:opt_minimax_rate_gaussian_probe_phase_est}
For a given minimax tolerance $\overline\delta > 0$, the Gaussian probe achieves the minimax error rate of
\begin{align}
    \overline{R}_{\mathrm{Gauss}}(\overline\delta) = \frac{\overline\delta}{2}.
\end{align}
\end{theorem}
The proof is presented in Section~\ref{ssec:minimax_analysis_of_phase_estimation} of the supplementary material.

Building on the previous result, we can also give a guarantee on the asymptotic tolerance achieved by the Gaussian probe.
\begin{observation}[Asymptotic tolerance of Gaussian probe]\label{obs:asymtptotic_tolerance_gauss}
For a given minimax success probability $\overline{\eta}$, the Gaussian probe achieves a minimax tolerance of 
\begin{align}
    \overline{\delta}_{\mathrm{Gauss}} = \frac{\alpha}{n+1},
\end{align}
where 
\begin{align}
\alpha \approx 2 \log \left( \frac{2}{\pi (1-\overline{\eta})} \right),
\end{align}
up to logarithmic factors.
\end{observation}
The argument is likewise presented in Section~\ref{ssec:minimax_analysis_of_phase_estimation} of the supplementary material.

\section{Extensions and connections to other fields}\label{sec:learning_from_quantum_systems}
This section is dedicated to exploring the various connections and possible generalizations of our definitions and results to other areas of quantum metrology and quantum information theory.

\subsection{Multi-parameter quantum metrology}\label{sec:multi_parameter_metrology}
Measuring multiple parameters at the same time
~\cite{liu_quantum_2020,albarelli_evaluating_2019} 
creates additional challenges, such as having to reconcile measurements~\cite{belliardo_incompatibility_2021,lu_incorporating_2021} and probe states~\cite{albarelli_probe_2022} that are optimal for each parameter but might be incompatible.
It furthermore offers a framework to study networks of quantum sensors~\cite{proctor_multi-parameter_2018}
and the optimal estimation of functions of multiple parameters~\cite{qian_heisenberg-scaling_2019}.

We can capture arbitrary instances of multivariate quantum metrology by replacing the parameter space $\bbR$ with an arbitrary set $\calX$.
We equip the set $\calX$ with a positive real-valued function $d(x,y)$ that
quantifies the estimation error associated with an estimate $y$ when the true value of the parameter is $x$.  A natural choice
for $d(x,y)$ might be a suitable distance measure. 
In this case, the success probability for a given set of states $x \mapsto \rho(x)$, prior distribution $\mu(x)$ and POVM $y \mapsto Q(y)$, is
\begin{align}
\eta(\delta, \mu, \rho, Q) = \int_{\calX} \diff \mu(x) \int_{\calX} \diff y \,
    w_{\delta}(d(x, y))
    \Tr[ \rho(x) Q(y)].
\end{align}
The definitions for the tolerance and sample complexity, the corresponding minimax quantities and the optimal quantities follow analogously as in the previous sections.
Some of our proofs extend naturally to general parameter sets $\calX$.
In Section~\ref{ssec:beyond_univariate_metrology} of the supplementary material, we %
give a generalization of Theorem~\ref{thm:succ_prob_upper_bound_mht_delta_window} to the multivariate case %
and use it to derive a multi-parameter analog of Corollary~\ref{thm:success-probability-upper-bound-via-two-point-method}.

\subsection{Confidence region tomography and shadow tomography}

A particularly well-studied variant of multivariate quantum metrology concerns the task of \emph{state tomography}, in which the parameter space $\calX$ is taken to be the quantum state space itself. 
A suitable distance measure, such as the infidelity or the trace distance, typically quantifies the estimation error.

This extension of our framework connects with a series of works in quantum tomography on establishing
confidence regions in state space given measurement data~\cite{Christandl2012_Tomo,BlumeKohout2012arXiv_Tomo,Faist2016PRL_practical,Wang2019PRL_polytopes}.
Confidence region estimators process the measurement data to output a subset of the state space (the confidence region)
in which the true state lies with high probability.  The region can furthermore be specified as the set of all states that are at least $\delta$-close to some reference state in a suitable distance measure. 

Another setting that connects to the multi-parameter estimation version of our framework is \emph{shadow tomography}~\cite{aaronson_shadow_2018,huang_predicting_2020}.
Shadow tomography aims at predicting the expectation values of a set of $M$ observables
$\calO = \{ O_i \}_{i=1}^M$ when evaluated on a given quantum state $\rho$.
If we define a distance measure
\begin{align}
    d_{\calO}(\rho, \sigma) = d_{\calO}(\rho - \sigma) \coloneqq \sup_{O \in \calO} | \Tr[ O (\rho - \sigma )]|,
\end{align}
then shadow tomography with precision $\delta$ is equivalent to finding an approximation of the quantum state
$\hat{\rho}$ that fulfills $d_{\calO}(\rho, \hat\rho) \leq \delta$.
Defining a measurement scheme then corresponds to a POVM with effects labeled by quantum states,
$Q(\hat\rho)$, and the minimax success probability is hence
\begin{align}
    \overline\eta(\delta, Q, \rho) = \min_{\rho} \int \diff \hat\rho \, \Tr[ \rho \, Q(\hat\rho)] w_{\delta}(d_{\calO}(\rho - \sigma)),
\end{align}
where both the minimization and integration are over all quantum states.
A probabilistic procedure introduced in Ref.~\cite{huang_predicting_2020} achieves the sample complexity
\begin{align}
    \overline{n}^{*}(\eta, \delta, \rho) \leq O\left( \frac{1}{\delta^2} \log \left(\frac{M}{1-\eta}\right) \max_{1 \leq i\leq M} \lVert O_i \rVert_{\mathrm{shadow}}^2\right),
\end{align}
where the \emph{shadow norm} $\lVert \cdot \rVert_{\mathrm{shadow}}$
captures properties of the particular randomized protocol used. 
We thus see that contemporary techniques like shadow tomography are captured by our PAC metrology framework.

In the same spirit, other tomography tasks, \textit{e.g.}\ the tomography of quantum channels and non-Markovian processes in the form of quantum combs, can be considered in the PAC metrology framework. Exploring the ultimate limitations of these tasks would constitute an intriguing direction of future work.

\subsection{Cryptography and adversarial parameter estimation}

There is %
an emerging subfield of quantum metrology concerned with its intersection with cryptography~\cite{shettell2022cryptographic,shettell2022private,faist2022time-energy}.
As an example, we might seek a metrological protocol where the precision with which an
eavesdropper might estimate a parameter should be as low as possible.
This regime corresponds to a regime of small success probability, $\eta \to 0$.

We %
make use of the fact that the success probability and tolerance have a functional
relationship that corresponds to an inversion of the success probability seen
as a function of the tolerance. %
This allows us to analyze the limit of small success probability for smooth, \emph{i.e.}, arbitrarily often differentiable, measurements $Q(\tau)$.
\begin{proposition}\label{prop:tolerance_for_small_eta}
For a given set of states $\rho(t)$, possibly with prior $\mu(t)$, a smooth measurement $Q(\tau)$ and a small success probability $\eta$, the Bayesian tolerance is given by
\begin{align}
    \delta(\eta, \mu, \rho, Q) &= \frac{\eta}{2} \left( \int \diff \mu(t) \, \Tr[ \rho(t) Q(t) ] \right)^{-1} + O(\eta^2),
\end{align}
whereas the minimax tolerance is given by
\begin{align}
    \overline\delta(\eta, \rho, Q) &= \frac{\eta}{2} \left( \inf_t \mathstrut \Tr[ \rho(t) Q(t) ] \right)^{-1} + O(\eta^2).
\end{align}
Both statements hold conditioned on the inverted quantity to be nonzero.
\end{proposition}
The proof is shown in Section~\ref{ssec:opt_window_size} of the supplementary material.

Let us now assume that Alice performs a quantum metrology protocol and obtains a quantum state $\rho_A(t)$ which it communicates to Bob via a quantum channel $\calN$, such that he receives the state $\rho_B(t) = \calN[\rho_A(t)]$. 
At the same time, an eavesdropper Eve tries to obtain as much information as possible about the transmitted state. The state Eve can obtain in the worst case is modeled by the \emph{complementary channel} $\calN_{\mathrm{c}}$~\cite{Wilde17_book}, which
describes the information the environment can obtain when viewing $\calN$
as part of a larger, unitary evolution.
As such, we assume that Eve holds the state $\rho_E(t) = \calN_{\mathrm{c}}[\rho_A(t)]$.
The above proposition tells us that, %
if we want to limit the precision with which Eve can estimate the parameter $t$ from the state $\rho_E(t)$, we need to choose the initial state of Alice such that the quantity
\begin{align}
\begin{split}
    &\sup_{Q(t)} \int \diff \mu(t) \, \Tr[ \rho_E(t) Q(t) ] \\
    &\qquad =\sup_{Q(t)} \int \diff \mu(t) \, \Tr[ \rho_A(t) \calN^{\dagger}_{\mathrm{c}} [ Q(t) ] ]    
\end{split}
\end{align}
is as small as possible.

\subsection{Optimization over constrained sets of measurements}

In the formulation of the success probability as a convex problem of Proposition~\ref{prop:sdp_formulation}, we optimize over all possible quantum measurements of the system. We already obtained the optimal post-processing for the practically important case of a fixed measurement. In this section, we discuss another possible way of including practical constraints that might limit the possible measurements by only optimizing over POVMs from a set $\calQ$ that represents the set of measurements that can be implemented on the system.

In this case, we obtain a restricted optimal success probability $\eta^{*}_{\calQ}(\delta, \mu, \rho)$ quantified as
\begin{align}
&\qquad\eta^{*}_{\calQ}(\delta, \mu, \rho)  \coloneqq \max\mathstrut \\
&\left\{\left. \int \diff t \, \Tr[ (w_{\delta} * [\mu \cdot \rho])(t) Q(t) ] \, \right|\right.\left.\, \int \diff t \, Q(t) = \bbI, Q \in \calQ \right\}.\nonumber
\end{align}
A particular case of interest appears when $\calQ$ is a convex set, in particular, if it is specified by semidefinite constraints.
In this case, the optimization above is a convex optimization problem.
Such a situation occurs, for instance, if we assume that Alice prepares a state $\rho(t)$ and sends it to a
noisy channel $\mathcal{E}$ to Bob, who attempts to estimate $t$ using any possible POVM on his system.
Bob's optimal POVM $Q(t)$ can be mapped to the POVM $ \mathcal{E}^\dagger(Q(t))$ on Alice's system through the adjoint map $\mathcal{E}^\dagger$ of $\mathcal{E}$.  
Consequently, Bob's optimization over any POVM can be equivalently expressed as Alice optimizing over all POVMs in the image of $\calE^{\dagger}$. The thus defined set $\calQ = \{ \calE^{\dagger}[Q] \pipe \int \diff t \, Q(t) = \bbI, Q(t) \geq 0\}$ is a convex set and Bob's optimal success probability can as such be computed efficiently.

The minimax variant, as well as the corresponding tolerance and sample complexity are then defined analogously
to their non-restricted counterparts.

\subsection{Estimation of properties beyond parameters}
We now consider the setting where we seek to estimate some property $x$, \emph{e.g.}, the expectation value
of an observable, of a general unknown state $\rho$. 
Crucially, multiple states might share the same property value $x$, hindering the use of our analysis which assumed that the unknown state belongs to a set of states that is fully specified by one parameter.

This task can be treated by the introduction of so-called \enquote{nuisance} parameters.
Nuisance parameters are additional parameters that ensure that each quantum state is associated with a distinct
set of parameter values. The consequence of introducing nuisance parameters is to reduce the property estimation problem to a multi-parameter estimation problem. In our case, we can perform a similar strategy that proceeds like the example of shadow tomography discussed above.
By taking the space of quantum states as the parameter space, we introduce the maximal
possible number of nuisance parameters, and by introducing a distance function
\begin{align}
    d_x(\rho, \sigma ) \coloneqq \lvert x(\rho) - x(\sigma) \rvert ,
\end{align}
we obtain a way to only pick out the relevant parameter $x$.
These settings are of extremely high importance for practical applications, \emph{e.g.}\ for near-term applications on NISQ devices~\cite{cerezo2021variational}. There exists a large variety of existing techniques to treat this particular application~\cite{tsang2020quantum}, but these tools pertain to the local estimation setting. It is therefore an intriguing direction of future research to see how these analytical approaches can be generalized to the framework of PAC metrology.

\section{Future directions}\label{sec:future_directions}

In the standard approach to quantum metrology built on top of the quantum Cramér-Rao bound, most fundamental questions have already been answered. In our single-shot PAC metrology framework, on the contrary, a broad collection of open question -- both fundamental and practical -- is still looking for answers. 
In this section, we highlight some questions of particular interest.

\paragraph*{Optimal measurements in the finite-sample regime.}

A significant open question that remains open is
the development of measurement schemes or protocols with guarantees on either the success probability or the estimation tolerance.
While the formulation of the success probability as a convex problem provides an efficient way to compute the optimal POVM associated with a discretized version of the estimation problem, a closed form of the optimal measurement remains elusive beyond the covariant case treated in Section~\ref{sec:covariant_setting}. 
Furthermore, solutions to the convex optimization problem are not likely to provide additional insight on measurement schemes that are perhaps sub-optimal but far more convenient to implement than the optimal measurement.

Natural candidates for measurement schemes are adaptive protocols~\cite{berry_optimal_2000} as well as the pretty good measurement that constitutes the optimal measurement in the pure and covariant setting of Theorem~\ref{thm:pgm_is_minimax_optimal}. Promising candidates might furthermore be constructed using the class of measurements studied in Section~\ref{sec:opt_fixed_measurement}.
\begin{openproblem}[Measurement schemes with performance guarantees]
    Develop measurement schemes with provable performance guarantees that are either practical to implement or that achieve close-to-optimal estimation error tolerance or success probability in the finite-sample regime.
\end{openproblem}
Such measurement schemes would significantly aid in deriving upper bounds on the optimal estimation tolerance and lower bounds on the success probability in various settings.  So far, such lower bounds have been elusive because of a lack of such schemes. For instance, an open question would be whether a measurement scheme is capable of achieving the rate given in Theorem~\ref{thm:rate_upper_bound_delta}.

\paragraph*{A finite-sample analogue of the quantum Cram\'er-Rao bound.}
The quantum Cramér-Rao bound of Eq.~\eqref{eqn:qcrb} is the fundamental cornerstone of the
standard approach to quantum metrology. It relates an operational quantity (the standard deviation of the optimal unbiased estimate of a parameter) to a geometric property of the underlying set of quantum states (the quantum Fisher information). We call the quantum Fisher information a geometric property because it quantifies the distance between quantum states whose parameters are close to each other when \enquote{distance} is measured through the fidelity of quantum states. States that are close in terms of the fidelity are difficult to distinguish, and we have already learned in Theorem~\ref{thm:succ_prob_upper_bound_mht_delta_window} that this is a prerequisite for successful parameter estimation. 

In our finite-sample approach, the estimation tolerance fulfills a role comparable to that of the standard deviation. We expect that the smallest achievable estimation tolerance should -- similar to the standard deviation -- be constrained by a quantity that captures geometric properties of the underlying set of quantum states.
As our framework pertains to cases where the estimation is not necessarily local, we expect that the we need quantities that go beyond the quantum Fisher information in the sense that they capture the geometry of the given state set at non-infinitesimal length scales.
Moreover, there is the additional factor of the desired success probability $\eta$ that factors into any relation between the optimal estimation tolerance and the structure of the given state set.

We managed to derive Theorem~\ref{thm:cramer_rao_like_bound} by quantifying how well the fidelity is approximated by a Gaussian as
\begin{align}
    F(\rho(t), \rho(t + \tau)) \approx \exp\left( - \frac{1}{8} \calF(t) \tau^2\right).
\end{align}
In that sense, we quantified how close we are to a case where the quantum Fisher information dictates not only the behavior in an infinitesimally small neighborhood but also for larger values of the perturbation $\tau$.
It is of immense interest if we can obtain more general statements of the same kind that resemble the Cramér-Rao bound in the following sense.
\begin{openproblem}[Finite-sample analogue of the Cramér-Rao bound]
Find improved lower bounds on the optimal Bayesian and minimax estimation tolerance
that put fundamental limits on the achievable estimation tolerance based on geometric properties of the underlying state set, \emph{i.e.},  bounds of the form
\begin{align}
    \overline{\delta}^{*}(\eta, \rho) \geq \frac{f(\eta)}{\sqrt{\calI}},
\end{align}
where $f(\eta)$ quantifies the dependence on the success probability and $\calI$ is a measure that captures the geometric structure of the given state set in the finite-sample setting. Note that $\calI$ can in general also depend on the desired success probability.
\end{openproblem}

An important improvement of our quantum Cramér-Rao-like bound of Theorem~\ref{thm:cramer_rao_like_bound} would be a bound that accurately predicts the achievable estimation tolerance in the i.i.d.\ limit -- \emph{i.e.}\ given a state $\rho^{\otimes n}(t)$ where $n\to\infty$ -- in the sense that there exists a matching upper bound. 
It would not be surprising if a different geometric quantity than the quantum Fisher information would appear in such a bound. This is because the converse on hypothesis testing involving the fidelity of quantum states we use to derive Theorem~\ref{thm:cramer_rao_like_bound} is known to not be tight in certain settings and that a proper converse should rely on the Chernoff divergence which is in turn associated to a different geometrical quantity, namely the Wigner-Yanase-Dyson information.

As we have argued above, Theorem~\ref{thm:cramer_rao_like_bound} builds on the insight that there are settings (\emph{e.g.},  the i.i.d.\ setting) where the quantum Fisher information carries information about the set of states beyond infinitesimal perturbations.
As such, we expect that the development of quantities that better capture the distinguishability of states in non-infinitesimal neighborhoods
could lead to a more general Cramér-Rao-like bound valid in the single-shot setting.

\paragraph*{Advantage of entangled measurement strategies.}
A further open question is to delineate the boundary of performance between different classes of protocols using entangled and non-entangled probe states, adaptive and non-adaptive processing and coherent and incoherent measurements. 
In particular, it is unclear whether optimal measurements in the finite-sample
regime require the use of large amounts of coherence.
When estimating a single parameter in the many-sample regime,
the Heisenberg scaling can be achieved only with the use of entangled probe states
and without the use of coherent measurements
~\cite{giovannetti_quantum_2006} -- does a similar statement also hold in the non-asymptotic case?
\begin{openproblem}[Understanding relevant resources]
What advantages can be gained from resources like entanglement of the probe state, coherence of the measurement and adaptivity? Can we quantify the gaps in tolerance and success probability between these different allowed resources?
\end{openproblem}
Naturally, finding general purpose strategies, \emph{e.g.}\ specific adaptive protocols, that give a competitive baseline success probability would significantly simplify addressing the above problem.

In our investigations of the phase estimation example, we observe an exact quadratic relation between the asymptotic rates for i.i.d.\ states and entangled states for small $\delta$. This echoes the quadratic relation between the standard quantum and the Heisenberg limits in the standard approach to quantum metrology. Understanding the generality of this phenomenon -- especially for channels where Heisenberg-limited scaling is impossible~\cite{demkowicz-dobrzanski2012elusive} -- would further deepen our understanding %
of the asymptotic rate for entangled inputs. 
\begin{openproblem}[Standard quantum and Heisenberg limit for rates]
For a given set of channels $\calN(t)$ with prior $\mu(t)$, do we in general have that
\begin{align}
\lim_{\delta \to 0} \frac{[R_{\mathrm{par}}^{*}(\delta, \mu, \calN)]^2}{R_{\mathrm{iid}}^{*}(\delta, \mu, \calN)} &= 1,
\end{align}
where the symbols $R_{\mathrm{par}}^{*}$ and $R_{\mathrm{iid}}^{*}$ have been defined in Section~\ref{sec:parametrized_quantum_channels}. 
\end{openproblem}

\paragraph*{The effect of noise on estimation performance.}
While noise can be implicitly treated in our formalism by including it in the construction of the parametrized set of states $\rho(t)$ or channels $\calN(t)$,
a much deeper understanding of the influence of noise is desirable, especially with an eye towards practical applications. It would be instructive to see how much the influence of noise destroys advantages of entanglement in this setting~\cite{PhysRevLett.79.3865,demkowicz-dobrzanski2012elusive,smirne_ultimate_2016}. As we have outlined in Section~\ref{sec:learning_from_quantum_systems}, the minimax success probability also pertains to communication tasks, where the study of noisy channels is of utmost importance. Can we therefore relate the decrease in success probability (or the increase in tolerance) to properties of the %
noise channel? 
A deeper understanding of noise could then be used to study the application of quantum error correction to quantum metrology~\cite{zhou_achieving_2018,faist2022time-energy} in this context. 

\paragraph*{Incompatibility in finite-sample multi-parameter quantum metrology.}
As discussed in Section~\ref{sec:multi_parameter_metrology}, our definitions extend in a natural way to multi-parameter quantum metrology. In this setting, challenges arise that are not present in the single-parameter setting, such as the incompatibility of measurements and probes that are optimal for different parameters~\cite{belliardo_incompatibility_2021,albarelli_probe_2022}. 
These phenomena particular to multi-parameter quantum metrology have been extensively studied in the asymptotic context. We expect that a quantification of incompatibility and other multi-parameter phenomena through our finite-sample framework could deepen our understanding of these effects.

\paragraph*{Quantifying generalized communication tasks.}
The minimax setting of quantum metrology has a further interesting interpretation, as it directly quantifies how well a sender Alice can communicate a scalar parameter $t \in \bbR$ to a receiver Bob when she encodes the parameter in a set of states $\rho(t)$. It is important to note here, that the worst-case nature of the minimax setting is critical in quantifying the performance, as Alice wants to be able to communicate any value $t$ with similar guarantees. In this sense, the minimax setting of quantum metrology quantifies a generalization of the commonly encountered task of communicating one of multiple discrete symbols~\cite{khatri2020principles}. In practical scenarios, Alice will usually send the state through a quantum channel $\calN$ that degrades the message, giving an additional impetus to study quantum metrology with the noise model $\rho(t) \mapsto \calN[\rho(t)]$~\cite{faist2022time-energy}.

\paragraph*{Further open questions of information-theoretic nature.}
In Section~\ref{sec:asymptotics_succ_prob}, we have given an upper bound on the asymptotic rate of quantum metrology, which is tight for commuting states. It is an open question to determine under which general conditions on the set of states this bound is tight. 
\begin{openproblem}[Asymptotic rate (states)]\label{op:asymptotic_rate_states}
For a given tolerance $\delta > 0$, a set of states $\rho(t)$, possibly with prior $\mu(t)$, when do we have that
\begin{align}
    R^{*}(\delta, \mu, \rho) = \overline{R}^{*}(\delta, \rho) = \inf_{|t-t'|>2\delta} C(\rho(t), \rho(t'))?
\end{align}
\end{openproblem}
Naturally, it is also interesting to quantify the second-order asymptotics
of the success probability, which will, however,
first require progress on the second-order asymptotics of the symmetric hypothesis testing error probability for quantum state discrimination.
Another quite natural extension is the case of \emph{mixed asymptotics}: 
We fix a desired scaling of the tolerance, decreasing slower than the optimal scaling; how fast can the success probability still attain unity?

As already outlined in Section~\ref{sec:parametrized_quantum_channels}, the setting of parametrized quantum channels offers much richer structure in its asymptotics, owing to the possibility of different strategies (separable use, parallel use, adaptive use and indefinite causal order). 
Given that Corollary~\ref{corr:succ_prob_upper_bound_mht_access_strategies} explicitly relates the success probability of quantum metrology under different access modes with the success probability of a corresponding hypothesis testing problem with similar access modes,
we expect that the asymptotics of quantum metrology with different kinds of strategies should relate to the asymptotics of the corresponding multi-hypothesis testing tasks.

\section{Discussion}\label{sec:discussion}

Our work extends the foundations of quantum metrology to the regime where few measurement samples can be obtained. 
To study this regime, we present a truly single-shot framework for quantum metrology that removes two important assumptions that are used to derive the quantum Cramér-Rao bound.
First, our approach quantifies estimation accuracy directly as the probability
that the true estimate lies close to the parameter value instead of quantifying
the variance of an unbiased estimator.
This definition guarantees operational significance even in regimes where the variance is only a poor indicator of single-shot performance.
Second, we remove the assumption of a local estimation setting. We have developed two ways of doing so: the Bayesian approach which allows us to reconnect to the local setting by choosing suitably narrow prior distributions and the minimax setting that truly quantifies the absence of knowledge about the underlying parameter.

In the setting of our framework (Fig.~\ref{fig:intro-super-simple-setup}), a parameter $t$ is encoded in a set of quantum states $t \mapsto \rho(t)$. We quantify the probability that a measurement embodied by a POVM $\tau \mapsto Q(\tau)$ produces an estimate $\hat{t}$ of $t$ that is within a tolerance $\delta$ of the true value. \enquote{Probability} is evaluated with respect to the random nature of the measurement outcomes,
either in the worst case over possible parameter values (minimax setting)
or in the case where the parameter value is sampled from a prior probability
distribution (Bayesian setting).
Our framework therefore captures any standard setting in quantum
metrology in which an unknown parameter in the quantum state, possibly imprinted
via a parameter-dependent dynamics, is to be estimated by the application of
a quantum measurement.

The optimal success probability optimized over all possible POVMs $\tau \mapsto Q(\tau)$ can be obtained from a semi-infinite program, an extension of semidefinite programming that enables the inclusion of a continuous POVM as a variable.
The rich structure offered by such a convex optimization
enables numerical computations as well as some of our proofs
including the derivation of
both upper and lower bounds on the success probability.
The solution to the convex optimization problem might be difficult to obtain, which is why there is a need for
good \enquote{general purpose} measurement strategies, \emph{i.e.}\ POVMs $Q(\tau)$ that give guarantees on an achievable success probability without the need to perform the convex optimization. 
A possible candidate for such a general purpose strategy could be the pretty good measurement.

Another practically relevant case arises when
the quantum measurement is 
fixed, for instance by experimental constraints, and only the post-processing of measurement outcomes into predictions of the parameters 
may be optimized. 
In this setting, we showed that a strategy generalizing maximum
a-posteriori estimation is optimal for the Bayesian setting and
gives bounds in the minimax setting.

Our quantum metrology setting
naturally extends
quantum
multi-hypothesis testing.
Instead of 
selecting one of finitely many alternatives, we need to discriminate states from the set \smash{$\{ \rho(t) \}_t$}.
This continuous generalization of hypothesis testing
requires %
a tolerance $\delta$ in the precision to which the parameter is to be estimated.
This difference seems to significantly complicate
extensions of existing error bounds for multi-hypothesis testing to our
continuous hypothesis testing setting: Such analyses typically analyze a protocol by
assuming a unique correct output
rather than a range of acceptable outputs.
We make this connection rigorous by
giving a general upper bound
on the success probability of the estimation procedure
in our metrology setup
by showing that
estimating a parameter to precision $\delta$
implies the ability to
successfully distinguish between states
with parameter values separated by at least
$2\delta$ in a hypothesis testing setting.
We exploit this result to give bounds on the success probability in terms of the fidelity of quantum states,
a more tractable and familiar quantity.

This connection provides an important application of quantum hypothesis
testing. The main use of quantum hypothesis testing so far
is the study the asymptotic behavior of entropy measures, which is then applied for 
instance to the study of quantum communication scenarios beyond the i.i.d.\ regime.
The application of quantum hypothesis testing to quantum metrology, an inherently
physical setting, highlights a need to extend existing bounds and protocols in
quantum hypothesis testing to more general settings relevant to metrology,
which include an error tolerance on the unknown parameter.
This connection is an exciting opportunity for quantum
information theory to inform the development of estimation procedures and the derivation
of fundamental bounds in quantum metrology.

The fact that quantum metrology can be seen as a generalization 
of multi-hypothesis testing also opens upon a wealth of open questions of a 
distinctly information-theoretic flavor.
We %
analyzed the asymptotic behavior of the success probability of quantum metrology
and have shown a bound on the asymptotic error rate in terms of the Chernoff divergence,
a bound which can be achieved in the case of commuting states.
As such, we established that the quantifier of the asymptotic properties of hypothesis testing also applies to quantum metrology.
Our metrology setting also opens new kinds of questions regarding the asymptotic behavior
of quantum information-theoretic quantities that appear in hypothesis testing tasks.
Theorem~\ref{thm:rate_upper_bound_delta} bounds the rate at which the success probability $\eta$ approaches
one as the number of i.i.d.\ copies $n$ goes to infinity, supposing that the estimation tolerance $\delta$ is kept constant.
The asymptotic behavior of $\delta$, keeping $\eta$ fixed, is in turn bounded
by Theorem~\ref{thm:cramer_rao_like_bound}. 
A particularly relevant regime to study is a ``mixed asymptotics'' regime where $\delta\to 0$ and $\eta\to 1$ simultaneously as $n\to\infty$. For instance, what is the behavior of $\eta$ 
if the tolerance decays as \smash{$\delta \sim 1/\sqrt{n}$}?

We take much inspiration from recent developments in single-shot
quantum information theory, where 
many results %
are expressed through entropic quantities that have explicit single-shot interpretations.
Popular examples are min-, max- and Rényi-relative entropies.
Similarly,
our single-shot framework for quantum metrology can be %
connected to entropic quantities.
In particular, we give an alternative definition of the optimal success probability of quantum metrology
as a generalized conditional min-entropy and bound the optimal tolerance through the hypothesis testing relative entropy.
One one hand, these relations further strengthen the connection between
quantum metrology and quantum information theory, demonstrating that concepts
deeply rooted in quantum information theory can offer alternative alternative
approaches to proving accuracy bounds in quantum metrology
(e.g.\@ \cite{walter_heisenberg_2014}).
On the other hand, these relations further demonstrate the broad usefulness
and applicability of the toolbox of single-shot entropy measures~\cite{PhDRenner2005,khatri2020principles,EntropyZoo}.
We further anticipate opportunities to exploit new relations between
quantum metrology and entropy measures to derive a deeper understanding
of the fundamental accuracy bounds through the lens of single-shot
quantum information theory.  

Next to the success probability, the estimation tolerance is the second important pillar of our single-shot framework for quantum metrology. 
It corresponds to the natural question of
quantifying the estimation accuracy that can be obtained
with a fixed success probability.
It is important because it fulfills a similar role as the standard deviation that is used as a measure of estimation precision in the standard approach to quantum metrology.
The importance of providing useful lower bounds on the optimal estimation
tolerance is further underscored by its use to manifestly express
advantages that can be obtained with entangled quantum states as opposed
to classical estimation strategies, such as Heisenberg
scaling~\cite{giovannetti_advances_2011}.
We derived a general lower bound on the estimation tolerance that
resembles the quantum Cramér-Rao bound. This bound involves the standard quantum
Fisher information and contains
corrections pertaining to the explicit single-shot nature of the estimation tolerance.
Obtaining improved bounds on the optimal tolerance -- especially bounds that are achievable in reasonable limits -- is a pressing question in the area of single-shot quantum metrology. 

Metrological problems are in many cases given by parametrized \emph{channels} $t \mapsto \calN(t)$ instead of parametrized quantum states. This setting offers 
a much richer structure of both the available ways of interacting with the quantum channel itself (e.g., adaptive versus non-adaptive protocols) as well as the underlying complications of the metrological protocols. 
We have adapted some of our central results, such as the formulation of the success probability as a convex program, as well as our central theorem about the connection to hypothesis testing, to this setting. 
There is a plethora of open questions of operational relevance, especially concerning the power of particular types of protocols with more restricted access to the parametrized quantum channel, which %
are detailed in the preceding section. 

An important class of parametrized channels are unitary evolutions under a given Hamiltonian. 
We study this setting by establishing a closed-form expression for the optimal minimax success probability associated to a covariant state set, \emph{i.e.}\ when a pure state evolves under a given Hamiltonian and the parameter range is identical to the recurrence time of said Hamiltonian. It is an interesting question of further research if metrology protocols that use adaptive processing with a memory system can achieve a higher minimax success probability.

We apply this result to the traditional problem of phase estimation using an ensemble of spin-$\tfrac{1}{2}$ particles. 
This is one of the most basic, but still technologically important, applications of quantum metrology.
The closed-form expression we derive for the optimal minimax success probability allows us to characterize the optimal probe state through a foundational result of Slepian~\cite{slepian_prolate_1978} and to compute the optimal asymptotic rate achievable via entangled probe states. 
It also facilitates numerical experiments for up to $n = 1000$ particles that we use to compare different kinds of probe states in terms of the achievable minimax success probability and minimax estimation tolerance.
We %
observe that %
reasonable guesses for good probe states, like the Holland-Burnett state or a Gaussian profile do only coincide with the optimal state in certain limits.
The GHZ state is a poor probe state this global
estimation task because its period is much shorter than the whole parameter range.
Another outcome of our numerical experiments was that there seem to be settings where the quantum Cramér-Rao bound faithfully predicts the optimal minimax estimation tolerance -- in our case in the limit of many i.i.d.\ copies of the optimal probe state when the success probability is fixed at the probability that a normally distributed random variable is within one standard deviation of the mean. 

The observation that the quantum Cramér-Rao bound can faithfully predict finite-sample performance is perhaps surprising, given that we could expect the finite-sample estimation tolerance to deviate significantly from the optimal estimator variance, whose operational significance only appears in the many-sample regime.
Indeed, there is much fine print to this observation -- the phase estimation setting is highly regular and we analyzed it in the pure state setting where many information measures collapse into the quantum Fisher information.
It nevertheless proves that the quantum Fisher information, which serves as a proxy for
the single-shot distinguishability of neighboring quantum states
$\rho(t)$ and $\rho(t+\diff t)$, is likely to still be relevant in some finite-sample estimation scenarios.
Theorem~\ref{thm:cramer_rao_like_bound} provides a specific connection between the
optimal estimation error tolerance and the quantum Fisher information which involves
the presence of additional error terms. Quantifying the magnitude of these error
terms in various settings is likely to provide clarity onto the regimes where the
optimal estimation tolerance is well approximated by the inverse square root of the quantum
Fisher information.

Our work reinforces the value of an operational approach to fundamental questions
in quantum metrology, especially in the finite-sample regime which is increasingly relevant
for current and near-term quantum technologies.
The foundations of our approach are formalized in our framework of
probably approximately correct (PAC) metrology.
Furthermore, the newly reinvigorated connection between quantum metrology and quantum information
theory offers exciting opportunities for progress in metrology using advanced techniques developed in quantum information theory, for instance, the use of matrix analysis and convex optimization to characterize quantum information entropy measures.
Contrary to quantum metrology in the standard many-sample regime,
where many fundamental questions are already answered,
our operational approach to quantum metrology finite-sample regime
offers a plethora of intriguing questions and research directions 
with the potential of uncovering new practical estimation procedures with increased accuracy
in quantum sensors and quantum clocks.

\section*{Acknowledgements}

The authors would like to thank Sergii Strelchuk, Wilfred Salmon, Francesco Albarelli, Jasminder Sidhu and Maximilian Reichert for valuable feedback on earlier versions of this manuscript. 

This work has been supported by the DFG (CRC 183, FOR 2724), by the BMBF (Hybrid), the BMWK (PlanQK, EniQmA), the Munich Quantum Valley (K-8), QuantERA (HQCC) and the Einstein Foundation (Einstein Research Unit on Quantum Devices). This work has also been funded by the DFG under Germany's Excellence Strategy – The Berlin Mathematics Research Center MATH+ (EXC-2046/1, project ID: 390685689).

\section*{Author contributions}

J.~J.~M.\ led the project, derived a significant majority of the technical results and wrote the manuscript draft.
All authors contributed substantially to the technical results and to writing the paper.

\bibliography{main}

\let\addcontentsline\oldaddcontentsline %

\clearpage
\setcounter{section}{0}

\def\tocname{}

\title{Supplementary Material:\texorpdfstring{\\}{} \papertitle}
\maketitle

\onecolumngrid
\vspace{-1.5cm}
\tableofcontents

\section{Relation to prior art}\label{ssec:relation_to_prior_art}
Because of its technological importance, quantum metrology is a very important subfield of quantum information theory. Most works in the literature focus on the asymptotic theory of local estimation centered around the quantum Fisher information -- a good overview can be found in Ref.~\cite{sidhu2019geometric}.   
In this work, we establish an inherently \emph{single-shot} characterization of quantum metrology via an operational characterization of the success probability of parameter estimation. From this characterization, we establish a coherent framework to quantify the quality of quantum metrology protocols by analyzing the success probability, the tolerance and the sample complexity. We develop a deep understanding of these quantities through analytical bounds and numerical experiments.

Due to the prominent role of quantum metrology in the field of quantum information science, it comes as no surprise that parts of our definitions and some analytical results concerning those have already appeared in the literature. In the following, we exemplify how our work goes well beyond the prior art by giving a detailed account of the similarities and differences between our results and the results already found in the literature.

In Ref.~\cite{hayashi2002two}, the author also defines a notion of success probability with a given tolerance similar to Definition~\ref{def:success_probability_no_optimization}. They study the properties of \emph{consistent} protocols that ensure that the success probability, asymptotically, approaches unity for all nonzero tolerances. For those protocols, a large-deviation analysis where the number of samples is finite but still assumed large is performed. This allows the author to bound the local curvature of the asymptotic rate defined in Eq.~\eqref{eqn:def_asymptotic_rate} of the main text. In our work, we go beyond this result by establishing bounds that hold in the single-shot setting and that we then use to establish bounds on the asymptotic rate that hold for fixed tolerance $\delta$ and that at the same time give tighter constraint on its local curvature (see Section~\ref{ssec:wigner_yanase_dyson_information}).

Refs.~\cite{sugiyama2011tomography,sugiyama2013tomography,sugiyama_thesis} contain what is essentially the multi-parameter generalization of Ref.~\cite{hayashi2002two} to the multi-parameter setting, with quantum state tomography as the most prominent application. It is a common feature of all these works that they assume the sequence of quantum measurements associated with the different numbers of possible copies of the state to be fixed. This allows the authors to make good use of tools from classical estimation theory. Especially in Ref.~\cite{sugiyama_precision-guaranteed_2015}, the author analyzes the scaling of the tolerance for a fixed success probability when a measurement is fixed and a maximum-likelihood estimate is performed. This can be seen as studying a particular prediction rule as introduced in Section~\ref{sec:opt_fixed_measurement}. We wish to emphasize that all our results except for the lower bound on the asymptotic rate and the developments of Section~\ref{sec:opt_fixed_measurement} pertain to the case where measurements are not fixed a priori but the optimal measurement can be chosen.

The setting of fixed measurements (compare Section~\ref{sec:opt_fixed_measurement}) was also studied in Ref.~\cite{Salmon}. There, the authors fixed on the question if there exists a measurement that is \emph{admissible} in a estimation-theoretical sense, \textit{i.e.}, if there is a measurement basis that works equally well for all possible values of the underlying parameter.

The optimal achievable tolerance was also studied in Ref.~\cite{walter_lower_2014,walter_heisenberg_2014}. The authors there studied parametrized quantum states as classical-quantum states, similar to what we use in Corollary~\ref{corollary:relation_to_cond_min_ent} to relate the success probability of quantum metrology to the conditional min-entropy. 

The authors of Ref.~\cite{yang2018stopwatch} also put forth a definition similar to the minimax tolerance of Definition~\ref{def:metrological_tolerance}, referring to it as \enquote{inaccuracy}. They furthermore obtain a lower bound on the minimax tolerance by making a particular multi-hypothesis testing reduction. Their lower bound, however, is limited by its dependence on the dimension of the underlying quantum systems. 

Parameter estimation of quantum channels using the Fisher information has been considered in Refs.~\cite{ji2008parameterestimationchannel,HS11,hayashi2011comparison,DCS17,yuan2017fidelityfisherquantumchannel,yuan2017parameterestimationdynamics,katariya2021geometric,katariya2021RLD,liu2023metrologyhierarchy}. We note that while in some of these works the number of channel uses is finite, the estimation strategies are still evaluated using the quantum Fisher information. In this work, on the other hand, we evaluate strategies for quantum channel parameter estimation using our inherently single-shot success probability.

Phase estimation, as studied in Section~\ref{sec:minimax_analysis_phase_estimation} of the main text, is the generally most-studied application of quantum metrology. It is an instance of a covariant estimation problem whose study dates back to the foundational works of Holevo and Helstrom, with important contributions from Belavkin and Maslov~\cite{helstrom1969quantum,Helstrom76_book,belavkin1988design,holevo2011probabilistic}. Already in these works, it has been shown that covariant measurements are optimal for covariant estimation problems, but not necessarily \emph{which} covariant measurement. The authors of Ref.~\cite{ban1997optimum} have shown that the pretty good measurement is optimal both for symmetric multi-hypothesis testing as well as covariant estimation with a maximum likelihood approach. In Section~\ref{sec:minimax_analysis_phase_estimation}, we extend this result by showing that the pretty good measurement also achieves the optimal minimax tolerance in the case of covariant estimation with pure states.

We note that our study of phase estimation goes well beyond these results, as it also covers the single-shot regime and as we also give results on the optimal minimax tolerance of phase estimation.
Our study of phase estimation has novel results valid in the single-shot regime and gives explicit bounds on the optimal asymptotic tolerance. Our results on the asymptotic rate (Theorem~\ref{thm:opt_minimax_rate_entangled_phase_est}) do however overlap significantly with prior work of Ref.~\cite{imai_fourier_2009}. The authors establish the distributions of measurement outcomes that can be realized asymptotically in phase estimation and build on the continuous version of the work of Slepian and Pollak~\cite{slepian1961prolate} to claim the asymptotic rate of Theorem~\ref{thm:opt_minimax_rate_entangled_phase_est}. The authors of Ref.~\cite{imai_fourier_2009} did not prove the required positivity of the DPSS of zeroth order necessary for such a claim, as such we filled a small gap in their proof. 
Furthermore, Refs.~\cite{durkin2007localglobalinterferometry,demkowicz-dobrzanski_optimal_2011,rubio2018nonasymptotic,chesi_protocol_2023} have also went beyond the asymptotic regime and considered the phase estimation problem from both the local and global perspectives.

Another approach to non-asymptotic metrology was studied in Refs.~\cite{rubio2018nonasymptotic,rubio2019limited}. The authors explicitly evaluate the mean-squared error for phase estimation problems with a small but increasing finite number of independent repetitions of a single-repetition estimation scheme, and they examine its convergence to the Cram\'{e}r-Rao bound as the number of repetitions increases. Similarly, in Ref.~\cite{liu2016valid}, the authors consider a finite number of independent repetitions of a single-repetition estimation scheme, but they also develop a Cram\'{e}r-Rao-like bound on the mean-squared error that applies to biased estimators. Notably, in these works, the estimation performance is evaluated using the mean-squared error~\cite{rubio2018nonasymptotic,rubio2019limited} and the quantum Fisher information~\cite{liu2016valid}, both of which are asymptotic quantities, while in this work our estimation procedure is evaluated using the inherently single-shot $\delta$-accurate estimation success probability presented in Definition~\ref{def:success_probability_no_optimization}.

It is worth mentioning that Ref.~\cite{tsang2012ziv-zakai} also uses a reduction of the parameter estimation problem to hypothesis testing, but it does so in a different way than in our work. In particular, in Ref.~\cite{tsang2012ziv-zakai}, the author defines a quantum version of the Ziv-Zakai bound from classical estimation theory, which provides a bound on the mean-squared error, in terms of binary hypothesis tests. The quantum Ziv-Zakai bound is then formulated in terms of the optimal binary symmetric hypothesis testing error probability. While the quantum Ziv-Zakai bound can improve upon the quantum Cramér-Rao bound in the regime of finite samples~\cite{rubio2018nonasymptotic}, it is still hampered by the fact that hypothesis testing is used to bound an inherently asymptotic quantity. Starting from the inherently operational single-shot definition of the success probability allows us to develop a way stronger connection to hypothesis testing, as evident in Section~\ref{sec:metrology_as_hypothesis_testing}.

Further work on continuous hypothesis testing includes Refs.~\cite{chase2009singleshot,tsang2012continuous}. The ``continuous'' in these works actually refers to measuring the (unknown) system at different points in time and then deciding which among two possible states the system was in initially, or about deciding among two possibilities for the dynamics of the system. However, notably, in Ref.~\cite{chase2009singleshot} (Section~IV), the authors already allude to a continuous version of quantum hypothesis testing as we view it.

\section{The Cram\'er-Rao bound in the presence of finite samples}\label{ssec:cramer_rao_finite_sample}
Here, we review how the operational relevance of the Cram\'er-Rao bound might be compromised
in the regime where limited data is available.  Consider a one-parameter family of
states $\rho(t)$. Suppose we know that the true value $t$ of the parameter is close
to some value $t_0$.  We seek an observable $T$ that reveals the true value $t$
in the neighborhood of $t_0$ in its expectation value, 
\ie
$\langle T \rangle_{\rho(t_0+\diff t)} = t_0 + \diff t + O(\diff t^2)$. The observable $T$ serves as
estimator for the true parameter.  A formulation of the
Cram\'er-Rao bound states that the minimal
variance achieved by such an estimator is the inverse of the quantum Fisher
information,
\begin{align}
    \min_{T} \langle (T - t_0)^2 \rangle_{\rho(t_0)}
    &= \frac1{\calF(\rho(t_0))}\ ,
\end{align}
where the quantum Fisher information $\calF(\rho(t))$ is defined as
$\calF(\rho(t)) := \Tr[\rho(t) L^2]$, where the \emph{SLD operator} $L$ is a solution to the equation
$(\rho(t) L + L\rho(t))/2 = \partial_t\rho(t)$.

The estimation strategy considered by the Cram\'er-Rao bound requires the estimation
of the expectation value of the corresponding observable $T$.  In practice, this
requires repeated measurements of $T$ and averaging the corresponding
individual outcomes.
However, in the presence of
limited data, an accurate estimation of the expectation value might require a large
amount of data.  In the following, we review such a situation.  In such a setting,
the fundamental accuracy limits might be fundamentally different than the one predicted
by the Cram\'er-Rao bound.

Consider the following situation detailed in
Refs.~\cite{Safranek2017PRA_discontinuities,Zhou2019arXiv_exact}.
Let $\omega>0$, then the one-parameter family of qubit states $\rho(t)$ and the
derivative $\partial_t\rho(t)$ are given as
\begin{align}
    \rho(t) &= \begin{pmatrix}
      \cos^2(\omega t/2) & 0
      \\
      0 & \sin^2(\omega t/2)
    \end{pmatrix} ,
    &
    \partial_t\rho(t) &=
    -\frac{\omega}{2} \sin(\omega t)\,\begin{pmatrix} 1 & 0 \\ 0 & -1 \end{pmatrix} .
\end{align}
This evolution actually induces a discontinuity in the quantum Fisher information:
As shown in Ref.~\cite{Safranek2017PRA_discontinuities}, the quantum Fisher information
$\calF(\rho(t))$ is
\begin{align}
    \calF(\rho(t)) = \begin{cases}
    0 \quad & \text{if $\omega t = m\pi$, $m\in\mathbb{Z}$,}\\
    \omega^2 & \text{otherwise.}
    \end{cases}
\end{align}
Now consider the regime where $t_0\approx 0$, $t_0\neq 0$.  According to the
Cram\'er-Rao bound, it is possible to estimate the value of $t$ close to $t_0$ with
optimal sensitivity $1/\omega^2$ with a suitable observable.  As computed e.g.\@ in
Ref.~\cite{faist2022time-energy} [Appendix~H.2], the optimal observable in question is
\begin{align}
    T = t_0\Id + \frac1\omega \begin{pmatrix}
        -\tan(\omega t_0/2) & 0\\
        0 & \frac1{\tan(\omega t_0/2)}
    \end{pmatrix}.
\end{align}
We can check indeed that
\begin{align}
    \Tr[T\rho(t_0+\diff t)] = t_0 + \diff t + O(\diff t^2)\ .
\end{align}
It is, however, instructive to write out the expectation value as a sum of two terms,
one associated with each outcome of a measurement of $T$.  With $t=t_0+\diff t$, the contributions
to $\Tr[T\rho(t)]$ are, with $T_j = \bra{j}{T}\ket{j} - t_0$,
\begin{align}
  T_0 \bra{0}{\rho(t)}\ket{0}
  &=\cos^2(\omega t/2)\,\frac1\omega\,
        \bigl(-\tan(\omega t_0/2)\bigr)
        \\
        \nonumber
  &= -\frac1\omega \biggl(\cos^2(\omega t_0/2)-\frac\omega2 \diff t\sin(\omega t_0)  + O(\diff t^2)\biggr)
        \,\tan(\omega t_0/2)
        \\
             \nonumber
  &= -\frac1{2\omega}\,\sin(\omega t_0) + \diff t\,\sin^2(\omega t_0/2) + O(\diff t^2)\ ,
  \\
  T_1 \bra1{\rho(t)}\ket1
  &= \sin^2(\omega t/2)\,\frac1\omega\,\frac1{\tan(\omega t_0/2)}
  \\
       \nonumber
  &= \frac1\omega \biggl( \sin^2(\omega t_0/2) + \frac\omega2 \diff t \sin(\omega t_0) + O(\diff t^2) \biggr)
      \frac1{\tan(\omega t_0/2)}
      \\
           \nonumber
  &= \frac1{2\omega} \sin(\omega t_0) + \diff t\,\cos^2(\omega t_0/2) + O(\diff t^2)\ .
\end{align}
Indeed, $\Tr[T\rho(t)] = t_0 + T_0 \bra0{\rho(t)}\ket0
+ T_1 \bra1{\rho(t)}\ket1 = t_0 + \diff t + O(\diff t^2)$.
However, as $t_0\to 0$, $t_0\neq 0$, we see that the value $T_1$ diverges
while the corresponding
probability $\bra1{\rho(t_0)}\ket1$ vanishes.  It turns out that that
large term times a tiny
term conspire to provide just exactly the required difference in
the expectation value so
that we have $\Tr[T\rho(t_0+\diff t)] \approx t_0 + \diff t$.
In the regime where limited data is available, the outcome $\ket1$ is never
observed, because it occurs too rarely.  As a consequence, only the first
outcome is observed and the reported estimate for the parameter, computed from the
averages of the samples of the outcomes of $T$, is
$t_0-[\tan(\omega t_0/2)]/\omega$.  This
value does not depend on $\diff t$; therefore, the estimation procedure
fails to accurately
reveal the value of $\diff t$ to the desired accuracy $1/\omega^2$.
The approach presented in this work aims to tackle settings such as the one above,
where the estimation procedure must rely on a \emph{limited number of samples}.

\section{Optimization through 
convex programming
}\label{ssec:optimization}
In this section, we consider the optimized success probabilities and show that we can cast them 
into the form of convex \emph{semi-infinite problems} (SIPs)~\cite{SemiInfinite,reemsten_ruckmann_SIP_book,charnes1980SIP}. 
These represent continuous analogues of semi-definite programs, which have to be discretized to be solved on a computer. We elaborate on this point in Section~\ref{ssec:numerical_implementation}. 

Throughout this section, we make use of the following standard forms of primal and dual semi-definite programs~\cite{khatri2020principles}: 
\begin{equation}\label{eqn:SDP_primal_dual_standard}
    \begin{array}{l l}
        \text{maximize} & \Tr[AX] \\[1ex] \text{subject to} & X\geq 0, \\[1ex] & \Phi[X]\leq B.
    \end{array}
    \qquad
    \begin{array}{l l}
        \text{minimize} & \Tr[BY] \\[1ex] \text{subject to} & Y\geq 0,\\[1ex] & \Phi^{\dagger}[Y]\geq A,
    \end{array}
\end{equation}
where $A$ and $B$ are Hermitian operators and $\Phi$ is a Hermiticity-preserving linear map. The convex semi-infinite programs that we consider in this work have the form of these semi-definite programs, with either given tuple $(A,B,\Phi)\equiv(A(t),B(t),\Phi(t))$, the optimization variables $X\equiv X(t)$ and $Y\equiv  Y(t)$, or both, being parameterized by $t\in\mathbb{R}$. As in the theory of semi-definite programs~\cite{BV96}, the notions of (weak and strong) duality and Slater's conditions carry over to such convex SIPs~\cite{charnes1980SIP,shapiro2009SIP}, and we make use these concept throughout in what follows.

\subsection{Measurement optimization}
It is quite straightforward that the optimal measurement in the Bayesian case can be determined using a 
convex program.

\begin{sproposition}[Bayesian measurement optimization]\label{prop:bayesian_sdp}
For a given set of states $\rho(t)$ with prior distribution $\mu(t)$ and a fixed window function $w(\tau)$, the optimal success probability $\eta^*(w,\mu,\rho)$ defined in Eq.~\eqref{eqn:opt_success_probability_bayesian} can be computed using the 
convex problem
\begin{align}
    \begin{array}{l l}
      \textnormal{maximize} & \int \diff t \, \Tr[ (w * [\mu \cdot \rho])(t) Q(t) ] \\ [1ex]
      \textnormal{subject to} & Q(t) \geq 0,\\[1ex]
       & \int \diff t \, Q(t) = \bbI.
    \end{array}
\end{align}
There is no duality gap and the associated dual program is
\begin{align}
    \begin{array}{l l}
      \textnormal{minimize} & \Tr [X] \\ [1ex]
      \textnormal{subject to} & X \geq 0,\\[1ex]
       & X \geq (w * [\mu \cdot \rho])(t)~~\forall~t.
    \end{array}
\end{align}
\end{sproposition}
Notably, the definition of the success probability arising from the dual coincides with a continuous version of the \enquote{least upper bound} for state discrimination~\cite{audenaert2014upper} which was also shown to be optimal in this case~\cite{yuen1975optimum}. 

Before proving Proposition~\ref{prop:bayesian_sdp}, let us make the following definitions. For a function $A:\mathbb{R}\to\mathrm{P}(\mathcal{H})$, $t\mapsto A(t)$, where $\mathrm{P}(\mathcal{H})$ is the set of positive semi-definite operators acting on a Hilbert space $\mathcal{H}$, and for a POVM $\{Q(t):t\in\mathbb{R}\}$, we let
\begin{align}
    G(A,Q)&\coloneqq\int\diff t \Tr[Q(t)A(t)],\\
    G_{\max}(A)&\coloneqq\sup_{\substack{Q(t)\geq 0\\\int\Diff{t}\!Q(t)=\mathbb{I}}} G(A,Q). \label{eqn:def_G_max}
\end{align}
From this, we see that the Bayesian success probability is given by
\begin{equation}\label{eq-Bayesian_success_Gmax}
    \eta^{\ast}(w,\mu,\rho)=G_{\max}(w\ast(\mu\cdot\rho)).
\end{equation}
The statement of Proposition~\ref{prop:bayesian_sdp} then follows from the following lemma about $G_{\max}$.

\begin{lemma}\label{lem:G_max_dual}
    For a function $A:\mathbb{R}\to\mathrm{P}(\mathcal{H})$, $t\mapsto A(t)$, the function $G_{\max}(A)$ can be computed via 
    a convex problem
    such that its dual formulation results in
    \begin{equation}\label{eq:G_max}
        G_{\max}(A)=\inf\{\Tr[Y]:Y\geq 0,\,Y\geq A(t)~~\forall~t\in\mathbb{R}\}=\inf_{\substack{\sigma\geq 0\\\Tr[\sigma]=1}}\sup_{t\in\mathbb{R}}\lambda_{\max}(\sigma^{-\frac{1}{2}}A(t)\sigma^{-\frac{1}{2}}),
    \end{equation}
    where $\lambda_{\max}$ denotes the largest eigenvalue.
\end{lemma}

\begin{remark}
    Note that, because $\lambda_{\max}(X)=\norm{X}_{\infty}$ for all $X\in\mathrm{P}(\mathcal{H})$, we equivalently have
    \begin{equation}
        G_{\max}(A)=\inf_{\substack{\sigma\geq 0\\\Tr[\sigma]=1}}\sup_{t\in\mathbb{R}}\norm{\sigma^{-\frac{1}{2}}A(t)\sigma^{-\frac{1}{2}}}_{\infty}.
    \end{equation}
\end{remark}

\begin{proof}
    By comparing the definition of $G_{\max}(A)$ in \eqref{eq:G_max} with the primal 
    convex problem
    in \eqref{eqn:SDP_primal_dual_standard}, we immediately see that $G_{\max}(A)$ is characterized by 
    a convex problem
    based on the following identifications:
    \begin{align}
        X&\coloneqq\int\Diff{t}\ketbra{t}{t}\otimes Q(t),\\
        A&\coloneqq\int\Diff{t}\ketbra{t}{t}\otimes A(t),\\
        \Phi[X]&\coloneqq \begin{pmatrix} \int\Diff{t}Q(t) & \\ & -\int\Diff{t}Q(t)   \end{pmatrix},\\
        B&\coloneqq \begin{pmatrix} \mathbb{I} & \\ & -\mathbb{I}   \end{pmatrix} .
    \end{align}
    Here, $\{\ket{t}\}_{t\in\mathbb{R}}$ refers to the (continuous) orthonormal basis of position-operator eigenstates, satisfying $\braket{t}{t'}=\delta(t-t')$ for all $t,t'\in\mathbb{R}$, where $\delta(t-t')$ is the Dirac delta function evaluated on $t-t'$.

    In order to obtain the dual 
    convex problem,
    we simply determine the adjoint of the map $\Phi$, as defined by the relation \begin{align}\Tr[Y\Phi[X]]=\Tr[\Phi^{\dagger}[Y]X].
    \end{align}
    Now, because $B$ is block-diagonal, it suffices to let the dual variable be of the form 
    \begin{align}
    Y=\begin{pmatrix} Y_1 & \\ & Y_2 \end{pmatrix},
    \end{align}
    such that $Y_1\geq 0$ and $Y_2\geq 0$. We then find that
    \begin{align}
        \Tr[Y\Phi[X]]&=\Tr\!\left[\begin{pmatrix} Y_1 & \\ & Y_2 \end{pmatrix}\begin{pmatrix} \int\Diff{t}Q(t) & \\ & -\int\Diff{t}Q(t) \end{pmatrix}\right]\\
        \nonumber 
        &=\Tr\!\left[(Y_1-Y_2)\left(\int\Diff{t}Q(t)\right)\right]\\
        &=\Tr\!\left[\left(\int\Diff{t}\ketbra{t}{t}\otimes(Y_1-Y_2)\right)\left(\int\Diff{t}\ketbra{t}{t}\otimes Q(t)\right)\right],
        \nonumber 
    \end{align}
    which means that we can identify $\Phi^{\dagger}[Y]$ as
    \begin{equation}
        \Phi^{\dagger}[Y]=\int\Diff{t}\ketbra{t}{t}\otimes (Y_1-Y_2).
    \end{equation}
    The dual %
    convex problem
    is thus
    \begin{equation}
        \begin{array}{l l} 
            \textnormal{minimize} & \Tr[Y_1-Y_2] \\[1ex] \textnormal{subject to} & Y_1\geq 0,\, Y_2\geq 0, \\[1ex] & \int\Diff{t}\ketbra{t}{t}\otimes (Y_1-Y_2)\geq \int\Diff{t}\ketbra{t}{t}\otimes A(t).
        \end{array}
    \end{equation}
    Now, the final constraint implies that $Y_1-Y_2\geq A(t)$ for all $t$. Furthermore, because only $Y_1-Y_2$ appears in the objective function and in the constraints, and because $A(t)\geq 0$ for all $t$, by a change of variable the dual optimization above simplifies to the 
    optimization problem
    \begin{equation}\label{eq:G_max_SDP_dual_pf1}
        \begin{array}{l l}
            \textnormal{minimize} & \Tr[Y] \\[1ex] \textnormal{subject to} & Y\geq 0, \\[1ex] & Y\geq A(t)~~\forall~t.
        \end{array}
    \end{equation}
    Finally, because strong duality holds, we have that the primal and dual 
    convex problems
    have the same optimal value, which concludes the proof of the first equality in \eqref{eq:G_max}.
    
    To prove the second equality in \eqref{eq:G_max}, we make another change of variable. For the 
    convex problem
    in \eqref{eq:G_max_SDP_dual_pf1}, we let $Y\equiv x\sigma$, such that $x\geq 0$, $\sigma\geq 0$ and $\Tr[\sigma]=1$. Then, $\Tr[Y]=x$, and the 
    convex problem
    in \eqref{eq:G_max_SDP_dual_pf1} becomes
    \begin{equation}\label{eq:G_max_SDP_dual_pf2}
        \begin{array}{l l}
            \textnormal{minimize} & x \\[1ex] \textnormal{subject to} & x \geq 0,\\[1ex] & x\sigma\geq A(t)~~\forall~t, \\[1ex] & \sigma\geq 0,\,\Tr[\sigma]=1.
        \end{array}
    \end{equation}  
    Next, observe that we can restrict the optimization to density operators $\sigma$ that have full rank, such that the inequality $x\sigma\geq A(t)$ is equivalent to $x\mathbb{I}\geq \sigma^{-\frac{1}{2}}A(t)\sigma^{-\frac{1}{2}}$. Furthermore, because $\sigma^{-\frac{1}{2}}A(t)\sigma^{-\frac{1}{2}}$ is positive semi-definite for all $t$, optimizing with respect to $x\geq 0$ is equivalent to optimizing with respect to all $x\in\mathbb{R}$. Therefore, because
    \begin{equation}
        \inf\{x:x\in\mathbb{R},\,x\mathbb{I}\geq H\}=\lambda_{\max}(H),
    \end{equation}
    where $H$ is an arbitrary Hermitian operator and $\lambda_{\max}(H)$ is the largest eigenvalue of $H$, we find that the 
    convex problem
    in \eqref{eq:G_max_SDP_dual_pf2} is equivalent to
    \begin{equation}
        \begin{array}{l l}
            \textnormal{minimize} & \sup_{t\in\mathbb{R}} \lambda_{\max}(\sigma^{-\frac{1}{2}}A(t)\sigma^{-\frac{1}{2}}) \\[1ex] \textnormal{subject to} & \sigma\geq 0,\,\Tr[\sigma]=1.
        \end{array}
    \end{equation}  
    This concludes the proof of the second equality in \eqref{eq:G_max}.
\end{proof}

We can also determine the optimal measurement in the minimax setting using a convex program.

\begin{sproposition}[Minimax measurement optimization]\label{prop:minimax_sdp}
For a given set of states $\rho(t)$ and a fixed window function $w(\tau)$, the optimal minimax success probability $\overline{\eta}^*(w,\rho)$ defined in Definition~\ref{def:success_probability_no_optimization}
can be computed using the following 
convex program:
\begin{align}\label{eq-minimax_succ_primal_supp}
    \begin{array}{l l}
      \textnormal{maximize} & \eta \\ [1ex]
      \textnormal{subject to} & Q(t) \geq 0, \, \eta\in[0,1],\\[1ex]
      & \int \diff t \, Q(t) = \bbI ,\\[1ex]
      & \Tr[ \rho(t) (w * Q)(t) ] \geq \eta~~\forall~t.
    \end{array}
\end{align}
There is no duality gap and the associated dual program is 
\begin{equation}\label{eq:minimax_sdp}
    \begin{array}{l l}
      \textnormal{minimize} & \Tr[X] \\ [1ex]
      \textnormal{subject to} & X \geq 0, \, \mu(t) \geq 0,\\
      [1ex]
      & \int \diff t \, \mu(t) = 1 ,\\[1ex]
      & X \geq \mu(t)(w\ast\rho)(t)~~\forall~t.
    \end{array}
\end{equation}
\end{sproposition}
\begin{proof}
    Comparing the primal problem in \eqref{eq-minimax_succ_primal_supp} with the primal problem in the left-hand side of \eqref{eqn:SDP_primal_dual_standard}, we can make the following identifications:
    \begin{align}
        X&\equiv\begin{pmatrix} \eta & \\ & \int\Diff{t}\ketbra{t}{t}\otimes Q(t)  \end{pmatrix},\\
        A&\equiv\begin{pmatrix} 1 & \\ & 0 \end{pmatrix},\\
        \Phi[X]&\equiv \begin{pmatrix} \int\Diff{t}\ketbra{t}{t}\left(\eta-\Tr[\rho(t)(w\ast Q)(t)]\right) & & \\ & \int\Diff{t}Q(t) & \\ & & -\int\Diff{t}Q(t)\end{pmatrix},\\
        B&\equiv \begin{pmatrix} 0 & & \\ & \mathbb{I} & \\ & & -\mathbb{I} \end{pmatrix}.
    \end{align}
    This establishes that the optimal minimax success probability is characterized by 
    a convex problem.

    Now, for the dual, because the operator $B$ defined above is block-diagonal, it suffices to let the dual variable be of the form
    \begin{equation}
        Y=\begin{pmatrix} \int\Diff{t}\mu(t)\ketbra{t}{t} & & \\ & Y_1 & \\ & & Y_2   \end{pmatrix},
    \end{equation}
    where $\mu(t)\geq 0$ for all $t$, and $Y_1,Y_2\geq 0$. Then, the adjoint of the map $\Phi$ defined above is given by the relation 
    \begin{align}
    \Tr[Y\Phi[X]]=\Tr[\Phi^{\dagger}[Y]X]. 
    \end{align}
    In particular,
    \begin{align}
        \Tr[Y\Phi[X]]&=\int\Diff{t}\mu(t)\left(\eta-\Tr[\rho(t)(w\ast Q)(t)]\right)+\int\Diff{t}\Tr[(Y_1-Y_2)Q(t)]\\
        \nonumber 
        &=\eta\int\Diff{t}\mu(t)-\int\Diff{t}\mu(t)\Tr[\rho(t)(w\ast Q)(t)]+\int\Diff{t}\Tr[(Y_1-Y_2)Q(t)].
    \end{align}
    Now,
    \begin{align}
        &\int\Diff{t}\mu(t)\Tr[\rho(t)(w\ast Q)(t)]\\
        \nonumber 
        &\quad =\int\Diff{t}\Diff{t'}\Tr[\mu(t)\rho(t)w(t-t')Q(t)]\\
        \nonumber 
        &\quad =\int\Diff{t'}\Tr\!\left[\left(\int\Diff{t}\mu(t)\rho(t)w(t-t')\right) Q(t')\right]\\
        \nonumber 
        &\quad =\int\Diff{t'}\Tr[(w\ast(\mu\cdot\rho))(t')Q(t')],
        \nonumber 
    \end{align}
    where to obtain the last line we have used the symmetry of the window function, \emph{i.e.}, $w(t-t')=w(t'-t)$. Therefore,
    \begin{equation}
        \Tr[Y\Phi[X]]=\eta\int\Diff{t}\mu(t)+\int\Diff{t}\Tr[(Y_1-Y_2-(w\ast(\mu\cdot\rho))(t))Q(t)],
    \end{equation}
    which implies that the adjoint of $\Phi$ is given by
    \begin{equation}
        \Phi^{\dagger}[Y]=\begin{pmatrix} \int\Diff{t}\mu(t) & \\ & \int\Diff{t}\ketbra{t}{t}\otimes (Y_1-Y_2-(w\ast(\mu\cdot\rho))(t) \end{pmatrix}.
    \end{equation}
    The dual %
    convex problem
    is therefore
    \begin{equation}
        \begin{array}{l l}
            \textnormal{minimize} & \Tr[Y_1-Y_2] \\[1ex]
            \textnormal{subject to} & Y_1\geq 0,\, Y_2\geq 0,\, \mu(t)\geq 0~~\forall~t,\\[1ex]
            & \int\Diff{t}\mu(t)\geq 1,\\[1ex]
            & \int\Diff{t}\ketbra{t}{t}\otimes(Y_1-Y_2-(w\ast(\mu\cdot\rho))(t))\geq 0.
        \end{array}
    \end{equation}
    The final constraint is equivalent to $Y_1-Y_2\geq (w\ast(\mu\cdot\rho))(t)$ for all $t$. Furthermore, because only $Y_1-Y_2$ appears in the objective function and in the constraints, and because $(w\ast(\mu\cdot\rho))(t)\geq 0$ for all $t$, the 
    convex problem
    above simplifies to the following:
    \begin{equation}
        \begin{array}{l l}
            \textnormal{minimize} & \Tr[Y] \\[1ex]
            \textnormal{subject to} & Y\geq 0,\,\mu(t)\geq 0~~\forall~t,\\[1ex]
            & \int\Diff{t}\mu(t)\geq 1,\\[1ex]
            & Y\geq (w\ast(\mu\cdot\rho))(t)~~\forall~t.
        \end{array}
    \end{equation}

    Finally, let us apply the complementary slackness condition~\cite{khatri2020principles} $\Phi^{\dagger}[Y]X=AX$. Based on the definitions above, this implies that $\eta\int\Diff{t}\mu(t)=\eta$, \textit{i.e.}, $\int\Diff{t}\mu(t)=1$. Therefore, we obtain
    \begin{equation}
        \begin{array}{l l}
            \text{minimize} & \Tr[Y] \\[1ex]
            \text{subject to} & Y\geq 0,\, \mu(t)\geq 0~~\forall~t,\\[1ex]
            & \int\diff{t}~\mu(t)= 1, \\[1ex]
            & Y\geq (w\ast(\mu\cdot\rho))(t)~~\forall~t,
        \end{array}
    \end{equation}
    as claimed. It is straightforward to verify strong duality, so that the primal and dual programs have the same optimal value.
\end{proof}

\subsection{Probe optimization}
In certain applications, especially considering real experiments where capabilities can be limited or pre-existing experiments should be used, the optimization of a probe state for \emph{fixed} measurement $Q(t)$ and encoding channel $\calN(t)$ 
needs to be considered. Note that fixing a measurement corresponds to fixing both the quantum part and the classical post-processing. In this case, the optimization takes a particularly simple form. Here, $\lVert
.\rVert_{\infty}$ denotes the infinity or spectral norm.

\begin{sproposition}[Probe state optimization]\label{prop:probe_state_opt}
For a given set of encoding channels $\calN(t)$ with prior probabilities $\mu(t)$, a fixed measurement $Q(t)$ and a fixed window function $w(\tau)$, the optimal success probability optimized over all probe states is given by 
\begin{align}\label{eqn:probe_opt_bayesian}
    \eta^{*}(w, \mu, \calN, Q) = \left\lVert \int \diff \mu(t) \, \calN^{\dagger}(t)[(w*Q)(t)] \right\rVert_{\infty}
\end{align}
and is achieved for the pure eigenstate of the operator $\int \diff \mu(t) \, \calN^{\dagger}(t)[(w*Q)(t)]$ corresponding to the largest eigenvalue. Similarly,
\begin{align}\label{eqn:probe_opt_minimax}
    \overline{\eta}^{*}(w, \calN, Q) = \min_t \left\lVert \calN^{\dagger}(t)[(w*Q)(t)] \right\rVert_{\infty},
\end{align} 
which is achieved for the pure eigenstate of the operator $\min_t \calN^{\dagger}(t)[(w*Q)(t)]$ corresponding to the largest eigenvalue.
\end{sproposition}
\begin{proof}
Writing out the objective of the optimization yields
\begin{align}
    \eta(w, \mu, \calN(\cdot)[\rho_0], Q) &= \int \diff \mu(t) \, \Tr [\calN(t)[\rho_0] (w*Q)(t) ] \\
    \nonumber
    &= \Tr \left[ \rho_0 \int \diff \mu(t) \, \calN^{\dagger}(t)[(w*Q)(t)] \right]
\end{align}
which is clearly maximized over quantum states $\rho_0$ for the largest eigenstate of $\int \diff \mu(t) \, \calN^{\dagger}(t)[(w*Q)(t)]$, as the latter is a positive semi-definite operator by construction. This proves the statement for the Bayesian case. The minimax case follows straightforwardly by noting that taking the infimum over priors will yield the minimum.
\end{proof}

\begin{figure}
    \centering
    \includegraphics[scale=1]{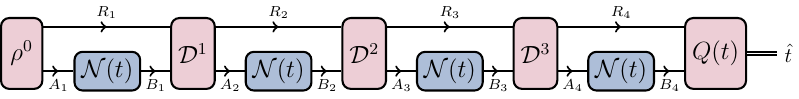}
    \caption{Depiction of an adaptive procedure for estimating the parameter $t$ encoded in a quantum channel $\mathcal{N}(t)$. The number of uses of the channel here is $n=4$.}\label{fig:adaptive_channel_uses}
\end{figure}

From the above proposition, we learn that the optimal probe states can always be assumed to be pure states, mixed states can only be admissible if the operators in Eqs.~\eqref{eqn:probe_opt_bayesian} and \eqref{eqn:probe_opt_minimax} have a degenerate subspace corresponding to the largest eigenvalue.

\subsection{Optimization with respect to strategies with definite causal order}

Typically, metrological problems are defined by a parametrized physical process modeled as a quantum channel $\calN(t)$ and an optimal \emph{combination} of probe state and measurement needs to be found to best extract the parameter $t$. This necessitates a joint optimization over both variables. Applying such an optimization naively, \textit{i.e.}, by optimizing over both variables in the expression
\begin{align}
    \eta^{*}(w, \mu, \calN) = \int \diff \mu(t) \, \Tr[ \calN(t)[\rho^0](w*Q)(t) ],
\end{align}
does not yield a semi-definite program, as it is quadratic in the variables $\rho^0$ and $Q(t)$. Another possible alternative would be to exploit the result of Proposition~\ref{prop:probe_state_opt} which gives the optimal probe for any measurement, and optimize over the measurement, \textit{i.e.},
\begin{align}
    \eta^{*}(w, \mu, \calN)=\sup \left.\left\{ \left\lVert \int \diff \mu(t) \, \calN^{\dagger}(t)[(w*Q)(t)] \right\rVert_{\infty} \, \right| \, Q(t) \geq 0, \int \diff t \, Q(t) = \bbI \right\}.
\end{align}
This, however, corresponds to \emph{maximizing} a convex function, and hence is not a convex optimization problem.

The above arguments might suggest that performing a joint optimization is impossible; however, we can circumvent these obstacles by a change of perspective. This is because the repeated use of a quantum channel, possibly in an adaptive way as shown in Fig.~\ref{fig:adaptive_channel_uses}, can be described by a \textit{quantum comb}~\cite{CDP09}, also known as a \textit{quantum strategy}~\cite{GW07}. In the following, we exploit the fact that the set of quantum combs is convex and formulate a 
convex problem
for the joint optimization of probe state and measurement.

To see how this works, before describing the general case, let us consider the example described above, with an input state $\rho^0$ and a measurement $t\mapsto Q(t)$, both of which we wish to optimize jointly. This scenario corresponds to the adaptive strategy depicted in Fig.~\ref{fig:adaptive_channel_uses} with $n=1$. The probability of the outcome $t'$ of the measurement, when the channel is $\mathcal{N}(t)$, is given by
\begin{equation}\label{eq-supp_meas_prob_1comb}
    \Tr[Q_{RB}(t')\mathcal{N}_{A\to B}(t)[\rho_{RA}^0]],
\end{equation}
where $A$ and $B$ are the input and output systems, respectively, of the channel, and $R$ is a memory system of arbitrary dimension. Let us write the output state $\mathcal{N}_{A\to B}(t)[\rho_{RA}^0]$ in terms of the Choi representation $C_{AB}^{\mathcal{N}(t)}$ of $\mathcal{N}(t)$ as follows~\cite{khatri2020principles}:
\begin{equation}
    \mathcal{N}_{A\to B}(t)[\rho_{RA}^0]=\Tr_A[(\rho_{RA}^{T_A}\otimes\mathbb{I}_B)(\mathbb{I}_R\otimes C_{AB}^{\mathcal{N}(t)})].
\end{equation}
Therefore, the probability in \eqref{eq-supp_meas_prob_1comb} can be written as $\Tr[P_{AB}^{T}(t')C_{AB}^{\mathcal{N}(t)}]\equiv P_{AB}(t')\star C_{AB}^{\mathcal{N}(t)}$, where the ``$\star$'' refers to the link product~\cite{CDP09} and
\begin{equation}
    P_{AB}(t')\coloneqq\Tr_R[(\rho_{RA}^0\otimes\mathbb{I}_B)(\mathbb{I}_B\otimes Q_{RA}^{T_A}(t')].
\end{equation}
Now, because $t\mapsto Q(t)$ is a POVM, we find that
\begin{equation}
    \int\diff t\, P_{AB}(t) = \Tr_R[\rho_{RA}^0]\otimes\mathbb{I}_B.
\end{equation}
In other words, for every state-measurement pair $(\rho_{RA}^0,t\mapsto Q(t))$, we can construct a positive semi-definite operator $t\mapsto P_{AB}(t)$ such that $\int\diff t\, P_{AB}(t)=\sigma_A\otimes\mathbb{I}_B$ for some quantum state $\sigma_A$. The converse is also true~\cite{GW07}~(Theorem~6), which implies that the optimal success probability, optimized with respect to both input probe state and measurement, can be obtained as the solution to the following convex problem (in the Bayesian setting):
\begin{equation}
    \begin{array}{l l}
        \text{maximize} & \int\diff t\, \mu(t)\Tr[(w\ast P_{AB})(t) C_{AB}^{\mathcal{N}(t)}] \\[1ex]
        \text{subject to} & P_{AB}(t)\geq 0~~\forall~t,\\[1ex]
        & \int\diff t\, P_{AB}(t)=\sigma_A\otimes\mathbb{I}_B,\\[1ex]
        & \sigma_A\geq 0,\,\Tr[\sigma_A]=1,
    \end{array}
\end{equation}
where $t\mapsto\mu(t)$ is the prior probability density function. In the minimax setting, the optimal success probability can be obtained as the solution to the following convex problem:
\begin{equation}
    \begin{array}{l l}
       \text{maximize}  &  \eta \\[1ex]
       \text{subject to}  &  \Tr[(w\ast P_{AB})(t) C_{AB}^{\mathcal{N}(t)}]\geq \eta~~\forall~t,\\[1ex]
       & \eta\in[0,1],\\[1ex]
       & P_{AB}(t)\geq 0~~\forall~t,\\[1ex]
       & \int\diff t\, P_{AB}(t)=\sigma_A\otimes\mathbb{I}_B,\\[1ex]
       & \sigma_A\geq 0,\,\Tr[\sigma_A]=1.
    \end{array}
\end{equation}
We provide a formal proof of these results, in the general context of multiple adaptive uses of the channel $t\mapsto\mathcal{N}(t)$, in Propositions~\ref{prop-succ_prob_combs_Bayesian} and \ref{prop-succ_prob_combs} below.

\begin{figure}
    \centering
    \includegraphics[scale=1]{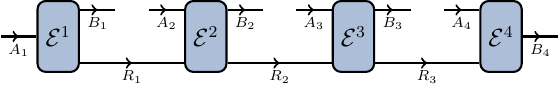}
    \caption{A general quantum comb with $n=4$ elements. The input and output systems are $A_j$ and $B_j$, respectively, and the memory systems are $R_j$.}\label{fig:quantum_comb}
\end{figure}

A general quantum comb of length $n=4$ is shown in Fig.~\ref{fig:quantum_comb}. The comb is simply a concatenation of quantum channels $\mathcal{E}^j$, with input systems $A_j$, output systems $B_j$, and memory systems $R_j$. We refer to the comb using the notation $\mathcal{E}^{[n]}$. It can be shown that the Choi representation $ C_{A_1^nB_1^n}^{\mathcal{E}^{[n]}}$ of the comb satisfies the following constraints:
\begin{align}
	\Tr_{B_n}\!\left[ C^{\mathcal{E}^{[n]}}_{A_1^nB_1^n}\right]&= C^{\mathcal{E}^{[n-1]}}_{A_1^{n-1}B_1^{n-1}}\otimes\mathbb{I}_{A_n},\label{eqn:comb_constr_1}\\
	\Tr_{B_k}\!\left[ C^{\mathcal{E}^{[k]}}_{A_1^kB_1^k}\right]&= C^{\mathcal{E}^{[k-1]}}_{A_1^{k-1}B_1^{k-1}}\otimes\mathbb{I}_{A_k}~~\forall~k\in\{2,3,\dotsc,r-1\},\label{eqn:comb_constr_2}\\
	\Tr_{B_1}\!\left[ C^{\mathcal{E}^{[1]}}_{A_1B_1}\right]&=\mathbb{I}_{A_1}.\label{eqn:comb_constr_3}
\end{align}
These constraints tell us that by iteratively tracing out the output systems we obtain Choi representations of the same comb but with one fewer round each time. Conversely, every set $\{C^{(k)}_{A_1^kB_1^k}\}_{k=1}^{n}$ of positive semi-definite operators satisfying the constraints in \eqref{eqn:comb_constr_1}--\eqref{eqn:comb_constr_3} gives us Choi representations corresponding to a quantum comb with length $n$; see
Ref.~\cite{GW07}~(Theorem~6). Note that these operators do not give us the Choi representations of the channels $\mathcal{E}^j$ themselves, only the Choi representations of the combs obtained by concatenating the channels in the manner shown in Fig.~\ref{fig:quantum_comb}.

Following Ref.~\cite{GW07}, for every $n\in\{1,2,\dotsc\}$, we define the set $\mathsf{S}_n(A_1^n,B_1^n)$ as 
\begin{multline}\label{eq-strategy_combs}
    \mathsf{S}_n(A_1^n,B_1^n)\coloneqq\left\{C_{A_1^nB_1^n}^{(n)}: C_{A_1^nB_1^n}^{(n)}\geq 0,\Tr_{B_k}[C_{A_1^kB_1^k}^{(k)}]=C_{A_1^{k-1}B_1^{k-1}}^{(k-1)}\otimes\mathbb{I}_{A_k},\,C_{A_1^{k-1}B_1^{k-1}}^{(k-1)}\geq 0~~\forall~k\in\{2,3,\dotsc,n\},\right.\\
    \left.\Tr_{B_1}[C_{A_1B_1}^{(1)}]=\mathbb{I}_{A_1},\,C_{A_1B_1}^{(1)}\geq 0 \right\}.
\end{multline}
In other words, $\mathsf{S}_n(A_1^n,B_1^n)$ is the set of all Choi representations of length-$n$ quantum combs with input systems $A_1,A_2,\dotsc,A_n$ and output systems $B_1,B_2,\dotsc,B_n$. Similarly, for the set of combs consisting of quantum state preparation at the beginning, known as \textit{co-strategies} (see the red comb in Fig.~\ref{fig:adaptive_channel_uses}), we let
\begin{multline}\label{eq-co_strategy_combs}
    \overline{\mathsf{S}}_n(A_1^{n-1},B_1^n)\coloneqq\left\{C_{A_1^{n-1}B_1^n}^{(n)}:C_{A_1^{n-1}B_1^n}^{(n)}\geq 0,\, \Tr_{B_k}[C_{A_1^{k-1}B_1^k}^{(k)}]=C_{A_1^{k-2}B_1^{k-1}}^{(k-1)}\otimes\mathbb{I}_{A_{k-1}},\right.\\\left.C_{A_1^{k-2}B_1^{k-1}}^{(k-1)}\geq 0~~\forall~k\in\{3,4,\dotsc,n\},\, \Tr_{B_2}[C_{A_1B_1^2}^{(2)}]=C_{B_1}^{(1)}\otimes\mathbb{I}_{A_1},\, C_{B_1}^{(1)}\geq 0,\, \Tr_{B_1}[C_{B_1}^{(1)}]=1\right\}.
\end{multline}

Now, returning to the parameter estimation problem, note that in Fig.~\ref{fig:adaptive_channel_uses}, we have a concatenation of two combs: One corresponding to the strategy itself (in red), and the other corresponding to the channels $\mathcal{N}(t)$ containing the parameter to be estimated. The Choi representation of the latter is simply a tensor product
\begin{equation}\label{eq-channel_n_uses_Choi}
     C_{A_1^nB_1^{n}}^{\mathcal{N}(t)^{[n]}}=\bigotimes_{j=1}^n C_{A_jB_{j}}^{\mathcal{N}(t)},
\end{equation}
because the channel uses are independent of each other. Using this, we can obtain the optimal success probability as the following primal-dual pair of convex problems, concretely semi-infinite problems, in both the Bayesian and minimax settings.
(We refer to Ref.~\cite{Chiribella12} for a similar result.)

\begin{sproposition}[Bayesian optimization of adaptive causal strategies]\label{prop-succ_prob_combs_Bayesian}
    Let $ C_{A_1^nB_1^n}^{\mathcal{N}(t)^{[n]}}$ be the Choi representation of the comb $\mathcal{N}(t)^{[n]}$ corresponding to $n$ uses of the paramterized quantum channel $t\mapsto \mathcal{N}(t)$, as shown in Fig.~\ref{fig:adaptive_channel_uses}. Also, let $t\mapsto\mu(t)$ be a prior probability density function. Then, the optimal Bayesian strategy for estimating the parameter $t$ can be determined using the %
    convex problem
    \begin{equation}\label{eqn:succ_prob_comb_Bayesian_primal}
        \begin{array}{l l}
            \textnormal{maximize} & \displaystyle\int\diff t\,\mu(t) \Tr[P_{A_1^nB_1^n}(t)(w\ast  C_{A_1^nB_1^n}^{\mathcal{N}^{[n]}})(t)] \\[0.2cm]
            \textnormal{subject to} & P_{A_1^{n}B_1^n}(t)\geq 0~~\forall~t,\\[0.2cm]
            & \int\diff t\,P_{A_1^{n}B_1^n}(t)=C_{A_1^nB_1^{n-1}}^{(n)}\otimes \mathbb{I}_{B_{n}},\\[0.2cm]
            & C_{A_1^nB_1^{n-1}}^{(n)}\in\overline{\mathsf{S}}_n(B_1^{n-1},A_1^n),
        \end{array}
    \end{equation}
    where the variable $C_{A_1^nB_1^{n-1}}^{(n)}$ represents a length-$n$ co-strategy quantum comb, excluding the measurement, (see the red comb in Fig.~\ref{fig:adaptive_channel_uses}). The variables $P_{A_1^{n}B_1^n}(t)$ correspond to the measurement.
   
    The 
    convex program dual to the one above is
    \begin{equation}
        \begin{array}{l l} \textnormal{minimize} & \lambda \\[0.2cm] \textnormal{subject to} & \lambda\geq 0, \\[0.1cm] & \lambda Y_{A_1^nB_1^n}^{(n)}\geq\mu(t)\left(w\ast C_{A_1^nB_1^n}^{\mathcal{N}(t)^{[n]}}\right)(t)~~\forall~t,\\[0.25cm]
        & Y_{A_1^nB_1^n}^{(n)}\in\mathsf{S}_n(A_1^n,B_1^n),\end{array}
    \end{equation}
    where the optimization is with respect to $\lambda\geq 0$ and length-$n$ quantum combs represented by the operator $Y_{A_1^nB_1^n}^{(n)}$. Furthermore, strong duality holds, so that the primal and dual problems have the same optimal value.
\end{sproposition}

\begin{proof}
    Starting with the primal problem in \eqref{eqn:succ_prob_comb_Bayesian_primal}, we can cast it into the standard form in \eqref{eqn:SDP_primal_dual_standard} as 
    \begin{align}
        X&=\left(\int\diff t\,\ketbra{t}{t}\otimes P_{A_1^nB_1^n}(t)\right)\oplus\left(\sum_{k=1}^n\ketbra{k}{k}\otimes C_{A_1^kB_1^{k-1}}^{(k)}\right),\\
        A&=\left(\int\diff t\,\ketbra{t}{t}\otimes \mu(t)(w\ast C_{A_1^nB_1^n}^{\mathcal{N}^{[n]}})(t)\right)\oplus\left(\sum_{k=1}^n\ketbra{k}{k}\otimes 0\right),\\
        \nonumber
        \Phi[X]&=\ketbra{0,0}{0,0}\otimes\left(\int\diff t\,P_{A_1^nB_1^n}(t)-C_{A_1^nB_1^{n-1}}^{(n)}\otimes\mathbb{I}_{B_n}\right)\\
        \nonumber
        &\quad +\ketbra{0,1}{0,1}\otimes\left(-\int\diff t\,P_{A_1^nB_1^n}(t)+C_{A_1^nB_1^{n-1}}^{(n)}\otimes\mathbb{I}_{B_n}\right)\\
        \nonumber
        &\quad +\ketbra{1,0}{1,0}\otimes\Tr_{A_1}[C_{A_1}^{(1)}]+\ketbra{1,1}{1,1}\otimes(-\Tr_{A_1}[C_{A_1}^{(1)}])\\
        \nonumber
        &\quad +\sum_{k=2}^n\ketbra{k,0}{k,0}\otimes\left(\Tr_{A_k}[C_{A_1^kB_1^{k-1}}^{(k)}]-C_{A_1^{k-1}B_1^{k-2}}^{(k-1)}\otimes\mathbb{I}_{B_{k-1}}\right)\\
        \nonumber
        &\quad +\sum_{k=2}^n\ketbra{k,1}{k,1}\otimes\left(-\Tr_{A_k}[C_{A_1^kB_1^{k-1}}^{(k)}]+C_{A_1^{k-1}B_1^{k-2}}^{(k-1)}\otimes\mathbb{I}_{B_{k-1}}\right),\nonumber\\
        B&=\ketbra{0,0}{0,0}\otimes 0+\ketbra{0,1}{0,1}\otimes 0 +\ketbra{1,0}{1,0}\otimes 1 +\ketbra{1,1}{1,1}\otimes (-1)\\
        &\quad + \sum_{k=2}^n\ketbra{k,0}{k,0}\otimes 0 + \sum_{k=2}^n\ketbra{k,1}{k,1}\otimes 0.
        \nonumber
    \end{align}
    Now, without loss of generality, we can let the dual variable $Y$ have the block-diagonal form
    \begin{multline}
        Y=\ketbra{0,0}{0,0}\otimes Y_{A_1^nB_1^n}^{(1)}+\ketbra{0,1}{0,1}\otimes Y_{A_1^nB_1^n}^{(2)}+\ketbra{1,0}{1,0}\otimes\alpha +\ketbra{1,1}{1,1}\otimes\beta\\
        +\sum_{k=2}^n\ketbra{k,0}{k,0}\otimes\widetilde{C}_{A_1^{k-1}B_1^{k-1}}^{(k,0)}+\sum_{k=2}^n\ketbra{k,1}{k,1}\otimes \widetilde{C}_{A_1^{k-1}B_1^{k-1}}^{(k,1)}
    \end{multline}
    From this, we obtain
    \begin{align}
        \Tr[Y\Phi[X]]&=\Tr\!\left[\left(\int\diff t\,P_{A_1^nB_1^n}(t)\right)\left(Y_{A_1^nB_1^n}^{(1)}-Y_{A_1^nB_1^n}^{(2)}\right)\right]\\
        \nonumber
        &\quad +\Tr\!\left[\left(C_{A_1^nB_1^{n-1}}^{(n)}\otimes\mathbb{I}_{B_n}\right)\left(-Y_{A_1^nB_1^n}^{(1)}+Y_{A_1^nB_1^n}^{(2)}\right)\right]\\
        \nonumber
        &\quad +\Tr_{A_1}[C_{A_1}^{(1)}](\alpha-\beta)\\
        \nonumber
        &\quad +\sum_{k=2}^{n}\Tr\!\left[\Tr_{A_k}\!\left[C_{A_1^kB_1^{k-1}}^{(k)}\right]\left(\widetilde{C}_{A_1^{k-1}B_1^{k-1}}^{(k,0)}-\widetilde{C}_{A_1^{k-1}B_1^{k-1}}^{(k,1)}\right)\right]\\
        \nonumber
        &\quad+\sum_{k=2}^n\Tr\!\left[\left(C_{A_1^{k-1}B_1^{k-2}}^{(k-1)}\otimes\mathbb{I}_{B_{k-1}}\right)\left(-\widetilde{C}_{A_1^{k-1}B_1^{k-1}}^{(k,0)}+\widetilde{C}_{A_1^{k-1}B_1^{k-1}}^{(k,1)}\right)\right]\\
        \nonumber
        &=\Tr\!\left[\left(\int\diff t\,\ketbra{t}{t}\otimes P_{A_1^nB_1^n}(t)\right)\left(\int\diff t\,\ketbra{t}{t}\otimes\left(Y_{A_1^nB_1^n}^{(1)}-Y_{A_1^nB_1^n}^{(2)}\right)\right)\right]\\
        \nonumber
        &\quad +\Tr\!\left[C_{A_1}^{(1)}\left( (\alpha-\beta)\mathbb{I}_{A_1}+\Tr_{B_1}\!\left[-\widetilde{C}_{A_1B_1}^{(2,0)}+\widetilde{C}_{A_1B_1}^{(2,1)}\right]  \right)\right]\\
        \nonumber
        &\quad +\sum_{k=2}^{n-1}\Tr\!\left[C_{A_1^kB_1^{k-1}}^{(k)}\left(\left(\widetilde{C}_{A_1^{k-1}B_1^{k-1}}^{(k,0)}-\widetilde{C}_{A_1^{k-1}B_1^{k-1}}^{(k,1)}\right)\otimes\mathbb{I}_{A_k}+\Tr_{B_k}\!\left[-\widetilde{C}_{A_1^kB_1^k}^{(k+1,0)}+\widetilde{C}_{A_1^kB_1^k}^{(k+1,1)}\right]\right)\right]\\
        \nonumber
        &\quad +\Tr\!\left[C_{A_1^{n-1}B_1^{n-1}}^{(n,0)}\left(\left(\widetilde{C}_{A_1^{n-1}B_1^{n-1}}^{(n,0)}-\widetilde{C}_{A_1^{n-1}B_1^{n-1}}^{(n,1)}\right)\otimes\mathbb{I}_{A_n}+\Tr_{B_n}\!\left[-Y_{A_1^nB_1^n}^{(1)}+Y_{A_1^nB_1^n}^{(2)}\right]\right)\right].
        \nonumber
    \end{align}
    This implies that
    \begin{align}
        \Phi^{\dagger}[Y]&=\int\diff t\,\ketbra{t}{t}\otimes\left(Y_{A_1^nB_1^n}^{(1)}-Y_{A_1^nB_1^n}^{(2)}\right)\\
        \nonumber
        &\quad\oplus\left(\left(\widetilde{C}_{A_1^{n-1}B_1^{n-1}}^{(n,0)}-\widetilde{C}_{A_1^{n-1}B_1^{n-1}}^{(n,1)}\right)\otimes\mathbb{I}_{A_n}+\Tr_{B_n}\!\left[-Y_{A_1^nB_1^n}^{(1)}+Y_{A_1^nB_1^n}^{(2)}\right]\right)\\
        \nonumber
        &\quad\oplus\sum_{k=2}^{n-1}\ketbra{k}{k}\otimes\left(\left(\widetilde{C}_{A_1^{k-1}B_1^{k-1}}^{(k,0)}-\widetilde{C}_{A_1^{k-1}B_1^{k-1}}^{(k,1)}\right)\otimes\mathbb{I}_{A_k}+\Tr_{B_k}\!\left[-\widetilde{C}_{A_1^kB_1^k}^{(k+1,0)}+\widetilde{C}_{A_1^kB_1^k}^{(k+1,1)}\right]\right)\\
        \nonumber
        &\oplus\left((\alpha-\beta)\mathbb{I}_{A_1}+\Tr_{B_1}\!\left[-\widetilde{C}_{A_1B_1}^{(2,0)}+\widetilde{C}_{A_1B_1}^{(2,1)}\right]\right).
        \nonumber
    \end{align}
    The dual problem is therefore
    \begin{equation}
        \begin{array}{l l}
            \text{minimize} & \alpha-\beta \\[1ex]
            \text{subject to} & \alpha\geq 0,\,\beta\geq 0,\,Y_{A_1^nB_1^n}^{(1)}\geq 0,Y_{A_1^nB_1^n}^{(2)}\geq 0,\,\widetilde{C}_{A_1^{k-1}B_1^{k-1}}^{(k,0)}\geq 0,\,\widetilde{C}_{A_1^{k-1}B_1^{k-1}}^{(k,1)}\geq 0,\,k\in\{1,2,\dotsc,n\},\\[1ex]
            & \Tr_{B_1}\!\left[\widetilde{C}_{A_1B_1}^{(2,0)}-\widetilde{C}_{A_1B_1}^{(2,1)}\right]\leq(\alpha-\beta)\mathbb{I}_{A_1},\\[1ex]
            & \Tr_{B_k}\!\left[\widetilde{C}_{A_1^kB_1^k}^{(k+1,0)}-\widetilde{C}_{A_1^kB_1^k}^{(k+1,1)}\right]\leq\left(\widetilde{C}_{A_1^{k-1}B_1^{k-1}}^{(k,0)}-\widetilde{C}_{A_1^{k-1}B_1^{k-1}}^{(k,1)}\right)\otimes\mathbb{I}_{A_k},\,k\in\{2,3,\dotsc,n-1\},\\[1ex]
            & \Tr_{B_n}\!\left[Y_{A_1^nB_1^n}^{(1)}-Y_{A_1^nB_1^n}^{(2)}\right]\leq\left(\widetilde{C}_{A_1^{n-1}B_1^{n-1}}^{(n,0)}-\widetilde{C}_{A_1^{n-1}B_1^{n-1}}^{(n,1)}\right)\otimes\mathbb{I}_{A_n},\\[1ex]
            & Y_{A_1^nB_1^n}^{(1)}-Y_{A_1^nB_1^n}^{(2)}\geq \mu(t)\left(w\ast C_{A_1^nB_1^n}^{\mathcal{N}^{[n]}}\right)(t)~~\forall~t.
        \end{array}
    \end{equation}
    It is straightforward to verify that strong duality holds, which means that the primal and dual problems have the same optimal value.
    
    Let us now make several simplifications to the dual optimization problem. We start with the following change of variables:
    \begin{align}
        \lambda&\equiv \alpha-\beta,\\
        Y_{A_1^kB_1^k}^{(k)}&\equiv \widetilde{C}_{A_1^kB_1^k}^{(k+1,0)}-\widetilde{C}_{A_1^kB_1^k}^{(k+1,1)},\quad k\in\{1,2,\dotsc,n-1\},\\
        Y_{A_1^nB_1^n}^{(n)}&\equiv Y_{A_1^nB_1^n}^{(1)}-Y_{A_1^nB_1^n}^{(2)}.
    \end{align}
    Then, because the operator $\mu(t)(w\ast C_{A_1^nB_1^n}^{\mathcal{N}^{[n]}})(t)$ is positive semi-definite for all $t$, we have that $Y_{A_1^nB_1^n}^{(n)}$ is positive semi-definite. Due to the second-last constraint in the above convex problem, this implies that $Y_{A_1^{n-1}B_1^{n-1}}^{(n-1)}\geq 0$, which in turn, from the third constraint in the convex problem above, implies that $Y_{A_1^kB_1^k}^{(k)}\geq 0$ for all $k\in\{1,2,\dotsc,n-2\}$, such that finally $\lambda\geq 0$ is also implied. Therefore, the convex problem above simplifies to
    \begin{equation}
        \begin{array}{l l}
            \text{minimize} & \lambda \\[1ex]
            \text{subject to} & \lambda\geq 0,\, Y_{A_1^kB_1^k}^{(k)}\geq 0~~\forall~k\in\{1,2,\dotsc,n\},\\[1ex]
            & \Tr_{B_1}[Y_{A_1B_1}^{(1)}]\leq\lambda\mathbb{I}_{A_1},\\[1ex]
            & \Tr_{B_k}[Y_{A_1^kB_1^k}^{(k)}]\leq Y_{A_1^{k-1}B_1^{k-1}}^{(k-1)}\otimes\mathbb{I}_{A_k}~~\forall~k\in\{2,3,\dotsc,n\},\\[1ex]
            & Y_{A_1^nB_1^n}^{(n)}\geq \mu(t)(w\ast C_{A_1^nB_1^n}^{\mathcal{N}^{[n]}})(t)~~\forall~t.
        \end{array}
    \end{equation}
    Let us now argue that the inequality constraints $\Tr_{B_1}[Y_{A_1B_1}^{(1)}]\leq\lambda\mathbb{I}_{A_1}$ and $\Tr_{B_k}[Y_{A_1^kB_1^k}^{(k)}]\leq Y_{A_1^{k-1}B_1^{k-1}}^{(k-1)}\otimes\mathbb{I}_{A_k}$, $k\in\{2,3,\dotsc,n\}$, for every feasible set of variables, can be made into equality constraints, without changing the value $\lambda$ of the objective function. First, by adding an appropriate positive multiple of the identity to $Y_{A_1B_1}^{(1)}$, we can obtain an operator $\widetilde{Y}_{A_1B_1}^{(1)}\geq 0$ such that $\widetilde{Y}_{A_1B_1}^{(1)}\geq Y_{A_1B_1}^{(1)}$ and $\Tr_{B_1}[\widetilde{Y}_{A_1B_1}^{(1)}]=\lambda\mathbb{I}_{A_1}$. Now, because $\widetilde{P}\geq P\Rightarrow \widetilde{P}\otimes\mathbb{I}\geq P\otimes\mathbb{I}$ for all $P\geq 0$, we have that 
    \begin{equation}
    \widetilde{Y}_{A_1B_1}^{(1)}\otimes\mathbb{I}_{A_2}\geq Y_{A_1B_1}^{(1)}\otimes\mathbb{I}_{A_1}\geq\Tr_{B_2}[Y_{A_1^2B_1^2}^{(2)}]. 
    \end{equation}
    This implies that there exists a $Q_{A_1^2B_1}^{(2)}\geq 0$ such that $\widetilde{Y}_{A_1B_1}^{(1)}\otimes\mathbb{I}_{A_2}=\Tr_{B_2}[Y_{A_1^2B_1^2}^{(2)}]+Q_{A_1^2B_1}^{(2)}$. Letting 
    \begin{equation}
    R_{A_1^2B_1^2}^{(2)}\coloneqq Q_{A_1^2B_1}^{(2)}\otimes\frac{\mathbb{I}_{B_2}}{d_{B_2}}
    \end{equation}
    and $\widetilde{Y}_{A_1^2B_1^2}^{(2)}\coloneqq Y_{A_1^2B_1^2}^{(2)}+R_{A_1^2B_1^2}^{(2)}$, we have that $\widetilde{Y}_{A_1^2B_1^2}^{(2)}\geq Y_{A_1^2B_1^2}^{(2)}$ and $\Tr_{B_2}[\widetilde{Y}_{A_1^2B_1^2}^{(2)}]=\widetilde{Y}_{A_1B_1}^{(1)}\otimes\mathbb{I}_{A_1}$. We can proceed analogously for all $k\in\{3,4,\dotsc,n\}$, defining new variables $\widetilde{Y}_{A_1^kB_1^k}^{(k)}$ such that $\widetilde{Y}_{A_1^kB_1^k}^{(k)}\geq Y_{A_1^kB_1^k}^{(k)}$ and $\Tr_{B_k}[\widetilde{Y}_{A_1^kB_1^k}^{(k)}]=\widetilde{Y}_{A_1^{k-1}B_1^{k-1}}^{(k-1)}\otimes\mathbb{I}_{A_k}$. In particular, for $k=n$, we obtain the constraint 
    \begin{equation}
    \widetilde{Y}_{A_1^nB_1^n}^{(n)}\geq Y_{A_1^nB_1^n}^{(n)}\geq (w\ast C_{A_1^nB_1^n}^{\mathcal{N}^{[n]}})(t)
    \end{equation}
    for all $t$. With this change of variables, the value $\lambda$ of the objective function does not change. Therefore, we have shown that the convex problem above is equivalent to
    \begin{equation}
        \begin{array}{l l}
            \text{maximize} & \lambda \\[1ex]
            \text{subject to} & \lambda\geq 0,\,Y_{A_1^kB_1^k}^{(k)}\geq 0~~\forall~k\in\{1,2,\dotsc,n\},\\[1ex]
            & \Tr_{B_1}[Y_{A_1B_1}^{(1)}]=\lambda\mathbb{I}_{A_1},\\[1ex]
            & \Tr_{B_k}[Y_{A_1^kB_1^k}^{(k)}]=Y_{A_1^{k-1}B_1^{k-1}}^{(k-1)}\otimes\mathbb{I}_{A_k}~~\forall~k\in\{2,3,\dotsc,n\},\\[1ex]
            & Y_{A_1^nB_1^n}^{(n)}\geq \mu(t)(w\ast C_{A_1^nB_1^n}^{\mathcal{N}^{[n]}})(t)~~\forall~t.
        \end{array}
    \end{equation}
    Finally, let us make one more change of variables. Let $\widetilde{Y}_{A_1^kB_1^k}^{(k)}=\frac{1}{\lambda}Y_{A_1^kB_1^k}^{(k)}$ for all $k\in\{1,2,\dotsc,n\}$. Then, we find that $\Tr_{B_1}[\widetilde{Y}_{A_1B_1}^{(1)}]=\mathbb{I}_{A_1}$, $\Tr_{B_k}[\widetilde{Y}_{A_1^kB_1^k}^{(k)}]=\widetilde{Y}_{A_1^kB_1^k}^{(k-1)}\otimes\mathbb{I}_{A_k}$ for all $k\in\{2,3,\dotsc,n\}$, and $\lambda Y_{A_1^nB_1^n}^{(n)}\geq \mu(t)(w\ast C_{A_1^nB_1^n}^{\mathcal{N}^{[n]}})(t)$ for all $t$. To conclude, we have that $Y_{A_1^nB_1^n}^{(n)}\in\mathsf{S}_n(A_1^n,B_1^n)$, based on the definition in \eqref{eq-strategy_combs}, which gives us the desired dual problem.
\end{proof}

\begin{sproposition}[Minimax optimization of adaptive causal strategies]\label{prop-succ_prob_combs}
    Let $ C_{A_1^nB_1^n}^{\mathcal{N}(t)^{[n]}}$ be the Choi representation of the comb $\mathcal{N}(t)^{[n]}$ corresponding to $n$ uses of the paramterized quantum channel $t\mapsto \mathcal{N}(t)$, as shown in Fig.~\ref{fig:adaptive_channel_uses}. Then, the optimal minimax strategy for estimating the parameter $t$ can be determined using the %
    convex problem
    \begin{equation}\label{eqn:succ_prob_comb_primal}
        \begin{array}{l l}
            \textnormal{maximize} & \eta \\[0.2cm]
            \textnormal{subject to} & \Tr[ C_{A_1^nB_1^n}^{\mathcal{N}(t)^{[n]}}(w\ast P_{A_1^nB_1^n})(t)]\geq\eta~~\forall~t,\\[0.3cm]
            &\eta\geq 0,\\[0.2cm]
            & P_{A_1^{n}B_1^n}(t)\geq 0~~\forall~t,\\[0.2cm]
            & \int \textnormal{d}t~P_{A_1^{n}B_1^n}(t)=C_{A_1^nB_1^{n-1}}^{(n)}\otimes \mathbb{I}_{B_{n}},\\[0.2cm]
            & C_{A_1^nB_1^{n-1}}^{(n)}\in\overline{\mathsf{S}}_n(B_1^{n-1},A_1^n),
        \end{array}
    \end{equation} 
    where the variable $C_{A_1^nB_1^{n-1}}^{(n)}$ represents a length-$n$ quantum comb, excluding the measurement, with a quantum state preparation at the beginning (see the red comb in Fig.~\ref{fig:adaptive_channel_uses}). The variables $P_{A_1^{n}B_1^n}(t)$ correspond to the measurement.
   
    The 
    convex program dual to the one above is
    \begin{equation}
        \begin{array}{l l} \textnormal{minimize} & \lambda \\[0.2cm] \textnormal{subject to} & \lambda\geq 0, \\[0.1cm] & \lambda Y_{A_1^nB_1^n}^{(n)}\geq\mu(t)(w\ast C_{A_1^nB_1^n}^{\mathcal{N}(t)^{[n]}})(t)~~\forall~t,\\[0.25cm]
        & Y_{A_1^nB_1^n}^{(n)}\in\mathsf{S}_n(A_1^n,B_1^n),\\[0.2cm]
        & \mu(t)\geq 0~~\forall~t,\,\int \textnormal{d}t~\mu(t)=1,
        \end{array}
    \end{equation}
    where the optimization is with respect to $\lambda\geq 0$, probability density functions $\mu(t)$, and length-$n$ quantum combs represented by the operator $Y_{A_1^nB_1^n}^{(n)}$. Furthermore, strong duality holds, so that the primal and dual programs have the same optimal value.
\end{sproposition}

The %
convex problems
in the above proposition are the continuous analogues of the 
semi-definite problems
for multiple channel discrimination \cite{CDP08b,CDP09,IM21}. Notably, as with Proposition~\ref{prop:minimax_sdp}, the optimal minimax success probability can be obtained via optimization of the Bayesian success probability with respect to all possible priors.

\begin{proof}
    The proof is analogous to the proof of Proposition~\ref{prop-succ_prob_combs_Bayesian}, so we omit some of the details. First, upon inspection of the primal problem in \eqref{eqn:succ_prob_comb_primal}, we find that it is of the standard form of the primal problem in \eqref{eqn:SDP_primal_dual_standard}, with
    \begin{align}
        X&=\left(\int\diff t\,\ketbra{t}{t}\otimes P_{A_1^nB_1^n}(t)\right)\oplus\left(\sum_{k=1}^n\ketbra{k}{k}\otimes C_{A_1^kB_1^{k-1}}^{(k)}\right)\oplus\eta,\\
        \nonumber
        A&=\left(\int\diff t\,\ketbra{t}{t}\otimes 0\right)\oplus\left(\sum_{k=1}^n\ketbra{k}{k}\otimes 0\right)\oplus 0, \\
        \nonumber
        \Phi[X]&=\left(\diff t\,\ketbra{t}{t}\left(\eta-\Tr\!\left[P_{A_1^nB_1^n}(t)(w\ast C_{A_!^nB_1^n}^{\mathcal{N}^{[n]}})(t)\right]\right)\right)\\
        \nonumber
        &\quad\oplus \ketbra{0,0}{0,0}\otimes\left(\int\diff t\,P_{A_1^nB_1^n}(t)-C_{A_1^nB_1^{n-1}}\otimes\mathbb{I}_{B_n}\right)\\
        \nonumber
        &\quad +\ketbra{1,0}{1,0}\otimes\Tr_{A_1}[C_{A_1}^{(1)}]+\ketbra{1,1}{1,1}\otimes (-\Tr_{A_1}[C_{A_1}^{(1)}])\\
        &\quad +\sum_{k=2}^{n}\ketbra{k,0}{k,0}\otimes\left(\Tr_{A_k}[C_{A_1^kB_1^{k-1}}^{(k)}]-C_{A_1^{k-1}B_1^{k-2}}^{(k-1)}\otimes\mathbb{I}_{B_{k-1}}\right)\\
        \nonumber
        &\quad +\sum_{k=2}^n\ketbra{k,1}{k,1}\otimes\left(-\Tr_{A_k}[C_{A_1^kB_1^{k-1}}^{(k)}]+C_{A_1^{k-1}B_1^{k-2}}\otimes\mathbb{I}_{B_{k-1}}\right),\\
        B&=0\oplus \ketbra{0,0}{0,0}\otimes 0+\ketbra{0,1}{0,1}\otimes 0+\ketbra{1,0}{1,0}\otimes 1+\ketbra{1,1}{1,1}\otimes (-1)\\
        \nonumber
        &\quad +\sum_{k=2}^n\ketbra{k,0}{k,0}\otimes 0+\sum_{k=2}^n\ketbra{k,1}{k,1}\otimes 0.
        \nonumber
    \end{align}
    Now, without loss of generality, we can let the dual variable $Y$ have the following block-diagonal form:
    \begin{multline}
        Y=\left(\int\diff t\,\mu(t)\ketbra{t}{t}\right)\oplus \ketbra{0,0}{0,0}\otimes Y_{A_1^nB_1^n}^{(1)}+\ketbra{0,1}{0,1}\otimes Y_{A_1^nB_1^n}^{(2)}+\ketbra{1,0}{1,0}\otimes\alpha+\ketbra{1,1}{1,1}\otimes\beta\\+\sum_{k=2}^n\ketbra{k,0}{k,0}\otimes\widetilde{C}_{A_1^{k-1}B_1^{k-1}}^{(k,0)}+\sum_{k=2}^n\ketbra{k,1}{k,1}\otimes\widetilde{C}_{A_1^{k-1}B_1^{k-1}}^{(k,1)}.
    \end{multline}
    This implies that
    \begin{align}
        \Tr[Y\Phi[X]]&=\int\diff t\,\mu(t)(\eta-\Tr[P_{A_1^nB_1^n}(t)(w\ast C_{A_1^nB_1^n}^{\mathcal{N}^{[n]}})(t)])\\
        &\quad +\Tr\!\left[\left(\int\diff t\,P_{A_1^nB_1^n}(t)\right)\left(Y_{A_1^nB_1^n}^{(1)}-Y_{A_1^nB_1^n}^{(2)}\right)\right]\\
        \nonumber
        &\quad +\Tr\!\left[\left(C_{A_1^nB_1^{n-1}}^{(n)}\otimes\mathbb{I}_{B_n}\right)\left(-Y_{A_1^nB_1^n}^{(n)}+Y_{A_1^nB_1^n}^{(2)}\right)\right]\\
        \nonumber
        &\quad +\Tr_{A_1}[C_{A_1}^{(1)}](\alpha-\beta)\\
        \nonumber
        &\quad +\sum_{k=2}^n\Tr\!\left[\Tr_{A_k}\!\left[C_{A_1^kB_1^{k-1}}^{(k)}\right]\left(\widetilde{C}_{A_1^{k-1}B_1^{k-1}}^{(k,0)}-\widetilde{C}_{A_1^{k-1}B_1^{k-1}}^{(k,1)}\right)\right]\\
        \nonumber
        &\quad +\sum_{k=2}^n\Tr\!\left[\left(C_{A_1^{k-1}B_1^{k-2}}^{(k-1)}\otimes\mathbb{I}_{B_{k-1}}\right)\left(-\widetilde{C}_{A_1^{k-1}B_1^{k-1}}^{(k,0)}+\widetilde{C}_{A_1^{k-1}B_1^{k-1}}^{(k,1)}\right)\right]\\
        \nonumber
        &=\Tr\!\left[\left(\int\diff t\,\ketbra{t}{t}\otimes P_{A_1^nB_1^n}(t)\right)\left(\int\diff t\,\ketbra{t}{t}\otimes (Y_{A_1^nB_1^n}^{(1)}-Y_{A_1^nB_1^n}^{(2)}-\mu(t)(w\ast C_{A_1^nB_1^n}^{\mathcal{N}^{[n]}})(t)\right)\right]\\
        \nonumber
        &\quad +\Tr\!\left[C_{A_1}^{(1)}\left((\alpha-\beta)\mathbb{I}_{A_1}+\Tr_{B_1}[-\widetilde{C}_{A_1B_1}^{(2,0)}+\widetilde{C}_{A_1B_1}^{(2,1)}]\right)\right]\\
        \nonumber
        &\quad +\sum_{k=2}^{n-1}\Tr\!\left[C_{A_1^kB_1^{k-1}}^{(k)}\left((\widetilde{C}_{A_1^{k-1}B_1^{k-1}}^{(k,0)}-\widetilde{C}_{A_1^{k-1}B_1^{k-1}}^{(k,1)})\otimes\mathbb{I}_{A_k}+\Tr_{B_k}[-\widetilde{C}_{A_1^kB_1^k}^{(k+1,0)}+\widetilde{C}_{A_1^kB_1^k}^{(k+1,1)}]\right)\right]\\
        \nonumber
        &\quad +\Tr\!\left[C_{A_1^nB_1^{n-1}}^{(n)}\left((\widetilde{C}_{A_1^{n-1}B_1^{n-1}}^{(n,0)}-\widetilde{C}_{A_1^{n-1}B_1^{n-1}}^{(n,1)})\otimes\mathbb{I}_{A_n}+\Tr_{B_n}[-Y_{A_1^nB_1^n}^{(1)}+Y_{A_1^nB_1^n}^{(2)}]\right)\right],
        \nonumber
    \end{align}
    which in turn implies that
    \begin{align}
        \Phi^{\dagger}[Y]&=\left(\int\diff t\,\ketbra{t}{t}\otimes (Y_{A_1^nB_1^n}^{(1)}-Y_{A_1^nB_1^n}^{(2)}-\mu(t)(w\ast C_{A_1^nB_1^n}^{\mathcal{N}^{[n]}})(t)\right)\\
        \nonumber
        &\quad\oplus \left((\widetilde{C}_{A_1^{n-1}B_1^{n-1}}^{(n,0)}-\widetilde{C}_{A_1^{n-1}B_1^{n-1}}^{(n,1)})\otimes\mathbb{I}_{A_n}+\Tr_{B_n}[-Y_{A_1^nB_1^n}^{(1)}+Y_{A_1^nB_1^n}^{(2)}]\right)\\
        \nonumber
        &\quad\oplus \sum_{k=2}^{n-1}\ketbra{k}{k}\otimes \left((\widetilde{C}_{A_1^{k-1}B_1^{k-1}}^{(k,0)}-\widetilde{C}_{A_1^{k-1}B_1^{k-1}}^{(k,1)})\otimes\mathbb{I}_{A_k}+\Tr_{B_k}[-\widetilde{C}_{A_1^kB_1^k}^{(k+1,0)}+\widetilde{C}_{A_1^kB_1^k}^{(k+1,1)}]\right)\\
        \nonumber
        &\quad\oplus \left((\alpha-\beta)\mathbb{I}_{A_1}+\Tr_{B_1}[-\widetilde{C}_{A_1B_1}^{(2,0)}+\widetilde{C}_{A_1B_1}^{(2,1)}]\right)\\
        &\quad\oplus \left(\int\diff t\,\mu(t)\right).
        \nonumber
    \end{align}
    The inequality $\Phi^{\dagger}[Y]\geq A$, therefore, implies that the dual problem is given by
    \begin{equation}
        \begin{array}{l l}
            \text{minimize} & \alpha-\beta \\[1ex]
            \text{subject to} & \alpha\geq 0,\,\beta\geq 0,\,Y_{A_1^nB_1^n}^{(1)}\geq 0,\,Y_{A_1^nB_1^n}^{(2)}\geq 0,\,\widetilde{C}_{A_1^{k-1}B_1^{k-1}}^{(k,0)}\geq 0,\,\widetilde{C}_{A_1^{k-1}B_1^{k-1}}^{(k,1)}\geq 0~~\forall~k\in\{1,2,\dotsc,n\},\\[1ex]
            & Y_{A_1^nB_1^n}^{(1)}-Y_{A_1^nB_1^n}^{(2)}-\mu(t)(w\ast C_{A_1^nB_1^n}^{\mathcal{N}^{[n]}})(t)~~\forall~t,\\[1ex]
            & (\widetilde{C}_{A_1^{n-1}B_1^{n-1}}^{(n,0)}-\widetilde{C}_{A_1^{n-1}B_1^{n-1}}^{(n,1)})\otimes\mathbb{I}_{A_n}+\Tr_{B_n}[-Y_{A_1^nB_1^n}^{(1)}+Y_{A_1^nB_1^n}^{(2)}]\geq 0,\\[1ex]
            & (\widetilde{C}_{A_1^{k-1}B_1^{k-1}}^{(k,0)}-\widetilde{C}_{A_1^{k-1}B_1^{k-1}}^{(k,1)})\otimes\mathbb{I}_{A_k}+\Tr_{B_k}[-\widetilde{C}_{A_1^kB_1^k}^{(k+1,0)}+\widetilde{C}_{A_1^kB_1^k}^{(k+1,1)}]~~\forall~k\in\{2,3,\dotsc,n-1\},\\[1ex]
            & (\alpha-\beta)\mathbb{I}_{A_1}+\Tr_{B_1}[-\widetilde{C}_{A_1B_1}^{(2,0)}+\widetilde{C}_{A_1B_1}^{(2,1)}]\geq 0,\\[1ex]
            & \int\diff t\,\mu(t)\geq 1,\,\mu(t)\geq 0~~\forall~t.
        \end{array}
    \end{equation}
    Strong duality is straightforward to show, which implies that the optimal solution to this dual problem is equal to the optimal solution of the primal problem.

    Next, by the complementary slackness condition $\Phi^{\dagger}[Y]X=AX$, we find that $\int\diff t\,\mu(t)=1$. We can further simplify the dual problem above via change of variables, in exactly the same way as we did in the proof of Proposition~\ref{prop-succ_prob_combs_Bayesian}. Doing so gives us the desired dual problem in the statement of the proposition.
\end{proof}

\subsection{Optimization with respect to strategies with indefinite causal order}

    In the previous section, we considered sequential/adaptive quantum metrology protocols in which every use of the parameterized channel $t\mapsto\mathcal{N}(t)$ is causally ordered. Let us now consider a more general class of protocols, based on \textit{non-causal} ordering of the channel uses. Following Refs.~\cite{oreshkov2012noncausal,bavaresco2022unitarydiscnoncausal}, we define a general, non-causal strategy for $n$ uses of the channel $\mathcal{N}_{A\to B}(t)$ by operators $t\mapsto P_{A_1^nB_1^n}(t)$ such that $P_{A_1^nB_1^n}(t)\geq 0$ for all $t$, and $W_{A_1^nB_1^n}\coloneqq\int\diff t\,P_{A_1^nB_1^n}(t)$ satisfies $\Tr[W_{A_1^nB_1^n}(C_{A_1B_1}^{(1)}\otimes C_{A_2B_2}^{(2)}\otimes\dotsb\otimes C_{A_nB_n}^{(n)})]=1$ for all Choi representations $C_{A_kB_k}^{(k)}$ of quantum channels (\emph{i.e.}, Hermitian operators $C_{A_kB_k}^{(k)}$ satisfying $C_{A_kB_k}^{(k)}\geq 0$ and $\Tr_{B_k}[C_{A_kB_k}^{(k)}]=\mathbb{I}_{A_k}$). We let
    \begin{equation}
        \mathsf{C}_n^{\textsf{prod}}(A_1^n,B_1^n)\coloneqq\left\{\bigotimes_{k=1}^n C_{A_kB_k}^{(k)}: C_{A_kB_k}^{(k)}\geq 0,\,\Tr_{B_k}[C_{A_kB_k}^{(k)}]=\mathbb{I}_{A_k},\,k\in\{1,2,\dotsc,n\}\right\}
    \end{equation}
    be the set of all tensor $n$-fold tensor products of Choi representations of quantum channel. Then, we define
    \begin{equation}\label{eq-non_causal_strategies_set}
        \mathsf{S}_n^{\textsf{ico}}(A_1^n,B_1^n)\coloneqq\left\{ W_{A_1^nB_2^n}: W_{A_1^nB_1^n}\geq 0,\, \Tr[W_{A_1^nB_1^n}Y_{A_1^nB_1^n}]=1~~\forall~Y_{A_1^nB_1^n}\in\mathsf{C}_n^{\textsf{prod}}(A_1^n,B_1^n)\right\}
    \end{equation}
    to be the set of all operators representing $n$-partite non-causal strategies. An explicit form for this set for arbitrary $n\in\{2,3,\dotsc\}$ can be found in Ref.~\cite{araujo2015noncausal}. As an example, for $n=2$,
    \begin{multline}
        \mathsf{S}_{2}^{\textsf{ico}}(A_1^2,B_1^2)=\left\{W_{A_1^2B_1^2}:W_{A_1^2B_1^2}\geq 0,\,\Tr[W_{A_1^2B_1^2}]=d_{B_1}d_{B_2},\, W_{A_1^2B_1^2}=(\mathcal{R}_{B_1}+\mathcal{R}_{B_2}-\mathcal{R}_{B_1B_2})(W_{A_1^2B_1^2}),\right.\\\left.\mathcal{R}_{A_1B_1}(W_{A_1^2B_1^2})=\mathcal{R}_{A_1B_1B_2}(W_{A_1^2B_1^2}),\,\mathcal{R}_{A_2B_2}(W_{A_1^2B_1^2})=\mathcal{R}_{A_2B_1B_2}(W_{A_1^2B_1^2})\right\},
    \end{multline}
    where $\mathcal{R}_A(X_{RA})\coloneqq\Tr_A[X_{RA}]\otimes\frac{\mathbb{I}_A}{d_A}$ is the completely depolarizing channel acting on a system $A$, which discards the state of the system $A$ and replaces it with the maximally-mixed state.
    
    \begin{lemma}\label{lem-non_causal_strategies_set}
        For every $n\in\{2,3,\dotsc\}$, it holds that
        \begin{multline}
            \mathsf{S}_{n}^{\mathsf{ico}}(A_1^n,B_1^n)=\left\{W_{A_1^nB_1^n}:W_{A_1^nB_1^n}\geq 0,\, \Tr[W_{A_1^nB_1^n}(C_{A_1B_1}^{(1)}\otimes C_{A_2B_2}^{(2)}\otimes\dotsb\otimes C_{A_nB_n}^{(n)})]=1,\right.\\\left. C_{A_kB_k}^{(k)}~\textnormal{Hermitian},\,\Tr_{B_k}[C_{A_kB_k}^{(k)}]=\mathbb{I}_{A_k},\,k\in\{1,2,\dotsc,n\} \right\}.
        \end{multline}
    \end{lemma}

    \begin{proof}
        We follow the arguments presented in Ref.~\cite{araujo2015noncausal}. The inclusion ``$\supseteq$'' is clear. For the inclusion ``$\subseteq$'', assume that $W_{A_1^nB_1^n}\in\mathsf{S}_n^{\mathsf{ico}}(A_1^n,B_1^n)$. We now show that $\Tr[W_{A_1^nB_1^n}(C_{A_1B_1}^{(1)}\otimes\dotsb\otimes C_{A_nB_n}^{(n)})]=1$ for every collection $C_{A_1B_1}^{(1)},\dotsc,C_{A_nB_n}^{(n)}$ of Hermitian operators satisfying $\Tr_{B_1}[C_{A_1B_1}^{(1)}]=\mathbb{I}_{A_1},\dotsc,\Tr_{B_n}[C_{A_nB_n}^{(n)}]=\mathbb{I}_{A_n}$. To that end, note that for every such $C_{A_kB_k}^{(k)}$, there exists $\alpha_k\geq 0$ such that $C_{A_kB_k}^{(k)}+\alpha_k\frac{1}{d_{B_k}}\mathbb{I}_{A_kB_k}\geq 0$. Furthermore, we can decompose $C_{A_kB_k}^{(k)}$ as
        \begin{align}
            C_{A_kB_k}^{(k)}&=(\alpha_k+1)C_{A_kB_k}^{(k,0)}-\alpha_k C_{A_kB_k}^{(k,1)},\\
            C_{A_kB_k}^{(k,0)}&=\frac{1}{\alpha_k+1}\left(C_{A_kB_k}^{(k)}+\alpha_k\frac{1}{d_{B_k}}\mathbb{I}_{A_kB_k}\right),\\
            C_{A_kB_k}^{(k,1)}&=\frac{1}{d_{B_k}}\mathbb{I}_{A_kB_k}.
        \end{align}
        Observe that $C_{A_kB_k}^{(k,0)}\geq 0$, $C_{A_kB_k}^{(k,1)}\geq 0$, and $\Tr_{B_k}[C_{A_kB_k}^{(k,0)}]=\Tr_{B_k}[C_{A_kB_k}^{(k,1)}]=\mathbb{I}_{A_k}$. Then, we have
        \begin{align}
            \Tr[W_{A_1^nB_1^n}(C_{A_1B_1}^{(1)}\otimes\dotsb\otimes C_{A_nB_n}^{(n)})]&=\sum_{x\in\{0,1\}^n}\left(\prod_{k=1}^n (\alpha_k+1)^{x_k}(-\alpha_k)^{1-x_k}\right)\underbrace{\Tr\!\left[W_{A_1^nB_1^n}\bigotimes_{k=1}^n C_{A_kB_k}^{(k,x)}\right]}_{=1~~\forall~x\in\{0,1\}^n}\\
            \nonumber
            &=\sum_{x\in\{0,1\}^n}\left(\prod_{k=1}^n (\alpha_k+1)^{x_k}(-\alpha_k)^{1-x_k}\right)\\
            \nonumber
            &=\prod_{k=1}^n\left(-\alpha_k+\alpha_k+1\right)\\
            \nonumber
            &=1,
        \end{align}
        which implies the desired result, because the operator $W_{A_1^nB_1^n}\in\mathsf{S}_n^{\mathsf{ico}}(A_1^n,B_1^n)$ was arbitrary.
    \end{proof}
    Let
    \begin{equation}
        \widetilde{\mathsf{C}}_n^{\textsf{prod}}(A_1^n,B_1^n)\coloneqq\left\{\bigotimes_{k=1}^n C_{A_kB_k}^{(k)}:C_{A_kB_k}^{(k)}\text{ Hermitian},\,\Tr_{B_k}[C_{A_kB_k}^{(k)}]=\mathbb{I}_{A_k},\,k\in\{1,2,\dotsc,n\}\right\}.
    \end{equation}
    Lemma~\ref{lem-non_causal_strategies_set} tells us that
    \begin{align}
        \mathsf{S}_n^{\textsf{ico}}(A_1^n,B_1^n)&=\left\{W_{A_1^nB_1^n}: W_{A_1^nB_1^n}\geq 0,\,\Tr[W_{A_1^nB_1^n}Y_{A_1^nB_1^n}]=1~~\forall~Y_{A_1^nB_1^n}\in\widetilde{\mathsf{C}}_n^{\textsf{prod}}(A_1^n,B_1^n)\right\}
        \nonumber\\
        &=\left\{W_{A_1^nB_1^n}: W_{A_1^nB_1^n}\geq 0,\, \Tr[W_{A_1^nB_1^n}Y_{A_1^nB_1^n}]=1~~\forall~Y_{A_1^nB_1^n}\in\text{aff}(\widetilde{\mathsf{C}}_n^{\textsf{prod}}(A_1^n,B_1^n))\right\},\label{eq-non_causal_strategies_set_alt2}
    \end{align}
    where $\textnormal{aff}(\mathsf{S})\coloneqq\{\sum_i x_i s_i:s_i\in\mathsf{S},\,x_i\in\mathbb{R},\,\sum_i x_i=1\}$ denotes the affine hull of a set $\mathsf{S}$.

    \begin{proposition}[Bayesian optimization of non-causal strategies]\label{prop-SDP_Bayesian_non_causal}
        Let $ C_{A_1^nB_1^n}^{\mathcal{N}(t)^{[n]}}$ be the Choi representation of the comb $\mathcal{N}(t)^{[n]}$ corresponding to $n$ uses of the parameterized quantum channel $t\mapsto\mathcal{N}(t)$. Also, let $t\mapsto\mu(t)$ be a prior probability density function. Then, the optimal non-causal Bayesian strategy for estimating the parameter $t$ can be determined using the convex problem
        \begin{equation}
            \begin{array}{l l}
                \textnormal{maximize} & \int\diff t\,\mu(t)\Tr[P_{A_1^nB_1^n}(t)(w\ast C_{A_1^nB_1^n}^{\mathcal{N}^{[n]}})(t)]\\[1ex]
                \textnormal{subject to} & P_{A_1^nB_1^n}(t)\geq 0~~\forall~t,\\[1ex]
                & \int\diff t\,P_{A_1^nB_1^n}(t)\in\mathsf{S}_n^{\mathsf{ico}}(A_1^n,B_1^n).
            \end{array}
        \end{equation} 
        The dual problem, which has the same optimal value as the primal problem above, is
        \begin{equation}
            \begin{array}{l l}
                \textnormal{minimize} & \lambda \\[1ex]
                \textnormal{subject to} & \lambda Y_{A_1^nB_1^n}\geq \mu(t)(w\ast C_{A_1^nB_1^n}^{\mathcal{N}^{[n]}})(t)~~\forall~t,\\[1ex]
                & Y_{A_1^nB_1^n}\in\widetilde{\mathsf{S}}_n^{\mathsf{NS}}(A_1^n,B_1^n),
            \end{array}
        \end{equation}
        where $\widetilde{\mathsf{S}}_n^{\mathsf{NS}}(A_1^n,B_1^n)$ is the set of all Choi representations of $n$-partite Hermiticity-preserving non-signaling superoperators with input systems $A_1,\dotsc,A_n$ and output systems $B_1,\dotsc,B_n$, defined as~\cite{gutoski2009sharedentanglement,CE16}
        \begin{equation}
            \widetilde{\mathsf{S}}_n^{\mathsf{NS}}(A_1^n,B_1^n)\coloneqq\left\{Y_{A_1^nB_1^n}:Y_{A_1^nB_1^n}\textnormal{ Hermitian},\,\Tr_{B_K}[Y_{A_1^nB_1^n}]=\mathbb{I}_{A_K}\otimes Y_{A_{\overline{K}}B_{\overline{K}}}',\,Y_{A_{\overline{K}}B_{\overline{K}}}'\geq 0,\,K\subseteq\{1,2,\dotsc,n\}\right\}.
        \end{equation}
        Here, $B_K$ denotes the $B$ systems labeled by the subset $K$, and $\overline{K}$ denotes the complement of $K$.
    \end{proposition}

    \begin{proof}
        Using \eqref{eq-non_causal_strategies_set_alt2}, we can write the primal problem as
        \begin{equation}
            \begin{array}{l l}
                \textnormal{maximize} & \int\diff t\,\mu(t)\Tr[P_{A_1^nB_1^n}(t)(w\ast C_{A_1^nB_1^n}^{\mathcal{N}^{[n]}})(t)]\\[1ex]
                \textnormal{subject to} & P_{A_1^nB_1^n}(t)\geq 0~~\forall~t,\\[1ex]
                & W_{A_1^nB_1^n}=\int\diff t\,P_{A_1^nB_1^n}(t),\\[1ex]
                & \Tr[W_{A_1^nB_1^n}Y_{A_1^nB_1^n}]=1~~\forall~Y_{A_1^nB_1^n}\in\text{aff}(\widetilde{\mathsf{C}}_n^{\textsf{prod}}(A_1^n,B_1^n)).
            \end{array}
        \end{equation}
        Now, let us pick a basis $\{Y_{A_1^nB_1^n}^j\}_j$ for the affine space $\text{aff}(\widetilde{\mathsf{C}}_n^{\textsf{prod}}(A_1^n,B_1^n))$. With this, the infinite number of constraints in the final line of the above convex problem can be made into a finite number of constraints, such that the primal problem can be formulated as
        \begin{equation}
            \begin{array}{l l}
                \textnormal{maximize} & \int\diff t\,\mu(t)\Tr[P_{A_1^nB_1^n}(t)(w\ast C_{A_1^nB_1^n}^{\mathcal{N}^{[n]}})(t)]\\[1ex]
                \textnormal{subject to} & P_{A_1^nB_1^n}(t)\geq 0~~\forall~t,\\[1ex]
                & W_{A_1^nB_1^n}=\int\diff t\,P_{A_1^nB_1^n}(t),\\[1ex]
                & \Tr[W_{A_1^nB_1^n}Y_{A_1^nB_1^n}^j]=1~~\forall~j.
            \end{array}
        \end{equation}
        We can now cast this into the standard form \eqref{eqn:SDP_primal_dual_standard} of a primal problem. Specifically, we have
        \begin{align}
            X&=\int\diff t\,\ketbra{t}{t}\otimes P_{A_1^nB_1^n}(t),\\
            A&=\int\diff t\,\ketbra{t}{t}\otimes \mu(t)(w\ast C_{A_1^nB_1^n}^{\mathcal{N}^{[n]}})(t),\\
            \Phi[X]&=\sum_j\Tr[W_{A_1^nB_1^n}Y_{A_1^nB_1^n}^j]\ketbra{j}{j}\oplus \sum_j (-\Tr[W_{A_1^nB_1^n}Y_{A_1^nB_1^n}^j])\ketbra{j}{j},\\
            B&=\sum_j\ketbra{j}{j}\oplus \sum_{j}(-1)\ketbra{j}{j}.
        \end{align}
        Now, for the dual variable, we can take it to be of the form $Y=\sum_j\alpha_j\ketbra{j}{j}\oplus\sum_j\beta_j\ketbra{j}{j}$, where $\alpha_j\geq 0$ and $\beta_j\geq 0$ for all $j$. From this, it is straightforward to show that
        \begin{equation}
            \Tr[Y\Phi[X]]=\Tr\!\left[\left(\int\diff t\,\ketbra{t}{t}\otimes P_{A_1^nB_1^n}(t)\right)\left(\int\diff t\,\ketbra{t}{t}\otimes\sum_j(\alpha_j-\beta_j)Y_{A_1^nB_1^n}^j\right)\right],
        \end{equation}
        so that
        \begin{equation}
            \Phi^{\dagger}[Y]=\int\diff t\,\ketbra{t}{t}\otimes \sum_j(\alpha_j-\beta_j)Y_{A_1^nB_1^n}^j.
        \end{equation}
        The dual problem is therefore
        \begin{equation}
            \begin{array}{l l}
                \text{minimize} & \sum_j(\alpha_j-\beta_j)\\[1ex]
                \text{subject to} & \alpha_j\geq 0,\,\beta_j\geq 0~~\forall j,\\[1ex]
                & \sum_j (\alpha_j-\beta_j)Y_{A_1^nB_1^n}^j\geq \mu(t)(w\ast C_{A_1^nB_1^n}^{\mathcal{N}^{[n]}})(t)~~\forall~t.
            \end{array}
        \end{equation}
        Strong duality is straightforward to verify, which means that this dual problem has the same optimal value as the primal problem. Now, let $\lambda\equiv\sum_j(\alpha_j-\beta_j)$. Then, $\lambda\in\mathbb{R}$ and
        \begin{equation}
            \sum_j(\alpha_j-\beta_j)Y_{A_1^nB_1^n}^j=\lambda\sum_j\frac{\alpha_j-\beta_j}{\lambda}Y_{A_1^nB_1^n}^j\equiv \lambda Y_{A_1^nB_1^n},
        \end{equation}
        where $Y_{A_1^nB_1^n}\in\text{aff}(\widetilde{\mathsf{C}}_n^{\textsf{prod}}(A_1^n,B_1^n))$. We thus conclude that the dual problem is given by
        \begin{equation}
            \begin{array}{l l}
                \textnormal{minimize} & \lambda \\[1ex]
                \textnormal{subject to} & \lambda Y_{A_1^nB_1^n}\geq \mu(t)(w\ast C_{A_1^nB_1^n}^{\mathcal{N}^{[n]}})(t)~~\forall~t,\\[1ex]
                & Y_{A_1^nB_1^n}\in\textnormal{aff}(\widetilde{\mathsf{C}}_n^{\textsf{prod}}(A_1^n,B_1^n)).
            \end{array}
        \end{equation}
        Finally, we use Ref.~\cite{gutoski2009sharedentanglement}~(Theorem~14), which implies that $\text{aff}(\widetilde{\mathsf{C}}_n^{\textsf{prod}}(A_1^n,B_1^n))=\widetilde{\mathsf{S}}_n^{\textsf{NS}}(A_1^n,B_1^n)$, completing the proof.
    \end{proof}

    \begin{proposition}[Minimax optimization of non-causal strategies]\label{prop-SDP_minimax_non_causal}
        Let $ C_{A_1^nB_1^n}^{\mathcal{N}(t)^{[n]}}$ be the Choi representation of the comb $\mathcal{N}(t)^{[n]}$ corresponding to $n$ uses of the parameterized quantum channel $t\mapsto\mathcal{N}(t)$. Then, the optimal non-causal minimax strategy for estimating the parameter $t$ can be determined using the convex problem
        \begin{equation}
            \begin{array}{l l}
                \textnormal{maximize} & \eta \\[1ex]
                \textnormal{subject to} & \Tr[P_{A_1^nB_1^n}(t)(w\ast C_{A_1^nB_1^n}^{\mathcal{N}^{[n]}})(t)]\geq\eta~~\forall~t,\\[1ex]
                & \eta\geq 0,\\[1ex]
                & P_{A_1^nB_1^n}(t)\geq 0~~\forall~t,\\[1ex]
                & \int\diff t\,P_{A_1^nB_1^n}(t)\in\mathsf{S}_n^{\mathsf{ico}}(A_1^n,B_1^n).
            \end{array}
        \end{equation}
        The dual problem, which has the same optimal value as the primal problem above, is
        \begin{equation}
            \begin{array}{l l}
                \text{minimize} & \lambda \\[1ex]
                \text{subject to} & \lambda Y_{A_1^nB_1^n}\geq\mu(t)(w\ast C_{A_1^nB_1^n}^{\mathcal{N}^{[n]}})(t)~~\forall~t,\\[1ex]
                & Y_{A_1^nB_1^n}\in\textnormal{aff}(\widetilde{\mathsf{C}}_n^{\mathsf{prod}}(A_1^n,B_1^n)),\\[1ex]
                & \mu(t)\geq 0,\,\int\diff t\,\mu(t)=1.
            \end{array}
        \end{equation}
    \end{proposition}

    \begin{proof}
        We proceed similarly to the proof of Proposition~\ref{prop-SDP_Bayesian_non_causal}. By picking a basis $\{Y_{A_1^nB_1^n}^j\}_j$ for $\text{aff}(\widetilde{\mathsf{C}}_n^{\mathsf{prod}}(A_1^n,B_1^n))$, we can write the primal problem as
        \begin{equation}
            \begin{array}{l l}
                \textnormal{maximize} & \eta \\[1ex]
                \textnormal{subject to} & \Tr[P_{A_1^nB_1^n}(t)(w\ast C_{A_1^nB_1^n}^{\mathcal{N}^{[n]}})(t)]\geq\eta~~\forall~t,\\[1ex]
                & \eta\geq 0,\\[1ex]
                & P_{A_1^nB_1^n}(t)\geq 0~~\forall~t,\\[1ex]
                & W_{A_1^nB_1^n}=\int\diff t\,P_{A_1^nB_1^n}(t),\\[1ex]
                & \Tr[W_{A_1^nB_1^n}Y_{A_1^nB_1^n}^j]=1~~\forall~j.
            \end{array}
        \end{equation}
        We can then cast this into the standard form \eqref{eqn:SDP_primal_dual_standard} as 
        \begin{align}
            X&=\eta\oplus\left(\int\diff t\,\ketbra{t}{t}\otimes P_{A_1^nB_1^n}(t)\right),\\
            A&=1\oplus\left(\int\diff t\,\ketbra{t}{t}\otimes 0  \right),\\
            \Phi[X]&=\left(\int\diff t\,\ketbra{t}{t}\otimes(\eta-\Tr[P_{A_1^nB_1^n}(t)(w\ast C_{A_1^nB_1^n}^{\mathcal{N}^{[n]}})(t)]\right)\nonumber\\
            &\qquad\oplus\left(\sum_j \Tr[W_{A_1^nB_1^n}Y_{A_1^nB_1^n}^j]\ketbra{j}{j}\right)\nonumber\\
            &\qquad\oplus\left(\sum_j (-\Tr[W_{A_1^nB_1^n}Y_{A_1^nB_1^n}^j])\ketbra{j}{j}\right),\\
            B&=\left(\int\diff t\,\ketbra{t}{t}\otimes 0\right)\oplus\left(\sum_j\ketbra{j}{j}\right)\oplus\left((-1)\ketbra{j}{j}\right).
        \end{align}
        Then, letting the dual variable $Y$ be
        \begin{equation}
            Y=\left(\int\diff t\,\ketbra{t}{t}\mu(t)\right)\oplus\left(\sum_j\ketbra{j}{j}\alpha_j\right)\oplus\left(\sum_j\ketbra{j}{j}\beta_j\right),
        \end{equation}
        with $\mu(t)\geq 0$ for all $t$ and $\alpha_j\geq 0$, $\beta_j\geq 0$ for all $j$, we obtain
        \begin{align}
            \Tr[Y\Phi[X]]&=\int\diff t\mu(t)\left(\eta-\Tr[P_{A_1^nB_1^n}(t)(w\ast C_{A_1^nB_1^n}^{\mathcal{N}^{[n]}})(t)]\right)\nonumber\\
            &\quad +\sum_j\alpha_j\Tr[W_{A_1^nB_1^n}Y_{A_1^nB_1^n}^j]-\sum_j\beta_j\Tr[W_{A_1^nB_1^n}Y_{A_1^nB_1^n}^j]\\
            &=\eta\int\diff t\,\mu(t)\nonumber\\
            &\quad +\Tr\!\left[\left(\int\diff t\,\ketbra{t}{t}\otimes P_{A_1^nB_1^n}(t)\right)\left(\int\diff t\,\ketbra{t}{t}\otimes\left(\sum_j(\alpha_j-\beta_j)Y_{A_1^nB_1^n}^j-\mu(t)(w\ast C_{A_1^nB_1^n}^{\mathcal{N}^{[n]}})(t)\right)\right)\right].
        \end{align}
        This implies that
        \begin{equation}
            \Phi^{\dagger}[Y]=\left(\int\diff t\,\mu(t)\right)\oplus\left(\int\diff t\,\ketbra{t}{t}\otimes\left(\sum_j(\alpha_j-\beta_j)Y_{A_1^nB_1^n}^j-\mu(t)(w\ast C_{A_1^nB_1^n}^{\mathcal{N}^{[n]}})(t)\right)\right).
        \end{equation}
        The dual problem is therefore
        \begin{equation}
            \begin{array}{l l}
                \text{minimize} & \sum_{j}\alpha_j-\beta_j\\[1ex]
                \text{subject to} & \alpha_j\geq 0,\,\beta_j\geq 0~~\forall~j,\,\mu(t)\geq 0~~\forall~t,\\[1ex]
                & \int\diff t\,\mu(t)\geq 1,\,\sum_j (\alpha_j-\beta_j)Y_{A_1^nB_1^n}^j\geq\mu(t)(w\ast C_{A_1^nB_1^n}^{\mathcal{N}^{[n]}})(t)~~\forall~t.
            \end{array}
        \end{equation}  
        Strong duality is straightforward to show, which means that this dual problem has the same optimal value as the primal problem.

        Now, the complementary slackness condition $\Phi^{\dagger}[Y]X=AX$ implies that $\int\diff t\,\mu(t)=1$. Furthermore, letting $\lambda\equiv\sum_j \alpha_j-\beta_j$, we find that
        \begin{equation}
            \sum_j (\alpha_j-\beta_j)Y_{A_1^nB_1^n}^j=\lambda\sum_j\frac{\alpha_j-\beta_j}{\lambda}Y_{A_1^nB_1^n}^j\equiv \lambda Y_{A_1^nB_1^n},\quad Y_{A_1^nB_1^n}\in\text{aff}(\widetilde{\mathsf{C}}_n^{\mathsf{prod}}(A_1^n,B_1^n)).
        \end{equation}
        With these simplifications, along with Ref.~\cite{gutoski2009sharedentanglement}~(Theorem~14), which states that $\text{aff}(\widetilde{\mathsf{C}}_n^{\mathsf{prod}}(A_1^n,B_1^n))=\widetilde{\mathsf{S}}_n^{\mathsf{NS}}(A_1^n,B_1^n)$, we obtain the desired dual problem.
    \end{proof}

    The convex programs presented in Propositions~\ref{prop-SDP_Bayesian_non_causal} and~\ref{prop-SDP_minimax_non_causal} are continuous analogues of the SDPs in Refs.~\cite{bavaresco2021channeldiscnoncausal,bavaresco2022unitarydiscnoncausal} for channel discrimination using strategies without causal ordering of the channel uses.

\subsection{Numerical implementation}\label{ssec:numerical_implementation}
The optimization programs presented in the preceding sections typically involve continuous objects, both in the inputs to the problem as well as in the optimization variables. As optimization over these quantities is not readily available in numerical solvers, we have to discretize the problems for actual implementation.
This means, we need to choose a number of points in the discretization which we denote with $k$. We then replace any function $f(t)$ with a vector of discrete values \smash{$(f_{\ell})_{\ell=1}^k$} corresponding to points in time \smash{$(t_{\ell})_{\ell=1}^k$}. 

Typically, $t$ varies only in a finite interval $[0,T]$. In this case, the most straightforward way to discretize is to fix a number of discretization steps $k$ and choose time points that are evenly spaced in time, \textit{i.e.}, $t_{\ell} = \ell T/k$. We then replace the time-varying quantity $f(t)$ with the average over the preceding interval in the discretization as
\begin{align}
    f_{\ell} = \frac{k}{T}\int_{t_{\ell-1}}^{t_{\ell}} \diff t \, f(t),
\end{align}
where we use the convention that $t_0 = 0$. This corresponds to replacing $f(t)$ with a piece-wise constant function
\begin{align}
    f(t) \to \sum_{\ell=1}^k f_{\ell} \, \chi[t_{\ell-1} < t \leq t_{\ell}],
\end{align}
where $\chi$ denotes an indicator function which takes the value one if the argument is true and zero otherwise.
As an example, we give the discretized version of the %
convex program 
of Proposition~\ref{prop:bayesian_sdp}. The objective of said program is given by
\begin{align}
    \int \diff t \, \diff \tau \, w(t - \tau) \Tr [\mu(t) \rho(t) Q(\tau) ].
\end{align}
We now replace all functions by their piece-wise approximations, to obtain
\begin{align}
    &\sum_{\ell=1}^{k}\sum_{m=1}^{k}\sum_{p=1}^{k} w_p \Tr [ \mu_{\ell} \rho_{\ell} Q_m  ] \int \diff t \, \diff \tau \, \chi[t_{\ell-1} < t \leq t_{\ell}]\chi[t_{p-1} < t-\tau \leq t_{p}]\chi[t_{m-1} < \tau \leq t_{m}]\\
    \nonumber
     =\mathstrut&\sum_{\ell=1}^{k}\sum_{m=1}^{k}\sum_{p=1}^{k} w_{\ell - m}\chi[1 \leq \ell - m \leq k] \Tr [ \mu_{\ell} \rho_{\ell} Q_m  ] \left(\frac{T}{k}\right)^2\\
      \nonumber
    =\mathstrut&\sum_{\ell=1}^{k}\sum_{m=\max(1,\ell - k)}^{\min(k,\ell - 1)}w_{\ell - m} \Tr [ \mu_{\ell} \rho_{\ell} Q_m  ] \left(\frac{T}{k}\right)^2.
\end{align}
The discretized version of the 
convex problem
would thus be
\begin{align}
    \eta^{*}(w, \mu, \rho) &\approx \eta^{*}( (w_{\ell}), (\mu_{\ell}), (\rho_{\ell}))\\
     \nonumber
    &= \begin{array}[t]{rc}
      \text{maximize} & \left(\frac{T}{k}\right)^2 \sum_{\ell=1}^{k}\sum_{m=\max(1,\ell - k)}^{\min(k,\ell - 1)}w_{\ell - m} \Tr [ \mu_{\ell} \rho_{\ell} Q_m  ] \\ [1ex]
      \text{such that} & Q_{\ell} \geq 0 \text{ for all } \ell \\[1ex]
      \text{such that} & \sum_{\ell} Q_{\ell} = \bbI.
    \end{array}
\end{align}
In this case, the sum over $m$ denotes a discrete convolution. Under the assumption of periodic boundary conditions, this convolution would be changed to wrap around the interval which would yield a simpler expression.

The other possibility, namely the case in which $t$ can take arbitrary values on the real line, makes sense only if we have access to a prior $\mu(t)$. This case can be treated in a similar manner, but the discretization is not uniform anymore. Instead, one can choose regions $R_{\ell} \subset \bbR$ of equal measure $\mu(R_{\ell}) = 1/k$ for all $\ell$. In this way, we are replacing the measure over values of $t$ with a uniform measure over the choice of region $R_{\ell}$. For each such region we replace the function value with its average over the region, \textit{i.e.}, a piece-wise approximation
\begin{align}
    f(t) \approx \sum_{\ell=1}^{k} f_{\ell} \chi[ t \in R_{\ell}], \qquad
    f_{\ell} = \frac{1}{|R_{\ell}|} \int_{R_{\ell}} \diff t \, f(t)
\end{align}
such that
\begin{align}
    \int \diff \mu(t) \, f(t) \approx \frac{1}{k} \sum_{\ell = 1}^k f_{\ell}.
\end{align}
Ideally, one would like to do away with the discretization completely. To do so, one could consider an expansion of the optimization variable, for example the measurement,
\begin{align}
Q(t) = \sum_{\omega} \hat{Q}_{\omega} \phi_{\omega}(t),
\end{align}
for some set of functions $\phi_{\omega}$. Typical choices, like the Fourier basis, however do not work, as the positive semi-definiteness constraint $Q(t) \geq 0$ is very difficult to enforce. It is an intriguing direction for future research if we can use bases of non-negative functions~\cite{freeman2021schauder} to get around this limitation. 

\subsection{Discretizing the success probability}\label{ssec:discretization_bouds_sip}

In this section we will discuss how to discretize and formalize the POVM optimization problem in Proposition~\ref{prop:sdp_formulation} in order to obtain mathematically formal statements. Recall that it was given by 
\begin{align}\label{eqn:equation_povm}
    \max_{Q(t) \geq 0}\mathstrut \left\{\left. \int \diff t \,  \Tr[ (w_{\delta} * [\mu \cdot \rho])(t) Q(t) ] \, \right| \, \int \diff t \, Q(t) = \bbI \right\}.
\end{align}
Here we will focus on the Bayesian case for simplicity, as the minimax case is completely analogous.
First, we note that to be formally correct in its formulations it is necessary to restrict the continuous-parameter POVMs we optimize over to make sense of the optimization and maximum above. Here, we will consider two natural possible restrictions: first, to POVMs that are Lipschitz with respect to the operator norm with a Lipschitz constant $C_\delta$, \ie for all $t,t'$:
\begin{align}
\|Q(t)-Q(t')\|_{\infty} \leq C_\delta|t-t'|.
\end{align}
We choose to make the Lipschitz constant depend explicitly on $\delta$ to emphasize that, as $\delta\to0$, it is also necessary to make $C_\delta\to+\infty$ to make sure we are optimizing over POVMs that are sufficiently sensitive to distinguish the underlying states. Indeed, if we want the POVMs to distinguish states that are $\delta$ apart it is necessary that $\|Q(t)-Q(t')\|_{\infty} =\Omega(1)$ for $|t-t'|$ of order $\delta$. Note that the set of Lipschitz POVMs is a compact, convex subset with respect to the operator norm, and that the success probability is a linear functional of the parametrized POVM. Thus, the $\max$ in Eq.~\eqref{eqn:equation_povm} is justified, as it is attained.
Another natural variation is to consider the set of continuous POVMs. As we let $C_\delta\to+\infty$ we can approximate any continuous function by a Lipschitz function, so this can be seen as a limiting case of the previous problem. 

We now discuss how to discretize the problem in Eq.~\eqref{eqn:equation_povm} and show convergence. For the setting of Lipschitz POVMs we will also obtain quantitative statements. 
To obtain such statements, we will assume that the curve of states $\rho(t)$ is Lipschitz with respect to the trace distance, \ie for all $t,t'$:
\begin{align}
\|\rho(t)-\rho(t')\|_1 \leq C_\rho |t-t'|.
\end{align}
We further need to discretize the set of possible measures $\mu$. We will assume here that $\mu$ has compact support and periodic boundary conditions, but the general case can be considered by considering a sequence of measures with compact support $\mu_n$ that approximate $\mu$.

We will pick a discretization parameter $\Delta$ satisfying $1>\delta>\Delta>0$ such that $\delta/\Delta\in\mathbb{N}$ and consider a discretization of the image of $\rho(t)$, \ie on the level of states. Let $T=|\textrm{supp}(\mu)|$ be the size of the support of $\mu$. We will discretize it into $N=\lceil T\Delta^{-1}\rceil$ points $\{ t_1,t_2,\ldots,t_N\}$ that are $\Delta$ apart. Furthermore, we will call $T_i=[t_i-\Delta/2,t_i+\Delta/2]$.
Given the $t_i$ and $\delta$, we will call 
\begin{align}
 N_{\delta}(i) \coloneqq \big\{j \, \big| \, |t_i-t_j|\leq \delta\big\}.
\end{align}

We then consider the following semidefinite program for a given parameter $\delta$, which gives the maximal success probability for this window function:
\begin{align}\label{eqn:discretized_SDP}
\operatorname{maximize}\quad &\sum_{i=1}^N\mu(T_i)\sum_{j\in N_\delta(i)}
\Tr\left[Q_j \rho(t_i)\right]\\ \nonumber
\textrm{subject to}\quad
&Q_i\geq0, \nonumber\\
&\sum_{i=1}^N Q_i=\bbI,\nonumber \\
&\|Q_i-Q_{i+1}\|_{\infty}\leq C_\delta\Delta \ \text{ for all } i \nonumber
\end{align}
The semidefinite program in Eq.~\eqref{eqn:discretized_SDP} can be solved in time that is polynomial in $N$ and the dimension of the states $\rho(t)$, and, as the number of constraints is linear in $\Delta^{-1}$, the complexity will also depend polynomially on $\Delta^{-1}$. Note further that the operator norm constraint can be recast as the linear matrix inequality $-C_\delta\Delta\leq Q_i-Q_{i+1}\leq C_\delta\Delta$, so it is a valid SDP constraint.

Of course, the central question is how fast the value of the above SDP converges to that of Eq.~\eqref{eqn:equation_povm}. Before we prove the convergence, let us give some Lemmas:
\begin{lemma}\label{lem:disc_states}
Let $\Delta$ be a discretization parameter satisfying $1>\delta>\Delta>0$, $T=|\textrm{supp}(\mu)|$ be the size of the support of $\mu$. Consider a discretization of the image of $\rho(t)$ into $N=\lceil T\Delta^{-1}\rceil$ points $\{ t_1,t_2,\ldots,t_N\}$ points that are $\Delta$ apart given by $\rho(t_i)$. Furthermore, we will call $T_i=[t_i-\Delta/2,t_i+\Delta/2]$. Further, assume that the curve of states is Lipschitz continuous:
\begin{align}
\|\rho(t_1)-\rho(t_2)\|_1\leq C_\rho|t_1-t_2|.
\end{align}
Then we have:
\begin{align}\label{equ:curve_LIipschitz}
    \left\|\left(\int_{T_i} \diff \mu(t)\, \rho(t)\right)-\mu(T_i)\rho(t_i)\right\|_1\leq C_\rho\Delta\mu(T_i).
\end{align}
\end{lemma}
\begin{proof}
Note that:
\begin{align}
    \left(\int_{T_i} \diff t \, \mu(t)\rho(t)\right)-\mu(T_i)\rho(t_i) = \int_{T_i} \diff t\, \mu(t)(\rho(t)-\rho(t_i)).
\end{align}
Thus, by the triangle inequality,
\begin{align}\label{equ:putnormintegral}
\left\|\left(\int_{T_i}\diff \mu(t) \, \rho(t)\right)-\mu(T_i)\rho(t_i)\right\|_1
\leq \int_{T_i} \diff \mu(t) \, \|\rho(t)-\rho(t_i)\|_1.
\end{align}
By the Lipschitz condition in Eq.~\eqref{equ:curve_LIipschitz} and our choice of $t_i$ we have that
\begin{align}
\|\rho(t)-\rho(t_i)\|_1\leq C_\rho|t-t_i|.
\end{align}
Inserting this bound in Eq.~\eqref{equ:putnormintegral} we get that:
\begin{align}
\int_{T_i}\diff \mu(t) \, \|\rho(t)-\rho(t_i)\|_1
\leq C_\rho \int_{T_i} \diff \mu(t) \, |t-t_i|
\leq C_\rho\Delta\mu(T_i),
\end{align}
which proves the claim.
\end{proof}

\begin{lemma}\label{lem:disc_POVM}
Let $\Delta$ be a discretization parameter satisfying $1>\delta>\Delta>0$, $T=|\textrm{supp}(\mu)|$ be the size of the support of $\mu$. Let $\{Q(t)\}$ be a continuous POVM that is Lipschitz with respect to the operator norm, i.e.
\begin{align}
   \|Q(t_1)-Q(t_2)\|_{\infty}\leq C_\delta|t_1-t_2|.
\end{align}
Consider the following discretization discretization of the POVM:
\begin{align}
Q_i=\int_{T_i} \diff t \, Q(t).
\end{align}
Then for all $1\leq i< T$:
\begin{align}\label{equ:discretization_Lip_POVM}
\|Q_{i+1}-Q_i\|_{\infty}\leq C_\delta\Delta^2
\end{align}
and for all $t\in T_i$ we have
\begin{align}\label{equ:windo_close}
   \left\|\int_{t-\tfrac{\delta}{2}}^{t+\tfrac{\delta}{2}} \diff s \, Q(s) - \sum_{j\in N_\delta(i)} Q_j \right\|_{\infty} \leq C_\delta\Delta\delta.
\end{align}
\end{lemma}
\begin{proof}
Let us start by proving Eq.~\eqref{equ:discretization_Lip_POVM}. Note that:
\begin{align}
    Q_{i+1}-Q_i = \int_{T_i} \diff s \, Q(s+\Delta)-Q(s).
\end{align}
By a triangle inequality followed by our Lipschitz condition we have:
\begin{align}
\|Q_{i+1}-Q_i\|_{\infty}
\leq\int_{T_i} \diff s \, \|Q(s+\Delta)-Q(s)\|_{\infty}
\leq C_\delta\Delta\int_{T_i}\diff s 
=C_\delta\Delta^2.
\end{align}
Let us now prove Eq.~\eqref{equ:windo_close} in a similar manner. We have that:
\begin{align}
\int_{t-\tfrac{\delta}{2}}^{t+\tfrac{\delta}{2}}\diff s \, Q(s)- \sum\limits_{j\in N_\delta(i)}Q_j=\int_{t-\tfrac{\delta}{2}}^{t+\tfrac{\delta}{2}}\diff s \, [Q(s)-Q(s+t_i-t)],
\end{align}
where we used the fact that $\delta/\Delta\in\mathbb{N}$. Thus, 
\begin{align}
\left\|\int_{t-\tfrac{\delta}{2}}^{t+\tfrac{\delta}{2}}\diff s \, [Q(s)-Q(s+t_i-t)] \right\|_{\infty}
&\leq \int_{t-\tfrac{\delta}{2}}^{t+\tfrac{\delta}{2}}\diff s \, \|Q(s)-Q(s+t_i-t)\|_{\infty} \\
&\leq C_\delta|t_i-t|\int_{t-\tfrac{\delta}{2}}^{t+\tfrac{\delta}{2}}\diff s \\
&\leq   C_\delta\delta\Delta.
\end{align}
\end{proof}
Our strategy will now consist in showing that the above discretization $\{Q_i\}$ of a continuous POVM gives a feasible point of the SDP in Eq.~\eqref{eqn:discretized_SDP} whose value does not change significantly from that of its continuous-time counterpart.
\begin{proposition}\label{prop:cont-to-disc}
  Under the same conditions as in Lemmas~\ref{lem:disc_states} and \ref{lem:disc_POVM} we we have that:
\begin{align}
\left|\int \diff t \,  \Tr[ (w_{\delta} * [\mu \cdot \rho])(t) Q(t) ]-\sum_i\mu(T_i)\sum_{j\in N_\delta(i)}\Tr[Q_j\rho(t_i)]\right|\leq C_\delta\delta\Delta+C_\rho\Delta.
\end{align}
\end{proposition}
\begin{proof}
We will proceed by showing the statement on each interval $T_i$. We have:
\begin{align}
\begin{split}
&\int_{T_i} \diff t \,  \Tr[ (w_{\delta} * [\mu \cdot \rho])(t) Q(t) ]-\mu(T_i)\sum_{j\in N_\delta(i)}\Tr[Q_j\rho(t_i)]\\
&\qquad=\int_{T_i} \diff \mu(t) \, \left\{\Tr[ \rho(t) (w_{\delta} * Q)(t) ]-\sum_{j\in N_\delta(i)}\Tr[Q_j\rho(t_i)]\right\}.    
\end{split}
\end{align}
First, note that as $Q(t)$ is a POVM, $\|(w_{\delta} * Q)(t) \|_{\infty} \leq1$. Thus, by Hölder's inequality:
\begin{align}
\int_{T_i} \diff \mu(t) \, \Tr[ [\rho(t)-\rho(t_i)] (w_{\delta} * Q)(t) ] 
&\leq \int_{T_i} \diff \mu(t) \, \|\rho(t)-\rho(t_i)\|_1\|(w_{\delta} * Q)(t) \|_{\infty} \\
&\leq \mu(T_i) C_\rho \Delta,
\end{align}
where we used Lemma~\ref{lem:disc_states} to bound the $1-$norm. From this we conclude that:
\begin{align}
&\left|\int_{T_i} \diff t \,  \mu(t)\left\{\Tr[ \rho(t) (w_{\delta} * Q)(t) ]-\sum_{j\in N_\delta(i)}\Tr[Q_j\rho(t_i)]\right\}\right| \\
&\qquad = 
\left|
\int_{T_i} \diff \mu(t) \, \left\{
\Tr[ [\rho(t) - \rho(t_i)] (w_{\delta} * Q)(t) ]
+\Tr\left[\rho(t_i) \left( (w_{\delta} * Q)(t)  - \sum_{j\in N_\delta(i)} Q_j
                \right)\right]
\right\}
\right| \\
&\qquad\leq \mu(T_i)C_\rho\Delta+\left|\int_{T_i} \diff \mu(t) \, \Tr\left[\rho(t_i) \left( (w_{\delta} * Q)(t)  - \sum_{j\in N_\delta(i)} Q_j
                \right)\right]\right|.
\end{align}
Let us now estimate the second term in the RHS of the last equation. Again applying a combination of Hölder and triangle inequalities, 
\begin{align}
\left|\int_{T_i} \diff \mu(t) \, \left\{\Tr\left[ \rho(t_i) \left( (w_{\delta} * Q)(t) -\sum_{j\in N_\delta(i)}Q_j\right)\right] \right\}\right|\leq \int_{T_i} \diff \mu(t) \, \lVert \rho(t_i) \rVert_1 \|(w_{\delta} * Q)(t) -\sum_{j\in N_\delta(i)}Q_j)\|_{\infty}.
\end{align}
By Lemma~\ref{lem:disc_POVM} we have that:
\begin{align}
\int_{T_i} \diff \mu(t) \, \|(w_{\delta} * Q)(t) -\sum_{j\in N_\delta(i)}Q_j)\|_{\infty} \leq \mu(T_i)\delta\Delta C_\delta. 
\end{align}
We conclude that 
\begin{align}
    \left|\int_{T_i} \diff \mu(t) \, \left\{ \Tr[ \rho(t) (w_{\delta} * Q)(t) ]-\sum_{j\in N_\delta(i)}\Tr[Q_j\rho(t_i)]\right\} \right| \leq  \mu(T_i)(\delta\Delta C_\delta+C_\rho\Delta).
\end{align}
An application of the triangle inequality and summing over all $i$ yields the claim.
\end{proof}
The proof above shows how, given one feasible POVM for the continuous-time version of the metrology problem, it is possible to construct a feasible POVM for the discretization without significantly changing the success probability as long as $\Delta$ is small enough. 

We now present the other direction: given one feasible point of the discretized problem, we construct a continuous-time version thereof that is Lipschitz-continuous and whose success probability does not differ significantly from the original value. 

Given a feasible POVM of the SDP in Eq.~\eqref{eqn:discretized_SDP}, $\{Q_i\}$, we define the continuous POVM $Q(t)$ by linear interpolation, \emph{i.e.}, for some $t=pt_i+(1-p)t_{i+1}$ as
\begin{align}\label{equ:cont_version_POVM}
Q(t)\coloneqq pQ_i+(1-p)Q_{i+1}.
\end{align}
We will now show that this is indeed a POVM and that it is Lipschitz:
\begin{lemma}\label{lem:disc-to-cont}
Let $\{Q_i\}$ be a family of POVMs such that:
\begin{align}
    \|Q_i-Q_{i+1}\|_{\infty} \leq C_\delta \Delta
\end{align}
and define $Q(t)$ as in Eq.~\eqref{equ:cont_version_POVM}. Then $Q(t)$ is a POVM and 
\begin{align}\label{equ:Lipschitz_cont_disc}
\|Q(s_1)-Q(s_2)\|_{\infty} \leq C_\delta \Delta |s_1-s_2|.
\end{align}
Furthermore, for $t\in T_i$,
\begin{align}\label{equ:close_window_no_window}
    \|(w_\delta* Q)(t)-\sum_{j\in N_\delta(i)}Q_j\|\leq C_\delta \delta\Delta.
\end{align}
\end{lemma}
\begin{proof}
Note that we have that:
\begin{align}
\int_{T_i}\diff t \, Q(t) = 
\int_{0}^1\diff p \, (p Q_i+(1-p)Q_{i+1})
=\frac{1}{2}(Q_i+Q_{i+1}).
\end{align}
summing over all $i$ (recall the periodic boundary conditions) we see that
\begin{align}
\int\diff t \, Q(t)=\bbI.
\end{align}
The fact that $Q(t)$ is positive semi-definite is obvious, as it is point-wise the convex combination of positive semi-definite operators, which shows that it is indeed a POVM.
Let us now show Eq.~\eqref{equ:Lipschitz_cont_disc}. For $s_1,s_2\in T_i$ we have that:
\begin{align}
Q(s_1)-Q(s_2)=(s_1-s_2)(Q_i-Q_{i+1}).
\end{align}
The claim follows by using the Lipschitz continuity of the $Q_i$. For $s_1,s_2$ in different intervals, say $s_1\in T_{i_1},s_2\in T_{i_2}$ we apply the same argument to the sequence of points $s_1,i_1+1$, $i_1+1,i_1+2$,.., $i_2,s_2$. To show Eq.~\eqref{equ:close_window_no_window} we can follow the same route as for Eq.~\eqref{equ:windo_close}, as we have already established that the POVM is Lipschitz.
\end{proof}

Now we have constructed a candidate for a feasible point of the continuous-time problem given a feasible point of the discrete problem. We will show below that the value they achieve is also close, which will lead us to conclude that the discretized and the continuous problems have comparable values for small enough values of $\Delta$.

\begin{proposition}\label{prop:discretizaion}
For given $\delta,C_\delta>0$ and $\Delta < \delta,<1$, let $\eta_{\delta}$ be the value of Eq.~\eqref{eqn:equation_povm} when restricted to POVMs that are Lipschitz with Lipschitz constant at most $C_\delta$. Furthermore, assume that $\rho(t)$ is Lipschitz with constant $C_\rho$. Then the value of the SDP in  Eq.~\eqref{eqn:discretized_SDP}, $\eta_{\delta,\Delta}$, satisfies:
\begin{align}
|\eta_{\delta}-\eta_{\delta,\Delta}|\leq (C_\rho+C_\delta)\Delta.
\end{align}
\end{proposition}
\begin{proof}
It follows from Proposition~\ref{prop:cont-to-disc} that, starting from the continuous version of the problem, we can construct a feasible point of the SDP whose success probability differs by at most $(C_\rho+C_\delta)\Delta$. This gives $\eta_{\delta}<\eta_{\delta,\Delta}-(C_\rho+C_\delta)\Delta$. On the other hand, Lemma~\ref{lem:disc-to-cont} shows how to construct a continuous-time Lipschitz POVM from a feasible point of the SDP whose success probability will differ by at most $(C_\rho+C_\delta)\Delta$. This follows from Eq.~\eqref{equ:close_window_no_window} and similar reasoning as in Proposition~\ref{prop:cont-to-disc}. In a nutshell, we first approximate $(w_\delta * Q)(t)$ by $\sum_{j\in N_i(\delta)}Q_j$ on each interval $T_i$. Then we approximate the averaged states on each interval by $\rho(t_i)$. This gives $\eta_\Delta<\eta_{\delta}-(C_\rho+C_\delta)\Delta$, which yields the claim.
\end{proof}
We can then obtain the convergence to continuous POVMs from the last statement:
\begin{corollary}
For given $\delta,C_\delta>0$ and $\Delta<\delta$, let $\eta_{{\delta,\Delta}}$ be the value of Eq.~\eqref{eqn:equation_povm} when restricted to POVMs that are Lipschitz with Lipschitz constant at most $C_\delta$ and $\eta$ when only requiring continuity. Furthermore, assume that $\rho(t)$ is Lipschitz with constant $C_\rho$. Then we have:
\begin{align}
\eta=\lim_{C_\delta\to\infty}\lim_{\Delta\to0}\eta_{{\delta,\Delta}}.
\end{align}
\end{corollary}
\begin{proof}
The statement follows from the fact that, by the Stone-Weierstrass theorem, any continuous function can be approximated arbitrarily well by Lipschitz functions. Thus, we can find a sequence of Lipschitz POVMs (with possibly diverging Lipschitz constant) that approximates the target POVM. In turn, these will be approximate arbitrarly well by the discretized SDP as we let $\tau\to0$ by Proposition~\ref{prop:discretizaion}. This gives the statement.
\end{proof}
It would be interesting to obtain statements about the Lipschitz constant of good POVMs for the metrology task to get more quantitative statements even in the continuous case. Indeed, one intuitively expects that it should not be too advantageous to have POVMs that vary significantly faster than the states $\rho(t)$ themselves. Thus, we leave to future work to investigate if we can always take $C_\delta=\mathcal{O}(C_\rho)$ to obtain a good approximation.

\section{Relation to entropy measures}\label{ssec:entropy_measures}

\subsection{Max-entropy radius}

We can generalize an argument of Ref.~\cite{audenaert2014upper} that gives an alternative characterization of the success probability in terms of the max-relative entropy~\cite{datta2009minmaxrelent} as
\begin{align}\label{eq:D_max_definition}
    D_{\max}(X \fatpipe Y ) = \log \inf \{ \gamma \, | \, X \leq \gamma Y\}=\log\lambda_{\max}(Y^{-\frac{1}{2}}XY^{-\frac{1}{2}})=\log\norm{Y^{-\frac{1}{2}}XY^{-\frac{1}{2}}}_{\infty},
\end{align}
for $X$ and $Y$ positive semi-definite and $\supp(X)\subseteq\supp(Y)$, where $\norm{\cdot}_{\infty}$ is the operator norm and the last equality holds because $Y^{-\frac{1}{2}}XY^{-\frac{1}{2}}$ is positive semi-definite.
Indeed, we have the following corollary of Proposition~\ref{prop:bayesian_sdp}.

\begin{scorollary}[Max-entropy radius]\label{cor:Bayesian_succ_max_ent_radius}
For a given set of states $\rho(t)$ with prior distribution $\mu(t)$, we define its max-relative entropy radius with respect to the window function $w(t)$ as
\begin{align}
    r_{\max}(w, \mu, \rho) = \inf_{\sigma} \sup_{t} D_{\max}( (w\ast[\mu \cdot \rho])(t) \Vert \sigma ),
\end{align}
where we optimize over arbitrary quantum states $\sigma$. We have that
\begin{align}
    r_{\max}(w, \mu, \rho) = \log \eta^{*}(w, \mu, \rho).
\end{align}
\end{scorollary}
\begin{proof}
This follows immediately from Proposition~\ref{prop:bayesian_sdp}, \eqref{eq-Bayesian_success_Gmax}, and Lemma~\ref{lem:G_max_dual}.
\end{proof}

The above corollary implies the following upper bound on the success probability
\begin{scorollary}
For any state $\sigma$, we have that
\begin{align}
    \eta^{*}(w, \mu, \rho) \leq \exp \sup_t D_{\max}( (w * [\mu \rho])(t) \Vert \sigma).
\end{align}
\end{scorollary}

\subsection{Conditional min-entropy}

    We can also relate the Bayesian success probability $\eta^{\ast}(w,\mu,\rho)$ to the \textit{conditional min-entropy}~\cite{koenig09operationalminmax}, which is defined as
    \begin{equation}
        H_{\min}(A\pipe B)_{\rho}\coloneqq -\inf_{\substack{\sigma_B\geq 0\\\Tr[\sigma_B]=1}}D_{\max}(\rho_{AB}\fatpipe\mathbb{I}_A\otimes\sigma_B),
    \end{equation}
    for arbitrary positive semi-definite operators $\rho_{AB}$, where $D_{\max}$ is defined in \eqref{eq:D_max_definition}.

    \begin{corollary}
        For a given set of states $\rho(t)$ with prior distribution $\mu(t)$, and a window function $w$, the optimal success probability $\eta^{\ast}(w,\mu,\rho)$ is given by
        \begin{equation}
            -\log\eta^{*}(w,\mu,\rho)=H_{\min}(T\pipe S)_X,
        \end{equation}
        where $P_{TS}$ is the classical--quantum operator~\cite{walter_lower_2014}
        \begin{equation}
            P_{TS}\coloneqq\int\Diff{t}\ketbra{t}{t}_T\otimes (w\ast[\mu\cdot\rho])(t)_S.
        \end{equation}
    \end{corollary}

    \begin{proof}
        This is a straightforward consequence of definitions. Using the definition of $D_{\max}$ in \eqref{eq:D_max_definition}, we have that, for all states $\sigma_S$,
        \begin{align}
            D_{\max}(P_{TS}\fatpipe\mathbb{I}_T\otimes\sigma_S)&=\log\Norm{\sigma_S^{-\frac{1}{2}}P_{TS}\sigma_S^{-\frac{1}{2}}}_{\infty}\\
            \nonumber
            &=\Norm{\int\Diff{t}\ketbra{t}{t}\otimes \sigma_S^{-\frac{1}{2}}(w\ast[\mu\cdot\rho_S])(t)\sigma_S^{-\frac{1}{2}}}_{\infty}\\
             \nonumber
            &=\sup_t\Norm{\sigma_S^{-\frac{1}{2}}(w\ast[\mu\cdot\rho])(t)\sigma_S^{-\frac{1}{2}}}_{\infty}\\
             \nonumber
            &=\sup_tD_{\max}((w\ast[\mu\cdot\rho])(t)\Vert\sigma).
             \nonumber
        \end{align}
        Therefore,
        \begin{equation}
            \inf_{\substack{\sigma_S\geq 0\\\Tr[\sigma_S]=1}}D_{\max}(P_{TS}\fatpipe\mathbb{I}_T\otimes\sigma_S)=\inf_{\substack{\sigma\geq 0\\\Tr[\sigma]=1}}\sup_tD_{\max}((w\ast[\mu\cdot\rho])(t)\fatpipe\sigma)=\log\eta^{\ast}(w,\mu,\rho),
        \end{equation}
        which implies the desired result, where for the last equality we have used Corollary~\ref{cor:Bayesian_succ_max_ent_radius}.
    \end{proof}

\section{Properties of the proposed quantifiers}\label{ssec:properties_of_quantifiers}
In this section, we outline some properties of the proposed quantities.

\subsection{Properties of the unoptimized quantifiers}
In this section, we outline some properties of the Bayesian success probability $\eta$ and the minimax success probability $\overline{\eta}$. 
We first establish that the Bayesian success probability is continuous in all arguments:
\begin{sproposition}[Continuity properties of the Bayesian success probability]\label{sprop:bayesian_success_prob_continuity_properties}
    The Bayesian success probability $\eta$ has continuity
    \begin{enumerate}\setlength{\itemsep}{0pt}
        \item[(i)] in the window function as
        \begin{align}
            |\eta(w, \mu, \rho, Q) - \eta(w', \mu, \rho, Q)| \leq \lVert w - w' \rVert_{\infty},
        \end{align}
        
        \item[(ii)] in the measure as
        \begin{align}
            |\eta(w, \mu, \rho, Q) - \eta(w, \mu', \rho, Q)| \leq \frac{1}{2}\lVert \mu - \mu' \rVert_1,
        \end{align}
        
        \item[(iii)] in the state set as
        \begin{align}
            |\eta(w, \mu, \rho, Q) - \eta(w, \mu, \rho', Q)| \leq \int \diff \mu(t)\, \lVert \rho(t) - \rho'(t) \rVert_1\leq  \max_t \, \lVert \rho(t) - \rho'(t) \rVert_1,
        \end{align}
        
        \item[(iv)] jointly in measure and state set as
        \begin{align}
            |\eta(w, \mu, \rho, Q) - \eta(w, \mu', \rho', Q)| \leq \int \diff t \, \lVert \mu(t) \rho(t) - \mu'(t) \rho'(t) \rVert_1,
        \end{align}
        
        \item[(v)] in the measurement as
        \begin{align}
            |\eta(w, \mu, \rho, Q) - \eta(w, \mu, \rho, Q')| &\leq\int \diff \mu(t) \, \left\lVert \int \diff \tau \,  w(t-\tau)  [Q(\tau) - Q'(\tau)] \right\rVert_{\infty} \\
        &\leq \left( \int \diff \tau \,  w(\tau) \right)\max_t \left\lVert Q(t) - Q'(t) \right\rVert_{\infty} 
        \nonumber.
        \end{align}
    \end{enumerate}
\end{sproposition}

\begin{proof}
    The properties are shown as follows.
\begin{enumerate}
    \item[(i)] follows because
    \begin{align}
        |\eta(w, \mu, \rho, Q) - \eta(w', \mu, \rho, Q)| &= 
        \left| \int \diff \mu(t) \, \diff \tau \,  [w(\tau) - w'(\tau)] \Tr [ Q(t-\tau) \rho(t) ] \right|\\
        \nonumber
        &=\int \diff \tau \,  [w(\tau) - w'(\tau)] \int \diff \mu(t) \, \Tr [ Q(t-\tau) \rho(t) ] \\
        \nonumber
        &\leq \int \diff \tau \,  |w(\tau) - w'(\tau)| \int \diff \mu(t) \, \Tr [ Q(t-\tau) \rho(t)] \\
        \nonumber
        &\leq \left(\max_t |w(t) - w'(t)|\right) \int \diff \tau \, \diff \mu(t) \, \Tr [ Q(t-\tau) \rho(t)]\\
        \nonumber
        &= \lVert w - w'\Vert_{\infty}.
    \end{align}
    
    \item[(ii)] follows from the change of measure inequality~\cite{ohnishi2020novel},
    \begin{align}
        \EV_{\mu(t)}[\phi(t)] \leq \EV_{\mu'(t)}[\phi(t)] + \frac{1}{2}\lVert \mu - \mu'\rVert_1,
    \end{align}
    which is valid if $0 \leq \phi(t) \leq 1$ when we note that 
    \begin{align}
    \eta(w, \mu, \rho, Q) &= \EV_{\mu} \, \{ \Tr [ \rho(t) (w * Q)(t) ] \}
    \end{align}
    is exactly of that form. The statement follows by symmetrizing through exchange of $\mu$ and $\mu'$.

    \item[(iii)] directly follows from item (iv).
    
    \item[(iv)] follows from the short rearrangement
    \begin{align}
        |\eta(w, \mu, \rho, Q) - \eta(w, \mu', \rho', Q)| &= \left|\int \diff t \, \diff \tau \,  w(t-\tau) \Tr [ Q(\tau) [\mu(t)\rho(t) - \mu'(t)\rho'(t)] ]\right| \\
        \nonumber
        &=\left| \int \diff t  \, \Tr \left[ \left( \int \diff \tau \,  w(t-\tau)  Q(\tau)\right) [\mu(t) \rho(t) - \mu'(t)\rho'(t)] \right]\right|   \nonumber\\
        \nonumber
        &\leq \left|\int \diff t \, \Tr \left[ \mu(t)\rho(t) - \mu'(t) \rho'(t) \right] \right|  \nonumber\\
        \nonumber
        &\leq \int \diff t \, \lVert \mu(t) \rho(t) - \mu'(t)\rho'(t)\rVert_1 , 
          \nonumber
    \end{align}
    where we have used the matrix Hölder inequality and exploited the fact that $\int \diff \tau \, w(t-\tau) Q(\tau) 
    \leq \,bI$ because $Q(\tau)$ is a POVM.
    
    \item[(v)] follows from
    \begin{align}
        |\eta(w, \mu, \rho, Q) - \eta(w, \mu, \rho, Q')| &= \left|\int \diff \mu(t) \, \diff \tau \,  w(t-\tau) \Tr [ [Q(\tau) - Q'(\tau)] \rho(t) ]\right| \\
         \nonumber
        &=\left| \int \diff \mu(t) \, \Tr \left[ \rho(t) \left( \int \diff \tau \,  w(t-\tau)  [Q(\tau) - Q'(\tau)] \right) \right]\right| \\
         \nonumber
        &\leq\int \diff \mu(t) \, \left\lVert \int \diff \tau \,  w(t-\tau)  [Q(\tau) - Q'(\tau)] \right\rVert_{\infty} \\
         \nonumber
        &\leq \max_t \left\lVert \int \diff \tau \,  w(t-\tau)  [Q(\tau) - Q'(\tau)] \right\rVert_{\infty}  \\
         \nonumber
        &\leq \left( \int \diff \tau \,  w(\tau) \right)\max_t \left\lVert Q(t) - Q'(t) \right\rVert_{\infty} ,
    \end{align}
    where we have used the matrix H{\"o}lder inequality.
\end{enumerate}
\end{proof}

We can exploit the linearity and positivity of the success probability in its arguments to establish majorization-type statements 
as follows.

\begin{sproposition}[Majorization properties of the Bayesian success probability]\label{sprop:bayesian_success_prob_properties_majorization}
    The Bayesian success probability $\eta$ has the following majorization properties:
    \begin{enumerate}\setlength{\itemsep}{0pt}
        \item[(i)] Let $w_{-}(t)$ and $w_{+}(t)$ be two functions such that $w_{-}(t) \leq w(t) \leq w_{+}(t)$ for all $t$. Then
        \begin{align}
             \eta(w_{-}, \mu, \rho, Q) \leq \eta(w, \mu, \rho, Q) \leq \eta(w_{+}, \mu, \rho, Q).
        \end{align}
        
        \item[(ii)] Let $g_{-}(t)X_{-}(t)$ and $g_{+}(t)X_{+}(t)$ be sets of operators such that $g_{-}(t)X_{-}(t) \leq \mu(t)\rho(t) \leq g_{+}(t)X_{+}(t)$ for all $t$. Then
        \begin{align}
            \eta(w, g_{-}, X_{-}, Q) \leq \eta(w, \mu, \rho, Q) \leq \eta(w, g_{+}, X_{+}, Q).
        \end{align}
        This directly implies similar statements when only measure or states are changed.
        
        \item[(iii)] Let $X_{-}(t)$ and $X_{+}(t)$ be sets of operators such that $X_{-}(t) \leq Q(t) \leq X_{+}(t)$ for all $t$. Then
        \begin{align}
            \eta(w, \mu, \rho, X_{-}) \leq \eta(w, \mu, \rho, Q) \leq \eta(w, \mu, \rho, X_{+}).
        \end{align}
    \end{enumerate}
\end{sproposition}
\begin{proof}
The properties directly follow from the linearity of $\eta$ in its arguments and the fact that the arguments which are not bounded above and below are non-negative.
\end{proof}

The minimax success probability also has comparable properties, some of which are inherited from the Bayesian success probability:
\begin{sproposition}[Continuity properties of the minimax success probability]\label{sprop:cont_prop_minimax}
    The minimax success probability $\overline{\eta}$ has continuity
    \begin{enumerate}
        \item[(i)] in the window function
        \begin{align}
            |\overline{\eta}(w, \rho, Q) - \overline{\eta}(w', \rho, Q)| \leq \lVert w - w' \rVert_{\infty},
        \end{align}
        
        \item[(ii)] in the state set
        \begin{align}
            |\overline{\eta}(w,  \rho, Q) - \overline{\eta}(w,  \rho', Q)| \leq \max_t \, \lVert \rho(t) - \rho'(t) \rVert_1,
        \end{align}
        
        \item[(ii)] in the measurement
        \begin{align}
            |\overline{\eta}(w,  \rho, Q) - \overline{\eta}(w,  \rho, Q')| \leq \left(\int \diff \tau \,  w(\tau)\right) \max_t \, \left\lVert Q(t)-  Q'(t)\right\rVert_{\infty}.
        \end{align}
\end{enumerate}
\end{sproposition}

\begin{proof}
    The properties are shown as follows.
\begin{enumerate}
    \item[(i)] Follows from the calculation
    \begin{align}
        \overline{\eta}(w, \rho, Q) &= \min_t \, \int \diff \tau \,  w(t-\tau) \Tr [ Q(\tau) \rho(t) ] \\
        &= \min_t \, \int \diff \tau \,  [w(t-\tau) + w'(t-\tau) - w'(t-\tau)] \Tr [ Q(\tau) \rho(t) ]   \nonumber\\
        &\leq \min_t \, \int \diff \tau \,  w'(t-\tau) \Tr [ Q(\tau) \rho(t) ] + \max_t \, \int \diff \tau \,  [w(t-\tau)-w'(t-\tau)]  \Tr [ Q(\tau) \rho(t)  ]   \nonumber\\
        &\leq \overline{\eta}(w', \rho, Q) + \max_t \, \{w(t-\tau)-w'(t-\tau)\}  \nonumber \\
        &\leq \overline{\eta}(w', \rho, Q) + \max_t \, |w(t-\tau)-w'(t-\tau)|   \nonumber\\
        &= \overline{\eta}(w', \rho, Q) + \lVert w - w'\rVert_{\infty},  \nonumber
    \end{align}
    where we have exploited the inequality $\min_t f(t) + g(t) \leq \min_t f(t) + \max_t g(t)$ as well as the fact that 
    \begin{align}\int \diff \tau \, g(t-\tau) Q(\tau) \leq \max_t g(t) \bbI ,
    \end{align}
    because $Q(\tau)$ is a POVM and the matrix Hölder inequality. The statement is 
    then obtained by symmetrizing, via exchange of $w$ and $w'$.
    
    \item[(ii)] follows from the short rearrangement
    \begin{align}
        \overline{\eta}(w, \rho, Q) &= \min_t \, \int \diff \tau \,  w(t-\tau) \Tr [ Q(\tau) \rho(t) ] \\
        &= \min_t \, \int \diff \tau \,  w(t-\tau) \Tr [ Q(\tau) [\rho(t) - \rho'(t) + \rho'(t)] ]   \nonumber\\
        &\leq \min_t \, \int \diff \tau \,  w(t-\tau) \Tr [ Q(\tau) \rho'(t) ] + \max_t \, \int \diff \tau \,  w(t-\tau) \Tr [ Q(\tau) [\rho(t) - \rho'(t)] ]   \nonumber\\
        &= \overline{\eta}(w, \rho', Q) + \max_t \, \Tr \left[ \left( \int \diff \tau \,  w(t-\tau)  Q(\tau)\right) [\rho(t) - \rho'(t)] \right]   \nonumber\\
        &\leq \overline{\eta}(w, \rho', Q) + \max_t \, \Tr \left[ \rho(t) - \rho'(t)\right]   \nonumber\\
        &\leq \overline{\eta}(w, \rho', Q) +\max_t \,  \lVert \rho(t) - \rho'(t)\rVert_1,   \nonumber
    \end{align}
    where we have exploited the inequality $\min_t f(t) + g(t) \leq \min_t f(t) + \max_t g(t)$ as well as 
    \begin{align}\int \diff \tau \, w(t-\tau) Q(\tau) \leq \bbI,
    \end{align}
    because $Q(\tau)$ is a POVM and the matrix Hölder inequality. The statement is then obtained by symmetrizing via exchange of $\rho$ and $\rho'$.
    
    \item[(iii)] follows similarly as
    \begin{align}
        \overline{\eta}(w, \rho, Q) &= \min_t \, \int \diff \tau \,  w(t-\tau) \Tr [ Q(\tau) \rho(t) ] \\
        &= \min_t \, \int \diff \tau \,  w(t-\tau) \Tr [ [Q(\tau) - Q'(\tau) + Q'(\tau)] \rho(t) ] 
        \nonumber
        \\
        &\leq \min_t \, \int \diff \tau \,  w(t-\tau) \Tr [ Q'(\tau) \rho(t) ] + \max_t \, \int \diff \tau \,  w(t-\tau) \Tr [ [Q(\tau)-Q'(\tau)] \rho(t)  ]   \nonumber\\
        &= \overline{\eta}(w, \rho, Q') + \max_t \, \Tr \left[\left[\left( \int \diff \tau \,  w(t-\tau)  Q(\tau)\right)-\left( \int \diff \tau \,  w(t-\tau)  Q'(\tau)\right)\right] \rho(t) \right]   \nonumber\\
        &\leq \overline{\eta}(w, \rho, Q') + \max_t \, \left\lVert  \int \diff \tau \,  w(t-\tau) \left[ Q(\tau)-  Q'(\tau)\right]\right\rVert_{\infty}   \nonumber\\
        &\leq \overline{\eta}(w, \rho, Q') + \left(\int \diff \tau \,  w(\tau)\right) \max_t \, \left\lVert Q(\tau)-  Q'(\tau)\right\rVert_{\infty}
          \nonumber
    \end{align}
    where we have exploited the inequality $\min_t f(t) + g(t) \leq \min_t f(t) + \max_t g(t)$ as well as the fact that $\int \diff \tau \, w(t-\tau) Q(\tau) \leq \bbI$ because $Q(\tau)$ is a POVM and the matrix Hölder inequality. 
    The statement is then obtained by symmetrizing via exchange of $\rho$ and $\rho'$.

     \item[(ii)] directly follows from the same argument as (i) when we note that
    \begin{align}
        \Tr [ Q(\tau) \rho_{-}(t)] \leq \Tr [ Q(\tau) \rho(t) ]\leq \Tr [ Q(\tau) \rho_{+}(t) ],
    \end{align}
    and hence the same ordering holds for the minimum.
    
     \item[(iii)] directly follows from the same argument as (i) when we note that
    \begin{align}
        \Tr [ Q_{-}(\tau) \rho(t)] \leq \Tr [ Q(\tau) \rho(t) ]\leq \Tr [ Q_{+}(\tau) \rho(t) ].
    \end{align}
    and hence the same ordering holds for the minimum.
\end{enumerate}
\end{proof}

Note that some of these continuity bounds are necessarily loose as the other parameters could be chosen in a particularly pathological way.

The minimax success probability has majorization properties similar to the ones of the Bayesian success probability:
\begin{sproposition}[Majorization properties of the minimax success probability]\label{Sprop:minimax_success_prob_properties_majorization}
    The minimax success probability $\overline{\eta}$ has the following majorization properties:
    \begin{enumerate}\setlength{\itemsep}{0pt}
        \item[(i)] Let $w_{-}(t)$ and $w_{+}(t)$ be two functions such that $w_{-}(t) \leq w(t) \leq w_{+}(t)$ for all $t$. Then
        \begin{align}
             \overline{\eta}(w_{-}, \rho, Q) \leq \overline{\eta}(w, \rho, Q) \leq \overline{\eta}(w_{+}, \rho, Q).
        \end{align}
        
        \item[(ii)] Let $X_{-}(t)$ and $X_{+}(t)$ be sets of operators such that $X_{-}(t) \leq \rho(t) \leq X_{+}(t)$ for all $t$. Then
        \begin{align}
            \overline{\eta}(w, X_{-}, Q) \leq \overline{\eta}(w, \rho, Q) \leq \overline{\eta}(w, X_{+}, Q).
        \end{align}
        
        \item[(iii)] Let $X_{-}(t)$ and $X_{+}(t)$ be sets of operators such that $X_{-}(t) \leq Q(t) \leq X_{+}(t)$ for all $t$. Then
        \begin{align}
            \overline{\eta}(w, \rho, X_{-}) \leq \overline{\eta}(w, \rho, Q) \leq \overline{\eta}(w, \rho, X_{+}).
        \end{align}
    \end{enumerate}
Additionally, it has 
\begin{enumerate}\setlength{\itemsep}{0pt}   
        \item[(iv)] monotonicity in the state set:
        \begin{align}
            \{ \rho'(t) \}_t \subseteq \{ \rho(t)\}_t \ \Rightarrow \ \overline{\eta}(w, \rho, Q) \leq \overline{\eta}(w, \rho', Q).
        \end{align}
    \end{enumerate}
\end{sproposition}
\begin{proof}
The properties directly follow from the fact that $\overline{\eta}$ is the result of a minimization of a linear function and the fact that the arguments which are not bounded above and below are non-negative. Item (iv) follows from the simple observation that the set of all measures over $\{ \rho(t)\}_t$ includes all measures over $\{ \rho'(t) \}_t$ and hence the optimization is bound to yield a higher value.
\end{proof}

\subsection{Properties of the optimized quantities}
As outcomes of a convex optimization, the optimized quantities fulfill convexity in the remaining parameters: 
\begin{sproposition}[Convexity]
The optimized Bayesian success probabilities $\eta^{*}(w, \mu, \rho)$, $\eta^{*}(w, \mu, \calN, Q)$ and $\eta^{*}(w, \mu, \calN)$ are convex in all arguments and especially fulfill a triangle inequality in the window function.
\end{sproposition}
\begin{proof}
The convexity of the Bayesian quantities is a direct consequence of linearity and the fact that 
\begin{align}
    \sup_Q f(Q) + g(Q) \leq \sup_Q f(Q) + \sup_Q g(Q).
\end{align}
The triangle inequality in the window function follows similarly.
\end{proof}

The optimized probabilities also fulfill a type of data processing inequality:
\begin{sproposition}[Data processing]
The optimized Bayesian success probabilities obey the following data-processing inequalities
\begin{align}
\eta^{*}(w, \mu, \calN[\rho(\cdot)]) &\leq \eta^{*}(w, \mu, \rho(\cdot)) ,\\
\eta^{*}(w, \mu, \calN(\cdot) \circ \calA, Q) &\leq \eta^{*}(w, \mu, \calN(\cdot), Q) ,\\
\eta^{*}(w, \mu, \calB \circ \calN(\cdot) \circ \calA) &\leq \eta^{*}(w, \mu, \calN(\cdot)),
\end{align}
which directly imply similar statements for the optimized minimax success probabilities
\begin{align}
\overline{\eta}^{*}(w, \calN[\rho(\cdot)]) &\leq \overline{\eta}^{*}(w,  \rho(\cdot)),\\
\overline{\eta}^{*}(w, \calN(\cdot) \circ \calA, Q) &\leq \overline{\eta}^{*}(w,  \calN(\cdot), Q), \\
\overline{\eta}^{*}(w, \calB \circ \calN(\cdot) \circ \calA) &\leq \overline{\eta}^{*}(w, \calN(\cdot)).
\end{align}
\end{sproposition}
\begin{proof}
Denote with $\calQ$ the set of all POVMs and with $\calN[\calQ]$ its image under a quantum channel $\calN$. We note that for all quantum channels, we have that $\calN^{\dagger}[\calQ] \subseteq \calQ$ due to the CPTP property of $\mathcal{N}$, which implies that $\mathcal{N}^{\dagger}$ is completely positive and unital. Then, we
find
\begin{align}
    \eta^{*}(w, \mu, \calN[\rho(\cdot)]) &= \sup_{Q(t) \in \calQ} \int \diff \mu(t) \, \diff \tau \, w(t - \tau)\Tr [\calN[\rho(t)] Q(\tau) ] \\
     \nonumber
    &= \sup_{Q(t) \in \calQ} \int \diff \mu(t) \, \diff \tau \, w(t - \tau)\Tr [\rho(t) \calN^{\dagger}[Q(\tau)] ] \\
     \nonumber
    &= \sup_{Q'(t) \in \calN^{\dagger}[\calQ]} \int \diff \mu(t) \, \diff \tau \, w(t - \tau)\Tr[\rho(t) Q'(\tau)]\\
     \nonumber
    &\leq \sup_{Q'(t) \in \calQ} \int \diff \mu(t) \, \diff \tau \, w(t - \tau)\Tr [\rho(t) Q'(\tau)]\\
     \nonumber
    &= \eta^{*}(w, \mu, \rho(\cdot)),
\end{align}
where we have used the fact that a value optimized over a subset can never exceed the value of the fully optimized case. The minimax result is implied as the Bayesian statement holds independently of the prior. Similar arguments give rise to the statements for optimization over input states and joint optimization of probe and measurement.
\end{proof}

\subsection{Subdivision trick}\label{ssec:proof_subdivision_trick}
We now prove the subdivision trick of the main text.

\begin{proof}[Proof of Lemma~\ref{lem:subdivision_trick}]
The proof follows a similar idea as the proof outlined for the reduction of multi-hypothesis testing to the binary case in Ref.~\cite{audenaert2014upper}. We consider the case in which, additionally to the (unknown) state $\rho(t)$ with $t$ sampled according to $\mu(t)$, an oracle supplies us the information that $t$ lies in a certain interval $I_{t'}$ where the oracle samples $t'$ uniformly from $I_t$, \textit{i.e.}, it uniformly randomly samples one of the intervals of size $T$ containing $t$. With this additional information available, we can restrict our attention to the interval $I_{t'}$ and perform a Bayesian update of our prior which means we now deal with $\mu|_{I_{t'}}$. We can then perform the optimal strategy for this prior. As the additional information can only improve our estimate, we obtain
\begin{align}
    \eta^{*}(\delta, \rho, \mu) &\leq \int \diff \mu(t) \, \int \diff t' \, \bbP[ t' \pipe t ] \eta^{*}(\delta, \rho, \mu|_{I_{t'}}) \\
    &= \int \diff \mu(t) \, \int \diff t' \, \frac{1}{T} \chi[ t' \in I_t ]  \eta^{*}(\delta, \rho, \mu|_{I_{t'}}) \\
    &= \int \diff \mu(t) \, \int \diff t' \, \frac{1}{T} \chi[ t \in I_{t'} ]  \eta^{*}(\delta, \rho, \mu|_{I_{t'}})\\
    &= \frac{1}{T}  \int \diff t' \, \mu(I_{t'})  \eta^{*}(\delta, \rho, \mu|_{I_{t'}}).
\end{align}
Renaming $t'$ to $t$ yields the first statement of the lemma. The second statement follows by bounding $\eta^{*}(\delta, \rho, \mu|_{I_{t'}})$ by its maximum over $t'$ and recognizing that 
\begin{align}
    \frac{1}{T} \int \diff t' \, \mu(I_{t'}) &=
    \frac{1}{T} \int \diff t' \, \int_{I_{t'}} \diff t \, \mu(t) \\
    &= \frac{1}{T} \int \diff t' \, \int \diff t \, \mu(t) \, \chi[|t-t'| \leq T/2] \\
    &= \frac{1}{T} \int \diff t \, \mu(t) \, T \\
    &= 1.
\end{align}
\end{proof}

\section{Optimal post-processing with fixed measurement}\label{ssec:maximum_weight_estimation}

In this section, we give additional content relative to Section~\ref{sec:opt_fixed_measurement} of the main text.
There, it was established in Eq.~\eqref{eqn:succ_prob_for_smap} that the success probability for the smoothed maximum a posteriori estimate relates to the function infinite norm of the smoothed posterior probability:
\begin{align}
    \eta(\delta, \mu, \rho, Q_{M, \tau_{\mathrm{SMAP}}^{*}}) &= \int \diff \nu(\lambda) \, \max_{\tau} \mathstrut(w_{\delta} * P(\cdot \pipe \lambda))(\tau) \\
&= \int \diff \nu(\lambda) \, \mathstrut \lVert w_{\delta} * P(\cdot \pipe \lambda) \rVert_{\infty}.\nonumber
\end{align}
The connection to the infinity norm allows us to derive some simple upper bounds by applying Young's convolution inequality~\cite{beckner_inequalities_1975}:
\begin{slemma}\label{slem:young_conv_inequ}
The success probability of the smoothed maximum a posteriori estimate obeys the upper bound
\begin{align}
    \eta(w, \mu, \rho, Q_{M, \tau_{\mathrm{SMAP}}^{*}})&\leq \lVert w \rVert_{p} \int \diff \nu(\lambda) \, \lVert P(\cdot \pipe \lambda) \rVert_{q}
\end{align}
for all $1/p + 1/q = 1$.
\end{slemma}
The above inequality immediately trivializes when choosing $p=\infty$, $q=1$ but yields a non-trivial upper bound otherwise. This can be useful, when for example an upper bound on the likelihood is known.

A particularly interesting property of the smoothed maximum a posteriori estimate, of which we will use a discrete analogue later to relate metrology to binary hypothesis testing, is the following bound on contributions to the error for the $\delta$ window, which makes the dependence on the window size explicit:
\begin{slemma}
For any interval $I$ of cardinality $|I| \leq 2\delta$, define its $\delta$-complement as
\begin{align}
    \Bar{I}^{\delta} \coloneqq \{ t \pipe \text{ there exists } t' \in I \text{ such that } |t-t'| > 2\delta \}.
\end{align}
Then, for any interval $I$ outside the smoothed maximum a posteriori interval, \textit{i.e.}, any interval contributing to the error, we have that
\begin{align}
    \int_I \diff t \, P(t \pipe \lambda) \leq \int_{\Bar{I}^{\delta}} \diff t \, P(t \pipe \lambda)
\end{align}
for all $\lambda$.
\end{slemma}
\begin{proof}
Let us denote the smoothed maximum a posteriori interval as $I^{*} = [\tau^{*}_{\mathrm{SMAP}}(\lambda) - \delta,\tau^{*}_{\mathrm{SMAP}}(\lambda) + \delta]$. Then, by definition of the smoothed maximum a posteriori estimate, we have that
\begin{align}
    \int_J \diff t \, P(t \pipe \lambda) \leq \int_{I^{*}} \diff t \, P(t \pipe \lambda)
\end{align}
for all compact intervals $J$ of cardinality $|J| = 2\delta$. We can exploit this and optimize over all intervals $J$ that contain the target interval $I$ to obtain
\begin{align}
    \int_I \diff t \, P(t \pipe \lambda)\leq \inf_{J \colon |J|=2\delta, I \subseteq J} \int_{I^{*} \backslash J} \diff t \, P(t \pipe \lambda),
\end{align}
where we implicitly made use of the assumption that $I$ lies outside of $I^{*}$, \textit{i.e.}, that $I \cap I^{*} = \emptyset$. By construction, $\Bar{I}^{\delta}$ is the complement of the union of all possible $J$ of cardinality $2\delta$ that contain $I$, which is exactly what we achieve as well on the right hand side by choosing the smallest interval $I^{*} \backslash J$, making the interval achieving the optimization a subinterval of $\Bar{I}^{\delta}$. Extending the integration to all of $\Bar{I}^{\delta}$ yields the statement of the lemma.
\end{proof}

The further study of upper and lower bounds for the smoothed maximum a posteriori strategy would be a promising direction for future research, especially to relate to concepts of classical statistics.

\section{Relation to hypothesis testing}\label{ssec:bin_hyp_testing}

\subsection{Multi-hypothesis testing as a special case of metrology}\label{ssec:mht_as_metrology}
The notions introduced above can be considered as a continuous generalization of the discrete multi-hypothesis testing problem for quantum states and quantum channels, respectively. In the multi-hypothesis testing problem for quantum states, one is given a set of states $\{ \rho_i \}_{i=1}^m$ -- in the Bayesian setting with associated prior probabilities $\{ p_i \}_{i=1}^m$ -- and is tasked to find a measurement given by POVM effects $\{ Q_i \}_{i=1}^m$ that maximizes the success probability~\cite{khatri2020principles,audenaert2014upper}
\begin{align}
    P_s(\{ p_i \rho_i\}_{i=1}^m) &= \sup_{\{ Q_i \}} \sum_{i=1}^m p_i \Tr[\rho_i Q_i].
\end{align}
In a similar way, one can define the associated minimax multi-hypothesis testing problem where we desired to find a measurement 
\begin{align}
    \overline{P}_s(\{ \rho_i\}_{i=1}^m) &= \sup_{\{ Q_i \}} \min_i \Tr[\rho_i Q_i]
\end{align}
with optimal worst-case performance.
The \emph{hypothesis testing} has in the binary case already been solved by Helstrom and Holevo \cite{helstrom1969quantum,Hol72}. In particular, the optimal success probability has been determined by them in seminal work to be
\begin{align}
    P_s(p \rho_1, (1-p)\rho_2) = \frac{1}{2} + \frac{1}{2}\lVert p \rho_1 - (1-p) \rho_2 \rVert_1,
\end{align}
where 
$\lVert\cdot\rVert_1$ denotes the trace or nuclear norm. In general, it is known that~\cite{koenig09operationalminmax}
\begin{equation}
    P_s(\{p_i\rho_i\}_{i=1}^m)=2^{-H_{\min}(X|B)_{\rho}},
\end{equation}
where $H_{\min}(X|B)=-\inf_{\substack{\sigma_B\geq 0,\Tr[\sigma_B]=1}} D_{\max}(\rho_{XB}\Vert\mathbb{I}_X\otimes\sigma_B)$ is the conditional min-entropy and $\rho_{XB}=\sum_{i=1}^mp_i\ketbra{i}{i}\otimes\rho_i$. It has been shown in Ref.~\cite{li2016discriminating} that the asymptotic rates for the Bayesian and minimax multi-hypothesis testing problem coincide and are given by the minimal pairwise Chernoff divergence
\begin{align}
    \overline{R}(\{ \rho_i\}_{i=1}^m) = \lim_{n\to\infty}-\frac{1}{n}\log (1-\overline{P}_s(\{  \rho_i\}_{i=1}^m)) =  -\log \min_{i\neq j}\min_{0\leq s \leq 1}\Tr [ \rho_i^{s} \rho_j^{1-s}].
\end{align}
It now becomes clear that our notion of success for quantum metrology encompasses the quantum multi-hypothesis testing problem when we consider the following metrological problem that embeds a multi-hypothesis testing problem. Consider the following prior over states,
\begin{align}\label{eqn:prior_for_embedding_discrete_problem}
    \mu(t) = \sum_{i=1}^m p_i \delta(t - i),
\end{align}
together with any parametrized state $\rho(t)$ such that $\rho(i) = \rho_i$ and a window function
\begin{align}\label{eqn:window_for_embedding_discrete_problem}
    w_{1/3}(t) = \begin{cases} 1 & \text{if } |t|\leq 1/3 \\ 0 & \text{else.}\end{cases}
\end{align}
It is obvious that
\begin{align}
    \eta^{*}(w_{1/3}, \mu, \rho) &= P_s(\{ p_i \rho_i \}_{i=1}^m).
\end{align}
Later in this manuscript, we make use of this property to derive upper bounds on the metrological success probability from binary state discrimination. 

In a similar way, we can consider the problem of multi-hypothesis testing for quantum channels. In this case, a discrete set of quantum channels \smash{$\{ \calN_i \}_{i=1}^m$}, possibly again with prior probabilities \smash{$\{p_i\}_{i=1}^m$}, is given and the optimal experimental prescription for distinguishing between these quantum channels is to be found. In the single-copy case, we need to find an input state $\rho_0$ and a measurement $\{Q_i\}_{i=1}^m$ that maximizes the success probability
\begin{align}
    P_{s}(\{ p_i \calN_i \}_{i=1}^m) = \sup_{\rho_0,\{ Q_i \}} \sum_{i=1}^m p_i \Tr [\calN_i[\rho_0] Q_i ].
\end{align}
Using again the prior $\mu(t)$ of Eq.~\eqref{eqn:prior_for_embedding_discrete_problem} together with any parametrized channel $\calN(t)$ such that $\calN(i) = \calN_i$ and the window function $w_{1/3}$ of Eq.~\eqref{eqn:window_for_embedding_discrete_problem}, we see that the success probability is given by the corresponding metrological success probability:
\begin{align}
    \eta^*(w_{1/3}, \mu, \calN) = P_{s}(\{ p_i \calN_i \}_{i=1}^m).
\end{align}
As we outlined in the preceding section, in the setting where multiple copies are available, there are different possible ways of using the quantum channel in question, corresponding to the i.i.d.\ case where the same input state is used repetitively, the separable case where only separable states are used as inputs, the parallel case where an entangled state is prepared and fed through the quantum channel and the adaptive case where a quantum comb is used. The success probabilities we defined in these cases naturally generalize the same notions available in the multi-hypothesis testing problem for channels.

\subsection{Upper bound on success probability from multi-hypothesis testing}\label{ssec:upper_bound_state_discr}
In this section, we prove Theorem~\ref{thm:succ_prob_upper_bound_mht_delta_window} of the main text and discuss its extensions to arbitrary window functions and quantum channels.

\begin{stheorem}\label{sthm:succ_prob_upper_bound_mht}
For a given window function $w$, fix any set $\calS = \{(\lambda, s)\}$ of prior probabilities $\lambda \geq 0$ and shifts $s \in \bbR$ such that $\sum_{\lambda \in \calS} \lambda = 1$. Then, for a state set $\rho(t)$, possibly with prior $\mu(t)$, we have the upper bounds
\begin{align}
    \eta^{*}(w, \mu, \rho) &\leq K \int \diff t \, P^{*}_s(\{ \lambda \, \mu(t+s) \rho(t+s)\}_{(\lambda, s) \in \calS}), \\
    \overline{\eta}^{*}(w, \rho) &\leq K  \inf_t \overline{P}^{*}_s(\{  \rho(t+s)\}_{s\in \calS}),
\end{align}
where we introduced the constant
\begin{align}\label{eqn:def_K_upper_bound}
    K \coloneqq \sup_t  \left\{\sum_{s \in \calS} w(t + s) \right\},
\end{align}
which measures the overlap of the windows for the different shifts.
\end{stheorem}
\begin{proof}
First, we recall the definition of the optimal multi-hypothesis testing success probability for a set of operators $\{ A_i\}$:
\begin{align}
    P^{*}_s(\{ A_i\}) &\coloneqq \sup_{\substack{0 \leq Q_i \leq \bbI \\ \sum_i Q_i = \bbI}} \sum_i \Tr[ A_i Q_i ].
\end{align}
We exploit that we can shift the time axis of the integration that computes the success probability arbitrarily, to observe that
\begin{align}
    \eta(w, \mu, \rho, Q)  = \sum_{(\lambda,s)\in\calS} \int \diff t \, \lambda \Tr[ \mu(t+s) \rho(t+s) (w * Q)(t+s)].
\end{align}
Using the definition of $K$ given in the theorem statement, we see that defining the operators
\begin{align}
    Q_s(t) \coloneqq \frac{1}{K} (w * Q)(t+s)
\end{align}
yields a valid sub-normalized POVM for all $t$ as
\begin{align}
    \sum_{s \in \calS} Q_s(t) &= \frac{1}{K}\sum_{s \in \calS} (w * Q)(t + s) \\
     \nonumber
    &\leq \frac{1}{K} \left(\left[ \sum_{s \in \calS} w(\cdot + s) \right] * Q\right)(t) \\
     \nonumber
    &\leq \frac{1}{K} ( K * Q )(t) \\
     \nonumber
    &= (1 * Q)(t) \\
     \nonumber
&= \bbI. \nonumber
\end{align}
This means that the operators $\{ Q_s \}_{s \in \calS}$ can serve as a candidate POVM in the optimization that computes $P_s( \{ \lambda \, \mu(t+s) \rho(t+s) \}_{(\lambda,s)\in \calS})$,and hence
\begin{align}
    \eta(w, \mu, \rho, Q)\leq K \int \diff t \, P^{*}_s(\{ \lambda \, \mu(t+s) \rho(t+s)\}_{(\lambda, s) \in \calS}\})
\end{align}
which implies the first statement of the theorem as the upper bound is independent of the chosen POVM $Q(t)$.

The minimax statement is derived in a similar fashion, observing that we can also apply the time shifting trick to obtain
\begin{align}
    \overline\eta(w, \rho, Q)  &=  \sum_{(\lambda,s)\in\calS} \lambda \inf_t  \Tr[ \rho(t+s) (w * Q)(t+s)] \\
    &\leq \inf_t \sum_{(\lambda,s)\in\calS} \lambda \Tr[ \rho(t+s) (w * Q)(t+s)].\nonumber
\end{align}
Here, we again make the argument that the $\{Q_s\}_{s \in \calS}$ form a candidate POVM and then optimize over all possible $\lambda$ to obtain the theorem statement.
\end{proof}
Note that if $K$ in the above theorem is larger than the inverse success probability, then the bound becomes vacuous. This means, as the success probability asymptotically approaches 1, any bound that should work asymptotically must have $K = 1$. Let us now prove the Theorem from the main text:
\begin{proof}[Proof of Theorem~\ref{thm:succ_prob_upper_bound_mht_delta_window}]
    In the case of a rectangular window with tolerance $\delta$, we have that as long as $|s - s'| >2\delta$ for any two shifts in $\calS$, that the rectangular windows do
     not overlap. Therefore, under the assumptions of Theorem~\ref{thm:succ_prob_upper_bound_mht_delta_window}, we have that $K = 1$ and the statement therefore directly follows from Theorem~\ref{sthm:succ_prob_upper_bound_mht}.
\end{proof}

Next, we present a corollary of Theorem~\ref{sthm:succ_prob_upper_bound_mht} that extends the statement to metrology protocols defined with respect to channels. In this setting, the success probability is defined as a joint optimization over the input state and the POVM, possibly using an ancillary system:
\begin{align}
    P^{*}_s(\{ p_i, \calN_i\}) &\coloneqq \sup \left\{ \left.\sum_i p_i \Tr[ (\bbI \otimes \calN_i)[\rho_0] Q_i ] \, \right| \, \rho_0 \geq 0, \Tr[\rho_0] = 1, 0 \leq Q_i \leq \bbI , \sum_i Q_i = \bbI\right\}.
\end{align}
The minimax success probability is defined analogously by optimizing the minimum over $i$. With these notions in place, we obtain the following statement:
\begin{scorollary}[Upper bound on success probability]\label{scorr:succ_prob_upper_bound_delta_window_multi_hypothesis_channel}
For a given window function $w$, fix any set $\calS = \{(\lambda, s)\}$ of prior probabilities $\lambda \geq 0$ and shifts $s \in \bbR$ such that $\sum_{\lambda \in \calS} \lambda = 1$. Then, for a channel set $\calN(t)$, possibly with prior $\mu(t)$, we have the upper bounds
\begin{align}
    \eta^{*}(\delta, \mu, \calN) &\leq K \int \diff t \, P^{*}_s(\{ \lambda \, \mu(t+s) \rho(t+s)\}_{(\lambda, s) \in \calS}), \\
    \overline{\eta}^{*}(\delta, \calN) &\leq K  \inf_t \overline{P}^{*}_s(\{  \rho(t+s)\}_{s\in \calS}),
\end{align}
where the constant $K$ is defined in Eq.~\eqref{eqn:def_K_upper_bound}.
\end{scorollary}
\begin{proof}
Theorem~\ref{sthm:succ_prob_upper_bound_mht} is valid for any set of states, which means it also applies when $\rho(t) = (\bbI \otimes \calN(t))[\rho_0]$ for the optimal probe state $\rho_0$. The fact that the optimal success probability for discriminating quantum channels is obtained by optimizing over $\rho_0$ and $Q(t)$ implies the corollary.
\end{proof}
Similar statements are readily obtained for adaptive discrimination of multiple channel copies and other variants of the channel metrology task.

\subsection{Upper bound via binary hypothesis testing}\label{ssec:upper_bound_from_mht}
In this section, we will derive some bounds using tools from symmetric hypothesis testing. Our first step is to derive a quantum analog of a Bretagnolle-Huber inequality (see \textit{e.g.},   Ref.~\cite{lumbreras_multi-armed_2022})
\begin{stheorem}[Binary hypothesis testing lower bound]\label{sthm:hyp_test_lower_bound}
Let $\rho$ and $\sigma$ be two quantum states and $0 \leq \lambda \leq 1$ a prior probability. Then, the optimal binary hypothesis testing error can be bounded from below via the fidelity as
\begin{align}
    P_e^{*}(\lambda \rho, (1-\lambda) \sigma) \geq \lambda(1-\lambda) F(\rho, \sigma)^2 = \lambda(1-\lambda) \exp\left( - \tilde{D}_{1/2}(\rho \, \lVert \, \sigma)\right),
\end{align}
where $\tilde{D}_{\alpha}$ denotes the sandwiched Rényi-relative entropy.
\end{stheorem}
\begin{proof}
We employ a strategy similar to the proof of Lemma 17 of Ref.~\cite{cheng_simple_2022} (compare also the proof of an analogue classical result in Ref.~\cite{lumbreras_multi-armed_2022}). To this end, we denote with $A = \lambda \rho$ and $B = (1-\lambda) \sigma$ and write the optimal hypothesis testing success and error probabilities as
\begin{align}
    P_s^{*}(A,B) &= \Tr[ A \Pi_A ] + \Tr[ B \Pi_B] ,\\
    P_e^{*}(A,B) &= \Tr[ A \Pi _B ] + \Tr[ B \Pi_A ], 
\end{align}
where $\Pi_A$ and $\Pi_B = \bbI - \Pi_A$ are the optimal POVM effects.
Then, we define a CPTP map
\begin{align}
    \Lambda( X \oplus Y ) = \Tr[ X \Pi_A \oplus Y \Pi_B ] \oplus \Tr[ X \Pi_B \oplus Y \Pi_A]
\end{align}
such that
\begin{align}
    \Lambda( A \oplus B ) &= P_s^{*}(A, B) \oplus P_e^{*}(A, B), \\
    \Lambda( B \oplus A ) &= P_e^{*}(A, B) \oplus P_s^{*}(A, B).
\end{align}
The data-processing property of the fidelity implies that
\begin{align}
    F( {A}  \oplus {B},  {B} \oplus {A}) = 2 \sqrt{\lambda (1-\lambda)} F(\rho, \sigma) 
    \leq F( \Lambda[{A}  \oplus {B} ], \Lambda[{B} \oplus {A}] ) 
    = 2 \sqrt{P_s^{*}(A,B) P_e^{*}(A,B)}.
\end{align}
Using $P_s^{*}(A,B) \leq 1$ and the definition of the sandwiched Rényi relative entropy then yields the statement of the Theorem.

\end{proof}
We can use the above theorem to deduce the following lower bound for metrology:
\begin{stheorem}[Two-point error probability lower bound]
For a given tolerance $\delta$ and a set of states $\rho(t)$, we have the lower bound 
\begin{align}
    1 - \overline{\eta}(\overline{\delta}, \rho) \geq \frac{1}{4}\exp\left(-\inf_{|t-t'| > 2\delta} D_{1/2}(\rho(t) \,\|\, \rho(t')) \right)
\end{align}
on the minimax success probability.
\end{stheorem}
\begin{proof}
We start from the upper bound on the minimax success probability derived in Theorem~\ref{thm:succ_prob_upper_bound_mht_delta_window}. Applied to a single time shift, it especially implies that
\begin{align}
    \overline\eta(\delta, \rho) \leq \inf_t \overline{P}_s^{*}(\rho(t), \rho(t + 2\delta)).
\end{align}
This is equivalent to a lower bound on the minimax error
\begin{align}
    1 - \overline\eta(\delta, \rho) \geq \sup_t \overline{P}_e^{*}(\rho(t), \rho(t + 2\delta)).
\end{align}
For the sake of simplicity, we will lower bound the optimal minimax error with the one obtained from a uniform prior, \ie
\begin{align}
    1 - \overline\eta(\delta, \rho) &\geq \sup_{|t-t'| > 2\delta} P_e^{*}\left(\frac{1}{2}\rho(t), \frac{1}{2}\rho(t')\right).
\end{align}
Now, applying Theorem~\ref{sthm:hyp_test_lower_bound} yields
\begin{align}
     1 - \overline\eta(\delta, \rho) &\geq \frac{1}{4} \sup_{|t-t'| > 2\delta} F(\rho(t), \rho(t'))^2 \\
     &= \frac{1}{4} \exp\left(-\inf_{|t-t'| > 2\delta} \tilde{D}_{1/2}(\rho(t) \,\|\, \rho(t')) \right)\\
     &\geq \frac{1}{4} \exp\left(-\inf_{t} \tilde{D}_{1/2}(\rho(t) \,\|\, \rho(t+2\delta)) \right).
\end{align}
\end{proof}
Corollary~\ref{corr:two_point_fidelity_bound} of the main text follows immediately.

\subsection{Upper bound on the success probability from asymmetric hypothesis testing}\label{ssec:upper_bound_succ_prob_asym_ht}

We can use similar reasoning as in the derivation of Theorem~\ref{thm:succ_prob_upper_bound_mht_delta_window} to obtain a lower bound that makes use of \emph{asymmetric} hypothesis testing. In asymmetric (binary) hypothesis testing, the goal is to determine a measurement $\{M,\mathbb{I}-M\}$ that distinguishes between two hypothesis $\rho$ and $\sigma$, such that the so-called type-II error $\Tr[M\sigma]$ is minimized while maintaining an upper bound of $\epsilon\in[0,1]$ on the type-I error probability $\Tr[(\mathbb{I}-M)\rho]$. In particular, the optimal type-II error probability is given by~\cite{khatri2020principles}
\begin{equation}\label{eqn:hyp_test_opt_typeII_error_prob}
    \beta_{\epsilon}(\rho\Vert\sigma)=\inf\{\Tr[M\sigma]:0\leq M\leq\mathbb{I},\,\Tr[M\rho]\geq 1-\epsilon\}.
\end{equation}
The \emph{hypothesis testing relative entropy} is then defined to be the optimal type-II error exponent, namely,
\begin{equation}
    D_{\mathrm{h}}^{\eta}(\rho\Vert\sigma)=-\log\beta_{1-\eta}(\rho\Vert\sigma),
\end{equation}
for $\eta\in[0,1]$.

We can establish the following theorem:
\begin{theorem}[Asymmetric hypothesis testing bound]\label{thm:succ_prob_upper_bound_mht_delta_window_asymmetric}
For a given tolerance $\delta$, fix any set of shifts $\calS = \{s\}$ such that $|s| > 2\delta$ and for all distinct $s, s' \in \calS$ we have that $|s-s'| > 2\delta$. Then, for a state set $\rho(t)$ with prior $\mu(t)$ we have the upper bound
\begin{align*}
    \eta(\delta, \rho, \mu, Q) \leq 1 - \int \diff \mu(t) \, \sum_{s \in \calS}  \beta^{H(t+s)}_{\mathrm{h}}(\rho(t+s) \fatpipe \rho(t)),
\end{align*}
where we defined the shorthand
\begin{align}
    H(t) \coloneqq \Tr[ (w_{\delta}*Q)(t)\rho(t)].
\end{align}
In the minimax case, we have that for all $\overline{\eta} \leq \overline{\eta}^{*}(\delta, \rho)$ that
\begin{align*}
    \overline{\eta}^{*}(\delta, \rho) \leq 1- \max_t \sum_{ s \in \calS} \beta_{\mathrm{h}}^{\overline{\eta}}(\rho(t+s) \fatpipe \rho(t)).
\end{align*}
\end{theorem}
An advantage of this bound is that the asymptotic behavior of the asymmetric hypothesis testing error is better understood than the asymptotics of the symmetric hypothesis testing error and that especially the second-order asymptotics are known.
\begin{proof}
Our strategy consists of using $(w*Q)(t)$ as a candidate POVM effect for the asymmetric hypothesis test. We first treat the minimax case, where the error probability can be expressed as
\begin{align}
    1-\overline{\eta}(\delta, \rho, Q) &= \max_t \Tr[ (\overline{w} * Q)(t) \rho(t)],
\end{align}
where $\overline{w}(t) = 1 - w(t)$ can be seen as the complement of the window function. We now make use of the fact that the set $\calS$ is defined such that we have $\overline{w}(t) \geq \sum_{s \in \calS} w(t+s)$ and hence
\begin{align}
    1-\overline{\eta}(\delta, \rho, Q) &\geq \max_t \sum_{s \in \calS} \Tr[ (w(\cdot + s) * Q)(t) \rho(t)]
\end{align}
As the POVM effect $(w * Q)(t +s)$ achieves $\Tr[ (w * Q)(t+s) \rho(t+s) ] \geq \overline{\eta}(\delta, \rho)$ by definition of the minimax success probability it is a candidate for a binary hypothesis testing and hence
\begin{align}
    1-\overline{\eta}(\delta, \rho, Q) &\geq \max_t \sum_{s \in \calS} \beta_{\mathrm{h}}^{\overline{\eta}(\delta, \rho)}(\rho(t+s) \fatpipe \rho(t)).
\end{align}
The claimed statement follows from the monotonicity of the asymmetric hypothesis testing error and by optimizing the left hand side over the POVM $Q$.
For the Bayesian case, we introduce the notation
\begin{align}
    \overline{\eta}(t) \coloneqq \Tr[ (w * Q)(t) \rho(t)]
\end{align}
such that $\overline{\eta} = \min_t \overline{\eta}(t)$ and $\eta = \int\diff\mu(t) \, \overline{\eta}(t)$. Now, we can use the exact similar reasoning as above and write
\begin{align}
    1 - \overline{\eta}(t) &= \Tr[ (\overline{w} * Q)(t) \rho(t)] \\
    &\geq \sum_{s \in \calS} \beta_{\mathrm{h}}^{\overline{\eta}(t+s)}(\rho(t+s) \fatpipe \rho(t)),
\end{align}
where the only difference is that the argument of the asymmetric hypothesis test is now a function of $t$. We therefore obtain
\begin{align}
    \eta \leq 1 - \int \diff \mu(t) \, \sum_{s \in \calS}  \beta^{\overline{\eta}(t+s)}_{\mathrm{h}}(\rho(t+s) \fatpipe \rho(t)).
\end{align}
\end{proof}

\subsection{Fano-type bounds for quantum multi-hypothesis testing}\label{ssec:bounds_for_quantum_mht}
In this section we derive some bounds for quantum multi-hypothesis testing.

We first establish an analogue of Fano's inequality for quantum multi-hypothesis testing. To this end, we need two lemmas. The first establishes the behavior of the relative entropy under a direct sum:
\begin{slemma}[Relative entropy and direct sum]\label{slem:relative_entropy_and_direct_sum}
We have the identity
\begin{align}
    D\left(\left. \bigoplus_{i=1}^M \mu_i \rho_i \, \right\| \, \bigoplus_{i=1}^M \nu_i \sigma_i \right) &=
    D(\mu \fatpipe \nu) + \sum_{i=1}^M \mu_i D(\rho_i \fatpipe \sigma_i). 
\end{align}
\end{slemma}
\begin{proof}
The proof is a straightforward algebraic manipulation based off the additivity of the matrix logarithm,
\begin{align}
D\left(\left. \bigoplus_{i=1}^M \mu_i \rho_i \, \right\| \, \bigoplus_{i=1}^M \nu_i \sigma_i \right) &= \sum_{i=1}^M D( \mu_i \rho_i \fatpipe \nu_i \sigma_i) \\
&= \sum_{i=1}^M \Tr[ \mu_i \rho_i \{ \log( \mu_i \rho_i) - \log ( \nu_i \sigma_i ) \}] \\
&= \sum_{i=1}^M \Tr[ \mu_i \rho_i \{ \log( \rho_i) - \log ( \sigma_i ) + \log ( \mu_i ) - \log(\nu_i)] \\
&= \sum_{i=1}^M\mu_i \Tr[ \rho_i \{ \log( \rho_i) - \log ( \sigma_i )\}] +\sum_{i=1}^M\mu_i \{ \log ( \mu_i ) - \log(\sigma_i)\} \\
&= D(\mu \fatpipe \nu) + \sum_{i=1}^M \mu_i D(\rho_i \fatpipe \sigma_i).
\end{align}
\end{proof}
The second lemma we need concerns the optimization of the above expression over the prior probabilities $\mu_i$.
\begin{lemma}[Relative-entropy regularized expectation value]\label{slem:entropy_plus_expval_optimization}
Choosing $\mu_i \propto \nu_i e^{- x_i }$ and subsequently normalizing yields
\begin{align}
    \inf_{\mu} \left\{ D(\mu \fatpipe \nu) + \sum_{i=1}^M \mu_i x_i \right\} = -\log \sum_{i=1}^M \nu_i e^{-x_i}.
\end{align}
\end{lemma}
\begin{proof}
We first expand the expression as
\begin{align}
    D(\mu \fatpipe \nu) + \sum_{i=1}^M \mu_i x_i &=
    \sum_{i=1}^M \mu_i \left(x_i + \log \frac{\mu_i}{\nu_i}\right)
\end{align}
and perform Lagrange optimization under the restriction $\sum_{i=1}^M \mu_i = 1$. The Lagrange function is given by
\begin{align}
    L(\mu, \lambda) = \sum_{i=1}^M \mu_i \left(x_i + \log \frac{\mu_i}{\nu_i}\right) + \lambda \left( 1 - \sum_{i=1}^M \mu_i \right).
\end{align}
The Karush-Kuhn-Tucker conditions enforce that $\partial_i L(\mu, \lambda) = 0$ for all $i$, \ie
\begin{align}
    \partial_i L(\mu, \lambda) = x_i + 1 + \log \frac{\mu_i}{\nu_i} - \lambda = 0.
\end{align}
We thus set
\begin{align}
    \mu_i = \nu_i \exp(\lambda - 1 - x_i) \propto \nu_i e^{-x_i}.
\end{align}
The constant $\lambda$ is implicitly chosen such that the above is normalized, which yields
\begin{align}
    \mu_i = \left( \sum_{j=1}^M \nu_j e^{- x_j} \right)^{-1} \nu_i e^{-x_i}.
\end{align}
The value of the optimization problem is then
\begin{align}
    \sum_{i=1}^M \mu_i \left(x_i + \log \frac{\mu_i}{\nu_i}\right) &= \sum_{i=1}^M \mu_i \left(x_i + \log \left( \sum_{j=1}^M \nu_j e^{- x_j} \right)^{-1}  e^{-x_i}\right) \\
    \nonumber
    &=   \sum_{i=1}^M \mu_i \left(- \log  \sum_{j=1}^M \nu_j e^{- x_j} \right) \\
    \nonumber
    &=-  \log  \sum_{j=1}^M \nu_j e^{-x_j}  ,
    \nonumber
\end{align}
as claimed.
\end{proof}
We are now ready to prove the Fano-type bound:
\begin{stheorem}[Fano-type bound for quantum multi-hypothesis testing]\label{sthm:bayesian_fano_type_bound}
Let $\{ \mu_i \rho_i \}_{i=1}^M$ be a quantum multi-hypothesis testing problem. For any reference state $\sigma$, the error probability obeys
\begin{align}
    - \log P_e^{*}( \{ \mu_i \rho_i \}_{i=1}^M  ) \leq \frac{M}{M-1} \left(h(1/M) + D(\mu \fatpipe u) + \sum_{i=1}^M \mu_i D(\rho_i \fatpipe \sigma)  \right),
\end{align}
where $u$ is the uniform distribution over $M$ elements and $h$ is the binary entropy function. Optimizing over $\mu$ yields the following bound on the minimax error probability:
\begin{align}
    - \log \overline{P}_e^{*}( \{  \rho_i \}_{i=1}^M  ) \leq \frac{M}{M-1} \left(h(1/M) -\log \frac{1}{M}\sum_{i=1}^M e^{-D(\rho_i \fatpipe \sigma)} \right).
\end{align}
\end{stheorem}
\begin{proof}
We will combine three ingredients to obtain the result. 
First, a result of Vazquez-Vilar that relates the multi-hypothesis testing error to a asymmetric hypothesis test~\cite{vazquez-vilar_multiple_2016}
\begin{align}
    P_e^{*}( \{ \mu_i \rho_i \}_{i=1}^M  ) &= \max_{\sigma} \beta^{1-1/M}_{\mathrm{h}} \left(\left. \bigoplus_{i=1}^M \mu_i \rho_i \, \right\| \, \frac{1}{M} \sigma^{\oplus M} \right) \\
    &= \exp\left( - \min_{\sigma} D^{1-1/M}_{\mathrm{h}} \left(\left. \bigoplus_{i=1}^M \mu_i \rho_i \, \right\| \, \frac{1}{M} \sigma^{\oplus M} \right)  \right),
    \nonumber
\end{align}
where $D_{\mathrm{h}}^{\eta}$ is the hypothesis testing relative entropy with success probability $\eta$. Next, we use the standard bound 
\begin{align}
     D^{\eta}_h(\rho \,\|\, \sigma) \leq \frac{1}{\eta}( D(\rho \, \| \, \sigma) + h(1-\eta))
\end{align}
that relates the hypothesis testing relative entropy to the regular relative entropy~\cite{khatri2020principles}.
The final ingredient is Lemma~\ref{slem:relative_entropy_and_direct_sum}. Putting everything together yields that for all $\sigma$,
\begin{align}
    - \log P_e^{*}( \{ \mu_i \rho_i \}_{i=1}^M ) 
    &\leq D_{\mathrm{h}}^{1-1/M} \left(\left. \bigoplus_{i=1}^M \mu_i \rho_i \, \right\| \, \frac{1}{M} \sigma^{\oplus M} \right)\\
    &\leq \frac{M}{M-1}\left[ D\left(\left. \bigoplus_{i=1}^M \mu_i \rho_i \, \right\| \, \frac{1}{M} \sigma^{\oplus M} \right) + h(1/M)\right]\nonumber\\
    &= \frac{M}{M-1}\left[ D(\mu\fatpipe u) + \sum_{i=1}^M \mu_i D\left(  \rho_i \fatpipe \sigma \right) + h(1/M)\right]\nonumber
\end{align}
as claimed. The result on the minimax success probability follows by minimizing the above bound over $\mu$, with the resulting bound obtained via Lemma~\ref{slem:entropy_plus_expval_optimization} using $x_i = D(\rho_i \fatpipe \sigma)$ and $\nu_i = 1/M$.
\end{proof}
Usually, in the literature for lower bounds in quantum information science (see, \textit{e.g.},   Ref.~\cite{haah2017sample-optimal}), Fano's inequality is used in a different form, namely the one that bounds the mutual information of a quantum channel. The mutual information in turn is bounded by the Holevo information of an ensemble of quantum states, usually taken to be a uniform mixture of the multi-hypothesis testing states. We recover a similar argument from the above bound by using a uniform distribution as the candidate for the minimax optimization and the expected state as the candidate state $\sigma = \sum_{i=1}^M \mu_i \rho_i$, in which case the average relative entropy is exactly the Holevo information of the ensemble $\{ \mu_i \rho_i \}_{i=1}^M$
\begin{align}
    \sum_{i=1}^M \mu_i D\left( \rho_i \, \left\| \, \sum_{j=1}^M \mu_j \rho_j \right.\right) &=
     \sum_{i=1}^M \mu_i \Tr\left[ \rho_i \left( \log \rho_i - \log \sum_{j=1}^M \mu_j \rho_j \right)\right]\\
     \nonumber
    &=   S\left( \sum_{i=1}^M \mu_i \rho_i\right)- \sum_{i=1}^M \mu_i S(\rho_i)  \\
    &= \chi(\{ \mu_i \rho_i \}_{i=1}^M).
    \nonumber
\end{align}
We note that the major improvement over this strategy is the logarithmic dependence on the hypothesis testing error, which, however, comes at the cost of a missing cross-dependence between the Holevo information and the number of samples. We can phrase this as the following corollary:
\begin{scorollary}[Fano-type bound with Holevo information]\label{scorollary:bayesian_fano_type_bound_holevo}
Let $\{ \rho_i \}_{i=1}^M$ be a quantum multi-hypothesis testing problem. We have the following bound on the minimax error probability:
\begin{align}
    - \log \overline{P}_e^{*}( \{  \rho_i \}_{i=1}^M  ) \leq \frac{M}{M-1} \left(h(1/M) + \chi\left( \left\{ \frac{1}{M} \rho_i \right\}_{i=1}^M \right) \right),
\end{align}
where $\chi$ is the Holevo information of an ensemble of quantum states.
\end{scorollary}

We can also use a similar strategy as in the proof of Theorem~\ref{sthm:hyp_test_lower_bound} to obtain a multi-hypothesis testing lower bound involving the fidelity. Let us extend the fidelity to mixtures of states as follows:
\begin{align}
    F(\{ \mu_i \rho_i\}_{i=1}^M, \{\nu_i \sigma_i \}_{i=1}^M) \coloneqq F\left(\bigoplus_{i=1}^M \mu_i \rho_i, \bigoplus_{i=1}^M \nu_i \sigma_i \right) = \sum_{i=1}^M \sqrt{\mu_i \nu_i} F(\rho_i, \sigma_i). 
\end{align}
We can then establish the following lower bound:
\begin{stheorem}[Multi-hypothesis testing lower bound]\label{sthm:multi_hyp_test_lower_bound_continuous}
Let $\{ \rho_i \}_{i=1}^M$ and $\sigma$ be quantum states and $\{ \mu_i \}_{i=1}^M$ be prior probabilities. We define
\begin{align}
    F_{\mathrm{avg}} &\coloneqq F(\{ \mu_i \rho_i\}_{i=1}^M, \{\sigma/M \}_{i=1}^M)
\end{align}
and obtain the lower bound
\begin{align}
 P_e^{*}(\{ \mu_i \rho_i \}) &\geq \frac{M}{M-1} \left(F_{\mathrm{avg}} - \sqrt{\frac{1}{M}} \right)^2 \geq \left(F_{\mathrm{avg}} - \sqrt{\frac{1}{M}} \right)^2. 
\end{align}
The bound is valid in the regime where $F_{\mathrm{avg}} \geq \sqrt{1/M}$.
\end{stheorem}
\begin{proof}
We employ a strategy similar to the proof of Theorem~\ref{sthm:hyp_test_lower_bound}. To this end, we denote with $A_i = \mu_i \rho_i$ and $X = \frac{1}{M}\sigma$. We can write the optimal hypothesis testing success and error probabilities as
\begin{align}
    P_s^{*}(\{ A_i\}) &= \sum_{i=1}^M \Tr[ A_i \Pi_i ] \\
    P_e^{*}(\{ A_i\}) &= \sum_{i=1}^M \Tr[ A_i (\bbI - \Pi_i) ] 
\end{align}
where $\{ \Pi_i\}_{i=1}^M$ are the optimal POVM effects.
Then, we define a CPTP map
\begin{align}
    \Lambda\left( \bigoplus_{i=1}^M Y_i \right) = \Tr\left[ \left( \bigoplus_{i=1}^M Y_i\right) \left( \bigoplus_{i=1}^M \Pi_i \right) \right] \oplus \Tr\left[ \left( \bigoplus_{i=1}^M Y_i\right) \left( \bigoplus_{i=1}^M \bbI - \Pi_i \right) \right]
\end{align}
such that
\begin{align}
    \Lambda\left( \bigoplus_{i=1}^M A_i \right) &= P_s^{*}(\{ A_i\}) \oplus P_e^{*}(\{ A_i\}) \\
    \Lambda\left( X^{\oplus M} \right) &= \Tr[X] \oplus (M-1) \Tr[X] = \frac{1}{M} \oplus \frac{M-1}{M}.
\end{align}
We can now again use the data-processing property of the fidelity. To this end, we observe that
\begin{align}
    F\left( \bigoplus_{i=1}^M A_i , X^{\oplus M} \right) &= F_{\mathrm{avg}}.
\end{align}
Now, applying the channel $\Lambda$ yields
\begin{align}
    F\left( \bigoplus_{i=1}^M A_i , X^{\oplus M} \right) &\leq F\left( \Lambda\left[\bigoplus_{i=1}^M A_i \right] , \Lambda\left[ X^{\oplus M}\right] \right) \\
    &= \sqrt{\frac{1}{M} P_s^{*}(\{ \mu_i \rho_i \})} +  \sqrt{\frac{M-1}{M} P_e^{*}(\{ \mu_i \rho_i \})}
    \nonumber
\end{align}
We now apply the bound $P_s^{*}(\{ \mu_i \rho_i \}) \leq 1$, rearrange and square to obtain the first Theorem statement. The validity range arises from the fact that the left hand side must remain non-negative for this to be sensible.
\end{proof}

\subsection{Further bounds}\label{ssec:further_succ_prob_bounds}

Similarly to our argument for Theorem~\ref{sthm:multi_hyp_test_lower_bound_continuous} above, which applied to quantum multi-hypothesis testing, we can do the same natively for quantum metrology. It involves the extension of the fidelity to parametrized quantum states defined as
\begin{align}
    F(\mu(t)\rho(t), \nu(t) \sigma(t)) &\coloneqq \int \diff t \, \sqrt{\mu(t)\nu(t)} F(\rho(t), \sigma(t)).
\end{align}
We obtain the following result:
\begin{stheorem}[Metrology error lower bound]\label{sthm:multi_hyp_test_lower_bound}
Let $\rho(t)$ a set of states with prior $\mu(t)$ supported on a compact interval of size $T$. Let $\sigma$ be a reference state. We define the random guessing probability $\kappa \coloneqq 2\delta/T$ and a measure of average fidelity as
\begin{align}
    F_{\mathrm{avg}} &\coloneqq F(\mu(t) \rho(t), \sigma/T).
\end{align}
With these definitions, the optimal Bayesian error probability can be lower bounded as
\begin{align}
1-\eta^{*}(\delta, \rho, \mu) &\geq \frac{1}{1-\kappa} \left(F_{\mathrm{avg}} - \sqrt{\kappa}\right)^2 \geq \left(F_{\mathrm{avg}} - \sqrt{\kappa}\right)^2.
\end{align}
The bound is valid in the regime where $F_{\mathrm{avg}} \geq \sqrt{\kappa}$. 
\end{stheorem}
\begin{proof}
We employ a strategy similar to the proof of Theorem~\ref{sthm:hyp_test_lower_bound}. 
The Bayesian success probability is given by
\begin{align}
\eta(\delta, \rho, \mu, Q) = \int \diff t \, \Tr[ \mu(t) \rho(t) (w_{\delta}*Q)(t)].
\end{align}
We can now see a set of states with prior as an operator-valued measure $\mu(t) \rho(t)$ with the trace map given as $\int \diff t \, \Tr[ \cdot ]$. 
In this sense, we can define a CPTP map
\begin{align}
    \Lambda\left( O(t) \right) = \left\{\int \diff t\, \Tr\left[ O(t) (w_{\delta} * Q)(t) \right] \right\} \oplus \left\{ \int \diff t\, \Tr\left[ O(t) ([1-w_{\delta}] * Q)(t) \right] \right\}
\end{align}
such that
\begin{align}
    \Lambda\left( \mu(t) \rho(t) \right) &= \eta \oplus 1-\eta \\
    \Lambda\left( \frac{\sigma}{T} \right) &= \frac{2\delta}{T} \oplus \frac{T-2\delta}{T} = \kappa \oplus 1 - \kappa,
\end{align}
where we introduced the random guessing probability $\kappa \coloneqq 2\delta/T$.
Our results will be formulated in terms of sandwiched Rényi relative entropy of order $1/2$,  
\begin{align}
    \tilde{D}_{1/2}(\rho \fatpipe \sigma) 
    &= -\frac{1}{2} \log \tilde{Q}_{1/2}(\rho\fatpipe\sigma) \\
    &= -\frac{1}{2} \log \Tr[(\sigma^{1/2} \rho \sigma^{1/2})^{1/2}] \\
    &= - \log F(\rho, \sigma),
\end{align}
essentially the log-fidelity.
Using the data-processing relation of the fidelity, we obtain a relation to the Bayesian success probability
\begin{align}
    F\left(\mu(t) \rho(t), \frac{\sigma}{T}\right) &\leq F\left(\eta \oplus 1-\eta , \kappa \oplus 1 - \kappa\right) \\
    \nonumber
    &= \sqrt{\eta \kappa } + \sqrt{(1-\eta) (1-\kappa)} \\
    \nonumber
    &\leq \sqrt{\kappa}+ \sqrt{(1-\eta) (1-\kappa)}. \nonumber
\end{align}
Combining the two preceding results allows us to deduce
\begin{align}
    1-\eta  &\geq \frac{1}{1-\kappa} \left(F_{\mathrm{avg}} - \sqrt{\kappa}\right)^2.
\end{align}
as long as 
\begin{align}
    F_{\mathrm{avg}} \geq \sqrt{\kappa}.
\end{align}
\end{proof}

Another possibility to get a metrology bound is to use the Fano-type bounds we derived before. Going along this path gives the following theorem:
\begin{stheorem}\label{sthm:metrology_lower_bound_from_fano}
Fix a sub-interval of size $T = 4 k \delta$ for $k \in \bbN$. We set $M = 2k+1$ and have that
\begin{align}
    &- \log 1 - \overline{\eta}^{*}(\delta, \rho) \leq \\
    &\max_t \frac{M}{M-1} \left(  h\left(\frac{1}{M}\right) - \log \frac{1}{M} \sum_{l=-k}^k e^{ - D( \rho(t + 2 \delta l) \fatpipe \rho(t))} \right).
    \nonumber
\end{align}
\end{stheorem}
\begin{proof}
We first use the subdivision trick of Lemma~\ref{lem:subdivision_trick} to establish the following lower bound relating to a subinterval $I = [t-T/2, t+T/2]$ of size $T= 4 k \delta$ centered around $t$:
\begin{align}
    1 - \overline{\eta}^{*}(\delta, \rho) \geq \max_t \left\{ 1 - \overline{\eta}^{*}(\delta, \rho|_{I})\right\},
\end{align}
where $\rho|_{I}$ is the restriction of $\rho(t)$ to $t \in I$. Next, we apply Theorem~\ref{thm:succ_prob_upper_bound_mht_delta_window} to the above, choosing a set of states $\{ \rho_l = \rho(t + 2\delta l) \}_{l=-k}^k$ that fulfill the condition that the associated times are at least $2\delta$ apart. This gives
\begin{align}
    1 - \overline{\eta}^{*}(\delta, \rho) \geq \max_t \overline{P}_e^{*}(\{ \rho_l = \rho(t + 2\delta l) \}_{l=-k}^k).
\end{align}
We then apply the above Theorem~\ref{sthm:bayesian_fano_type_bound} with the reference state $\sigma = \rho(t)$ to obtain the statement of the Theorem.
\end{proof}

\subsection{Upper bound on asymptotic rate}\label{ssec:upper_bound_asymptotic_rate_mht}
Theorem~\ref{sthm:succ_prob_upper_bound_mht} allows us to get bounds on the asymptotic rate for quantum metrology from the corresponding upper bounds for binary hypothesis testing. These bounds rely on the fact that the asymptotic rate for multi-hypothesis testing is given by the smallest quantum Chernoff divergence~\cite{audenaert2007discriminating}
\begin{align}
    C(\rho, \sigma) \coloneqq - \inf_{0 \leq s \leq 1} \log \Tr [\rho^s \sigma^{1-s}]
\end{align}
among two states that are to be tested~\cite{li2016discriminating}.
To be able to establish such bounds, we first need the bound of Theorem~\ref{sthm:succ_prob_upper_bound_mht} not to be vacuous which corresponds to enforcing $K=1$, as the success probability asymptotically approaches 1, but $K$ is independent of the number of copies of the state that are used. 
As a further ingredient to establish the bounds on the asymptotic rate, we need the following Lemma that generalizes Laplace's method in a way relevant to our work.
\begin{slemma}[Laplace principle]\label{slem:laplace_principle}
Let $\mu(t)$ be a probability measure on $\bbR$ and $f(t)$ a bounded measurable function. The essential infimum with respect to $\mu(t)$ is defined as
\begin{align}
    \operatornamewithlimits{ess \, inf}_{\mu(t)} f(t) \coloneqq \sup \, \{ b \in \bbR \, | \, \mu(\{ f(t) < b \}) = 0 \}.
\end{align}
We have that
\begin{align}
    -\lim_{n\to \infty} \frac{1}{n} \log \int \diff \mu(t) \, e^{-n f(t)} = \operatornamewithlimits{ess \, inf}_{\mu(t)} f(t).
\end{align}
\end{slemma}
\begin{proof}
    We assume without loss of generality that 
\begin{align}
    f^*=\operatornamewithlimits{ess \, inf}_{\mu(t)} f(t) = \sup \, \{ b \in \bbR \, | \, \mu(\{ f(t) < b \}) = 0 \}=0,
\end{align}
as we can otherwise consider the function $f-f^*$. First, note that
\begin{align}
 \frac{1}{n} \log \int \diff \mu(t) \, e^{-n f(t)} \leq 0,
\end{align}
as $e^{-nf(t)}\leq 1$ almost everywhere. Thus, 
\begin{align}\label{equ:limsup}
    \limsup\limits_{n\to\infty}\frac{1}{n} \log \int \diff \mu(t) \, e^{-n f(t)} \leq 0.
\end{align}
By the definition of the essential minimum, for every $\epsilon>0$ we have that there exists a $\delta>0$ such that
\begin{align}
\mu(\{f(t)<\epsilon\})\geq \delta.
\end{align}
From this, we obtain that for all $n$
\begin{align}
    \int \diff \mu(t) \, e^{-n f(t)}\geq e^{-n\epsilon}\delta.
\end{align}
Taking the $\log$ and dividing by $n$ we obtain:
\begin{align}
    \frac{1}{n}\log\left(\int \diff \mu(t) \, e^{-n f(t)}\right)\geq -\epsilon+\frac{\log(\delta)}{n}.
\end{align}
Taking the $\liminf$ of both sides we obtain that
\begin{align}
    \liminf\limits_{n\to\infty}\frac{1}{n}\log\left(\int \diff \mu(t) \, e^{-n f(t)}\right)\geq -\epsilon.
\end{align}
As $\epsilon>0$ was arbitrary, we conclude that 
\begin{align}\label{equ:liminf}
    \liminf\limits_{n\to\infty}\frac{1}{n}\log\left(\int \diff \mu(t) \, e^{-n f(t)}\right)\geq 0.
\end{align}
The claim then follows by combining 
Eq.~\eqref{equ:liminf} with Eq.~\eqref{equ:limsup}.
\end{proof}

We are now equipped to show that Theorem~\ref{sthm:succ_prob_upper_bound_mht} implies the following bound on the asymptotic rate.
\begin{stheorem}[Upper bound on asymptotic rate]\label{sthm:rate_upper_bound}
For a given window function $w$, state set $\rho(t)$ and possibly a prior $\mu(t)$ that has non-vanishing support on the parameter domain,  the Bayesian and the minimax rate obey the upper bounds
\begin{align}
    \overline{R}^{*}(w, \rho) \leq  R^{*}(w, \mu, \rho) \leq \inf \left\{\left. C(\rho(t), \rho(t')) \vphantom{\sup_x} \, \right| \, t, t' \colon \sup_x \{ w(t + x) + w(t' + x) \} = 1\right\}.
\end{align}
\end{stheorem}
\begin{proof}
Let us first recall the definition of the asymptotic Bayesian rate
\begin{align}
    R^{*}(w, \mu, \rho) &= \lim_{n \to \infty}  -\frac{1}{n} \log \left( 1 -\eta^{*}(w, \mu, \rho^{\otimes n}) \right).
\end{align}
We will apply Theorem~\ref{sthm:succ_prob_upper_bound_mht} for a set $\calS = \{ (1/2, 0), (1/2, s)\}$ where $s$ is chosen such that the constant $K$ is equal to 1, \textit{i.e.}, such that the shifted windows do not overlap. In the case of binary discrimination, we can make sue of the Helstrom formula
\begin{align}
P^{*}_s(A, B) = \frac{1}{2}\Tr[A+B] + \frac{1}{2}\lVert A - B \rVert_1.
\end{align}
The Theorem then implies the following lower bound on the error:
\begin{align}
    1 -\eta^{*}(w, \mu, \rho^{\otimes n}) &\geq 1 - \int \diff t \, P^{*}_s\left( \frac{1}{2}\mu(t) \rho^{\otimes n}(t),  \frac{1}{2} \mu(t+s) \rho^{\otimes n}(t+s) \right)\\
    \nonumber
    &=1 - \frac{1}{2}\left(\int \diff t \, \Tr\left[ \frac{1}{2}\mu(t) \rho^{\otimes n}(t) + \frac{1}{2}\mu(t + s) \rho^{\otimes n}(t+s) \right] + \left\lVert \frac{1}{2}\mu(t) \rho^{\otimes n}(t) + \frac{1}{2}\mu(t + s) \rho^{\otimes n}(t+s) \right\rVert_1\right)\\
     \nonumber
    &= \frac{1}{2}\left(\int \diff t \, \Tr\left[ \frac{1}{2}\mu(t) \rho^{\otimes n}(t) + \frac{1}{2}\mu(t + s) \rho^{\otimes n}(t+s) \right] - \left\lVert \frac{1}{2}\mu(t) \rho^{\otimes n}(t) + \frac{1}{2}\mu(t + s) \rho^{\otimes n}(t+s) \right\rVert_1\right)\\
     \nonumber
    &= \int \diff t \, P^{*}_e\left( \frac{1}{2}\mu(t) \rho^{\otimes n}(t),  \frac{1}{2} \mu(t+s) \rho^{\otimes n}(t+s) \right),
\end{align}
where we have used the fact that
\begin{align}
    \int \diff t \, \Tr\left[ \frac{1}{2}\mu(t) \rho^{\otimes n}(t) + \frac{1}{2}\mu(t + s) \rho^{\otimes n}(t+s) \right] = 1.
\end{align}
We can now make use of the fact that the asymptotic scaling of the error probability for binary hypothesis tests is known~\cite{nussbaum2009chernoff}. We only have to take care of the additional measures $\mu(t)$ and $\mu(t+s)$ that appear in the expression we look at. 
For the sense of brevity, we do not reproduce the whole proof of Theorem 2.2 of Ref.~\cite{nussbaum2009chernoff} but we just note that using the relation $\min\{ a b , cd\} \geq \min\{ a,c\}\min\{b,d\}$ for non-negative real numbers $a,b,c,d$ after Eq.~(12) of Ref.~\cite{nussbaum2009chernoff} implies that
\begin{align}
    P_e^{*}\left( \frac{1}{2}\mu(t)\rho^{\otimes n}(t),  \frac{1}{2}\mu(t+s)\rho^{\otimes n}(t+s) \right) \geq \frac{1}{2}\min(\mu(t), \mu(t+s)) \exp\left(-n (C(\rho(t), \rho(t+s)) + o(1) )\right).
\end{align}
Therefore,
\begin{align}
    1 -\eta^{*}(w, \mu, \rho^{\otimes n}) &\geq \frac{1}{2}\int \diff t \, \min\{\mu(t), \mu(t+s)\}\exp\left(-n (C(\rho(t), \rho(t+s)) + o(1) )\right).
\end{align}
For the next step, we define $\nu(t) = \min\{\mu(t), \mu(t+s)\}$ and have that $\int \diff t \, \nu(t) = c > 0$ by the assumptions of the theorem statement. With this at hand, we have that $\nu(t) / c$ is a proper measure so that
\begin{align}
    \lim_{n\to\infty} -\frac{1}{n}\log  (1 -\eta^{*}(w, \mu, \rho^{\otimes n})) \leq \lim_{n\to\infty} -\frac{1}{n}\log  \frac{c}{2}\int \frac{\diff \nu(t)}{c} \, \exp\left(-n C(\rho(t), \rho(t+s)\right)\exp(-o(n)) .
\end{align}
We can asymptotically take care of the $o(n)$ term by adding $-n\epsilon$ with an arbitrary $\epsilon > 0$ to the exponent and then apply Lemma~\ref{slem:laplace_principle} to obtain
\begin{align}
    \lim_{n\to\infty} -\frac{1}{n}\log  (1 -\eta^{*}(w, \mu, \rho^{\otimes n})) 
    &\leq \lim_{\epsilon \to 0}\lim_{n\to\infty} -\frac{1}{n}\log \int \frac{\diff \nu(t)}{c} \, \exp\left(-n [C(\rho(t), \rho(t+s)+\epsilon]\right) \\
     \nonumber
    &=  \lim_{\epsilon \to 0} \operatornamewithlimits{ess \, \inf}_{\nu(t)} C(\rho(t), \rho(t+s) +\epsilon\\
     \nonumber
    &= \operatornamewithlimits{ess \, \inf}_{\nu(t)} C(\rho(t), \rho(t+s).
\end{align}
In the above statement, the prior enters only through defining the support of $\nu(t)$, which by virtue of the assumptions of the theorem is the full support of the domain of the metrological problem which we left implicit. This means that asymptotic rate bound applies both to the Bayesian and the minimax case. The theorem statement follows by writing $s = t-t'$ and the condition $\sup_x \{ w(t + s) + w(t'+s) \} = 1$ enforces $K=1$ as desired.
\end{proof}
The statement of Theorem~\ref{thm:rate_upper_bound_delta} in the main text follows straightforwardly by noting that the condition on $t$ and $t'$ reduces to $|t-t'| > 2\delta$ for the rectangular window with tolerance $\delta$.

\subsection{Relation to Wigner-Yanase-Dyson information}\label{ssec:wigner_yanase_dyson_information}
Under continuity assumptions, the most similar states -- and thus the states that are hardest to distinguish from each other -- are the ones that are close in time. In the case of a rectangular window with very small tolerance $\delta$ we thus expect the rate to be limited by the states that are $2\delta$-close. In this case, we can expand the Chernoff divergence as~\cite{jarzyna2020geometric}
\begin{align}
    C(\rho(t), \rho(t + 2\delta)) = \frac{1}{2}\delta^2 \calI(t) + O(\delta^3),
\end{align}
where $\calI(t)$ is the \emph{Wigner-Yanase-Dyson} information. We can alternatively define it through the \emph{affinity} of quantum states
\begin{align}
    A(\rho, \sigma) = \Tr [\rho^{1/2} \sigma^{1/2}],
\end{align}
such that 
\begin{align}
    A(\rho(t), \rho(t+2\delta)) = 1-\frac{1}{4} \delta^2 \calI(t)+ O(\delta^3).
\end{align}
Compare this to the Bures fidelity, which expands into the quantum Fisher information $\calF(t)$ as
\begin{align}
    F(\rho(t), \rho(t+2\delta)) = \lVert \rho(t)^{1/2} \rho(t+2\delta)^{1/2}\rVert_1 = 1 - \frac{1}{2} \delta^2 \calF(t) + O(\delta^3).
\end{align}
For small $\delta$, we thus expect the asymptotic rate to be bounded by
\begin{align}
    \overline{R}^{*}(w_{\delta}, \rho) \leq \frac{1}{2}\delta^2 \min_t \calI(t).
\end{align}
As the Wigner-Yanase-Dyson information is always smaller than the \emph{Bogoliubov-Kubo-Mori} information associated to the quantum relative entropy, this is tighter than the characterization given by Hayashi in Eq.~(71) of Ref.~\cite{hayashi2002two}.

\subsection{Lower bound on success probability}

In this section, we will establish a lower bound on the success probability. Such lower bounds are obtained by exhibiting a POVM $Q(\tau)$ for which we can guarantee a certain performance. In our case, we will assume a suitable measurement $\{ M(\lambda) \}$ has already been chosen and we then compute guarantees for the smoothed maximum a posteriori strategy introduced in Section~\ref{ssec:maximum_weight_estimation}. Our bounds allow us to reduce the problem to the hardest binary hypothesis testing problem for two states $\rho(t)$ and $\rho(t')$ that can never be in the same window. We use our bounds to establish a lower bound on the asymptotic rate for the $\delta$ window and show that it matches the upper bound of Theorem~\ref{sthm:rate_upper_bound} for commuting states. We conjecture that a suitable lower bound can be derived in the general case as well and sketch a possible direction to do so.

To establish a lower bound on the success probability -- or equivalently an upper bound on the error -- we will discretize the problem in question, establish the bound in the discrete case and then lift the result again to the continuous case by taking the appropriate limits. To this end, we will first establish the discrete version of the success probability:
\begin{align}
    \eta^{*}(W, \{ \rho_i \}, \{ p_i \}) \coloneqq \sup_{\{ Q_j \}} \sum_{i} \sum_j W_{i,j} p_i \Tr[ \rho_i Q_j ].
\end{align}
The first argument is a matrix $W$ that takes the role of the window function, whereas the other arguments are a collection of quantum states and a collection of prior probabilities. With this notation settled, we can establish the following lemma that establishes a lower bound through a specific discretization:
\begin{slemma}[Lower bound through discretization]\label{slem:discretization_lower_bound}
For a given window function $w$, state set $\rho(t)$ and prior $\mu(t)$, we define a discretization with respect to a set of mutually disjoint intervals $\{ T_i \}_{i=1}^N$ as
\begin{align}
    p_i \coloneqq \mu(T_i), \qquad
    \rho_i \coloneqq \frac{1}{p_i} \int_{T_i} \diff \mu(t) \, \rho(t), \qquad
    W_{i,j} \coloneqq \inf_{t_i \in T_i, t_j \in T_j} w(t-t').
\end{align}
Then, we have that
\begin{align}
    \eta^{*}(w, \rho, \mu) \geq \eta^{*}(W, \{ \rho_i \}_{i=1}^N, \{ p_i \}_{i=1}^N).
\end{align}
\end{slemma}
\begin{proof}
As a first step, we can obtain a lower bound on the success probability by restricting the optimization over all POVMs $Q(\tau)$ to discrete POVMs $Q(t) = \sum_{j=1}^N \delta(t - t_j) Q_j$ associated to predictions $t_j \in T_j$. This yields
\begin{align}
    \eta^{*}(\delta, \rho, \mu) &= \sup_{Q(\tau)} \int \diff \mu(t) \, \diff \tau \, w(t-\tau) \Tr[\rho(t) Q(\tau)]  \\
    &\geq \sup_{ \{ Q_j \} } \sum_j \int \diff \mu(t) \, w(t-t_j) \Tr[\rho(t) Q_j].\nonumber
\end{align}
Next, we split the integration over $t$ into integrals over the intervals $T_i$ and exploit the definition of $W_{i,j}$ as a lower bound. The statement of the lemma follows from collecting the resulting terms
\begin{align}
    \eta^{*}(\delta, \rho, \mu) &\geq \sup_{ \{ Q_j \} }  \sum_{i=1}^N \sum_j \int_{T_i} \diff \mu(t) \, w(t-t_j) \Tr[\rho(t) Q_j] \\
    \nonumber
    &\geq \sup_{ \{ Q_j \} }  \sum_{i=1}^N \sum_{j=1}^N\int_{T_i} \diff \mu(t) \, W_{i,j} \Tr[\rho(t) Q_j] \\
     \nonumber
    &= \sup_{ \{ Q_j \} }  \sum_{i=1}^N \sum_{j=1}^N p_i W_{i,j} \Tr[\rho_i Q_j] \\
     \nonumber
    &= \eta^{*}(w, \{ \rho_i \}_{i=1}^N, \{ p_i \}_{i=1}^N).
     \nonumber
\end{align}
\end{proof}

Our next step is to provide a bound for the discrete problem that achieves a reduction to binary hypothesis testing. We will now define the discrete analogues of the smoothed maximum a posteriori estimation strategy of Section~\ref{ssec:maximum_weight_estimation}. For a given measurement $\{ M(\lambda) \}$, we define the \emph{discrete likelihood function}
\begin{align}
    \Lambda(\lambda \pipe i) \coloneqq \Tr[ \rho_i M(\lambda)].
\end{align}
Analogously to Section~\ref{sec:opt_fixed_measurement}, we also define the marginal probability of observing $\lambda$ as
\begin{align}
    \nu(\lambda) = \sum_{i=1}^N p_i \Lambda_i(\lambda)
\end{align}
such that the \emph{discrete posterior distribution} is given by
\begin{align}
    P(i \pipe \lambda) = \frac{ p_i \Lambda(\lambda \pipe i) }{ \nu(\lambda)}.
\end{align}
The discrete likelihood function corresponds to a state
\begin{align}\label{eqn:def_likelihood_state}
    \Lambda( \lambda \pipe i) \Leftrightarrow \Lambda(i) \coloneqq \int \diff \lambda \, |\lambda \rangle\!\langle \lambda | \, \Lambda(\lambda \pipe i)
    = \frac{1}{p_i} \int \diff \nu(\lambda) \, |\lambda \rangle\!\langle \lambda | \, P(i \pipe \lambda)
\end{align}
that captures the conditional distribution of measurement outcomes conditioned on the underlying state being $\rho_i$.

A given strategy $\tau^{*}(\lambda)$ induces a POVM
\begin{align}
    Q_j = \int \diff \lambda \, M(\lambda)\chi[\tau^{*}(\lambda) = j].
\end{align}
The discrete success probability for that strategy is then
\begin{align}
    \eta(w, \{ \rho_i \}_{i=1}^N, \{ p_i \}_{i=1}^N, \{ Q_j \}_{j=1}^N) &= \int \diff \lambda \, \sum_{i=1}^N \sum_{j=1}^N p_i W_{i,j} \Tr[\rho_i M(\lambda)] \chi[\tau^{*}(\lambda) = j] \\
    &= \int \diff \lambda \, \sum_{i=1}^N \sum_{j=1}^N W_{i,j} p_i \Lambda(\lambda \pipe i) \chi[\tau^{*}(\lambda) = j] \nonumber\\
    &= \int \diff \lambda \, \sum_{i=1}^N \sum_{j=1}^N W_{i,j} \nu(\lambda) P(i \pipe \lambda) \chi[\tau^{*}(\lambda) = j].\nonumber
\end{align}
The smoothed maximum a posteriori estimation strategy corresponds to
\begin{align}
    \tau^{*}_{\mathrm{SMAP}}(\lambda) = \argmax_{1 \leq i \leq N} \sum_{j=1}^N W_{i,j} P(j \pipe \lambda).
\end{align}
The smoothed maximum a posteriori estimation strategy allows us to derive the following proposition that relates the error to binary hypothesis testing of the output distributions of a fixed measurement for any window which takes only values in 0 and 1, like the $\delta$ window.
\begin{sproposition}[Error bound in the discrete case]\label{sprop:discrete_error_bound}
For a given discrete set of states $\{ \rho_i \}_{i=1}^N$ with prior probabilities $\{ p_i \}_{i=1}^N$ and a window matrix $W$ with entries that are either $0$ or $1$, we have that the posterior states $\{ \Upsilon_i \}_{i=1}^N$ provide the upper bound
\begin{align}
    \sum_{i=1}^N p_i - \eta^{*}(W, \{ \rho_i \}_{i=1}^N, \{ p_i \}_{i=1}^N) &\leq N^2 \max_{\substack{1 \leq i,j \leq N \\ \not\exists k\colon W_{i,k} = W_{j,k} = 1 }} P_e(p_i \Lambda(i), p_j \Lambda(j)).
\end{align}
The condition on the indices $i$ and $j$ means we optimize over all states that cannot be in the same window at the same time.
\end{sproposition}
\begin{proof}
By the construction of our measurement, we know that for all $i$
\begin{align}
    \sum_{j=1}^N W_{i,j} P(j \pipe \lambda) \leq \sum_{j=1}^N W_{\tau^{*}(\lambda) j} P(j \pipe \lambda).
\end{align}
We can reformulate this by introducing a \emph{neighborhood} of an index $i$, $K(i) \coloneqq \{ j \pipe W_{i,j} = 1 \}$, as
\begin{align}
    \sum_{j \in K(i)} P(j \pipe \lambda) \leq \sum_{j \in K(\tau^{*}(\lambda)) } P(j \pipe \lambda),
\end{align}
where we implicitly made use of the fact that $W_{i,j}$ can be only either zero or one by assumption. We will denote the complement of a neighborhood as $\bar{K} = [N]\, \backslash \, K$. Note that the total error of our construction will be
\begin{align}
    \sum_{i=1}^N p_i - \eta(W, \{ \rho_i \}_{i=1}^N, \{ p_i \}_{i=1}^N, \{ Q_j \}) = \int \diff \nu(\lambda) \, \sum_{i \in \bar{K}(\tau^{*}(\lambda)} P(i \pipe \lambda).  
\end{align}
The first key observation we can draw from the above is that for any neighborhood $K(k)$ that contains $i$, we have that
\begin{align}
    P(i \pipe \lambda) \leq \sum_{j \in K(\tau^{*}(\lambda)) }^N P(j \pipe \lambda) - \sum_{l \in K(k) }^N P(l \pipe \lambda). 
\end{align}
This upper bound will be non-trivial for any $i \in \Bar{K}(\tau^{*}(\lambda))$. We can remove the overlapping terms in the upper bound to obtain the following inequality which holds for all $i \in \Bar{K}(\tau^{*}(\lambda))$ and all $k$ such that for $i \in K(k)$,
\begin{align}
    P(i \pipe \lambda) \leq \sum_{j \in K(\tau^{*}(\lambda) \backslash K(k)} P(j \pipe \lambda).
\end{align}
Hence, we can conclude that if we define the union of all neighborhoods that contain $i$ as $U(i) \coloneqq \bigcup \{ K(k) \pipe i \in K(k)\}$, we have that
\begin{align}
    P(i \pipe \lambda) \leq \sum_{j \not\in U(i) } P(j \pipe \lambda).
\end{align}
The next crucial step in our derivation is to observe that the above immediately implies 
\begin{align}\label{eqn:discrete_pairwise_bound_maximum_weight}
    P(i \pipe \lambda) \leq \sum_{j \not\in U(i) } \min\{ P(i \pipe \lambda), P(j \pipe \lambda)\}
\end{align}
for all likelihood values that contribute to the error of the smoothed maximum a posteriori estimate.
We can thus bound the total error as
\begin{align}
    \sum_{i=1}^N p_i - \eta(W, \{ \rho_i \}_{i=1}^N, \{ p_i \}_{i=1}^N, \{ Q_j \}) &= \int\diff \nu(\lambda) \, \sum_{i \in \bar{K}(\tau^{*}(\lambda))} P(i \pipe \lambda) \\
    \nonumber
    &\leq \int\diff \nu(\lambda) \, \sum_{i \in \bar{K}(\tau^{*}(\lambda))}\sum_{j \not\in U(i) } \min\{ P(i \pipe \lambda), P(j \pipe \lambda)\} \\
    \nonumber
    &\leq \int\diff \nu(\lambda) \, \sum_{i=1}^N \sum_{j \not\in U(i) } \min\{ P(i \pipe \lambda), P(j \pipe \lambda)\} \\
    \nonumber
    &=  \sum_{i=1}^N \sum_{j \not\in U(i) } \int \diff \nu(\lambda)  \, \min\{ P(i \pipe \lambda), P(j \pipe \lambda)\} \\
    \nonumber
    &= \sum_{i=1}^N \sum_{j \not\in U(i) } P_e(p_i \Lambda(i), p_j \Lambda(j)) \\
    &\leq N^2 \max_{\substack{1 \leq i, j \leq N\\ j \not\in U(i)}} P_e(p_i \Lambda(i), p_j \Lambda(j)).
    \nonumber
\end{align}
The first equality is the error under the smoothed maximum a posteriori strategy, the first inequality is Eq.~\eqref{eqn:discrete_pairwise_bound_maximum_weight}, the second inequality -- a further crucial step -- extends the summation over $i$ to include all possible indices. The second equality exchanges the order of integration and summation and the third equality compares to Eq.~\eqref{eqn:def_likelihood_state} and recognizes the term $\int \diff \nu(\lambda) \,\min\{ P(i \pipe \lambda), P(j \pipe \lambda)\}$ as the minimum attainable error when discriminating the classical states $\Lambda(i)$ and $\Lambda(j)$ with prior probabilities $p_i$ and $p_j$. The final inequality bounds the sum via the maximum.
The statement of the proposition follows by recognizing that $j \not\in U(i)$ is equivalent to $\not\exists k \colon W_{i,k} = W_{jl} = 1$.
\end{proof}

Let us now turn to the $\delta$ window and lift this proposition to a bound in the continuous case. To this end, recall the following definitions: We have the likelihood function
\begin{align}
    \Lambda(\lambda \pipe t) &\coloneqq \Tr[ \rho(t) M(\lambda) ]
\end{align}
capturing the conditional distribution of measurement outcomes $\lambda$ for a given ground truth $t$, the associated likelihood state
\begin{align}
    \Lambda(t) \coloneqq \int \diff \lambda \, |\lambda \rangle\!\langle \lambda | \,   \Lambda(\lambda \pipe t),
\end{align}
and the marginal distribution of measurement outcomes and the posterior distribution:
\begin{align}
    \nu(\lambda) &\coloneqq \int \diff \mu(t) \, \Tr[ \rho(t) M(\lambda) ] \\
    P(t \pipe \lambda) &\coloneqq \frac{\mu(t)}{\nu(\lambda)} \Lambda(\lambda \pipe t).
\end{align}
With these definitions at hand, we can establish the following proposition pertaining to a rectangular window function:
\begin{sproposition}\label{sprop:continuous_upper_bound_from_discrete_case}
Assume a rectangular window with fixed tolerance $\delta$, state set $\rho(t)$ with prior $\mu(t)$ and a fixed measurement $M(\lambda)$.
We can then take any given compact interval $T$ and discretize it into $N \in \bbN$ equally sized sub-intervals of size $\Delta = |T|/N$. We define a smoothed state as
\begin{align}
\Upsilon^{\Delta}(t) &\coloneqq \frac{1}{\Delta}\int_{-\Delta/2}^{\Delta /2} \diff \tau \, \mu(t+\tau) \Lambda(t+\tau). 
\end{align}
Then, we have the bound
\begin{align}
    \mu(T) - \eta^{*}(\delta, \rho, \mu) \leq |T| N \max_{\substack{t,t' \in T\\|t-t'|>2(\delta - \Delta) } } P_e(\Upsilon^{\Delta}(t), \Upsilon^{\Delta}(t')).
\end{align}
\end{sproposition}
\begin{proof}
Let us use the discretization introduced in Lemma~\ref{slem:discretization_lower_bound} for a set of mutually disjoint intervals $T_i$ of cardinality bounded as $|T_i| \leq \Delta$ such that $T = \bigcup_{i=1}^N T_i$ where $N = \lceil T/\Delta \rceil$. We combine the discretization lower bound of Lemma~\ref{slem:discretization_lower_bound} with Proposition~\ref{sprop:discrete_error_bound}. Recall that the discretization involves the definitions
\begin{align}
    p_i &= \mu(T_i) \\
    \rho_i &= \frac{1}{p_i} \int_{T_i} \diff \mu(t) \, \rho(t).
\end{align}
Subsequent application of Proposition~\ref{sprop:discrete_error_bound} then gives the bound
\begin{align}
    \mu(T) - \eta^{*}(\delta, \rho, \mu) \leq N^2 \max_{\substack{1 \leq i,j \leq N \\ \not\exists k\colon W_{i,k} = W_{j,k} = 1 }} P_e(p_i \Lambda(i), p_j \Lambda(j)),
\end{align}
where we have used $\sum_{i=1}^N p_i = \mu(T)$ and, for $t_i$ the midpoint of $T_i$, we have that
\begin{align}
    p_i \Lambda(i) &= p_i \int \diff \lambda \, |\lambda \rangle\!\langle \lambda | \, \Tr[\rho_i M(\lambda)] \\
    \nonumber
    &= p_i \int \diff \lambda \, |\lambda \rangle\!\langle \lambda | \, \Tr\left[\frac{1}{p_i} \int_{T_i} \diff \mu(t) \, \rho(t) M(\lambda)\right] 
    \nonumber\\
    &=  \int \diff \lambda \, |\lambda \rangle\!\langle \lambda | \,\int_{T_i} \diff \mu(t) \, \Lambda(\lambda \pipe t) 
    \nonumber\\ 
    &=  \int_{T_i} \diff \mu(t) \, \Lambda(t) 
    \nonumber\\ 
    &= \Delta \Upsilon^{\Delta}(t) .
    \nonumber
\end{align}
Hence,
\begin{align}
    \mu(T) - \eta^{*}(\delta, \rho, \mu) &\leq \Delta N^2 \max_{\substack{1 \leq i,j \leq N \\ \not\exists k\colon W_{i,k} = W_{j,k} = 1 }} P_e(\Upsilon^{\Delta}(t_i), \Upsilon^{\Delta}(t_j)).
\end{align}
To bring this inequality to its final form, we revisit the definition of $W_{i,k}$ for the $\delta$ window
\begin{align}
    W_{i,k} &= \inf_{\substack{t_i' \in T_i \\ t_k' \in T_k}} w(t_i' - t_k') \\
    \nonumber
    &= \inf_{\substack{t_i' \in T_i \\ t_k' \in T_k}} \chi[|t_i' - t_k'| \leq \delta] \\
    \nonumber
    &= \inf_{-\Delta/2 \leq \epsilon, \epsilon' \leq \Delta/2 } \chi[|t_i - t_k + \epsilon - \epsilon'| \leq \delta] \\
    \nonumber
    &= \chi[|t_i - t_k| \leq \delta - \Delta],
    \nonumber
\end{align}
where we recall that $t_i$ and $t_k$ are the midpoints of $T_i$ and $T_k$, respectively. The condition that there should not be a $k$ such that $W_{i,k}$ and $W_{j,k}$ are one at the same time thus translates to
\begin{align}
    \not\exists k \colon W_{i,k} = W_{j,k}  = 1 \ \Leftrightarrow \ |t_i - t_j| > 2 (\delta - \Delta).
\end{align}
We, therefore, have
\begin{align}
    \mu(T) - \eta^{*}(\delta, \rho, \mu) &\leq \Delta N^2 \max_{\substack{1 \leq i,j \leq N \\ 
    \nonumber|t_i-t_j| > 2(\delta - \Delta) }}  P_e(\Upsilon^{\Delta}(t_i), \Upsilon^{\Delta}(t_j)) \\
    \nonumber
    &\leq \Delta N^2 \sup_{\substack{t, t' \in T \\ |t-t'| > 2(\delta - \Delta) }}  P_e(\Upsilon^{\Delta}(t), \Upsilon^{\Delta}(t')),
    \nonumber
\end{align}
where the second inequality follows from the fact that we can always extend the optimization in the maximum to also include points that are not the midpoints of the discretization intervals. The statement of the proposition follows from $\Delta = |T|/N$.
\end{proof}

We will now use the above proposition to get a lower bound on the asymptotic rate for the $\delta$ window. We consider the case of a defined sequence of measurements $\{ M^{(n)} \}$ for $n \in \bbN$. We denote the channel that maps states to their output distributions over $\lambda$ as
\begin{align}
    \calM^{(n)}[\rho] = \int \diff \lambda \, |\lambda \rangle\!\langle \lambda | \, \Tr[ \rho M^{(n)}(\lambda)].
\end{align}
This sequence achieves the rate 
\begin{align}
    R^{*}(\rho, \sigma, \{ M^{(n)}(\lambda) \})  \coloneqq - \lim_{n\to \infty} \frac{1}{n} \log P^{*}_e(\calM^{(n)}[\rho^{\otimes n}], \calM^{(n)}[\sigma^{\otimes n}])
\end{align}
for binary state discrimination.
With this notation in place, we can now proceed with the proof of Theorem~\ref{thm:rate_lower_bound_delta_window} of the main text:
\begin{proof}[Proof of Theorem~\ref{thm:rate_lower_bound_delta_window}]
We will use Proposition~\ref{sprop:continuous_upper_bound_from_discrete_case} and choose an interval $T = [-T_0, T_0]$, and a number of discretization intervals $N \in \bbN$ such that $\Delta = 2 T_0 / N < \delta$. We actually have to choose $\Delta$ much smaller than that as will become apparent later.
In our case, the likelihood state is given by
\begin{align}
   \Lambda^{(n)}(t) =  \calM^{(n)}[\rho^{\otimes n}(t)]
\end{align}
and its smoothed counterpart by
\begin{align}
    \Upsilon^{\Delta,(n)}(t) = \frac{1}{\Delta} \int_{-\Delta/2}^{\Delta/2} \diff \tau \, \mu(t+\tau) \Lambda^{(n)}(t + \tau).
\end{align}
The upper bound on the error obtained from Proposition~\ref{sprop:continuous_upper_bound_from_discrete_case} then takes the form
\begin{align}
    \mu([-T_0, T_0]) - \eta^{*}(\delta, \rho^{\otimes n}, \mu) &\leq 2 T_0 N \sup_{\substack{-T_0 \leq t, t' \leq T_0 \\ |t-t'| > 2(\delta-\Delta)}} P_e(\Upsilon^{\Delta,(n)}(t), \Upsilon^{\Delta,(n)}(t')).\label{eqn:upper_bound_from_prop}
\end{align}
We will introduce the smoothed state set and measure as
\begin{align}
    \mu^{\Delta}(t) &\coloneqq \frac{1}{\Delta} \int_{-\Delta/2}^{\Delta/2} \diff \tau \, \mu(t+\tau) = \frac{1}{\Delta}\mu([t - \Delta/2, t + \Delta/2]), \\
    \rho^{\Delta,(n)} &= \frac{1}{\mu([t - \Delta/2, t + \Delta/2])} \int_{-\Delta/2}^{\Delta/2} \diff \tau \, \mu(t) \rho^{\otimes n}(t).\label{eqn:def_smoothed_state_set}
\end{align}
With these notions, we have that
\begin{align}
    \Upsilon^{\Delta,(n)}(t) &= \calM^{(n)}\left[ \frac{1}{\Delta} \int_{-\Delta/2}^{\Delta/2} \diff \tau \, \mu(t + \tau) \rho^{\otimes n}(t + \tau) \right] \\
    \nonumber
    &= \calM^{(n)}[ \mu^{\Delta}(t) \rho^{\Delta,(n)}(t)].
\end{align}
For technical reasons, we will need a full rank state in the following derivations. We therefore introduce the perturbed state set
\begin{align}
    \rho_{\gamma}(t) \coloneqq (1-\gamma) \rho(t) + \gamma \omega,
\end{align}
where $0 < \gamma < 1$ and $\omega$ is the maximally mixed state. We will use $\rho^{\Delta,(n)}_{\gamma}(t)$ to denote the associated smoothed state sets analogously defined as in Eq.~\eqref{eqn:def_smoothed_state_set}. We can use the upper bound
\begin{align}
    \rho(t) \leq \frac{1}{1-\gamma} \rho_{\gamma}(t) \ \Rightarrow \ \rho^{\otimes n}(t) \leq e^{n \log(1/(1-\gamma))} \rho_{\gamma}^{\otimes n}(t).
\end{align}
The state $\rho^{\Delta,(n)}$ appearing in the upper bound of Eq.~\eqref{eqn:upper_bound_from_prop} is a mixture of i.i.d.\ states, but for our purposes we need an i.i.d.\ state. We therefore use the chain of inequalities
\begin{align}
    \rho^{\Delta,(n)}(t) \leq e^{n \log(1/(1-\gamma))} \rho^{\Delta,(n)}_{\gamma}(t) \leq e^{n [ \log(1/(1-\gamma)) + D_{\max}^{T_0, \Delta, \gamma}]} [\rho^{\Delta}_{\gamma}(t)]^{\otimes n},
\end{align}
where we have introduced the quantity
\begin{align}
    D_{\max}^{T_0, \Delta, \gamma} \coloneqq \sup_{-T_0 \leq t \leq T_0} \sup_{-\Delta/2 \leq \tau \leq \Delta/2} D_{\max}(\rho_{\gamma}(t + \tau), \rho_{\gamma}^{\Delta}(t)).
\end{align}
We have thus reduced from the state $\rho^{\Delta,(n)}(t)$ to the state $[\rho^{\Delta}_{\gamma}(t)]^{\otimes n}$ at the cost of a correction to the asymptotic rate given by $\log(1/(1-\gamma)) + D_{\max}^{T_0, \Delta, \gamma}$, as becomes apparent from the following chain of inequalities:
\begin{align}
    & \mu([-T_0, T_0]) - \eta^{*}(\delta, \rho^{\otimes n}, \mu) \\
    \nonumber
    &\leq 2 T_0 N \sup_{\substack{-T_0 \leq t, t' \leq T_0 \\ |t-t'| > 2(\delta-\Delta)}} P_e(\calM^{(n)}[ \mu^{\Delta}(t) \rho^{\Delta,(n)}(t)], \calM^{(n)}[ \mu^{\Delta}(t') \rho^{\Delta,(n)}(t')])\\
    \nonumber
    &\leq 2 T_0 N \sup_{\substack{-T_0 \leq t, t' \leq T_0 \\ |t-t'| > 2(\delta-\Delta)}} e^{n [ \log(1/(1-\gamma)) + D_{\max}^{T_0, \Delta, \gamma}]} P_e(\calM^{(n)}[ \mu^{\Delta}(t) [\rho^{\Delta}_{\gamma}(t)]^{\otimes n}], \calM^{(n)}[ \mu^{\Delta}(t') [\rho^{\Delta}_{\gamma}(t')]^{\otimes n}])\\
    \nonumber
    &\leq 2 T_0 N \sup_{\substack{-T_0 \leq t, t' \leq T_0 \\ |t-t'| > 2(\delta-\Delta)}} e^{n [ \log(1/(1-\gamma)) + D_{\max}^{T_0, \Delta, \gamma}]}
 \max\{\mu^{\Delta}(t), \mu^{\Delta}(t')\}P_e(\calM^{(n)}[ [\rho^{\Delta}(t)]^{\otimes n}], \calM^{(n)}[ [\rho^{\Delta}(t')]^{\otimes n}])\\
 \nonumber
    &\leq 2 T_0 N \sup_{\substack{-T_0 \leq t, t' \leq T_0 \\ |t-t'| > 2(\delta-\Delta)}} e^{n [ \log(1/(1-\gamma)) + D_{\max}^{T_0, \Delta, \gamma}]}
  P_e(\calM^{(n)}[ [\rho^{\Delta}(t)]^{\otimes n}], \calM^{(n)}[ [\rho^{\Delta}(t')]^{\otimes n}]).
  \nonumber
\end{align}
The first inequality is Eq.~\eqref{eqn:upper_bound_from_prop}, the second inequality uses the facts that for the optimal binary hypothesis testing error we have that $A \leq A', B\leq B'$ implies that $P_e^{*}(A, B) \leq P_e^{*}(A', B')$ as well as that $P_e^{*}(\alpha A, \alpha B) = \alpha P_e^{*}(A, B)$. Both of these facts are readily observable from the Helstrom formula and the %
convex problem
formulation of the hypothesis testing error. The third inequality extracts the measure via maximization and the fourth inequality upper-bounds the maximum of the smoothed measures by one. We can conclude that the asymptotic rate of approaching $\mu([-T_0, T_0])$ is thus at least
\begin{align}
    -\lim_{n \to \infty} \log [\mu([-T_0, T_0]) - \eta^{*}(\delta, \rho^{\otimes n}, \mu)] 
    &\geq  \inf_{|t-t'| > 2(\delta - \Delta)} R(\rho_{\gamma}^{\Delta}(t), \rho_{\gamma}^{\Delta}(t'), \{ M^{(n)}\})  - \log(1/(1-\gamma)) - D_{\max}^{T_0, \Delta, \gamma}.
\end{align}
Having established this bound, we now desire to let the number of discretization steps $N \to \infty$ and thus $\Delta \to 0$. By assumption, $t\mapsto \rho(t)$ is a continuous function and hence
\begin{align}\label{eqn:state_limit}
\lim_{\Delta \to 0}\left\lVert \rho_{\gamma}(t) - \rho_{\gamma}^{\Delta}(t) \right\rVert &=
(1-\gamma) \lim_{\Delta \to 0}\left\lVert \rho(t) - \rho^{\Delta}(t) \right\rVert \\
\nonumber
&=(1-\gamma)\lim_{\Delta \to 0}\left\lVert \rho(t) - \frac{1}{\Delta} \int_{-\Delta/2}^{\Delta/2} \diff \tau \, \rho(t + \tau)  \right\rVert \\
\nonumber
&= (1-\gamma)\lim_{\Delta \to 0} \frac{1}{\Delta} \left\lVert \int_{-\Delta/2}^{\Delta/2}  \rho(t) - \rho(t + \tau)  \right\rVert\\
\nonumber
&\leq \lim_{\Delta \to 0} \max_{-\Delta/2 \leq \tau \leq \Delta/2} \lVert \rho(t)  - \rho(t+\tau) \rVert \\&= 0.
\nonumber
\end{align}
This implies that
\begin{align}
    \lim_{\Delta \to 0} R(\rho_{\gamma}^{\Delta}(t), \rho_{\gamma}^{\Delta}(t'), \{ M^{(n)}\}) = R(\rho_{\gamma}(t), \rho_{\gamma}(t'), \{ M^{(n)}\})
\end{align}
as $R$ is a continuous function in the first two arguments with respect to the trace norm by the uniform limit theorem as it is a composition of continuous functions. 

To wrap up our proof, we also show that the quantity $D_{\max}^{T_0, \Delta, \gamma}$ vanishes. To do so, we rely on the formulation of the max-relative entropy that involves the pseudoinverse:
\begin{align}
    D_{\max}(\rho_{\gamma}(t+\delta),\rho_{\gamma}^{\Delta}(t)) &= \log \mathstrut \lVert \rho_{\gamma}(t+\delta) \rho_{\gamma}^{\Delta}(t)^{+}\rVert_{\infty}.
\end{align} 
We can then make use of some results on the perturbation theory of the pseudoinverse. Namely, that if $\operatorname{rank}(A) = \operatorname{rank}(A + X)$ and $\lVert X \rVert_{\infty} < 1/\lVert A^{+}\rVert_{\infty}$, we have that~\cite{wedin_perturbation_1973}
\begin{align}
    \lVert A^{+} - (A+X)^{+} \rVert_{\infty} &\leq 3 \lVert A^{+} \rVert_{\infty} \lVert (A+X)^{+}\rVert_{\infty} \lVert X \rVert_{\infty} \\
    &\leq 3 \frac{ \lVert A^{+} \rVert_{\infty}^2 \lVert X \rVert_{\infty}}{1 -  \lVert A^{+} \rVert_{\infty} \lVert X \rVert_{\infty}} .
    \nonumber
\end{align}
In our case, we would like to show that $A+X = \rho_{\gamma}^{\Delta}(t)^{+}$ is not too far from $A = \rho_{\gamma}(t+\delta)^{+}$. As we perturbed all states with the maximally mixed state, we know that $\rho_{\gamma}^{\Delta}(t)$ and $\rho_{\gamma}(t+\delta)$ both have full rank, fulfilling the first requirement of the above result.
Next, we use a similar argument as in Eq.~\eqref{eqn:state_limit} to establish that as long as $|\delta| \leq \Delta/2$
\begin{align}
\lim_{\Delta \to 0}\left\lVert \rho_{\gamma}(t + \delta) - \rho_{\gamma}^{\Delta}(t) \right\rVert &=(1-\gamma) \lim_{\Delta \to 0}\left\lVert \rho(t + \delta) - \rho^{\Delta}(t) \right\rVert \\
\nonumber
&\leq \lim_{\Delta \to 0}\left\lVert \rho(t+\delta) - \frac{1}{\Delta} \int_{-\Delta/2}^{\Delta/2} \diff \tau \, \rho(t + \tau)  \right\rVert \\
\nonumber
&= \lim_{\Delta \to 0} \frac{1}{\Delta} \left\lVert \int_{-\Delta/2}^{\Delta/2} \diff \tau \,  \rho(t+\delta) - \rho(t + \tau)  \right\rVert\\
\nonumber
&\leq \lim_{\Delta \to 0} \max_{-\Delta/2 \leq \tau \leq \Delta/2} \lVert \rho(t+\delta)  - \rho(t+\tau) \rVert \\&= 0,
\nonumber
\end{align}
by continuity. This implies that we can make $\lVert \rho_{\gamma}^{\Delta}(t) - \rho_{\gamma}(t+\delta) \rVert_{\infty}$ arbitrarily small when decreasing $\Delta$. This especially means that we can make fulfill the second requirement of the above result on the magnitude of the perturbation. This is because for sufficiently small $\gamma$, we have that $1/\lVert \rho_{\gamma}(t+\delta)^{+} \rVert_{\infty} = O(1/\gamma)$ which is independent of $\Delta$. Hence, there exists a sufficiently small $\Delta$ such that we can apply the result. This allows us to conclude
\begin{align}
    \lim_{\Delta \to 0} D_{\max}(\rho_{\gamma}(t+\delta),\rho_{\gamma}^{\Delta}(t)) &= \lim_{\Delta \to 0} \log \lVert  \rho_{\gamma}(t+\delta) \rho_{\gamma}^{\Delta}(t)^{+}\rVert_{\infty} \\
    \nonumber
    &= \lim_{\Delta \to 0} \log \lVert \rho_{\gamma}(t+\delta) [\rho_{\gamma}^{\Delta}(t)^{+} - \rho_{\gamma}(t+\delta)^{+} + \rho_{\gamma}(t+\delta)^{+}]\rVert_{\infty}\\
    \nonumber
    &\leq \lim_{\Delta \to 0} \log \left(\lVert \rho_{\gamma}(t+\delta) \rho_{\gamma}(t+\delta)^{+}\rVert_{\infty} + \lVert \rho_{\gamma}(t+\delta)[\rho_{\gamma}^{\Delta}(t)^{+} - \rho_{\gamma}(t+\delta)^{+}] \rVert_{\infty} \right)\\
    \nonumber
    &\leq \lim_{\Delta \to 0} \log \left(1 + \lVert \rho_{\gamma}(t+\delta)\rVert_{\infty}\lVert[\rho_{\gamma}^{\Delta}(t)^{+} - \rho_{\gamma}(t+\delta)^{+}] \rVert_{\infty} \right)\\
    \nonumber
    &\leq \lim_{\Delta \to 0} \log \left(1 + \lVert \rho_{\gamma}^{\Delta}(t)^{+} - \rho_{\gamma}(t+\delta)^{+} \rVert_{\infty} \right)\\
    \nonumber
    &\leq \lim_{\Delta \to 0} \lVert[\rho_{\gamma}^{\Delta}(t)^{+} - \rho_{\gamma}(t+\delta)^{+}] \rVert_{\infty}\\
    \nonumber
    &\leq \lim_{\Delta \to 0} 3 \frac{ \lVert \rho_{\gamma}(t+\delta)^{+} \rVert_{\infty}^2 \lVert \rho_{\gamma}^{\Delta}(t) - \rho_{\gamma}(t+\delta) \rVert_{\infty}}{1 -  \lVert \rho_{\gamma}(t+\delta)^{+} \rVert_{\infty} \lVert\rho_{\gamma}^{\Delta}(t) - \rho_{\gamma}(t+\delta) \rVert_{\infty}}\\
    \nonumber
    &= 0.
\end{align}
The last line follows by choosing $\Delta$ sufficiently small. This immediately implies the desired relation
\begin{align}
   \lim_{\Delta \to 0} D_{\max}^{T_0, \Delta, \gamma} = 0.
\end{align}
and we obtain that
\begin{align}
    -\lim_{n \to \infty} \log [\mu([-T_0, T_0]) - \eta^{*}(\delta, \rho^{\otimes n}, \mu)] 
    &\geq  \inf_{\substack{-T_0 \leq t, t' \leq T_0 \\ |t-t'| > 2\delta}} R(\rho_{\gamma}(t), \rho_{\gamma}(t'), \{ M^{(n)}\}) - \log(1 / (1-\gamma)).
\end{align}
Finally, we let $\gamma \to 0$ and $T_0 \to \infty$, which recovers the theorem statement.
\end{proof}

Looking at the structure of the above proofs, we see that Open Problem~\ref{op:asymptotic_rate_states} could be solved if there exists a measurement $M(\lambda))$ such that in the discrete setting the pairwise discrimination error of the outcome distributions $\{\Lambda_i = \calM[\rho_i] \}_{i=1}^N$ fulfills
\begin{align}
    \max_{\substack{1 \leq i, j \leq N \\ \not\exists k \colon W_{i,k} = W_{j,k}=1}} P_e(p_i\Lambda_i, p_j\Lambda_j) \leq C \max_{\substack{1 \leq i, j \leq N, \\ \not\exists k \colon W_{i,k} = W_{j,k}=1}} \min_{0 \leq s \leq 1} \Tr[\rho_i^s \rho_j^{1-s}],
\end{align}
where the constant $C$ can be polynomial in the number of discrete states and the dimension of the underlying system, as i.i.d.\ states live in the symmetric subspace.

\section{Optimal tolerance}\label{ssec:opt_window_size}

In this section, we will focus our attention on the rectangular window with tolerance $\delta$. Until now, we have analyzed the problem of finding and optimizing the success probability we can guarantee for a fixed window size $\delta$. It is however also operationally meaningful to ask the reverse question: How small can we make the window tolerance $\delta$ while keeping the probability of success constant? This is especially interesting, as this quantity compares more naturally to the usual quantifiers in quantum metrology, namely the standard deviation.

We defined the success probability of a metrology protocol over states $\rho(t)$ with prior $\mu(t)$ as a quantity dependent on $\delta$ as $\eta(\delta) = \eta(\delta, \mu, \rho, Q)$. In the same spirit, we now define the optimal tolerance (as in the main text) as
\begin{align}
    \delta(\eta, \mu, \rho, Q) \coloneqq \inf \mathstrut\{ \delta' \geq 0 \pipe \eta(\delta', \mu, \rho, Q) \geq \eta \}.
\end{align}
If the function $\eta(\delta, \mu, \rho, Q)$ is injective in $\delta$, this is functionally equivalent to the inverse of this function seen as a map from $\delta$ to $\eta$. Contrary to $\eta$, the quantity $\delta$ cannot be written as a semi-definite program, because the dependence of $\eta$ onto $\delta$ is non-linear. 
The associated minimax quantities are defined likewise.

We can learn something about the relation between $\eta$ and $\delta$ by performing a Taylor expansion:
\begin{sproposition}[Limit for smooth POVMs]\label{sprop:limit_smooth_POVM}
For a rectangular window with small tolerance $\delta$, a set of states $\rho(t)$, possibly with with prior $\mu(t)$ and a fixed smooth POVM $Q(\tau)$, we have that
\begin{align}
    \eta(\delta, \mu, \rho, Q) &= 2\delta \int \diff \mu(t) \, \Tr[ \rho(t) Q(t)] + O(\delta^3), \\
    \overline{\eta}(\delta, \rho, Q) &= 2\delta \min_t \Tr[ \rho(t) Q(t)] + O(\delta^3)     .
\end{align}
\end{sproposition}
\begin{proof}
A simple Taylor expansion of $Q$ around $Q(t)$ gives
\begin{align}
    Q(t + \tau) = Q(t) + \tau \dot{Q}(t) + \frac{1}{2}\tau^2 \ddot{Q}(t) + O(\tau^3).
\end{align}
Integrating this from $-\delta$ to $\delta$ yields
\begin{align}
    \int_{-\delta}^{\delta} \diff \tau \, Q(t+\tau) = 2 \delta Q(t) + \frac{1}{3} \delta^3 \ddot{Q}(t) + O(\delta^4).
\end{align}
As 
\begin{align}
    \eta(\delta, \mu, \rho, Q) = \int \diff \mu(t) \, \int_{-\delta}^{\delta} \diff \tau \, \Tr[\rho(t) Q(t+\tau)],
\end{align}
the statement of the proposition in the Bayesian case follows. The minimax case is also evident when recognizing that we just have to take a minimum of the same expression.
\end{proof}
The above proposition highlights that in the limit of small $\delta$, the optimal POVM is independent of $\delta$, at least if we optimize over smooth POVMs only. We can also use it to give a simple proof of  Proposition~\ref{prop:tolerance_for_small_eta} of the main text:
\begin{proof}[Proof of Proposition~\ref{prop:tolerance_for_small_eta}]
The formula for the derivative $\partial_{\delta} \eta$ at $\delta = 0$ can be readily read off the result of Proposition~\ref{sprop:limit_smooth_POVM}, the statement then follows by applying the formula for the first derivative of the inverse function.
\end{proof}

\subsection{Lower bounds via asymmetric hypothesis testing}\label{ssec:opt_window_size_hyp_test}

In this section, we provide lower bounds on the size of the window function in terms of asymmetric hypothesis testing. We refer to Section~\ref{ssec:upper_bound_succ_prob_asym_ht} for a brief description of asymmetric hypothesis testing.

\begin{stheorem}
\label{sthm:window_width_lower_bound_integral_t_fixed_eta}
  For a given window function $w$ and set of states $\rho(t)$, we have for every quantum state $\sigma$ and constant $0<\eta_0\leq \overline{\eta}^*(w, \rho)$ that
  \begin{align}
    \int \diff t \, w(t)
    \geq
    \int\Diff{t}\beta_{1-\eta_0}(\rho(t)\Vert\sigma).%
\label{eqn:window_width_lower_bound_integral_t_fixed_eta}
  \end{align}
\end{stheorem}
\begin{proof}
  By %
  duality of semi-definite problems, it holds that the optimal type-II error probability defined in \eqref{eqn:hyp_test_opt_typeII_error_prob} can be expressed as~\cite{khatri2020principles}
\begin{equation}
    \beta_{\epsilon}(\rho\Vert\sigma)=\sup\{\mu(1-\epsilon)-\Tr[X]: \mu\geq 0,\, X\geq 0, \mu\rho\leq\sigma+X\}.
\end{equation}
  Define $\bar\mu(t), \bar{X}(t)$ to be the optimal candidates in this 
  convex problem,
  such that
  $\bar{\mu}(t) \rho(t) \leq \sigma + \bar{X}(t)$ and
  $\beta_{1-\eta_0}(\rho(t)\Vert\sigma) = \bar{\mu}(t){\eta_0} -
  \Tr [\bar{X}(t)]$.
  The strategy of the proof is to find dual candidates in the
  convex
  program for $\overline{\eta}^*(w, \rho)$ in \eqref{eq:minimax_sdp}.

  Let $\nu = \bigl[ \int\diff t\,\bar{\mu}(t) \bigr]^{-1}$ and let
  $\mu(t) = \nu\bar{\mu}(t)$, and observe that
  $\int \diff t \, \mu(t) = 1$.  Let
  \begin{align}\label{eqn:9u03qeiubfhieohfiow}
    \delta
    &= \frac{1}{2} \int\diff t\, w(t).
  \end{align}
  For any $t'\in\mathbb{R}$, we find
  \begin{align}
    \int\diff t\,\mu(t)\,w(t - t')\,\rho(t)
    & =
    \nu\int\diff t\,w(t - t')\, \bar\mu(t)\rho(t)\\
    \nonumber
    &\leq
    \nu\int\diff t\,w(t - t')\, \bigl[ \sigma + \bar{X}(t) \bigr]\\
    \nonumber
    &\leq
    2\nu \delta\sigma + \nu\int \diff t \, \bar{X}(t)\ ,
    \nonumber
  \end{align}
  where for the second term in the last line we have used the fact that $w(t - t')\leq 1$.
  Defining
  \begin{align}
    X = 2\nu\delta\sigma + \nu\int \diff t\,\bar{X}(t)
  \end{align}
  thus ensures that
  \begin{align}
      \int \diff t\,\mu(t)\,w(t-t')\,\rho(t)
      \leq X.
  \end{align}
  Therefore, $\mu(t)$ and $X$ are feasible candidates in the dual problem
  for $\overline{\eta}^*(w, \rho)$. The objective value attained
  by this choice of variables directly gives us an upper bound on the
  optimal value $\overline{\eta}^*(w, \rho)$, \textit{i.e.},
  \begin{align}
    \overline{\eta}^*(w, \rho)
    \leq \Tr [X]
    = 2\nu\delta + \nu\int \diff t \,\Tr[\bar{X}(t)].
    \label{eqn:9r3iurgebfhijdnsk}
  \end{align}
  The second term of~\eqref{eqn:9r3iurgebfhijdnsk} is
  \begin{align}
    \nu\int \diff t\,\Tr [\bar{X}(t)]
    &= \nu \int \diff t\,\bigl(\bar{\mu}(t) {\eta_0} - \beta_{1-\eta_0}(\rho(t)\Vert\sigma) \bigr)
    \nonumber\\
    &= {\eta_0} - \nu\int\diff t\,\beta_{1-\eta_0}(\rho(t)\Vert\sigma),
  \end{align}
  where the first equality follows from the properties of the optimal candidates
  in the 
  convex
  problem defining hypothesis testing entropy, and the
  second by the definition of $\nu$. 
  Plugging this
  into Eq.~\eqref{eqn:9r3iurgebfhijdnsk}, and recalling that $\eta_0\leq\overline{\eta}^*(w,\rho)$, we find
  \begin{align}
    {\eta_0} &\leq \overline{\eta}^*(w, \rho)
    \leq {\eta_0} +
    \nu\biggl[ 2\delta - \int\diff t\,\beta_{1-\eta_0}(\rho(t)\Vert\sigma) \biggr],
  \end{align}
  which implies that
  \begin{equation}
      \nu\biggl[ 2\delta - \int\diff t\,\beta_{1-\eta_0}(\rho(t)\Vert\sigma)\biggr]\geq 0.
  \end{equation}
  Finally, because $\nu\geq 0$, we obtain
  \begin{align}
      2\delta \geq \int\diff t\,\beta_{1-\eta_0}(\rho(t)\Vert\sigma)\ .
  \end{align}
  This proves the claim, recalling the definition of $\delta$ in Eq.~\eqref{eqn:9u03qeiubfhieohfiow}.
\end{proof}

Alternatively, the theorem states that for any $0<\eta_0\leq 1$,
any attempt to use a window function
that is not as wide as prescribed by~\eqref{eqn:window_width_lower_bound_integral_t_fixed_eta} will result
in a success probability $\overline{\eta}^*(w, \rho)$
that is less than $\eta_0$.
The left hand side of~\eqref{eqn:window_width_lower_bound_integral_t_fixed_eta}
is a measure of the width of the window function.

We can now give the proof of the asymptotic lower bound of Theorem~\ref{sprop:asymptotic_delta_lower_bound_from_asym_ht} from the main text. For this, we need the famous Laplace's method which we use in the following simplified version (see, \textit{e.g.}, Ref.~\cite{wong_asymptotic_1989}):
\begin{slemma}[Laplace's method]\label{slem:laplace_method}
    Let $\phi(x)$ be a twice continuous differentiable function such that it has a unique minimum $\phi(x_0) = 0$. Then, for any interval $I$ that contains $x_0$ in its interior, we have that
    \begin{align}
        \int_{I} \diff x \, \exp( -n \phi(x)) = \sqrt{\frac{2\pi}{n \phi''(x_0)}} + O(n^{-3/2}).
    \end{align}
\end{slemma}
We can now prove the desired statement:

Using a $\delta$ window and a suitably chosen state $\sigma$ we can obtain the following corollary.

\begin{scorollary}\label{scorr:delta_lower_bound_sup_Dh_t_tp}
For a rectangular window function with tolerance $\delta$ and a state set $\rho(t)$, we have for any $\eta_0 \leq \overline{\eta}^*(w_{\delta}, \rho)$ that
\begin{align}
  \delta \geq \frac{1}{2}\int \diff t \, \exp(-D_{\mathrm{h}}^{\eta_0}( \rho(t)  \Vert{\rho(t')}))
  \geq \sup_{\delta'>0} \delta' \sup_{t} \inf_{|t-t'| \leq 2\delta'}  \exp(-D_{\mathrm{h}}^{\eta_0}( \rho(t)  \Vert{\rho(t')})).
\end{align}
\end{scorollary}
\begin{proof}
As Theorem~\ref{sthm:window_width_lower_bound_integral_t_fixed_eta} holds for any state, we can also choose any $\rho(t')$. We therefore have
\begin{align}
    2 \delta \geq \int \diff t \, \exp(-D_{\mathrm{h}}^{\eta_0}( \rho(t)  \Vert{\rho(t')})). 
\end{align}
We can now go further and restrict the integration to any interval $T$ of size $|T| = 2\delta'$ that contains $t'$ and maximize over these intervals:
\begin{align}
    2 \delta \geq \sup_{\substack{|T| = 2\delta' \\ t' \in T}} \int_T \diff t \, \exp(-D_{\mathrm{h}}^{\eta_0}( \rho(t)  \Vert{\rho(t')})).
\end{align}
We can now lower-bound the integral by the lower-bound of the integrand over the interval to obtain
\begin{align}
    2 \delta \geq 2 \delta'  \sup_{\substack{|T| = 2\delta' \\ t' \in T}} \inf_{t \in T} \exp(-D_{\mathrm{h}}^{\eta_0}( \rho(t)  \Vert{\rho(t')})).
\end{align}
As $t'$ has been arbitrary, we can instead just optimize over $T$ as well to obtain
\begin{align}
    2 \delta \geq 2 \delta'  \sup_{\substack{|T| = 2\delta'}} \sup_{t' \in T} \inf_{t \in T} \exp(-D_{\mathrm{h}}^{\eta_0}( \rho(t)  \Vert{\rho(t')})).
\end{align}
As $t$ and $t'$ have to lie in the same interval, we have that
\begin{align}
    2 \delta \geq 2 \delta' \sup_{t} \inf_{|t-t'| \leq 2\delta'}  \exp(-D_{\mathrm{h}}^{\eta_0}( \rho(t)  \Vert{\rho(t')})).
\end{align}
The statement of the Corollary follows by optimizing over $\delta'$ and dividing by 2.
\end{proof}

The above result allows us to exploit Laplace's method and the relation of the hypothesis testing relative entropy with the sandwiched Rényi relative entropies of Eq.~\eqref{eqn:def_sandwiched_renyi_relative_entropy} to get an asymptotic bound. 
We will use Laplace's method in the following simplified version (see, \textit{e.g.}, Ref.~\cite{wong_asymptotic_1989}):
\begin{slemma}[Laplace's method]\label{slem:laplace_method_2}
    Let $\phi(x)$ be a twice continuous differentiable function such that it has a unique minimum $\phi(x_0) = 0$. Then, for any interval $I$ that contains $x_0$ in its interior, we have that
    \begin{align}
        \int_{I} \diff x \, \exp( -n \phi(x)) = \sqrt{\frac{2\pi}{n \phi''(x_0)}} + O(n^{-3/2}).
    \end{align}
\end{slemma}
The asymptotic behavior of the hypothesis testing bound is dominated by the hypothesis testing relative entropy of states close in time. In this case, the information measure induced by the sandwiched Rényi relative entropy~\cite{meyer2021fisher} becomes relevant. We denote it by $\tilde{\calI}_{\alpha}(t)$ and it is implicitly defined via~\cite{takahashi_information_2017}
\begin{align}
    \tilde{D}_{\alpha}(\rho(t)\fatpipe \rho(t + \delta)) = \frac{\alpha}{2} \delta^2 \tilde{\calI}_{\alpha}(t) + O(\delta^3).
\end{align}
We obtain the following asymptotic result that reproduces the scaling of the standard quantum limit:
\begin{sproposition}[Asymptotic lower bound]\label{sprop:asymptotic_delta_lower_bound_from_asym_ht}
For all $\alpha > 1$, the optimal minimax tolerance for a given minimax success probability $\overline{\eta}$ obeys the inequality
\begin{align}
    \overline{\delta}^{*}(\overline{\eta}, \rho^{\otimes n}) &\geq \frac{1}{2} \overline{\eta}^{\frac{\alpha}{\alpha-1}} \sqrt{\frac{2\pi}{\alpha n \tilde\calI_{\alpha}}} + O\left(\frac{1}{n^{3/2}}\right)
\end{align}
where 
\begin{align}
    \tilde\calI_{\alpha} \coloneqq \min_t \tilde\calI_{\alpha}(\rho(t)).
\end{align}
\end{sproposition}
As Theorem ~\ref{thm:cramer_rao_like_bound}, the above theorem visibly reminds us of the Quantum Cramér-Rao in the scaling of $n$, but the dependence on $\overline{\eta}$ is inferior.
\begin{proof}
We start from Proposition~\ref{prop:tolerance_for_small_eta}. To obtain a lower bound, we employ the standard upper bound 
\begin{align}
    D_{\mathrm{h}}^{\eta_0}( \rho \fatpipe \sigma) \leq \tilde{D}_{\alpha}(\rho \fatpipe \sigma) + \frac{\alpha}{\alpha - 1} \log \frac{1}{\eta_0}
\end{align}
on the hypothesis testing relative entropy via the sandwiched Rényi relative entropies~\cite{khatri2020principles},
which holds for $1 < \alpha < \infty$. We now set $\sigma = \rho^{\otimes n}(t')$ for some $t'$ to be determined later. Exploiting the additivity of the sandwiched Rényi relative entropy, we obtain
\begin{align}
    \overline{\delta}(\overline{\eta},\rho^{\otimes n}) &\geq \frac{1}{2} \int \diff t \, \exp\left( - n \tilde{D}_{\alpha}(\rho(t) \fatpipe \rho(t')) - \frac{\alpha}{\alpha - 1} \log \frac{1}{\overline{\eta}}\right) \\
    &=\frac{1}{2} \overline{\eta}^{\frac{\alpha}{\alpha - 1}} \int \diff t \, \exp\left( - n \tilde{D}_{\alpha}(\rho(t) \fatpipe \rho(t'))\right).
    \nonumber
\end{align}
The integral on the right hand side can be evaluated by a simple application of Laplace's method as given in Lemma~\ref{slem:laplace_method_2}, recognizing that $\tilde{D}_{\alpha}(\rho(t) \lVert \rho(t'))$ achieves its minimum for $t = t'$, yielding
\begin{align}
    \overline{\delta}(\overline{\eta},\rho^{\otimes n}) &\geq\frac{1}{2} \overline{\eta}^{\frac{\alpha}{\alpha - 1}}
    \sqrt{\frac{2\pi}{n \tilde{D}_{\alpha}''(t')}} + O(n^{-3/2}),
\end{align}
where we have denoted
\begin{align}
    \tilde{D}_{\alpha}''(t') = \left.\frac{\partial^2}{\partial \Delta^2} \tilde{D}_{\alpha}(\rho(t' + \Delta)\lVert \rho(t')) \right|_{\Delta = 0}.
\end{align}
The second order expansion of the sandwiched Rényi relative entropy was studied in Ref.~\cite{takahashi_information_2017}. We have that 
\begin{align}
    \tilde{D}_{\alpha}''(t') = \alpha \Tilde{\calI}_{\alpha}(t'),
\end{align}
where $\Tilde{\calI}_{\alpha}$ interpolates between the Bogoliubov-Kubo-Mori (BKM) information in the limit $\alpha \to 1$ and other information measures. The associated Petz function is given by
\begin{align}
    f_{\alpha}(t) = (\alpha - 1) \frac{t^{1/\alpha}}{1 - t^{(1-\alpha)/\alpha}}.
\end{align}
Optimizing over $t'$ yields $\tilde{\calI}_{\alpha} \coloneqq \min_{t'} \Tilde{\calI}_{\alpha}(t')$. Putting this into the bound then gives
\begin{align}
    \overline{\delta}(\overline{\eta},\rho^{\otimes n}) &\geq\frac{1}{2} \overline{\eta}^{\frac{\alpha}{\alpha - 1}}
    \sqrt{\frac{2\pi}{n \alpha \tilde{\calI}_{\alpha}}} + O(n^{-3/2}).
\end{align}
\end{proof}

\subsection{Lower bound via symmetric hypothesis testing}\label{ssec:tolerance_lower_bound_symmetric_ht}
In this section, we will give the proof of the non-asymptotic Cramér-Rao like bound presented in Theorem~\ref{thm:cramer_rao_like_bound} of the main text.

\begin{proof}[Proof of Theorem~\ref{thm:cramer_rao_like_bound}]
Our derivation starts from Corollary~\ref{corr:two_point_fidelity_bound}, which states that
\begin{align}
    1 - \overline{\eta}^{*}(\delta, \rho) &\geq \frac{1}{4} \sup_{|t-t'| > 2\delta} F(\rho(t), \rho(t'))^2.
\end{align}
Using the fact that the sandwiched Rényi relative entropy of order $1/2$ is given by
\begin{align}
    \tilde{D}_{\frac{1}{2}}(\rho \fatpipe \sigma) = -\frac{1}{2} \log F(\rho, \sigma)
\end{align}
and choosing $t' = t + 2\delta$, we obtain
\begin{align}\label{eqn:corollary_10_as_entropy_inequality}
    \log\frac{1}{4(1-\overline{\eta})} \leq 4 \inf_t \tilde{D}_{\frac{1}{2}}(\rho(t) \fatpipe \rho(t+ 2\delta)),
\end{align}
as we can readily compare to Eq.~\eqref{eqn:corollary_10_as_log_fidelity} of the main text. Our desire is now to determine the scale of $\delta$ that we are allowed to choose. To this end, we perform a Taylor expansion of $\tilde{D}_{\frac{1}{2}}$:
\begin{align}\label{eqn:supplementary_taylor_expansion_D_12}
    \tilde{D}_{\frac{1}{2}}(\rho(t) \fatpipe \rho(t + \tau)) &= \sum_{k=2}^{\infty} \frac{f_k \tau^k}{k!} \\
    &= \frac{1}{2} f_2(t) \tau^2 + \frac{1}{6} f_3(t) \tau^3 + \frac{1}{24} f_4(t) \tau^4 + \dots
\end{align}
We assume that the Taylor expansion is valid in a radius of convergence $|\tau| < R(t)$.
The coefficients $f_k(t)$ are given by
\begin{align}
    f_k(t) \coloneqq \left.\frac{\partial^k}{\partial \tau^k} \tilde{D}_{\frac{1}{2}}(\rho(t) \fatpipe \rho(t + \tau))\right|_{\tau=0}.
\end{align}
As the quantum Fisher information can be defined via~\cite{meyer2021fisher}
\begin{align}
    \calF(t) \coloneqq - 2 \left.\frac{\partial^2}{\partial \tau^2} F(\rho(t), \rho(t+\tau))^2 \right|_{\tau=0},
\end{align}
we have that -- as, for example discussed in Appendix C of Ref.~\cite{meyer2021fisher} --
\begin{align}
    f_2(t) &= -\frac{1}{2} \left.\frac{\partial^2}{\partial \tau^2}  \log F(\rho(t), \rho(t+\tau))\right|_{\tau=0} \\
    &= -\frac{1}{4} \left.\frac{\partial^2}{\partial \tau^2}  \log F(\rho(t), \rho(t+\tau))^2 \right|_{\tau=0} \\
    &= -\frac{1}{4} \left(\left.\frac{\partial}{\partial F} \log(F)\right|_{F=1}\right) \left( \left.\frac{\partial^2}{\partial \tau^2} F(\rho(t), \rho(t+\tau))^2 \right|_{\tau=0}\right) \\
    &= \frac{1}{8} \calF(t).
\end{align}
We note that the higher derivatives of the sandwiched Rényi relative entropy do not coincide anymore with the higher derivatives of the fidelity up to a constant. Having established the relation between the leading terms of the expansions of fidelity and the sandwiched Rényi relative entropy, we can now turn to the actual scale of $\delta$. As was made intuitive in the main text, for i.i.d.\ copies we are able to asymptotically choose $\delta = O(1/\sqrt{f_2}) = O(1/\sqrt{\calF})$, the same scaling we expect from the quantum Cramér-Rao bound. We will thus make the ansatz 
\begin{align}
    \tau = 2 \delta = \frac{\gamma(t)}{\sqrt{f_2(t)}}
\end{align}
in Eq.~\eqref{eqn:supplementary_taylor_expansion_D_12}. In this case, we have
\begin{align}\label{eqn:expansion_with_gamma}
    \tilde{D}_{\frac{1}{2}}(\rho(t)\fatpipe \rho(t+ \tau)) = \frac{1}{2} \gamma^2(t) + \frac{1}{3!} \frac{f_3(t)}{f_2^{2/3}(t)} \gamma^3(t)+ \frac{1}{4!} \frac{f_4(t)}{f_2^{2}(t)} \gamma^4(t) + \dots.
\end{align}
To obtain a general upper bound for this expression, we define the constant
\begin{align}
    q(t) \coloneqq \sup_{3 \leq p \in \bbN} \left| \frac{f_p(t)}{f_2^{p/2}(t)} \right|^{\mathrlap{\frac{1}{p-2}}},
\end{align}
such that we can bound the fractions $f_p(t)/f_2^{p/2}(t)$ in Eq.~\eqref{eqn:expansion_with_gamma} as
\begin{align}
    \tilde{D}_{\frac{1}{2}}(\rho(t)\fatpipe \rho(t+ \tau)) &\leq \frac{1}{2} \gamma^2(t) + \frac{1}{3!} q(t) \gamma^3(t)+ \frac{1}{4!} q^2(t)\gamma^4(t) + \dots \\
    &= \frac{1}{q^2(t)} \left( e^{q(t) \gamma(t)} - 1 - q(t) \gamma(t) \right).
\end{align}
We now insert this bound into Eq.~\eqref{eqn:corollary_10_as_entropy_inequality} to obtain
\begin{align}
    \frac{1}{4}\log\frac{1}{4(1-\overline{\eta})} \leq \frac{1}{q^2(t)} \left( e^{q(t) \gamma(t)} - 1 - q(t) \gamma(t) \right).
\end{align}
We observe that the right hand side is always non-negative and that this only gives a nontrivial bound if $\overline\eta > 3/4$, similarly to Corollary~\ref{corr:two_point_fidelity_bound}. We will henceforth assume that this condition is met. 
Substituting the left hand side as $a = -\log(4(1-\overline{\eta}))/4$, we can now solve the above for equality using Mathematica
\begin{lstlisting}[language=Mathematica,basicstyle=\small\ttfamily]
Solve[
    {(Exp[q\[Gamma]]-1-q\[Gamma])/q^2 == a, q>0}, 
    \[Gamma]
]
\end{lstlisting}
to obtain
\begin{align}\label{eqn:def_gamma_of_q_and_eta}
\gamma_{=}(t) = -\frac{1}{q(t)}\left[ 1 + a q^2(t) + W_k(-e^{-1 - a q^2(t)})\right],
\end{align}
where $W_k(x)$ is the product logarithm function, \textit{i.e.}, the solution of $w e^w = x$. The integer $k$ identifies the corresponding branch, in our case $k = -1$ is relevant because the argument $-e^{-1 - a q^2(t)}$ lies between $-1/e$ and $0$. $\gamma_{=}(t)$ identifies the smallest admissible $\gamma$ we can take to still satisfy Corollary~\ref{corr:two_point_fidelity_bound} for fixed reference time $t$.
We can now continue to place bounds on $\gamma_{=}(t)$. To do so, we exploit the results of Ref.~\cite{chatzigeorgiou2013bounds} that show that
\begin{align}
    -1 - \sqrt{2u} - u < W_{-1}(-e^{-u-1}) < -1 - \sqrt{2u} - \frac{2}{3}u.
\end{align}
Identifying $u = a q^2(t)$ and inserting into Eq.~\eqref{eqn:def_gamma_of_q_and_eta} yields
\begin{align}
    \sqrt{2 a} - \frac{1}{3} a q(t) < \gamma_{=}(t) < \sqrt{2 a}.
\end{align}
Substituting $a = -\log(4(1-\overline\eta))/4$ then gives
\begin{align}
    \sqrt{\frac{1}{2} \log \frac{1}{4(1-\overline{\eta})}} - \frac{q(t)}{12} \log \frac{1}{4(1-\overline{\eta})} < \gamma_{=}(t) < \sqrt{\frac{1}{2} \log \frac{1}{4(1-\overline{\eta})}}.
\end{align}
We now use the fact that we chose $2\delta(t) = \gamma(t)/\sqrt{f_2(t)}$ and $f_2(t) = \calF(t)/8$ to deduce
\begin{align}\label{eqn:delta_lower_bound_gamma_equals}
    \overline{\delta}(\overline{\eta},\rho) &\geq \sup_t \frac{1}{2} \frac{\gamma_{=}(t)}{\sqrt{f_2(t)}} \\
    &= \sup_t \frac{\sqrt{2} \gamma_{=}(t)}{\sqrt{\calF(t) } }. 
\end{align}
To obtain the Theorem statement, we define
\begin{align}
    q &\coloneqq \sup_t q(t).
\end{align}
The coefficient $\Gamma(t) \coloneqq \sqrt{2} \gamma(t)$ now fulfills the inequality
\begin{align}
    \frac{q}{6\sqrt{2}} \log \frac{1}{4(1-\overline{\eta})} > \sqrt{\log \frac{1}{4(1-\overline{\eta})}} - \Gamma(t) > 0.
\end{align}
This holds especially for the $t$ achieving $\inf_t \calF(t)$ in Eq.~\eqref{eqn:delta_lower_bound_gamma_equals} and as such completes the Theorem statement.
\end{proof}

\section{Optimal sample complexity}\label{ssec:optimal_sample_complexity}
In this section, we provide supplementary material for Section~\ref{sec:sample_complexity} of the main text.

\begin{proof}[Proof of Corollary~\ref{corr:sample_complexity_scaling_bound}]
Theorem~\ref{thm:cramer_rao_like_bound} states that
\begin{align}
    \overline{\delta}(\overline\eta, \rho) \geq \frac{\Gamma}{\sqrt{\inf_t \calF(t)}},
\end{align}
where
\begin{align}
      \frac{q}{6\sqrt{2}} \log \frac{1}{4(1-\overline{\eta})} > \sqrt{\log \frac{1}{4(1-\overline{\eta})}} - \Gamma > 0.
\end{align}
We thus have the lower bound
\begin{align}
    \overline{\delta}(\overline\eta, \rho) \geq \frac{\sqrt{\log \frac{1}{4(1-\overline{\eta})}}}{\sqrt{\inf_t \calF(t)}} - \frac{\frac{q}{6\sqrt{2}} \log \frac{1}{4(1-\overline{\eta})}}{\sqrt{\inf_t \calF(t)}}.
\end{align}
Inserting the assumptions of the Corollary then establishes the scalings
\begin{align}
    \overline{\delta}(\overline\eta, \rho) \geq O(n^{-\frac{\alpha}{2}})\sqrt{\log \frac{1}{4(1-\overline{\eta})}} - o(1)O(n^{-\frac{\alpha}{2}}) \log \frac{1}{4(1-\overline{\eta})}.
\end{align}
After rearranging, we have
\begin{align}
    O(n^{\frac{\alpha}{2}}) \geq \frac{\sqrt{\log \frac{1}{4(1-\overline{\eta})}}}{\overline{\delta}(\overline\eta, \rho)} - \frac{o(1) \log \frac{1}{4(1-\overline{\eta})}}{\overline{\delta}(\overline\eta, \rho)},
\end{align}
which means that
\begin{align}
    n \geq O\left( \left[ \frac{\log\frac{1}{1-\overline\eta}}{\overline\delta^{2}} \right]^{\frac{1}{\alpha}}\right),
\end{align}
as desired.
\end{proof}

\section{Beyond univariate metrology}\label{ssec:beyond_univariate_metrology}

In the case of arbitrary parameter spaces, let us consider a parameters $x$ from a set $\calX$ over which the prior distribution $\mu$ is defined such that $x \mapsto \rho(x)$. We assume that $\calX$ is equipped with a distance measure $d(x, y)$ that is not necessarily symmetric.
In this case, the definition of the Bayesian success probability becomes
\begin{align}
    \eta(w, \mu, \rho) = \int_{\calX} \diff \mu(x) \, \int_{\calX} \diff y \, w(d(x,y)) \Tr[\rho(x) Q(y)].
\end{align}
Likewise, the minimax success probability is given by
\begin{align}
    \overline\eta(w, \mu, \rho) = \inf_{x\in\calX} \int_{\calX} \diff y \, w(d(x,y)) \Tr[\rho(x) Q(y)].
\end{align}
In this definition, we associate any parameter $x$ the accepting POVM effect
\begin{align}
    \tilde{Q}_{w}(x) \coloneqq \int_{\calX} \diff y \, w(d(x,y)) Q(y),
\end{align}
which takes the role of $(w * Q)(t)$ in the univariate case. To see how our results extend in this realm, we introduce a generalization of Theorem~\ref{sthm:succ_prob_upper_bound_mht}. For it, we only need the further notion of a \emph{space-preserving transformation}:
\begin{definition}[Space-preserving transformation]
    We call a transformation $\calT\colon \calX \to \calX$ space-preserving if it is invertible and $\calT[\calX] = \calX$. 
\end{definition}
We now give a generalization of Theorem~\ref{sthm:succ_prob_upper_bound_mht} which in turn generalizes Theorem~\ref{thm:succ_prob_upper_bound_mht_delta_window} of the main text:
\begin{stheorem}\label{sthm:succ_prob_upper_bound_mht_multivariate}
For a given parameter space $\calX$ with distance function $d$ and window function $w$, fix any set $\calS = \{(\lambda, \calT)\}$ of prior probabilities $\lambda \geq 0$ and space-preserving transformations $\calT$ such that $\sum_{\lambda \in \calS} \lambda = 1$. Then, for a state set $\rho(x)$, possibly with prior $\mu(x)$, we have the upper bounds
\begin{align}
    \eta^{*}(w, \mu, \rho) &\leq K \int_{\calX} \diff x \, P^{*}_s(\{ \lambda \, \mu(\calT[x]) \rho(\calT[x])\}_{(\lambda, \calT) \in \calS}), \\
    \overline{\eta}^{*}(w, \rho) &\leq K  \inf_{x\in\calX} \overline{P}^{*}_s(\{  \rho(\calT[x])\}_{\calT \in \calS}),
\end{align}
where we have introduced the constant
\begin{align}
    K \coloneqq \sup_{x, y \in \calX}  \left\{\sum_{\calT \in \calS}  w(d(\calT[x],y)) \right\},
\end{align}
which measures the overlap of the windows for the different transformations.
\end{stheorem}
\begin{proof}
First, we recall the definition of the optimal multi-hypothesis testing success probability for a set of operators $\{ A_i\}$:
\begin{align}
    P^{*}_s(\{ A_i\}) &\coloneqq \sup_{\substack{0 \leq Q_i \leq \bbI \\ \sum_i Q_i = \bbI}} \sum_i \Tr[ A_i Q_i ].
\end{align}
We exploit that the transformations are space-preserving and we can thus transform the domain of integration that computes the success probability arbitrarily, to observe that
\begin{align}
    \eta(\delta, \mu, \rho, Q)  = \sum_{(\lambda,\calT)\in\calS} \int_{\calX} \diff x \, \lambda \Tr[ \mu(\calT[x]) \rho(\calT[x]) \Tilde{Q}_w(\calT[x])].
\end{align}
Using the definition of $K$ given in the theorem statement, we see that defining the operators
\begin{align}
    \overline{Q}_{\calT}(x) \coloneqq \frac{1}{K}\tilde{Q}_w(\calT[x]) 
\end{align}
yields a valid sub-normalized POVM for all $x$ as
\begin{align}
    \sum_{\calT \in \calS} \overline{Q}_{\calT}(x) &= \frac{1}{K}\sum_{\calT \in \calS} \tilde{Q}_w(\calT[x]) \\
    \nonumber
    &= \frac{1}{K}\sum_{\calT \in \calS} \int_{\calX} \diff y \, w(d(\calT[x],y)) Q(y)\\
     \nonumber
    &= \frac{1}{K}\int_{\calX} \diff y \, \left( \sum_{\calT \in \calS}  w(d(\calT[x],y)) \right)  Q(y)\\
     \nonumber
    &\leq \frac{1}{K}\int_{\calX} \diff y \, K  Q(y)\\
     \nonumber
    &\leq \bbI.
\end{align}
This means that the operators $\{ \overline{Q}_{\calT}(x) \}_{\calT \in \calS}$ can serve as a candidate POVM in the optimization that computes $P_s( \{ \lambda \, \mu(\calT[x]) \rho(\calT[x]) \}_{(\lambda,\calT)\in \calS})$,and hence
\begin{align}
    \eta(w, \mu, \rho, Q)\leq K \int \diff t \, P^{*}_s(\{ \lambda \, \mu(\calT[x]) \rho(\calT[x])\}_{(\lambda, \calT) \in \calS}\})
\end{align}
which implies the first statement of the theorem as the upper bound is independent of the chosen POVM $Q(y)$.

The minimax statement is derived in a similar fashion, observing that we can also apply the transformation trick to obtain
\begin{align}
    \overline\eta(w, \rho, Q)  &=  \sum_{(\lambda,\calT)\in\calS} \lambda \inf_{x \in \calX}  \Tr[ \rho(\calT[x]) \tilde{Q}_{w}(\calT[x])] \\
    &\leq \inf_{x\in\calX} \sum_{(\lambda,\calX)\in\calS} \lambda \Tr[ \rho(\calT[x])\tilde{Q}_{w}(\calT[x])].
    \nonumber
\end{align}
Here, we again make the argument that the $\{ \overline{Q}_{\calT}(x) \}_{\calT \in \calS}$ form a candidate POVM and then optimize over all possible $\lambda$ to obtain the theorem statement.
\end{proof}
Note that if $K$ in the above theorem is larger than the inverse success probability, then the bound becomes vacuous. This means, as the success probability asymptotically approaches 1, any bound that should work asymptotically must have $K = 1$. Let us now turn to the practically important task of the rectangular window with tolerance $\delta$:
We can define a distance ball around a points as
\begin{align}
    \calB_{\delta}(x) \coloneqq \{ y \in \calX \pipe d(x, y) \leq \delta \}.
\end{align}
In this case, the definition of the accepting POVM effect becomes
\begin{align}
    \tilde{Q}_{\delta}(x) \coloneqq \int_{\calB_{\delta}(x)} \diff y \, Q(y).
\end{align}
The transformation that ensure that different balls do not overlap and hence $K=1$ are the ones that make sure that for all $x, y \in \calX$ there is at most one among the $\{\calT[x]\}$ such that $d(\calT[x], y) \leq \delta$. In other words, if we define the union of all balls of size $\delta$ around $x$ as
\begin{align}
\calU_{\delta}(x) \coloneqq \bigcup \{ \calB_{\delta}(y) \pipe x \in \calB_{\delta}(y) \},
\end{align}
then 
\begin{align}
    K = 1 \ \Leftrightarrow \ \text{ for all } x \in \calX, \calT \in \calS \colon x \not\in \calU_{\delta}(\calT[x]).
\end{align}
This is analogous to the notion that metrology is as hard as distinguishing two points that are at least $2\delta$ apart in the univariate case. For very small $\delta$, one expects that $\calU_{\delta}(x) \approx \calB_{2\delta}(x)$.

We can cast the above reasoning into a corollary that can be seen as an analogue of Le Cam's two-point method:
\begin{scorollary}[Generalized two-point method]
For a given parameter space $\calX$ with distance function $d$ and given tolerance $\delta$, we have that
\begin{align}
    \overline{\eta}^{*}(\delta, \rho) \leq \inf_{\substack{x,y \in \calX\\ y \not\in \calU_{\delta}(x)}} \overline{P}_s^{*}(\rho(x), \rho(y)).
\end{align}
\end{scorollary}

\section{The covariant case: Pure Hamiltonian evolution}\label{ssec:pure_covariant_evolution}
In this section, we initiate the analytical study of the minimax success probability in a group-covariant setting. We consider a set of states $|\psi(t)\rangle$ generated by unitary evolution of a pure initial state $|\psi \rangle$ under a Hamiltonian $H$ reflecting closed system quantum mechanical evolution,
\ie
\begin{align}
    |\psi(t)\rangle &= e^{-i t H} |\psi\rangle = U(t)\ket{\psi}.
\end{align}
To ensure the group structure, $H$ must have eigenvalues such that all differences between eigenvalues are integer, in which case the recurrence time of the Hamiltonian is guaranteed to be $2\pi$. 

Let now $H$ decompose as $H = \sum_{\lambda} \lambda \Pi_{\lambda}$, where $\lambda$ are the different eigenvalues and $\Pi_{\lambda}$ are the projectors onto the possibly degenerate eigenspaces. Then, we can write
\begin{align}
    \ket{\psi} = \sum_{\lambda} \psi_{\lambda} \ket{\psi_{\lambda}},
\end{align}
where we have the normalized projections of $\ket{\psi}$ onto the eigenspaces of $H$ such that $\Pi_{\lambda}\ket{\psi} = \psi_{\lambda}\ket{\psi_{\lambda}}$. Then,
\begin{align}
    \ket{\psi(t)} &= \sum_{\lambda} \psi_{\lambda}e^{-i t \lambda} \ket{\psi_{\lambda}}
\end{align}
and
\begin{align}
    \psi(t) &= |\psi(t) \rangle\!\langle \psi(t)|\\
     \nonumber
    &=\sum_{\lambda, \lambda'}  \psi_{\lambda}\psi^{*}_{\lambda'}e^{-i t (\lambda-\lambda')} |{\psi_{\lambda}}\rangle\!\langle {\psi_{\lambda'}} | \\
     \nonumber
    &= \sum_{\omega} e^{-i \omega t} \sum_{\lambda} \psi_{\lambda}\psi^{*}_{\lambda-\omega} |{\psi_{\lambda}}\rangle\!\langle {\psi_{\lambda-\omega}} |  \label{eqn:hat_psi_omega_def}\\
     \nonumber
    &= \sum_{\omega} e^{-i \omega t} \hat{\psi}_{\omega},
\end{align}
where we have made use of $\psi(t)$ to denote the density matrix associated with 
the state $\ket{\psi(t)}$ and $\hat{\psi}_{\omega}$ are the coefficients of its Fourier transform given by
\begin{align}
\hat{\psi}_{\omega} = \frac{1}{2\pi}\int \diff t \, e^{-i \omega t} \psi(t).
\end{align}
In this setting, we can identify the optimal measurement strategy for any window function.
\begin{stheorem}[PGM is minimax optimal]\label{sthm:pgm_is_minimax_optimal}
For a state set $\rho(t)$ given by a pure initial state $\rho_0 = |\psi\rangle\!\langle \psi |$ evolving under a Hamiltonian with integer eigenvalue differences for time $t \in [0, 2\pi]$, the 
\emph{pretty good measurement} (PGM)
\begin{align}
    Q_{\mathrm{PGM}}(t) &:= R^{-1/2} \psi(t) R^{-1/2} \text{ where } R = \int \diff t \, \psi(t),
\end{align}
achieves the optimal minimax success probability for any window function $w(\tau)$. The minimax success probability is given by
\begin{align}\label{eqn:covariant_minimax_success_probability}
    \overline{\eta}^*(w, \psi) &= \sum_{\lambda, \lambda'} |\psi_{\lambda}| |\psi_{\lambda'}| \hatw_{\lambda - \lambda'},
\end{align}
where $\hat{w}_{\omega}$ is the Fourier transform of $w(\tau)$ at frequency $\omega$.
\end{stheorem}

\begin{proof}
The pretty good measurement is by definition a valid POVM and thus a feasible point for the primal problem given in Proposition~\ref{prop:minimax_sdp}. The associated success probability is now independent of $t$, and we can thus let go of the minimum, to get
\begin{align}
    \overline{\eta}_{\mathrm{PGM}} &= \min_t \Tr[ \psi(t) (w * Q_{\mathrm{PGM}})(t) ]\\
     \nonumber
    &= \min_t \int \diff \tau \, \Tr[ \psi(t) w(\tau) Q(t-\tau) ] \\
     \nonumber
    &= \int \diff \tau \, \Tr[ \psi(0) w(\tau) Q_{\mathrm{PGM}}(-\tau) ] \\
     \nonumber
    &= \int \diff \tau \, \Tr[ \psi(0) w(\tau) Q_{\mathrm{PGM}}(\tau) ] \\
     \nonumber
    &= \int \diff \tau \, w(\tau) f(\tau),
\end{align}
where we have used the symmetry of $\tau$ and the fact that we integrate over a cyclic interval. Now, if we look closely at the definition of the pretty good measurement and compare to Eq.~\eqref{eqn:hat_psi_omega_def}, we see that
\begin{align}
    R = 2\pi \hat{\psi}_0 =  2\pi \sum_{\lambda} |\psi_{\lambda}|^2 |{\psi_{\lambda}}\rangle\!\langle {\psi_{\lambda}} |,
\end{align}
which in turn implies that
\begin{align}
    R^{-1/2} = \frac{1}{\sqrt{2\pi}} \sum_{\lambda} \frac{1}{|\psi_{\lambda}|} |{\psi_{\lambda}}\rangle\!\langle {\psi_{\lambda}} |.
\end{align}
We can now use this to complete the POVM elements
\begin{align}
    Q(t) &= \frac{1}{2\pi} \left(\sum_{\lambda} \frac{1}{|\psi_{\lambda}|} |{\psi_{\lambda}}\rangle\!\langle {\psi_{\lambda}} |\right)\left( \sum_{\omega} e^{-i \omega t} \sum_{\lambda} \psi_{\lambda}\psi^{*}_{\lambda-\omega} |{\psi_{\lambda}}\rangle\!\langle {\psi_{\lambda-\omega}} |\right)
    \left(\sum_{\lambda} \frac{1}{|\psi_{\lambda}|} |{\psi_{\lambda}}\rangle\!\langle {\psi_{\lambda}} |\right) \\
 \nonumber
 &= \frac{1}{2\pi} \sum_{\omega} e^{-i \omega t} \sum_{\lambda} \frac{\psi_{\lambda}\psi^{*}_{\lambda-\omega}}{|\psi_{\lambda}||\psi_{\lambda-\omega}|}|{\psi_{\lambda}}\rangle\!\langle {\psi_{\lambda-\omega}} |
\end{align}
and evaluate $\overline{\eta}_{\mathrm{PGM}}$ via
\begin{align}
    f(\tau) &= \Tr [ \psi(0) Q_{\mathrm{PGM}}(\tau) ] \\
     \nonumber
    &= \frac{1}{2\pi} \sum_{\omega} e^{-i \omega \tau} \Tr \left\{  \left(\sum_{\lambda,\lambda'} \psi_{\lambda}\psi^{*}_{\lambda'}|{\psi_{\lambda}}\rangle\!\langle {\psi_{\lambda'}} |  \right) \left(\sum_{\lambda} \frac{\psi_{\lambda}\psi^{*}_{\lambda-\omega}}{|\psi_{\lambda}||\psi_{\lambda-\omega}|}|{\psi_{\lambda}}\rangle\!\langle {\psi_{\lambda-\omega}} |\right) \right\} \\
     \nonumber
    &=  \sum_{\omega} e^{-i \omega \tau} \sum_{\lambda} \frac{|\psi_{\lambda}||\psi_{\lambda-\omega}|}{2\pi} \\
    &= \sum_{\omega} e^{-i \omega \tau}\hat{f}_{\omega}.
     \nonumber
\end{align}
With this, we can Parseval's theorem to determine
\begin{align}\label{eqn:minimax_success_probability_formula}
    \overline{\eta}_{\mathrm{PGM}} &= \int \diff \tau \, w(\tau) f(\tau) \\
     \nonumber
    &= 2\pi \sum_{\omega} \hat{w}_{\omega} \hat{f}_{\omega} \\
     \nonumber
    &= \sum_{\omega}\hatw_{\omega} \sum_{\lambda} {|\psi_{\lambda}||\psi_{\lambda-\omega}|} \\
    &= \sum_{\lambda, \lambda'} |\psi_{\lambda}| |\psi_{\lambda'}| \hatw_{\lambda - \lambda'}
     \nonumber
\end{align}
as claimed in the theorem statement. Note that we can also see $f(\tau) = |p(\tau)|^2$ where
\begin{align}
    p(\tau) = \frac{1}{\sqrt{2\pi}} \sum_{\lambda} e^{-i \lambda \tau} |\psi_{\lambda}|. 
\end{align}
To show the optimality of the PGM, we
use complementary slackness and give a feasible point of the dual. First, we use $\mu(t) = 1/2\pi$ as the dual prior. The only non-trivial complementary slackness condition then is
\begin{align}
    X Q(t) = \frac{1}{2\pi} (w * \psi)(t) Q(t).
\end{align}
We first expand the right hand side in the frequency picture, using the duality of multiplication and convolution under the Fourier transform,
\begin{align}
    (w * \psi)(t) Q(t) = \sum_{\omega}e^{-i\omega t} \sum_{\delta}\hatw_{\delta} \hat{\psi}_{\delta} \hat{Q}_{\omega - \delta}.
\end{align}
The complementary slackness condition can therefore be written as
\begin{align}
    \sum_{\omega} e^{-i \omega t} \left(X\hat{Q}_{\omega} - \sum_{\delta}\hatw_{\delta} \hat{\psi}_{\delta} \hat{Q}_{\omega - \delta} \right) = 0.
\end{align}
As we want the above to hold for all $t$, we expect this to require to hold independently for all frequencies, especially $\omega = 0$. This gives
\begin{align}
    X\hat{Q}_{0} = \sum_{\delta}\hatw_{\delta} \hat{\psi}_{\delta} \hat{Q}_{- \delta}.
\end{align}
But, as we know that $Q$ is a valid POVM, we have that $\hat{Q}_0 = \bbI/2\pi$, and hence we get a formula for $X$ given by
\begin{align}
    X &= 2\pi \sum_{\delta}\hatw_{\delta} \hat{\psi}_{\delta} \hat{Q}_{- \delta} \\
     \nonumber
    &= \sum_{\delta}\hatw_{\delta} \hat{\psi}_{\delta} \hat{\psi}_0^{-\frac{1}{2}}\hat{\psi}_{- \delta}\hat{\psi}_0^{-\frac{1}{2}} \\
     \nonumber
    &= \sum_{\delta} \hatw_{\delta} 
    \left(\sum_{\lambda} \psi_{\lambda}\psi^{*}_{\lambda-\delta} |{\psi_{\lambda}}\rangle\!\langle {\psi_{\lambda-\delta}} | \right)
    \left(\sum_{\lambda} \frac{1}{|\psi_{\lambda}|} |{\psi_{\lambda}}\rangle\!\langle {\psi_{\lambda}} | \right)
    \left(\sum_{\lambda} \psi_{\lambda-\delta}\psi^{*}_{\lambda} |{\psi_{\lambda-\delta}}\rangle\!\langle {\psi_{\lambda}} | \right)
    \left(\sum_{\lambda} \frac{1}{|\psi_{\lambda}|} |{\psi_{\lambda}}\rangle\!\langle {\psi_{\lambda}} | \right)
    \\
     \nonumber
    &= \sum_{\delta} \hatw_{\delta} \sum_{\lambda} {|\psi_{\lambda}||\psi_{\lambda-\delta}|} |\psi_{\lambda}\rangle\!\langle \psi_{\lambda}|.
\end{align}
We can immediately see that the dual value is correct, as
\begin{align}
    \Tr [X] 
    &= \sum_{\delta} \hatw_{\delta}  \sum_{\lambda} {|\psi_{\lambda}||\psi_{\lambda-\delta}|}.
\end{align}
It is left to be shown that $X$ and $\mu(t) = 1/2\pi$ constitute a feasible point of the dual, \textit{i.e.}, that
\begin{align}
    X \geq \frac{1}{2\pi} (w * \psi)(t).
\end{align}
Note that by construction, $X$ commutes with the Hamiltonian and is therefore invariant under time evolution, which means it is sufficient to check the above condition at $t=0$, where it evaluates to
\begin{align}
    X \geq \frac{1}{2\pi} \int \diff \tau \, w(\tau) \psi(\tau) = \sum_{\omega} \hatw_{\omega} \hat{\psi}_{\omega},
\end{align}
where we again used Parseval's theorem. We can then collect the hermitian conjugate terms to obtain
\begin{align}
    \sum_{\omega \geq 0} \hatw_{\omega} \sum_{\lambda}(|\psi_{\lambda}||\psi_{\lambda-\omega}| + |\psi_{\lambda}||\psi_{\lambda+\omega}|) |\psi_{\lambda}\rangle\!\langle \psi_{\lambda}|
    \geq \sum_{\omega \geq 0}\hatw_{\omega} \sum_{\lambda} \psi_{\lambda}\psi^{*}_{\lambda-\omega} |{\psi_{\lambda}}\rangle\!\langle {\psi_{\lambda-\omega}} | + \psi_{\lambda}\psi^{*}_{\lambda+\omega} |{\psi_{\lambda}}\rangle\!\langle {\psi_{\lambda+\omega}} | 
    ,
\end{align}
where we have used that $\hatw_{\omega} = \hatw_{-\omega}$ because it is a symmetric function. As $w$ is still arbitrary, this should hold for all $\omega$ independently. We can exploit that the sum runs over all $\lambda$, which allows us to replace $(\lambda, \lambda + \omega) \to (\lambda - \omega, \lambda)$ and obtain
\begin{align}
    \sum_{\lambda}  |\psi_{\lambda}| |\psi_{\lambda-\omega}|  |\psi_{\lambda} \rangle \!\langle \psi_{\lambda} | + |\psi_{\lambda-\omega}| |\psi_{\lambda}|  |\psi_{\lambda-\omega} \rangle \!\langle \psi_{\lambda-\omega} |\geq \\
     \nonumber
    \sum_{\lambda} \psi_{\lambda}\psi_{\lambda-\omega}^{*} | \psi_{\lambda} \rangle\!\langle \psi_{\lambda-\omega} | +\psi_{\lambda-\omega}\psi_{\lambda}^{*} | \psi_{\lambda-\omega} \rangle\!\langle \psi_{\lambda} |.
\end{align}
We can write this as a sum over matrix-inequalities on the subspaces spanned by $\ket{\psi_{\lambda}}$ and $\ket{\psi_{\lambda - \omega}}$ to arrive at the final
\begin{align}
    \sum_{\lambda} \begin{bmatrix}
    |\psi_{\lambda}||\psi_{\lambda-\omega}| & \psi_{\lambda}^{*} \psi_{\lambda-\omega} \\
    \psi_{\lambda-\omega}^{*} \psi_{\lambda} & |\psi_{\lambda}| |\psi_{\lambda-\omega}| 
    \end{bmatrix}_{\lambda, \lambda-\omega} \geq 0.
\end{align}
Denoting now for brevity $a = \psi_{\lambda}^{*} \psi_{\lambda-\omega}$ we can chat the inequality by computing the characteristic polynomial
\begin{align}
    \operatorname{det} \begin{bmatrix} |a|-\alpha & a \\ a^{*} & |a|-\alpha \end{bmatrix} = \alpha^2 - 2 |a| \lambda = \alpha ( \alpha - 2 |a|)
\end{align}
which has the two non-negative roots $0$ and $2 |a|$, which, together with the fact that a sum of positive semi-definite matrices is positive-semi-definite, concludes the proof that $X$ is indeed a feasible point of the dual and hence the pretty good measurement is optimal.
\end{proof}

We can exploit the optimality of the pretty good measurement together with the formula for the minimax success probability to optimize the probe state. If we define the vector $\ppsi = (|\psi_{\lambda_1}|, |\psi_{\lambda_2}|, \dots, |\psi_{\lambda_d}|)$ and the matrix $W_{\lambda, \lambda'} = \hatw_{\lambda - \lambda'}$, we can write the minimax success probability of Eq.~\eqref{eqn:covariant_minimax_success_probability} as a quadratic form $\overline{\eta} = \langle \ppsi, W \ppsi \rangle$. 
This implies that the optimal probe state is the solution of a quadratic program over the positive orthant
\begin{align}
    \overline{\eta}^*(w, U(\phi))= \begin{array}[t]{rc}
      \text{maximize} & \langle \ppsi, W \ppsi \rangle  \\ [1ex] 
      \text{over} & \psi_i \geq 0 \\ [1ex] 
      \text{such that} & \sum_{i=1}^d \psi_i^2 = 1.
    \end{array}
\end{align}
Programs of this form are in general NP-hard to solve
in worst case complexity, but as the optimization is over $n \times n$ matrices where $n$ is the number of channel repetitions, we can still use numerical methods to solve this problem for large numbers of qubits, as direct simulation is not required.
Note that the normalization condition for $\ppsi$ implies that
\begin{align}
    \overline{\eta}^*(w, U(\phi)) \leq \lVert W \rVert_{\infty}.
\end{align}
By the Perron-Frobenius theorem, this bound is tight if the matrix $W$ has only positive entries, which is only guaranteed if the Fourier transform of the window function is a non-negative function, which will most of the time not be the case.

Every candidate for the optimal probe state will allow us to determine a lower bound on $\overline{\eta}^*(w, U(\phi))$. One such candidate is a uniform superposition of energy eigenstates $\ppsi = (1/\sqrt{d},1/\sqrt{d}, \dots, 1/\sqrt{d})$. This will give a near optimal success probability in the limit where $\frac{1}{n}$ is large compared to the spectral width of the window. This can be seen by noting that in this limit the first order expansion of the Fourier transform is approximately constant
\begin{align}
    \hatw_{\lambda} \approx \hatw_0
\end{align}
In this case, the success probability is dominated by the one-norm of the vector $\lVert \ppsi \rVert_1$:
\begin{align}
    \overline{\eta}^{*}(w, \psi) &\approx \sum_{\lambda, \lambda'} |\psi_{\lambda}||\psi_{\lambda'}| \hatw_0 = \hatw_0 \lVert \ppsi\rVert_1^2.
\end{align}
This is only meant as a intuitive argument. 

Let us make this more concise in the case of a $\delta$ window in the limit $\delta \to 0$. We can show the following proposition:
\begin{sproposition}
In the phase sensing example, for a rectangular window $w_{\delta}$ and a probe state $\psi(0)$, we have that
\begin{align}
    \overline{\eta}^*(w, U(\phi)) \leq \frac{\delta}{\pi}\left[ \Tr[ \hat{\psi}_0^{\frac{1}{2}} ]^2 -  \frac{\delta^2}{2}\frac{2}{3} \operatorname{Var}_{\psi(0)}(H) + O(\delta^4) \right].
\end{align}
\end{sproposition}
\begin{proof}
We use the second order expansion
\begin{align}
    \hatw_{\omega} &= \frac{\sin \omega \delta}{\pi \omega} = \frac{\delta}{\pi} \left[1 - \frac{1}{2}\frac{\delta^2}{3} \omega^2 + O(\delta^4) \right].
\end{align}
If we use this together with the formula for the success probability, where $\omega = \lambda - \lambda'$, we obtain 
\begin{align}
    \overline{\eta}^{*}(w, \psi) &= \sum_{\lambda, \lambda'} |\psi_{\lambda}| |\psi_{\lambda'}| \hatw_{\lambda - \lambda'} \\
    \nonumber
    &= \frac{\delta}{\pi} \sum_{\lambda, \lambda'} |\psi_{\lambda}| |\psi_{\lambda'}| \left[1 - \frac{1}{2}\frac{\delta^2}{3} \omega^2 + O(\delta^4) \right] \\
    \nonumber
    &= \frac{\delta}{\pi} \left[ \lVert \ppsi \rVert_1^2 -  \frac{1}{2}\frac{\delta^2}{3} \sum_{\lambda, \lambda'} |\psi_{\lambda}| |\psi_{\lambda'}| (\lambda - \lambda')^2 + O(\delta^4)\right].
    \nonumber
\end{align}
The quadratic term can be recast into
\begin{align}
    \sum_{\lambda, \lambda'} |\psi_{\lambda}| |\psi_{\lambda'}| (\lambda - \lambda')^2  &=
    \sum_{\lambda, \lambda'} |\psi_{\lambda}| |\psi_{\lambda'}| (\lambda^2 + \lambda'^2 - 2 \lambda \lambda')\\
    \nonumber
    &= 2 \lVert \ppsi \rVert_1 \sum_{\lambda} |\psi_{\lambda}| \lambda^2 -2  \left(\sum_{\lambda} |\psi_{\lambda}| \lambda\right)^2 \\
    \nonumber
    &= 2 \left( \Tr [\hat{\psi}_0^{\frac{1}{2}}] \Tr[ H^2 \hat{\psi}_0^{\frac{1}{2}} ] - \Tr[ H \hat{\psi}_0^{\frac{1}{2}} ]^2 \right).
    \nonumber
\end{align}
Now, we show that the map
\begin{align}
    V( X ) = \Tr [ X ] \Tr [ H^2 X ] - \Tr [ H X ]^2
\end{align}
is an operator monotone. We can rewrite it as
\begin{align}
    V(X) = \frac{1}{2}\Tr [ (H^2 \otimes \bbI + \bbI \otimes H^2 - 2 H\otimes H) (X \otimes X) ]
\end{align}
and can explicitly check the positivity of the operator by checking its action on an products of eigenstates of $H \ket{\lambda} = \lambda \ket{\lambda}$ given by
\begin{align}
    (H^2 \otimes \bbI + \bbI \otimes H^2 - 2 H\otimes H) \ket{\lambda \lambda'}
    &= (\lambda^2 + \lambda'^2 - 2 \lambda \lambda')\ket{\lambda \lambda'} \\
    &= (\lambda - \lambda')^2 \ket{\lambda \lambda'}. 
    \nonumber
\end{align}
We can then also use the facts that
\begin{align}
    A \geq B\geq 0 \ &\Rightarrow \ A \otimes A \geq B \otimes B, \\
    A \geq B \ &\Rightarrow \ \Tr [ M A ] \geq  \Tr [ M B ]  \text{ for all } M \geq 0,
\end{align}
to conclude that
\begin{align}
    A \geq B \ &\Rightarrow \ V(A) \geq V(B).
\end{align}
Considering  
\begin{align}
    \bbI \geq \psi(0) \geq 0 \ \Rightarrow \ \hat{\psi}_0^{\frac{1}{2}} \geq \psi(0)
\end{align}
together with the operator monotonicity of $V(X)$ allows us to conclude that
\begin{align}
    V(\hat{\psi}_0^{\frac{1}{2}}) \geq V(\psi(0)) = \operatorname{Var} (H).
\end{align}
Putting this back into the success probability we obtain
\begin{align}
    \overline{\eta}^{*}(w, \psi) &= \frac{\delta}{\pi} \left[ \lVert \ppsi \rVert_1^2 -  \frac{\delta^2}{2}\frac{2}{3} V(\hat{\psi}_0^{\frac{1}{2}}) + O(\delta^4) \right] \\
    &\leq \frac{\delta}{\pi} \left[ \Tr[ \hat{\psi}_0^{\frac{1}{2}} ]^2 -  \frac{\delta^2}{2}\frac{2}{3} \operatorname{Var}(H) + O(\delta^4) \right].\nonumber
\end{align}
\end{proof}
The above proposition implies that as long $\delta n \ll 1$, the uniform superposition probe state will work well.

\section{Minimax analysis of phase estimation}\label{ssec:minimax_analysis_of_phase_estimation}

We now treat the practically important case of phase estimation. We consider without loss of generality a phase $t \in [-\pi, \pi]$ to be encoded via $U(t) = e^{-i t H}$ for the Hamiltonian $H = \operatorname{diag}(0, 1)$. If we perform $n$ repetitions of the experiment, we encode with
\begin{align}
    U^{\otimes n}(t) = e^{-i t \sum_{i=1}^n H_i} \qquad H_i = \bbI^{\otimes i-1} \otimes H \otimes \bbI^{\otimes n-i}.
\end{align}
The effective accessible spectrum is then
\begin{align}
    \Lambda_n = \{ 0, 1, \dots, n\}.
\end{align}
In the following discussions, we always use the rectangular window function
\begin{align}
    w_{\delta}(t) = \begin{cases} 1 & \text{ if } |t| \leq \delta \\ 0 & \text{ else}.\end{cases}
\end{align}

We first give a proof of Theorem~\ref{thm:opt_minimax_rate_entangled_phase_est} of the main text, which establishes the asymptotic rate of the optimal probe for phase estimation on a spin chain.

\begin{proof}[Proof of Theorem~\ref{thm:opt_minimax_rate_entangled_phase_est}]
 Slepian gives an asymptotic rate with which the largest eigenvalue of $W$ approaches $1$ as
 \begin{align}
 \overline{R}^{*}_{\mathrm{par}} &= \log \left(1 + \frac{2 \sqrt{1 - \cos \delta}}{\sqrt{2} - \sqrt{1 - \cos \delta}}\right) \\
 \nonumber
 &= \log \left( \frac{\sqrt{2} + \sqrt{1 - \cos \delta}}{\sqrt{2} - \sqrt{1 - \cos \delta}}\right)\\
 &= \log \left( \frac{1 + \sin \frac{\delta}{2}}{1 - \sin \frac{\delta}{2}} \right).
 \nonumber
\end{align}
A quick Taylor expansion yields the theorem statement. We do, however, still need to establish positivity of the associated eigenvector, the DPSS of zeroth order. To do so, we rely on another result of Slepian, namely that the DPSS of zeroth order also corresponds to the largest eigenvector of the tridiagonal 
matrix $\tilde{W}$ with entries
\begin{align}
    \tilde{W}_{\lambda, \lambda'} = \begin{cases}
        \frac{\lambda (n + 1 - \lambda)}{2} & \text{for } \lambda' = \lambda - 1 \\
        \left(\frac{n}{2}-\lambda\right)^2 \cos \delta & \text{for } \lambda' = \lambda \\
        \frac{(\lambda+1) (n - \lambda)}{2} & \text{for } \lambda' = \lambda + 1 \\
        0 & \text{else.}
    \end{cases}
\end{align}
As $0 \leq \lambda \leq n$, the entries of this matrix are positive as long as $\delta \leq \pi/2$ as required by the Theorem as well as being the only non-trivial parameter range. Because of the positive entries, we can apply the Perron-Frobenius theorem to conclude that the eigenvector associated to the largest eigenvalue -- the DPSS of zeroth order -- must be non-negative.
\end{proof}

Next, we give a proof of Theorem~\ref{thm:opt_minimax_rate_gaussian_probe_phase_est} of the main text that establishes the minimax rate for a Gaussian probe. In the course of the proof, we will optimize the width of the Gaussian to have an optimal trade off between $\delta$ and $n$. Note that in the below proof, $\sigma$ refers to the standard deviation of the distribution of estimates, and is thus the \emph{inverse} of the standard deviation of the probe state itself.
\begin{proof}[Proof of Theorem~\ref{thm:opt_minimax_rate_gaussian_probe_phase_est}]
Before we start the proof, we recall the Fourier transform
\begin{align}
    f(t) &= \sum_{\omega = -\infty}^{\infty} e^{-i \omega t} \hat{f}(\omega),\\
    \hat{f}(\omega) &= \frac{1}{2\pi} \int_{-\pi}^{\pi} \diff t \, f(t) e^{i \omega t}.
\end{align}
We prove the statement by constructing a probe state that achieves the given asymptotic rate. Our target will be the construction of a probe state whose associated PGM fidelity function is a wrapped normal distribution 
\begin{align}\label{eqn:target_function_gaussian}
    f_{\sigma}(t) = \frac{1}{\sqrt{2\pi} \sigma}\sum_{k=-\infty}^{\infty} \exp\left({-\frac{1}{2}\left(\frac{t + 2\pi k}{\sigma}\right)^2}\right)
\end{align}
with standard deviation $\sigma$. In the following, we use $\Tilde{f}_{\sigma}(t)$ to denote an \emph{unwrapped} normal distribution.
Our aim is to let $\sigma$ vanish asymptotically in the number of repetitions $n$. Because the available frequencies of the Hamiltonian are $\Lambda_n = \{ 0 , 1  , \dots, n\}$, the Fourier transform of the actual PGM fidelity function must be supported on the frequencies frequencies $\{-n, \dots, n\}$, which are all possible differences of frequencies in $\Lambda_n$. 
Our strategy will be to assume that we choose Fourier coefficients of the PGM fidelity function such that they are equal to our target in Eq.~\eqref{eqn:target_function_gaussian} on the frequencies $\{-n, \dots, n\}$ and zero outside. With this approximation, we will make an error, which we will see below. Then, we need to judiciously choose $\sigma$ so that we keep both the approximation error in check as well as making $\sigma$ as small as possible to fit most of the Gaussian into the window.

Let us first analyze the error we make when we approximate $f_{\sigma}(t)$ on frequencies from $-n$ to $n$. We denote this approximation to $f_{\sigma}$ as $f_{\sigma}^{\leq n}$ so that $f_{\sigma}(t) = f_{\sigma}^{\leq n} + f_{\sigma}^{> n}$. Note that $f_{\sigma}^{\leq n}$ is still a properly normalized function as only $\hat{f}(0)$ conributes to the integral. 
The expected error probability for this probe is then
\begin{align}
    1- \overline{\eta}_{\sigma} &= 1- \int_{-\pi}^{\pi} \diff t \, w(t) f_{\sigma}^{\leq n}(t) \\
     \nonumber
    &=1 - \int_{-\pi}^{\pi} \diff t \, w(t) f_{\sigma}(t) + \int_{-\pi}^{\pi} \diff t \, w(t) f^{>n}_{\sigma}(t) \\
     \nonumber
     \nonumber
    &=\int_{-\pi}^{\pi} \diff t \, f_{\sigma}(t) - \int_{-\delta}^{\delta} \diff t \, f_{\sigma}(t) + \int_{-\pi}^{\pi} \diff t \, w(t) f^{>n}_{\sigma}(t) \\
     \nonumber
    &=  2 \int_{\delta}^{\pi} \diff t \, f_{\sigma}(t) + \int_{-\pi}^{\pi} \diff t \, w(t) f^{>n}_{\sigma}(t) \\
     \nonumber
    &=  2 \int_{\delta}^{\pi} \diff t \, f_{\sigma}(t) + 2\pi \sum_{\omega = -\infty}^{\infty}\hatw(\omega) \hatf_{\sigma}^{>n}(\omega) 
    \\
     \nonumber
    &=  2 \int_{\delta}^{\pi} \diff t \, f_{\sigma}(t) + 4 \pi \sum_{\omega = n+1}^{\infty}\hatw(\omega) \hatf_{\sigma}(\omega),
\end{align}
where we have used the definition of $w_{\delta}$ and Parseval's theorem, which for real-valued functions and our Fourier transform conventions reads 
\begin{align}
    \int_{-\pi}^{\pi} \diff t \, f(t)g(t) = 2\pi \sum_{\omega = -\infty}^{\infty} \hat{f}(\omega)\hat{g}(\omega).
\end{align}
We just expanded the error in terms of two tails of $f_{\sigma}$, one in real space and one in frequency space. The next step is to choose $\sigma$ judiciously to balance the two tails to achieve a minimal error. The first tail is easy to treat,
\begin{align}
    2 \int_{\delta}^{\pi} \diff t \, f_{\sigma}(t) 
    &= 2 \frac{1}{\sqrt{2\pi}\sigma} \sum_{k=-\infty}^{\infty} \int_{\delta}^{\pi} \exp\left( -\frac{1}{2}\left(\frac{t + 2\pi k}{\sigma}\right)^2\right)\\
     \nonumber
    &=2 \frac{1}{\sqrt{2\pi}\sigma} \sum_{k=-\infty}^{\infty} \int_{\delta-2\pi k}^{\pi - 2\pi k} \exp\left( -\frac{1}{2}\left(\frac{t}{\sigma}\right)^2\right) \\
     \nonumber
    &\leq 2 \frac{1}{\sqrt{2\pi}\sigma} \int_{\delta}^{\infty} \exp\left( -\frac{1}{2}\left(\frac{t}{\sigma}\right)^2\right) \\
     \nonumber
    &\leq 2 \int_{\delta}^{\infty} \diff t \, \tilde{f}_{\sigma}(t) \\
     \nonumber
    &= \operatorname{Erfc}\left(\frac{\delta}{\sqrt{2} \sigma}\right) \\
     \nonumber
    &\leq \exp\left(-\frac{1}{2}\left(\frac{\delta}{\sigma}\right)^2\right) \sqrt{\frac{2}{\pi}}\frac{\sigma}{\delta},
\end{align}
where we have used the definition of the 
cumulative error function $x\mapsto \operatorname{Erfc}(x)$ and the standard tail bound. 
Note that the first inequality follows from the positivity of the normal distribution.

The second term will be treated in a comparable manner, however we have to be careful when relating the discrete summation to an integral. We will need the Fourier transforms of the involved functions, given by
\begin{align}
    \hatw_{\delta}(\omega) &= \frac{\sin \delta \omega}{\pi \omega} \leq \frac{1}{\pi \omega} ,\\
    \hatf_{\sigma}(\omega) &= \frac{1}{2\pi} \exp\left(-\frac{1}{2}(\sigma \omega)^2\right).
\end{align}
We will apply the bound
\begin{align}
    4\pi \hatw_{\delta}(\omega) \hatf_{\sigma}(\omega) \leq \frac{2}{\pi \omega}\exp\left(-\frac{1}{2}(\sigma \omega)^2\right) \leq \frac{2}{\pi}\exp\left(-\frac{1}{2}(\sigma \omega)^2\right),
\end{align}
where the last inequality is valid for $\omega \geq 1$. In this approximation, we can use the integral approximation formula for positive and monotonically decreasing functions
\begin{align}
    4\pi \sum_{\omega = n+1}^{\infty}\hatw(\omega) \hatf_{\sigma}(\omega) &\leq \sum_{\omega = n+1}^{\infty} \frac{2}{\pi}\exp\left(-\frac{1}{2}
    (\sigma \omega)^2\right) \\
     \nonumber
    &\leq \frac{2}{\pi }\exp\left(-\frac{1}{2}(\sigma (n+1) )^2\right)  + \frac{2}{\pi}\int_{n+1}^{\infty} \diff \omega \, \exp\left(-\frac{1}{2}(\sigma \omega)^2\right)\\
     \nonumber
    &\leq \frac{2}{\pi}\exp\left(-\frac{1}{2}(\sigma (n+1) )^2\right)  + \frac{1}{\pi}\operatorname{Erfc}\left(\frac{(n+1) \sigma}{\sqrt{2}}\right)\\
     \nonumber
    &\leq \frac{2}{\pi}\exp\left(-\frac{1}{2}(\sigma (n+1) )^2\right)  +\frac{1}{\pi} \exp\left(-\frac{1}{2}(\sigma(n+1))^2\right) \sqrt{\frac{2}{\pi}}\frac{1}{\sigma(n+1)}\\
     \nonumber
    &= \frac{2}{\pi}\exp\left(-\frac{1}{2}(\sigma (n+1) )^2\right)\left[1  + \frac{1}{\sqrt{2\pi} \sigma}\right].
     \nonumber
\end{align}
Putting both upper bounds together, we see that
\begin{align}
    1 - \overline{\eta}_{\sigma} &\leq  \exp\left(-\frac{1}{2}\left(\frac{\delta}{\sigma}\right)^2\right) \sqrt{\frac{2}{\pi}}\frac{\sigma}{\delta} + \frac{2}{\pi}\exp\left(-\frac{1}{2}(\sigma (n+1) )^2\right)\left[1  + \frac{1}{\sqrt{2\pi} \sigma}\right].
\end{align}
In practice, the error will be dominated by the larger exponent, so we can optimize it by making both exponents equal, \emph{i.e.}, choosing $\sigma$ such that
\begin{align}
    \frac{\delta}{\sigma} = \sigma (n+1) \ \Leftrightarrow \ \sigma = \sqrt{\frac{\delta}{n+1}},
\end{align}
in which case we have that
\begin{align}\label{eqn:error_rate_upper_bound_gauss}
    1 - \overline{\eta}_{\sigma} \leq \exp\left(-\frac{\delta(n+1)}{2} \right)\left\{ \sqrt{\frac{2}{\pi}}\sqrt{\frac{1}{\delta(n+1)}} + \frac{2}{\pi}\left[1  + \frac{1}{\sqrt{2\pi}}\sqrt{\frac{n+1}{\delta}}\right]\right\}.
\end{align}
The asymptotic rate of the above quantity is $\delta/2$ as claimed in the theorem statement.

To round off the proof of the theorem, we need to show that the desired fidelity function can actually be realized. Note that the PGM fidelity function fulfills
\begin{align}
f(t) = |p(t)|^2
\end{align}
where the Fourier transform of $t\mapsto p(t)$ has coefficients
\begin{align}
    \hat{p}(\omega) = \frac{1}{\sqrt{2\pi}} |\psi_{\omega}|.
\end{align}
The Fourier transform is such that
\begin{align}
    \hatf(\omega) = (\hat{p} * \hat{p}^{*})(\omega).
\end{align}
Here $\hat{p}^{*}(\omega) = \hat{p}(-\omega)$ is a time reversed version of $\hat{p}$ which coincides with the complex conjugate when $p(t)$ is real-valued as we will choose below.

As the weights $|\psi_{\omega}|$ need to correspond to a valid quantum state and $\omega \in \{0, 1, \dots, n\}$, we require that
\begin{align}
    \hat{p}(\omega) \geq 0 \text{ and } \sum_{\omega = 0}^n \hat{p}(\omega)^2 = \frac{1}{2\pi}.
\end{align}
We also note that shifting the support of the Fourier transform of $p(\tau)$ only introduces a phase. We will therefore consider the Fourier transform to be supported on the interval $\{-n/2, \dots, n/2\}$.

As we let $\sigma \to 0$, we can work with the unwrapped Gaussian instead of the wrapped Gaussian which will significantly simplify our calculations. Formally, this is 
\begin{align}
\frac{f_{\sigma}(t)}{\tilde{f}_{\sigma}(t)} \leq C
\end{align}
for some $C$ that can be chosen arbitrarily close to $1$ for appropriately large $n$. As we want to construct $f_{\sigma}(t)$, we can exploit the fact that
\begin{align}
    \tilde{f}_{\sigma}(t) = (\tilde{f}_{\sigma\sqrt{2}}(t))^2
\end{align}
and choose 
\begin{align}
p(t) = \tilde{f}^{\leq n/2}_{\sigma\sqrt{2}}(t).
\end{align}
By approximating this expression, we make an error that we need to correct for through re-normalization, the factor of which is given by
\begin{align}
    \frac{\calN}{2\pi} &= \frac{1}{4\pi^2}\sum_{\omega = -n/2}^{n/2} \exp(-\frac{1}{2}(\sigma\omega)^2) \\
     \nonumber
    &= \frac{1}{2\pi} - \frac{1}{2\pi^2} \sum_{\omega = n/2}^{\infty} \exp(-\frac{1}{2}(\sigma\omega)^2).
\end{align}
We see that the right term is again the tail of the Gaussian in frequency space which becomes arbitrarily small for $n \to \infty$. We therefore know that there exists an $n$ such that $1/2 \leq \calN \leq 1$. This means the actual probe we will use is given by
\begin{align}
    \hat{p}(\omega) &= \calN^{-1}\hatf^{\leq n}_{\sigma\sqrt{2}}(\omega) , \\
    \psi_{\omega} &= \sqrt{2\pi}\calN^{-1}\hatf^{\leq n}_{\sigma\sqrt{2}}(\omega).
\end{align}
The correction factor, as well as the error of approximating the wrapped Gaussian with the unwrapped Gaussian is bounded and will thus not contribute to the asymptotic rate. This concludes the proof.
\end{proof}

Next, we will discuss the constant factor of the Heisenberg scaling of the Gaussian probe and give a proof of Theorem~\ref{obs:asymtptotic_tolerance_gauss} of the main text.
\begin{proof}[Argument for Observation~\ref{obs:asymtptotic_tolerance_gauss}]
    We use the same setting as in the above proof of Theorem~\ref{thm:opt_minimax_rate_gaussian_probe_phase_est}.
    We start from Eq.~\eqref{eqn:error_rate_upper_bound_gauss} and set 
    \begin{align}
        \delta = \frac{\alpha}{n+1}
    \end{align}
    to obtain
    \begin{align}
        1 - \overline{\eta} &\leq \frac{2}{\pi} e^{-\frac{\alpha}{2}} \left( 1 + \sqrt{\frac{1}{\alpha}}\left[ \sqrt{\frac{\pi}{2}} + \sqrt{\frac{1}{2\pi} } (n+1)\right]\right) \\
        &= \frac{2}{\pi} \exp\left( - \frac{\alpha}{2} + O(\log \alpha) + O(\log n) \right).
    \end{align}
    Rearranging then yields
    \begin{align}
        2 \log \frac{2}{\pi(1-\overline{\eta})} \geq  \alpha + O(\log \alpha) + O(\log n)
    \end{align}
    which implies the Observation.
\end{proof}

\subsection{Chernoff bound and entanglement advantage}

In the previous section, we have shown that the optimal rate for the phase sensing problem is lower-bounded by $\delta/2$ in the case of an entangled probe. In this section, we also provide an upper bound on the best rate that can be achieved with tensor power inputs. 
We can exploit the upper bound on the asymptotic rate for the tensor power case derived in Theorem~\ref{sthm:rate_upper_bound} to establish the following theorem.

\begin{stheorem}\label{sthm:upper_bound_iid_rate_phase_estimation}
The optimal minimax rate for tensor power inputs for the phase sensing problem is upper-bounded by
\begin{align}
    \overline{R}_{\mathrm{iid}}^{*}(w_{\delta}, U(t)) \leq - \log \cos^2(\delta) = \delta^2 + O(\delta^4).
\end{align}
We can, therefore, guarantee a quadratic advantage over the i.i.d.\ case through the use of entanglement.
\end{stheorem}
\begin{proof}
    We know from Theorem~\ref{sthm:rate_upper_bound} that we need to bound the Chernoff coefficient, which for pure quantum states relates to the fidelity according to
    \begin{align}
        C(\psi(t), \psi(t+2\delta)) &= -\log \min_{0\leq s \leq 1} \Tr [\psi(t)^s \psi(t+2\delta)^{1-s}] \\
        &= - \log \Tr [\psi(t) \psi(t+2\delta)]\nonumber \\
        &= - \log |\langle \psi(t) | \psi(t+2\delta)\rangle|^2.\nonumber
    \end{align}
    Note that due to the covariance property it is sufficient to study the fidelity for $t=0$. As for our setting, we have used a single qubit as our building block we can expand
    \begin{align}
        |\langle \psi(0) | \psi(2\delta)\rangle|^2
        &= | |\psi_0|^2 + e^{-i 2\delta} |\psi_1|^2 |^2 \\
        &= |\psi_0|^4 + |\psi_1|^4 + 2  |\psi_0|^2|\psi_1|^2 \cos\delta\nonumber\\
        &= (|\psi_0|^2 + |\psi_1|^2)^2 - 2 ( 1 - \cos 2\delta) |\psi_0|^2|\psi_1|^2\nonumber \\
        &= 1 - 2 ( 1 - \cos 2\delta) |\psi_0|^2|\psi_1|^2.\nonumber
    \end{align}
    For the best protocol, this quantity has to be as small as possible, hence the optimal probe has $|\psi_0|^2 = |\psi_1|^2 = 1/2$. Therefore, for the optimal probe state 
    \begin{align}\label{eqn:single_qubit_fidelity}
        |\langle \psi(0) | \psi(2\delta)\rangle|^2 &=
        1 - \frac{1}{2} ( 1 - \cos 2\delta) =
        \cos^2(\delta).
    \end{align}
    We conclude by the Taylor expansion of $-\log \cos^2(\delta)$.
\end{proof}

\subsection{Additional numerics}\label{ssec:additional_numerics}
In this section, we present additional numerical results for the minimax analysis of phase estimation performed in Section~\ref{sec:minimax_analysis_phase_estimation} of the main text.

\begin{figure}
    \centering    \includegraphics{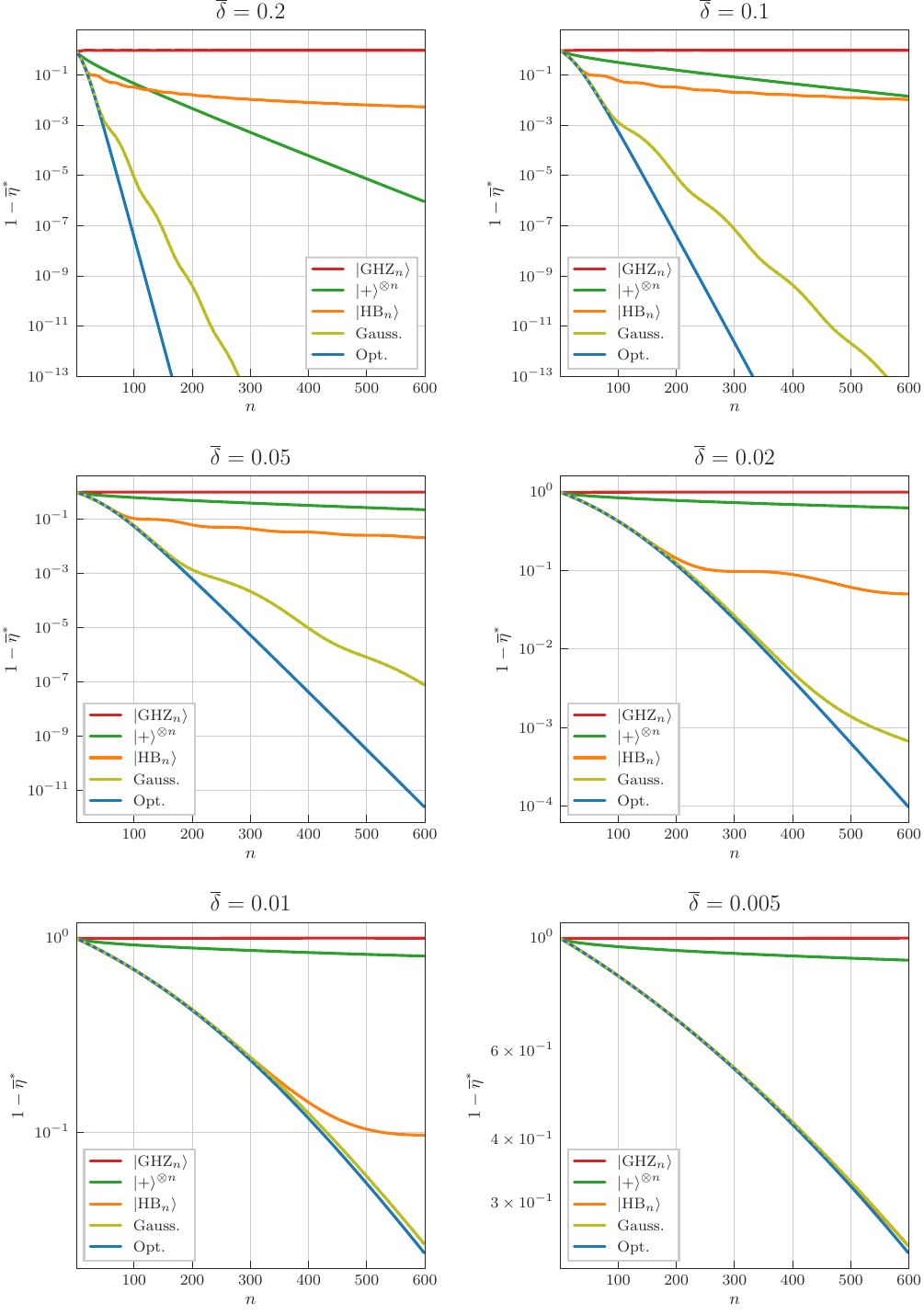}
    \caption{Optimal error probability of quantum metrology that can be guaranteed for any prior distribution for different probes in the phase estimation scenario for different values of the target tolerance $\overline{\delta}$. We compare a generalized GHZ state (Eq.~\eqref{eqn:def_ghz_n}, red), an tensor power of plus states (green), the Holland-Burnett state (Eq.~\eqref{eqn:def_HB_n}, orange), the Gaussian state (Eq.~\eqref{eqn:def_gauss_probe}, yellow) and the optimal state (Eq.~\eqref{eqn:opt_probe_state}, blue).}    \label{sfig:success_probability_phase_estimation_publication_overview.pdf}
\end{figure}

\begin{figure}
    \centering    \includegraphics{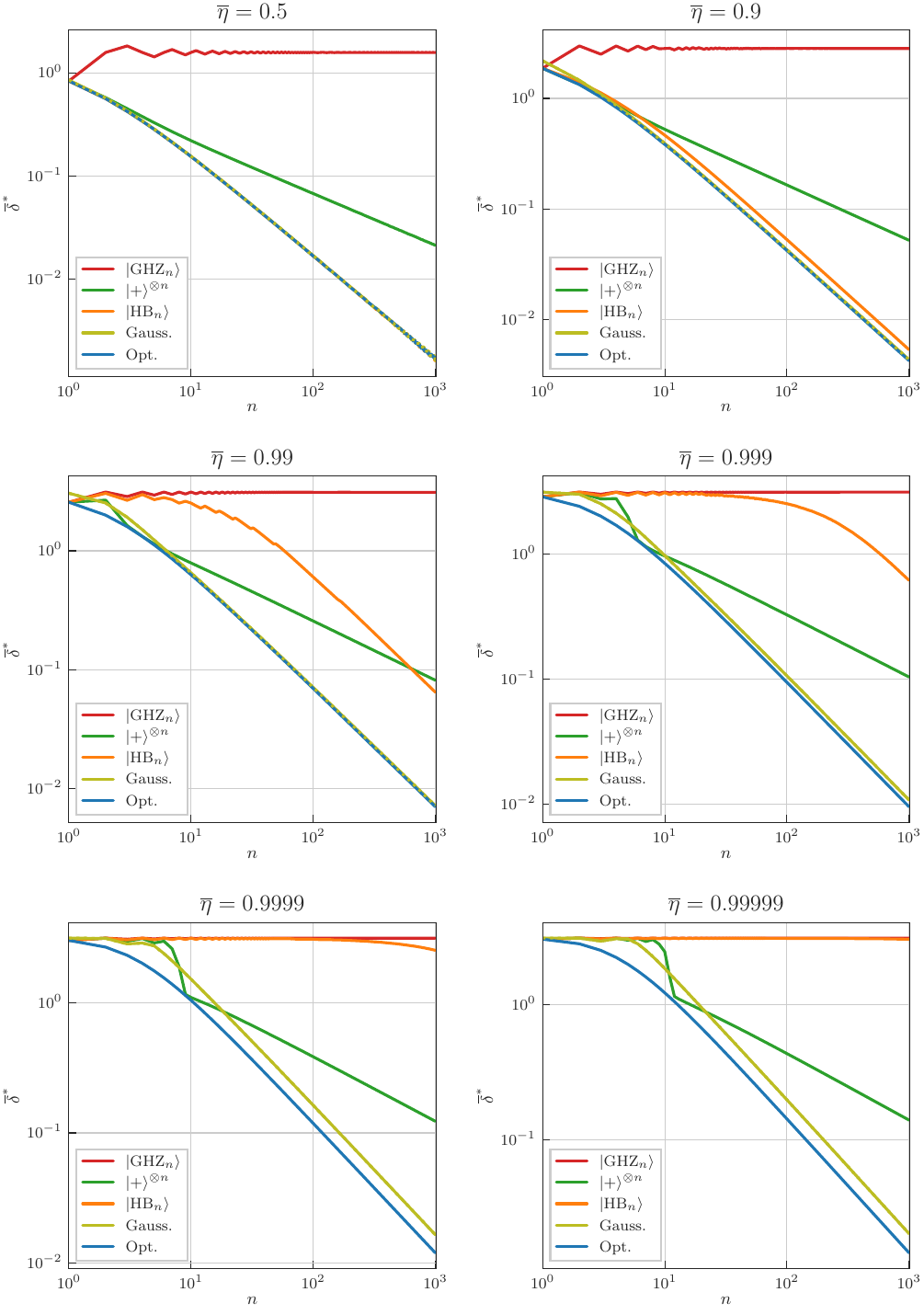}
    \caption{Optimal tolerance of quantum metrology that can be guaranteed for any prior distribution for different probes in the phase estimation scenario for different values of the target success probability $\overline\eta$. We compare a generalized GHZ state (Eq.~\eqref{eqn:def_ghz_n}, red), an tensor power of plus states (green), the Holland-Burnett state (Eq.~\eqref{eqn:def_HB_n}, orange), the Gaussian state (Eq.~\eqref{eqn:def_gauss_probe}, yellow) and the optimal state (Eq.~\eqref{eqn:opt_probe_state}, blue).}    \label{sfig:tolerance_phase_estimation_publication_loglog_overview}
\end{figure}

\end{document}